%% file: L2e16.170418.nocolor.tex
\documentclass[12pt,reqno]{amsart}
\usepackage{amsthm,amsfonts,amssymb,euscript,xcolor}

\newtheorem{proposition}{Proposition}
\newtheorem{theorem}{Theorem}
\newtheorem{remark}{Remark}

\newtheorem{corollary}[proposition]{Corollary}
\usepackage[margin=1in,dvips]{geometry}
\usepackage{bbm}
\usepackage{graphics}
\usepackage{enumerate}

\def\ub {\underline{u}}
\def\th {\theta}
\def\Lb {\underline{L}}
\def\Hb {\underline{H}}
\def\chib {\underline{\chi}}
\def\chih {\hat{\chi}}
\def\chibh {\hat{\underline{\chi}}}
\def\omegab {\underline{\omega}}
\def\etab {\underline{\eta}}
\def\betab {\underline{\beta}}
\def\alphab {\underline{\alpha}}

\def\hot{\widehat{\otimes}}

\def\rhoc{\check{\rho}}
\def\sigmac{\check{\sigma}}
\def\Psic{\check{\Psi}}
\def\kappab{\underline{\kappa}}
\def\mub{\underline{\mu}}

\def\a {\alpha}
\def\b {\beta}
\def\ab {\alphab}
\def\bb {\betab}
\def\nab {\nabla}
\def\f {\frac}

\renewcommand{\div}{\mbox{div }}
\newcommand{\curl}{\mbox{curl }}
\newcommand{\trchb}{\mbox{tr} \chib}
\def\trch{\mbox{tr}\chi}

\newcommand{\eps}{{\epsilon} \mkern-8mu /\,}
\newcommand{\Ls}{{\mathcal L} \mkern-10mu /\,}

\title[Nonlinear Interaction of Impulsive Gravitational Waves]{Nonlinear Interaction of Impulsive Gravitational Waves for the Vacuum Einstein Equations}

\author{Jonathan Luk}
\address{Department of Mathematics, Stanford University, CA 94305, USA}
\email{jluk@stanford.edu}

\author{Igor Rodnianski}
\address{Department of Mathematics, Princeton University, Princeton NJ 08544, USA}
\email{irod@math.princeton.edu}

\begin{document}

\begin{abstract}
In this paper, we study the problem of the nonlinear interaction of impulsive gravitational waves for the Einstein vacuum equations. The problem is studied in the context of a characteristic initial value problem with data given on two null hypersurfaces and containing curvature delta singularities. We establish an existence and uniqueness result for the spacetime arising from such data and show that the resulting spacetime represents the interaction of two impulsive gravitational waves germinating from the initial singularities. In the spacetime, the curvature delta singularities propagate along 3-dimensional null hypersurfaces intersecting to the future of the data. To the past of the intersection, the spacetime can be thought of as containing two independent, non-interacting impulsive gravitational waves and the intersection represents the first instance of their nonlinear interaction. Our analysis extends to the region past their first interaction and shows that the spacetime still remains smooth away from the continuing propagating individual waves. The construction of these spacetimes are motivated in part by the celebrated explicit solutions of Khan-Penrose and Szekeres. The approach of this paper can be applied to an even larger class of characteristic data and in particular implies an extension of the theorem on formation of trapped surfaces by Christodoulou and Klainerman-Rodnianski, allowing non-trivial data on the initial incoming hypersurface.
\end{abstract}

\maketitle

\tableofcontents

\section{Introduction}

\subsection{Impulsive Gravitational Waves}

In this paper, we study spacetime solutions $(\mathcal M, g)$ to the vacuum Einstein equations
\begin{equation}\label{Einstein}
R_{\mu\nu}=0
\end{equation}
representing a nonlinear interaction of two impulsive gravitational waves. Informally, an impulsive gravitational spacetime is a vacuum spacetime which contains a null hypersurface supporting a curvature delta singularity. Explicit solutions with such properties have been constructed by Penrose \cite{Penrose72}, and its origin can be traced back to the cylindrical waves of Einstein-Rosen \cite{EinsteinRosen} and the plane waves of Brinkmann \cite{Brinkmann}.

Impulsive gravitational waves have been first studied within the class of \emph{pp}-waves that was discovered by Brinkmann \cite{Brinkmann}, for which the metric takes the form 
$$g=-2d\ub dr+H(\ub,X,Y)d\ub^2+dX^2+dY^2,$$
and (\ref{Einstein}) implies that 
\begin{equation}\label{H}
\frac{\partial^2 H}{\partial X^2}+\frac{\partial^2 H}{\partial Y^2}=0.
\end{equation}
These include the special case of sandwich waves, where $H$ is compactly supported in $\ub$. Originally, impulsive gravitational waves have been thought of as a limiting case of the \emph{pp}-wave with the function $H$ admitting a delta singularity in the variable $\ub$. Precisely, explicit impulsive gravitational spacetimes were discovered and studied by Penrose \cite{Penrose72} who gave the metric in the following double null coordinate form:
\begin{equation}\label{Penrosesol}
g=-2dud\ub+(1-\ub\Theta(\ub))dx^2+(1+\ub\Theta(\ub))dy^2,
\end{equation}
where $\Theta$ is the Heaviside step function. In the Brinkmann coordinate system, the metric has the \emph{pp}-wave form and an obvious delta singularity:
\begin{equation}\label{Penrosesol2}
g=-2d\ub dr-\delta(\ub)(X^2-Y^2)d\ub^2+dX^2+dY^2,
\end{equation}
where $\delta(\ub)$ is the Dirac delta. Despite the presence of the delta singularity for the metric in the Brinkmann coordinate system, the corresponding spacetime is Lipschitz and it is only the Riemann curvature tensor (specifically, the only non-trivial $\alpha$ component\footnote{See \eqref{curv.def} for the definition of $\alpha$. In this specific example, these are the $R(\frac{\partial}{\partial\ub},\frac{\partial}{\partial X},\frac{\partial}{\partial\ub},\frac{\partial}{\partial X})$, $R(\frac{\partial}{\partial\ub},\frac{\partial}{\partial X},\frac{\partial}{\partial\ub},\frac{\partial}{\partial Y})$ and $R(\frac{\partial}{\partial\ub},\frac{\partial}{\partial Y},\frac{\partial}{\partial\ub},\frac{\partial}{\partial X})$ components.} of it) that has a delta function supported on the plane null hypersurface $\{\ub=0\}$. This spacetime turns out to possess remarkable global geometric properties \cite{Penrose65}. In particular, it exhibits strong focusing properties and is an example of a non-globally hyperbolic spacetime.

In a previous paper, we initiated a comprehensive study of impulsive gravitational spacetimes in the context of the characteristic initial value problem. We were able to construct a large class of spacetimes which can be thought of as representing impulsive gravitational waves parametrized by the data given on an outgoing and an incoming hypersurface such that the curvature on the outgoing hypersurface has a delta singularity supported on a 2-dimensional slice. Our construction in particular provides the first instance of an impulsive gravitational wave of compact extent and does not require any symmetry assumptions.

\subsection{Collision of Impulsive Gravitational Waves}

Returning to the explicit examples, one of the interesting features of \emph{plane} gravitational waves is that they enjoy a principle of linear superposition provided that the direction and polarization of the waves are fixed. This is not the case when one tries to combine two plane gravitational waves propagating in different directions. Nonetheless, explicit solutions to the vacuum Einstein equations modelling the interaction of two plane sandwich waves have been constructed by Szekeres \cite{Szekeres1}. Khan-Penrose \cite{KhanPenrose} later discovered an explicit solution representing the collision of two plane impulsive gravitational waves. Further analysis of the Khan-Penrose solution was carried out by Szekeres \cite{Szekeres2}.

\begin{figure}[htbp]
\begin{center}
 
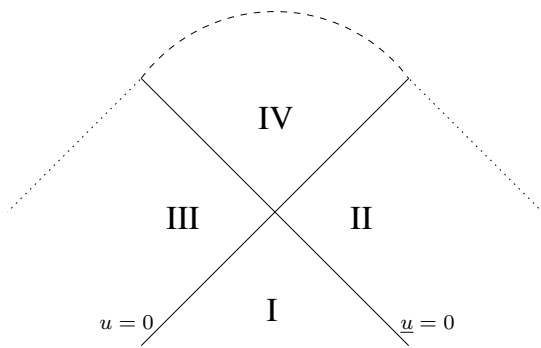
 
\caption{The Khan-Penrose Solution}
\end{center}
\end{figure}

The Khan-Penrose solution can be represented by Figure 1. The null hypersurfaces $\{u=0\}$ and $\{\ub=0\}$ have delta singularities in the Riemann curvature tensor. In region I, where $u<0$ and $\ub<0$, the metric is flat and takes the form 
$$g=-2dud\ub+dx^2+dy^2.$$
In region II, where $u<0$ and $\ub>0$, the metric is also flat and takes the form
$$g=-2dud\ub+(1-\ub)dx^2+(1+\ub)dy^2.$$
Across the null hypersurface $\{\ub=0\}$ between regions I and II, the curvature has a delta singularity. In fact, when $u <0$, the Khan-Penrose solution coincides with the Penrose solution (\ref{Penrosesol}) of one impulsive gravitational wave. The region III, where $u>0$ and $\ub<0$, is symmetric to region II, and the metric takes the form
$$g=-2dud\ub+(1-u)dx^2+(1+u)dy^2.$$
The intersection of the null hypersurfaces $u=0$ and $\ub=0$ represents the interaction of the impulsive gravitational waves. Thus region IV, where $u>0$ and $\ub>0$, is interpreted as the region after the interaction. Here, the metric takes the form
\begin{equation*}
\begin{split}
g=&-\frac{2(1-u^2-\ub^2)^{\frac 32}}{\sqrt{(1-u^2)(1-\ub^2)}(u\ub+\sqrt{(1-u^2)(1-\ub^2)})^2}du d\ub\\
&+(1-u^2-\ub^2)\left(\frac{1-u\sqrt{1-\ub^2}-\ub\sqrt{1-u^2}}{1+u\sqrt{1-\ub^2}+\ub\sqrt{1-u^2}}dx^2+\frac{1+u\sqrt{1-\ub^2}+\ub\sqrt{1-u^2}}{1-u\sqrt{1-\ub^2}-\ub\sqrt{1-u^2}}dy^2\right).
\end{split}
\end{equation*}
Even the spacetime is flat and plane symmetric in regions I, II and III, the curvature is nonzero and the plane symmetry is destroyed in region IV, signaling that the two plane impulsive gravitational waves have undergone a nonlinear interaction. Nevertheless, the metric is smooth when $u>0$, $\ub>0$ and $u^2+\ub^2 <1$. Towards $u^2+\ub^2=1$, the spacetime has a spacelike singularity.

As seen from (\ref{Penrosesol2}), the Penrose solution of one impulsive gravitational wave in particular belongs to the class of \emph{linearly polarized} \emph{pp}-waves, which takes the general form
$$g=-2d\ub dr-H(\ub)(\cos\alpha(X^2-Y^2)+2\sin\alpha XY)d\ub^2+dX^2+dY^2.$$
The constant $\alpha$ is defined to be the polarization of the wave. Thus the Khan-Penrose solution represents the interaction of two linearly polarized impulsive gravitational waves with \emph{aligned} polarization. The Khan-Penrose construction was later generalized by Nutku-Halil \cite{NutkuHalil} who wrote down explicit solutions modelling the interaction of two plane impulsive gravitational waves with non-aligned polarization. These spacetimes have the same singularity structure as that of Khan-Penrose. 

Further examples of interacting \emph{plane} impulsive gravitational waves were constructed via solving the characteristic initial value problem with data prescribed on the boundary of region IV. This was undertaken by Szekeres \cite{Szekeres2} and Yurtsever \cite{Yurtsever88} for the case of aligned polarization via the Riemann method. The general case of non-aligned polarization has been studied in a series of papers of Hauser-Ernst \cite{HE1}, \cite{HE2}, \cite{HE3} by reducing it to the matrix homogeneous Hilbert problem. The construction of even more general \emph{plane} distributional solutions for the vacuum Einstein equations that include colliding impulsive gravitational waves was carried out in \cite{LeSm}, \cite{LeSte2}.

We refer the readers to \cite{Gr}, \cite{GrPo}, \cite{BaHo}, \cite{Bicak} and the references therein for further description and more examples of spacetimes with colliding impulsive gravitational waves.

The solutions of Khan-Penrose, Szekeres and Nutku-Halil as well as the Hauser-Ernst solutions are all constructed within the class of plane symmetry. This imposes the assumptions that the wavefronts are flat and that the waves are of infinite extent. It has been speculated that the singular structure of the Khan-Penrose solution is an artifact of plane symmetry \cite{Tipler}. Concerning the assumption of plane wavefronts, Szekeres \cite{Szekeres2} wrote

\begin{quote}
The eventual singular behavior is just another aspect of Penrose's result that plane gravitational waves act as a perfect astigmatic lens. It is certainly false for waves with curved fronts, but such waves may still act as imperfect lenses providing a certain degree of focusing and amplification for each other... Clearly a better understanding of the interaction of gravitational waves with more realistic wavefronts is a problem of considerable importance.
\end{quote}
A partial remedy has been suggested by Yurtsever \cite{Yurtsever}, who did a heuristic study of ``almost plane waves'' and their interactions, allowing waves of large but finite extent. Our present paper considers the interaction of impulsive gravitational waves with finite extent and with wavefronts having arbitrary curvature. Locally, this in particular includes the case that the wavefronts are flat. Nevertheless, even in this special case, we do not require either of the waves to be linearly polarized.

\subsection{Interaction of Coherent Structures}

The nonlinear interaction of gravitational waves in general relativity can be viewed in the wider context of nonlinear interaction of coherent structures such as solitons, vortices, etc. in evolutionary gauge theories, nonlinear wave and dispersive equations. The completely integrable models KdV \cite{GGKM}, 1-dimensional cubic Schr\"{o}dinger equation \cite{ZS} and Sine-Gordon equation \cite{AKNS} not only admit individual solitary waves, but also exact solutions representing their superposition. In the past, these solutions have an asymptotic form of individual propagating solitary waves. For the period after nonlinear interaction, which can be described explicitly and typically results in a phase shift, a new superposition of new individual propagating solitary waves emerges in the distant future. These solutions are analyzed by means of the inverse scattering method. For the non-integrable models, our knowledge is much more limited and only partial results are available. In those cases, most of the results concerned perturbative interaction of coherent structures in the regimes which are either close to integrable or corresponding to interactions with high relative velocity or in which one of the objects is significantly larger than the other one. In this context, we should mention the work of Stuart on the dynamics of abelian Higgs vortices \cite{Stuart1} and the Yang-Mills-Higgs equation \cite{Stuart2} and the recent breakthrough work of Martel-Merle on the nonlinear solitary interaction for the generalized KdV equation \cite{MM1}, \cite{MM2}.

Returning to the present work, one of the main challenges in treating the interaction of impulsive gravitational waves is their singular nature, i.e., not only do we want to describe precisely how gravitational waves affect each other during the interaction, but we also need to contend with the fact that each impulsive gravitational wave separately is a singular object. We should note that partially because of this challenge, no results of this kind are available even for semilinear, let alone quasilinear, model problems. On the other hand, model problems may not be even suitable for studying the phenomena discovered in this work since it is precisely the special structure of the Einstein equations that plays a crucial role in our analysis and its conclusions.

\subsection{Previous Work on Impulsive Gravitational Spacetimes}

In a previous paper \cite{LR}, we studied the (characteristic) initial value problem for spacetimes representing a single propagating impulsive gravitational wave. Corresponding to such spacetimes, we considered data that have a curvature delta singularity supported on an embedded 2-sphere $S_{0,\ub_s}$ on an outgoing null hypersurface, and is smooth on an incoming null hypersurface. We showed that such data give rise to a unique impulsive gravitational spacetime satisfying the vacuum Einstein equations. Moreover, the curvature has a delta singularity supported on a null hypersurface emanating from the initial singularity on $S_{0,\ub_s}$ and the spacetime metric remains smooth away from this null hypersurface (see Figure 2).

\begin{figure}[htbp]\label{propagationfig}
\begin{center}
 
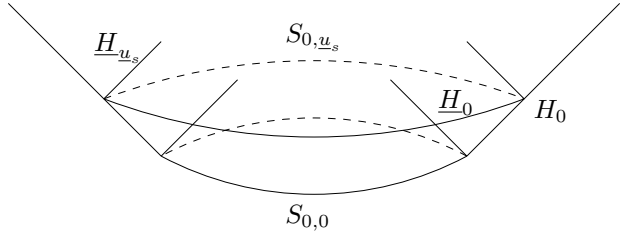
 
\caption{Propagation of One Impulsive Gravitational Wave}
\end{center}
\end{figure}

\subsection{Description of Results in this Paper}

In this paper, we begin the study of the (characteristic) initial value problem for spacetimes which represent the nonlinear interaction of two impulsive gravitational waves. For such a problem, the initial data have delta function singularities supported on embedded 2-spheres $S_{0,\ub_s}$ and $S_{u_s,0}$ on the initial null hypersurfaces $H_0$ and $\Hb_0$ respectively see Figure 3). According to the results that were obtained in \cite{LR}, before the interaction of the two impulsive gravitational waves, i.e., for $u<u_s$ or $\ub<\ub_s$, a unique solution to the vacuum Einstein equations exists, and the singularity is supported on the null hypersurfaces emanating from the initial singularities.

\begin{figure}[htbp]
\begin{center}
 
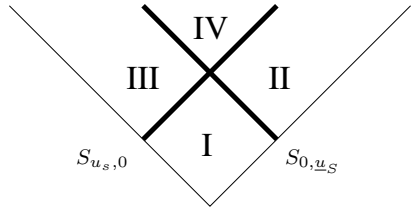
 
\caption{Nonlinear Interaction of Impulsive Gravitational Waves}
\end{center}
\end{figure}

Our focus here will be to understand the spacetime ``beyond'' the first interaction (region IV in Figure 3). We will show that the resulting spacetime will be a solution to the vacuum Einstein equations with delta function singularities in the curvature on the corresponding null hypersurfaces germinating from the initial singularities. Surprisingly, the spacetime remains smooth locally in region IV after the interaction of the impulsive gravitational waves. Our main result for the collision of impulsive gravitational waves is described by the following theorem:

\begin{theorem}\label{cgiwthm}
Suppose the following hold for the initial data set:
\begin{itemize}
\item The data on $H_0$ are smooth except across a two sphere $S_{0,\ub_s}$, where the traceless part of the second fundamental form of $H_0$ has a jump discontinuity.
\item The data on $\Hb_0$ are smooth except across a two sphere $S_{u_s,0}$, where the traceless part of the second fundamental form of $\Hb_0$ has a jump discontinuity.
\end{itemize}
Then
\begin{enumerate}[(a)]
\item For such initial data and $\epsilon$ sufficiently small, there exists a unique spacetime $(\mathcal M,g)$ endowed with a double null foliation $u$, $\ub$ that solves the characteristic initial value problem for the vacuum Einstein equations in the region $0\leq u\leq u_*$, $0\leq \ub\leq\ub_*$, whenever $u_*\leq\epsilon$ or $\ub_*\leq\epsilon$.
\item Let $\Hb_{\ub_s}$ (resp. $H_{u_s}$) be the incoming (resp. outgoing) null hypersurface emanating from $S_{0,\ub_s}$ (resp. $S_{u_s,0}$). Then the curvature components $\alpha_{AB}=R(e_A,e_4,e_B,e_4)$ and $\alphab_{AB}=R(e_A,e_3,e_B,e_3)$ are measures with singular atoms supported on $\Hb_{\ub_s}$ and $H_{u_s}$ respectively. 
\item All other components of the curvature tensor can be defined in $L^2$. Moreover, the solution is smooth away from $\Hb_{\ub_s}\cup H_{u_s}$.
\end{enumerate}
\end{theorem}

\begin{remark}
The norms that we use allow us to choose $u_s<\epsilon$ and $\ub_s<\epsilon$ so that the solution indeed represents the collision of two impulsive gravitational waves. See the statement of Theorem \ref{rdthmv1}.
\end{remark}

Our approach relies on an extension of the renormalized energy estimates introduced in \cite{LR}. As in \cite{LR}, our concern is not just the existence of weak solutions admitting two colliding impulsive gravitational waves, but also their uniqueness. The uniqueness property follows from the a priori estimates developed in this paper and leads to strong solutions of the vacuum Einstein equations.

Parts (b) and (c) of Theorem \ref{cgiwthm} can be interpreted as results on the propagation of singularity that is conormal with respect to a pair of transversally intersecting characteristic hypersurfaces. Similar problems have been studied for general hyperbolic equations with a much \emph{weaker} singularity such that classical well-posedness theorems can be applied \cite{Beals}, \cite{Alinhac1}. In the case of second order equations, it is known that no new singularities appear after the interaction of the weak conormal singularities. In general, however, a third order semilinear hyperbolic equation can be constructed so that new singularities form after the interaction of two weak conormal singularities \cite{RauchReed}. In this paper, we address stronger conormal singularities such that in general, even for \emph{semilinear} hyperbolic systems, only the local propagation of \emph{one} conormal singularity has been proved \cite{Metivier}. For conormal singularities of this strength, no general theorem is known to address the interaction of propagating singularities even for semilinear, let alone quasilinear, equations. By contrast, in this work, the special structure of the Einstein equations in the double null foliation gauge has been heavily exploited to show that even for the stronger conormal singularities that we consider, the spacetime remains smooth after their interaction.

In this paper, as in \cite{LR}, we prove a more general theorem on the existence and uniqueness of solutions to the vacuum Einstein equations that in particular implies Theorem \ref{cgiwthm}(a). In addition to allowing non-regular characteristic initial data on both $H_0$ and $\Hb_0$, our main existence theorem extends the results in \cite{LR} in two other ways. First, we consider the characteristic initial value problem with initial data such that the traceless parts of the null second fundamental forms and their angular derivatives are only in $L^2$ in the null directions as opposed to being in $L^\infty$ in the previous work. Second, in \cite{LR}, the constructed spacetime lies in the range of the double null coordinates corresponding to $\{0\leq u\leq \epsilon\}\cap\{0\leq \ub\leq \epsilon\}$. In this paper, using some ideas in \cite{L}, we extend the domain of existence and uniqueness to a region that is not symmetric in $u$ and $\ub$, i.e., in $(\{0\leq u\leq \epsilon\}\cap\{0\leq \ub\leq I_1\})\cup (\{0\leq \ub\leq \epsilon\}\cap\{0\leq u\leq I_2\})$, where $I_1$ and $I_2$ are finite but otherwise arbitrarily large (see Figure 4). We refer the readers to Sections \ref{thmfirstver} and \ref{maintheorem} for precise formulations of the existence and uniqueness theorem.

\begin{figure}[htbp]
\begin{center}
 
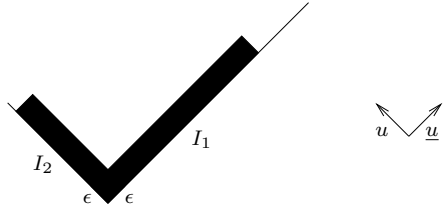
 
\caption{Region of Existence}
\end{center}
\end{figure}

One of the unexpected consequences of our approach in this paper is that we can also apply it to the problem on the formation of trapped surfaces. The work of Christodoulou \cite{Chr} was a major breakthrough in solving the problem of the evolutionary formation of a trapped surface and this was later extended and simplified in \cite{KlRo}, \cite{KlRo1}. In all of those works, characteristic initial data were prescribed on $H_0\cap\{0\leq \ub\leq\epsilon\}$ and $\Hb_0$ with sufficient conditions for data on $H_0\cap\{0\leq \ub\leq\epsilon\}$ formulated in such a way as to guarantee the appearance of a trapped surface in the causal future of $H_0\cap\{0\leq \ub\leq\epsilon\}$ and $\Hb_0$ (see Figure 5).

\begin{figure}[htbp]
\begin{center}
 
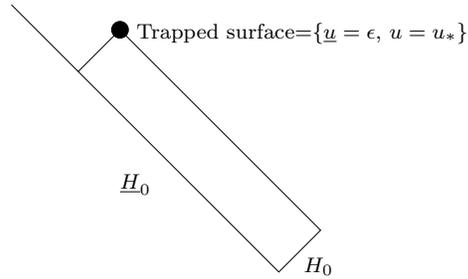
 
\caption{Formation of a Trapped Surface}
\end{center}
\end{figure}

The sufficient condition on $H_0\cap\{0\leq \ub\leq\epsilon\}$ required that certain geometric quantities are large with respect to $\epsilon$ and thus lead to the problem of constructing a semi-global large data solution to the Einstein equations. In all those works, to control the dynamics of the Einstein equations, the largeness of geometric quantities associated to $H_0\cap\{0\leq \ub\leq\epsilon\}$ was offset by requiring the data on $\Hb_0$ to be the trivial Minkowski data.

Our new approach allows us to eliminate the requirement that the data on $\Hb_0$ have to be trivial. It can be replaced by a condition that the data on $\Hb_0$ are merely ``not too large'' and still guarantee the formation of a trapped surface in the causal future of $H_0\cap\{0\leq \ub\leq\epsilon\}$ and $\Hb_0$. We refer the readers to Section \ref{sectrapped} for a more precise formulation of the theorem on the formation of trapped surfaces.

\subsection{A Toy Model}

One of the most challenging aspects of the vacuum Einstein equations is its quasilinear and tensorial nature. Nonetheless, it may be instructive to examine a related phenomenon in a toy model of a scalar semilinear wave equation satisfying the null condition in $\mathbb R^{3+1}$
\begin{equation}\label{wavenull}
\Box \phi = -(\partial_t\phi)^2+\sum_{i\leq 3}(\partial_{x_i}\phi)^2,
\end{equation}
(or more generally a system
$\Box \Phi = Q(\Phi,\Phi)$, where $\Phi: \mathbb R^{3+1} \to \mathbb R^n$ and $Q(\Phi,\Phi)$ is a null form)
with the characteristic initial data
\begin{equation*}
\begin{split}
f(\ub,\theta)=&\partial_{\ub}\phi(\ub,u=0,\theta),\\
g(u,\theta)=&\partial_{u}\phi(\ub=0,u,\theta),
\end{split}
\end{equation*}
prescribed on the light cones $H_0=\{\ub:=t+r=0\}$ and $\Hb_0=\{u:=t-r+2=0\}$ respectively and
$$h(\theta)=\phi(\ub=0,u=0,\theta)$$
prescribed on the initial 2-sphere defined by $\{\ub=0, u=0\}$.

For this toy model, the analogue of the problem addressed in Theorem \ref{cgiwthm} is the local existence and uniqueness result for \eqref{wavenull} in the region $\{0\leq \ub\leq \epsilon\}\cup\{0\leq u\leq 1\}$ for the data 
$$f = f_1+\mathbbm 1_{\{\ub-\frac{\epsilon}{2}\geq 0\}} f_2 $$ 
and 
$$g = g_1+\mathbbm 1_{\{u-\frac{1}{2}\geq 0\}} g_2,$$ 
where $f_1$, $f_2$, $g_1$, $g_2$, $h$ are smooth functions and $\mathbbm 1$ is the indicator function. For these data, $\partial_{\ub} f$ and $\partial_u g$ have delta singularities supported on the 2-spheres $\{u=0\}\cap\{\ub=\frac{\epsilon}{2}\}$ and $\{\ub=0\}\cap\{u=\frac{1}{2}\}$ respectively. It turns out that the corresponding solution is smooth away from the set $\{\ub=\frac{\epsilon}{2}\}\cup\{u=\frac{1}{2}\}$, but yet $\partial_{\ub}\phi$ (resp. $\partial_u\phi$) remains discontinuous across $\{\ub=\frac{\epsilon}{2}\}$ (resp. $\{u=\frac{1}{2}\}$)\footnote{assuming, of course, that the initial data $f_2$ (resp. $g_2$) is non-zero for $\ub=\frac{\epsilon}{2}$ (resp. $u=\frac{1}{2}$).}.

Theorem \ref{cgiwthm} is embedded in a more general local existence and uniqueness result (stated precisely in Theorem \ref{rdthmv1} below). Its analogue for the above toy model is the local existence for \eqref{wavenull} with the data $f$, $g$ and $h$ only satisfying
$$\sum_{i\leq 4}||\Omega^i f||_{L^2(H_0(0,\epsilon))}\leq C,$$
$$\sum_{i\leq 4}||\Omega^i g||_{L^2(\Hb_0(0,1))}\leq C,$$
and
$$\sum_{i\leq 4}||\Omega^i h||_{L^2(S_{0,0})}\leq C,$$
where $\Omega\in\{x_1\partial_{x_2}-x_2\partial_{x_1}, x_2\partial_{x_3}-x_3\partial_{x_2}, x_3\partial_{x_1}-x_1\partial{x_3}\}$. The corresponding solution exists in the region $\{0\leq \ub\leq \epsilon\}\cup\{0\leq u\leq 1\}$ and obeys the following estimates:
$$\sup_{0\leq u\leq 1}\sum_{i\leq 3}||\Omega^i \partial_{\ub}\phi||_{L^2(H_u)}\leq C',$$
$$\sup_{0\leq\ub\leq \epsilon}\sum_{i\leq 3}||\Omega^i \partial_u\phi||_{L^2(\Hb_{\ub})}\leq C',$$
$$\sup_{0\leq u\leq 1}\sum_{i\leq 4}||\Omega^i\phi||_{L^2(H_u)}+\sup_{0\leq\ub\leq \epsilon}\sum_{i\leq 4}||\Omega^i \phi||_{L^2(\Hb_{\ub})}\leq C'.$$

Even though this model hardly reflects the difficulties of the nonlinear structure of the vacuum Einstein equations, such local existence, uniqueness and propagation of singularity results to our knowledge are not known for this type of equations but follow from the methods\footnote{In particular, we show that in order to guarantee the existence of the solution, it suffices to commute the equation \eqref{wavenull} only with angular derivatives $\Omega$.} used in this paper.

\subsection{First Version of the Theorem}\label{thmfirstver}

Our general approach is based on energy estimates and transport equations in the double null foliation gauge. This gauge was used in our previous work \cite{LR}. The general approach in the double null gauge has been carried out in \cite{KN}, \cite{Chr} and \cite{KlRo}.

The spacetime in question will be foliated by families of outgoing and incoming null hypersurfaces $H_u$ and $\Hb_{\ub}$ respectively. Their intersection is assumed to be a 2-sphere denoted by $S_{u,\ub}$. Define a null frame $\{e_1,e_2,e_3,e_4\}$, where $e_3$ and $e_4$ are null, as indicated in Figure 6, and $e_1$, $e_2$ are vector fields tangent to the two spheres $S_{u,\ub}$. $e_4$ is tangent to $H_u$ and $e_3$ is tangent to $\Hb_{\ub}$.

\begin{figure}[htbp]
\begin{center}
 
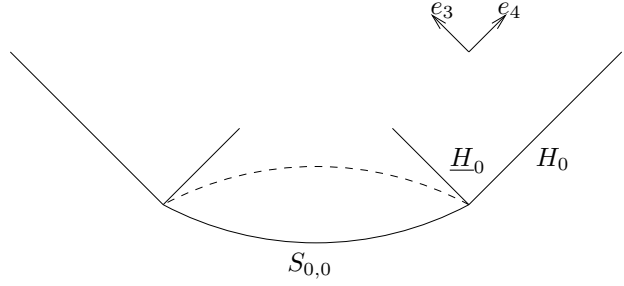
 
\caption{The Null Frame}
\end{center}
\end{figure}
Decompose the Riemann curvature tensor with respect to the null frame $\{e_1,e_2,e_3,e_4\}$:
\begin{equation}\label{curv.def}
\begin{split}
\a_{AB}&=R(e_A, e_4, e_B, e_4),\quad \, \,\,   \ab_{AB}=R(e_A, e_3, e_B, e_3),\\
\b_A&= \frac 1 2 R(e_A,  e_4, e_3, e_4) ,\quad \bb_A =\frac 1 2 R(e_A,  e_3,  e_3, e_4),\\
\rho&=\frac 1 4 R(e_4,e_3, e_4,  e_3),\quad \sigma=\frac 1 4  \,^*R(e_4,e_3, e_4,  e_3),
\end{split}
\end{equation}
where $^*R$ denotes the Hodge dual of $R$.
In the context of the interaction of impulsive gravitational waves, the $\alpha$ and $\alphab$ components of curvature can only be understood as measures. In the main theorem below, we do not require $\alpha$ and $\alphab$ to even be defined.

\noindent Define also the following Ricci coefficients with respect to the null frame:
\begin{equation*}
\begin{split}
&\chi_{AB}=g(D_A e_4,e_B),\, \,\, \quad \chib_{AB}=g(D_A e_3,e_B),\\
&\eta_A=-\frac 12 g(D_3 e_A,e_4),\quad \etab_A=-\frac 12 g(D_4 e_A,e_3),\\
&\omega=-\frac 14 g(D_4 e_3,e_4),\quad\,\,\, \omegab=-\frac 14 g(D_3 e_4,e_3),\\
&\zeta_A=\frac 1 2 g(D_A e_4,e_3).
\end{split}
\end{equation*}
Let $\chih$ (resp. $\chibh$) be the traceless part of $\chi$ (resp. $\chib$). For the problem of the interaction of impulsive gravitational waves, we prescribe initial data on $H_0$ (resp. $\Hb_0$) such that $\chih$ (resp. $\chibh$) has a jump discontinuity across $S_{0,\ub_s}$ (resp. $S_{u_s,0}$) but smooth otherwise.

As mentioned before, we prove a theorem concerning existence and uniqueness of spacetimes for a larger class of initial data than that for the interacting impulsive gravitational waves. The following is the main theorem in this paper on existence and uniqueness of solutions to the vacuum Einstein equations.
\begin{theorem}\label{rdthmv1}
Let $\th^A$ be transported coordinates on the 2-spheres\footnote{see definition in Section \ref{coordinates}} $S_{u,\ub}$ and $\gamma$ be the spacetime metric restricted to $S_{u,\ub}$. Prescribe data such that\footnote{for $2\Omega^{-2}=-g(L',\Lb')$, where $L'$ and $\Lb'$ are defined to be null geodesic vector fields (see Section \ref{secdnf}).} $\Omega=1$. Suppose, in every coordinate patch on $H_0$ and $\Hb_0$,
$$\det\gamma \geq c,$$
$$\sum_{i\leq 4} |(\frac{\partial}{\partial\th})^i\gamma_{AB}|+\sum_{i\leq 3} |(\frac{\partial}{\partial\th})^i\zeta_A|\leq C.$$
On $H_0$,
$$\sum_{i\leq 3}\int_0^{I_1} |(\frac{\partial}{\partial\th})^i\chih_{AB}|^2 d\ub +\sum_{i\leq 3} |(\frac{\partial}{\partial\th})^i\trch|\leq C,$$
and on $\Hb_0$,
$$\sum_{i\leq 3}\int_0^{I_2} |(\frac{\partial}{\partial\th})^i\chibh_{AB}|^2 du+\sum_{i\leq 3} |(\frac{\partial}{\partial\th})^i\trchb|\leq C.$$
Then for $\epsilon$ sufficiently small depending only on $c$, $C$, $I_1$ and $I_2$, there exists a unique spacetime solution $(\mathcal M,g)$ that solves the characteristic initial value problem for the vacuum Einstein equations in the region\footnote{The variables $u$ and $\ub$ will be defined to be null, i.e., the region $\{0\leq u\leq \epsilon\}\cap\{0\leq \ub\leq I_1\}$ is given geometrically as the spacetime region to the future of the initial data and bounded by the hypersurfaces emanating from the initial spheres $S_{\epsilon,0}$ and $S_{0,I_1}$.} $(\{0\leq u\leq \epsilon\}\cap\{0\leq \ub\leq I_1\})\cup (\{0\leq \ub\leq \epsilon\}\cap\{0\leq u\leq I_2\})$. Associated to the spacetime a double null coordinate system $(u,\ub,\th^1,\th^2)$ exists, relative to which the spacetime is in particular Lipschitz and retains higher regularity in the angular directions.
\end{theorem}

Due to the symmetry in $u$ and $\ub$, it suffices to prove the Theorem in $0\leq u\leq I$, $0\leq \ub\leq \epsilon$. In the sequel, we will focus on the proof in this region. The other case can be treated similarly. A more precise formulation of the theorem can be found in Section \ref{maintheorem}. 

In this paper, local existence and uniqueness is proved under the assumption that the spacetime is merely $W^{1,2}$. In terms of differentiability, this is even one derivative weaker than the recently resolved $L^2$ curvature conjecture (\cite{L21}, \cite{L22}, \cite{L23}, \cite{L24}, \cite{L25}). Of course the $W^{1,2}$ assumption refers to the worst possible behavior observed in our data and our result heavily relies on the structure of the Einstein equations which allows us to efficiently exploit the better behavior of the other components.

Theorem \ref{rdthmv1} in particular shows the existence and uniqueness of solutions for the initial data of nonlinearly interacting impulsive gravitational waves. An additional argument, based on the estimates in the proof of Theorem \ref{rdthmv1}, will be carried out to show the regularity of the spacetime with colliding impulsive gravitational waves, i.e., parts (b) and (c) in Theorem \ref{cgiwthm}.

Theorem \ref{rdthmv1} also forms the basis for the theorem on the formation of trapped surfaces (Theorem \ref{trappedsurface}).\footnote{In fact, one of the motivations for formulating Theorem \ref{rdthmv1} for a finite but arbitrarily long $u$ region is for proving Theorem \ref{trappedsurface}.} In particular, Theorem \ref{rdthmv1} extends the existence theorem of Christodoulou \cite{Chr} to data that is not necessarily small on $\Hb_0$ while allowing the data to be large on $H_0$. Moreover, the estimates obtained in Theorem \ref{rdthmv1} show that for a large class of data on $\Hb_0$ that is not necessarily close to Minkowski space, there exists an open set of initial data on $H_0$ such that a trapped surface is formed in evolution.

\subsection{Strategy of the Proof}

Without symmetry assumptions, all known proofs of existence and uniqueness of spacetimes satisfying the Einstein equations are based on $L^2$-type estimates for the curvature tensor and its derivatives or the metric components and their derivatives. One of our main challenges in \cite{LR} and this paper is that for an impulsive gravitational wave the curvature tensor can only be defined as a measure and is not in $L^2$. 

Let $\Psi$ denote the curvature components and $\Gamma$ denote the Ricci coefficients. In \cite{LR} where we studied the propagation of one impulsive gravitational wave, the curvature component $\alpha$ is non-$L^2$-integrable. Nevertheless, we showed that the $L^2$-type energy estimates for the components of the Riemann curvature tensor
\begin{equation}\label{EEsch}
\int_{H_u} \Psi^2+\int_{\Hb_{\ub}} \Psi^2\leq \int_{H_0} \Psi^2+\int_{\Hb_0} \Psi^2 +\int_0^{\ub} \int_0^u\int_{S_{u
',\ub'}} \Gamma\Psi\Psi du' d\ub'.
\end{equation}
coupled together with the null transport equations for the Ricci coefficients
$$\nab_3\Gamma=\Psi+\Gamma\Gamma,\quad\nab_4\Gamma=\Psi+\Gamma\Gamma$$
can be renormalized and closed avoiding the singular curvature component $\alpha$.

In this paper, we consider spacetimes with two interacting impulsive gravitational waves and therefore both curvature components $\alpha$ and $\alphab$ are not $L^2$-integrable. We thus need to extend the renormalization in \cite{LR} and to close the energy estimates circumventing both $\alpha$ and $\alphab$.

In the remainder of this subsection, we will explain the main ideas for proving a priori estimates. Note that since we are working at a very low level of regularity, a priori estimates alone do not imply the existence and uniqueness of solutions. An additional argument to go from a priori estimates to existence and uniqueness was carried out in \cite{LR} in which we studied the convergence of a sequence of smooth solutions of the vacuum Einstein equations to the non-regular solution. A direct but tedious modification of that argument can be carried out in the context of this paper, giving the desired existence and uniqueness result. We, however, will be content to prove a priori estimates in this paper and refer the readers to \cite{LR} for more details.

After we explain the ideas for proving the a priori estimates, we will then return to sketch the ideas in the proofs of the regularity for colliding impulsive gravitational waves (Theorem \ref{cgiwthm}(b),(c)) and the formation of trapped surfaces.

\subsubsection{Renormalized Energy Estimates}

In \cite{LR}, we introduced the renormalized energy estimates for the vacuum Einstein equations. This allowed us to avoid any information of $\alpha$ while deriving the a priori estimates. In this paper, since in addition to an incoming impulsive gravitational wave there is an outgoing impulsive gravitational wave, both $\alpha$ and $\alphab$ are non-$L^2$-integrable. We thus need to renormalize the curvature components in a way that avoids both $\alpha$ and $\alphab$.

To this end, we view the vacuum Einstein equations as a coupled system for the Ricci coefficients $\Gamma$ and the curvature components $\Psi$, which is traditionally treated by a combination of estimates for the transport equations for $\Gamma$ coupled with the energy estimates for curvature. The renormalization used in this paper replaces the full set of curvature components $\Psi$ with the new quantities
\begin{eqnarray*}
\begin{cases}
&\Psic=\Psi+\Gamma\Gamma\quad \mbox{for} \quad \Psi=\beta,\rho,\sigma,\betab,\\
&\Psic=0 \quad\mbox{otherwise}.
\end{cases}
\end{eqnarray*}
We also replace the full set of transport equations for $\Gamma$ with a subset which does not involve the prohibited curvature components $\alpha$, $\alphab$ (or rather, involves only the renormalized components $\Psic$). Similarly, we consider a subset of Bianchi equations. We then show that the reduced system can still be closed by a combination of transport-energy type estimates. 

To illustrate the renormalization, we first prove the energy estimates for $\beta$ on $H_u$ and for $(\rho,\sigma)$ on $\Hb_{\ub}$ by considering the following set of Bianchi equations:
$$\nabla_4\rho=\div\beta - \frac 12 \chibh\cdot\alpha+\Gamma\Psi,$$
$$\nab_4\sigma=-\div ^*\beta+\frac 12\chibh\wedge\alpha+\Gamma\Psi,$$
$$\nab_3\beta=\nab\rho+\nab^*\sigma+\Gamma\Psi,$$
where $\Psi$ denotes the regular curvature components. However, the curvature component $\alpha$ still appears in the nonlinear terms in these equations. In order to deal with this problem, we consider the equations for the \emph{renormalized} curvature components $\rhoc=\rho-\frac 12 \chih\cdot\chibh$ and $\sigmac=\sigma+\frac 12 \chibh\wedge\chih$ instead.
Using the equation
$$\nabla_4\chih=-\alpha+\Gamma\Gamma,$$
we notice that the equations can be rewritten as
$$\nabla_4\rhoc=\div\beta+\Gamma\Psic+\Gamma\nab\Gamma+\Gamma\Gamma\Gamma,$$
$$\nab_4\sigmac=-\div ^*\beta+\Gamma\Psic+\Gamma\nab\Gamma+\Gamma\Gamma\Gamma,$$
$$\nab_3\beta=\nab\rhoc+\nab^*\sigmac+\Gamma\Psic+\Gamma\nab\Gamma+\Gamma\Gamma\Gamma.$$
We now have a set of renormalized Bianchi equations that does not contain $\alpha$. Using these equations, we derive the renormalized energy estimate
$$\int_{H_u} \Psic^2+\int_{\Hb_{\ub}} \Psic^2\leq \int_{H_0} \Psic^2+\int_{\Hb_0} \Psic^2 +\int_0^{\ub} \int_0^u\int_{S_{u
',\ub'}} \left(\Gamma\Psic\Psic+\Gamma\nab\Gamma\Psic+\Gamma\Gamma\Gamma\Psic\right) du' d\ub',$$
in which $\alpha$ does not appear in the error term.

It turns out that the same renormalization $\rhoc$ and $\sigmac$ that was used to avoid $\alpha$ also can also be applied to circumvent $\alphab$. For example, $\alphab$ enters as source terms in the following Bianchi equations,
$$\nabla_3\rho=-\div\betab - \frac 12 \chih\cdot\alphab+\Gamma\Psi,$$
$$\nab_3\sigma=-\div ^*\betab-\frac 12 \chih\wedge\alphab+\Gamma\Psi.$$
Using the equation
$$\nab_3\chibh=-\alphab+\Gamma\Gamma,$$
we see that $\alphab$ does not appear in the equations for $\nab_3\rhoc$ and $\nab_3\sigmac$.

As a consequence, we obtain a set of $L^2$ curvature estimates which do not explicitly couple to the singular curvature components $\alpha$ and $\alphab$. However, we say explicitly that a priori it is not obvious for the Ricci coefficients $\Gamma$ appearing in the nonlinear error for the energy estimates to be bounded independent of $\alpha$ and $\alphab$.

\subsubsection{Mixed Norm Estimates for the Ricci Coefficients}

In order to close the estimates, it is necessary to obtain control of the Ricci coefficients via the transport equations
\begin{equation}\label{schematictransport}
\nab_3\Gamma=\Psic+\Gamma\Gamma,\quad\nab_4\Gamma=\Psic+\Gamma\Gamma.
\end{equation}
In \cite{LR}, we showed that $\Gamma$ can be estimated in $L^\infty$ by considering a subset of the transport equations that do not involve the singular curvature component $\alpha$ (and involve only the renormalized curvature components $\Psic$).

In the setting of this paper, in addition to proving bounds on $\Gamma$ without any information on both singular curvature components $\alpha$ and $\alphab$, an extra challenge is that unlike in \cite{LR}, not all Ricci coefficients are bounded in the initial data. In fact, for the class of initial data considered in this paper, $\chih$ (resp. $\chibh$) is only assumed to be in $L^2_{\ub}H^3(S)$ (resp. $L^2_u H^3(S)$), where $H^3(S)$ refers to the $L^2$ norm of the third angular derivatives on the 2-spheres. Therefore, (\ref{schematictransport}) at best implies that $\chih$ (resp. $\chibh$) can be estimated in $L^2_{\ub}L^\infty_uL^\infty(S)$ (resp. $L^2_{u}L^\infty_{\ub}L^\infty(S)$), where the $L^\infty$ norms on the sphere and along the $u$ (resp. $\ub$) direction are taken first, before the $L^2$ norm in $\ub$ (resp. $u$) is taken.

Because of the weaker assumption on the Ricci coefficients in the initial data, we only prove estimates for the Ricci coefficients in mixed norms. In fact, we prove different mixed norm bounds for different Ricci coefficients. Using a schematic notation $\psi\in\{\trch,\trchb,\eta,\etab\}$, $\psi_H\in\{\chih,\omega\}$ and $\psi_{\Hb}\in\{\chibh,\omegab\}$, we only control $\psi$ in $L^\infty_uL^\infty_{\ub} L^\infty(S)$, $\psi_H$ in $L^2_{\ub}L^\infty_u L^\infty(S)$ and $\psi_{\Hb}$ in $L^2_u L^\infty_{\ub}L^\infty(S)$.

It is a remarkable fact that the Einstein equations possess a structure such that these mixed norm bounds are sufficient to close all the estimates for the Ricci coefficients using the transport equations, as well as the energy estimates for the curvature components.

As an example, in order to estimate $\psi$ in $L^\infty_uL^\infty_{\ub} L^\infty(S)$, we use the transport equation
$$\nab_3\psi=\betab+\rhoc+\nab\psi+(\psi+\psi_{\Hb})(\psi+\psi_{\Hb}).$$ 
Notice that the term $\psi_H$ does not appear as the source of this equation. Therefore, with the control of the Ricci coefficients in the mixed norms, all terms on the right hand side can be bounded after integrating in the $e_3$ (i.e. $u$) direction to obtain the desired bound for $\psi$.

On the other hand, the transport equation for $\psi_H$ contains both $\psi_H$ and $\psi_{\Hb}$ in the inhomogeneous term:
$$\nab_3\psi_H=\psi_H\psi_{\Hb}+...$$
Integrating this equation, we get
\begin{equation}\label{psiHschematic}
||\psi_H||_{L^\infty(S_{u,\ub}))}\leq ||(\psi_{\Hb})_0||_{L^\infty(S_{u,0})}+||\psi_{\Hb}||_{L^2_uL^\infty(S)}||\psi_{H}||_{L^{\infty}_uL^\infty(S)}+....
\end{equation}
The initial data term $||(\psi_{\Hb})_0||_{L^\infty(S_{0,\ub})}$ and the factor $||\psi_{H}||_{L^{\infty}_uL^\infty(S)}$ in the second term are not bounded uniformly in $\ub$. Nevertheless, since we are only aiming to prove estimates for $\psi_H$ in $L^2_{\ub}L^\infty_u L^\infty(S)$, we can take the $L^2_{\ub}$ norm in (\ref{psiHschematic}) and every term on the right hand side is controlled by the mixed norms. This allows us to prove the mixed norm estimates for all the Ricci coefficients.

Even more remarkable is that the bounds we obtain for the Ricci coefficients in mixed norms are also sufficient to close the energy estimates for the renormalized curvature components. Schematically, the renormalized energy estimates read as follows:
\begin{equation*}
\begin{split}
&||(\beta,\rhoc,\sigmac)||_{L^\infty_u L^2_{\ub}L^2(S)}+||(\rhoc,\sigmac,\betab)||_{L^\infty_{\ub} L^2_uL^2(S)}\\
\leq& \mbox{ Initial Data }+||\Gamma\Psic^2||_{L^1_uL^1_{\ub}L^1(S)}+||\Gamma^5||_{L^1_uL^1_{\ub}L^1(S)}+...
\end{split}
\end{equation*}
The error terms on the right hand side have to be controlled by the $L^2$ curvature bounds on the left hand side together with the estimates for the Ricci coefficients in the mixed norms. As an example, an error term $\psi_H\rhoc\rhoc$ can be controlled after applying Cauchy-Schwarz as follows:
$$||\psi_H\rhoc\rhoc||_{L^1_uL^1_{\ub}L^1(S)}\leq ||\rhoc||_{L^\infty_{\ub}L^2_uL^2(S)}^2||\psi_H||_{L^2_{\ub}L^\infty_u L^\infty(S)}.$$
Here, it is important to note that using the mixed norms for $\psi_H$, we can estimate in $L^\infty$ first, before taking the $L^2$ norm. On the other hand, an error term of the type $\psi_H\beta\beta$ cannot be controlled in $L^1_uL^1_{\ub}L^1(S)$ since each of the three factors can only be bounded after taking the $L^2_{\ub}$ norm. Miraculously, such terms never arise as error terms in the energy estimates!

A similar structure also arises in the error terms of the form
$$||\Gamma^5||_{L^1_uL^1_{\ub}L^1(S)}.$$
For this term, $\psi_H$ (or $\psi_{\Hb}$) appears at most twice, allowing us to estimate each of them in $L^2_{\ub}$ (or $L^2_u$).

In order to close all the estimates, we need to prove mixed norm estimates for higher derivatives of the Ricci coefficients and energy estimates for higher derivatives of the curvature components. This is achieved using only angular covariant derivatives $\nab$ as commutators. For such estimates, the singular curvature components $\alpha$ and $\alphab$ never arise in the nonlinear error terms (see Proposition \ref{commutation.prop} in Section \ref{commutation}). Moreover, there is a structure similar to that described above for the higher order estimates that allows us to close merely with the mixed norm bounds.

\subsubsection{Estimates in an Arbitrarily Long $u$ Interval}

In our main theorem, we prove existence, uniqueness and a priori estimates in a region such that only the $\ub$ interval is assumed to be short, while the $u$ interval can be arbitrarily long (but finite). This poses an extra challenge since when we control the nonlinear error terms integrated over the $u$ interval, we do not gain a smallness constant.

This difficulty already arises in the problem of existence in such a region with smooth initial data. This was studied in \cite{L}\footnote{In \cite{L}, the a priori estimates were proved in the case where the $u$ interval is assumed to be short and the $\ub$ interval is allowed to be arbitrarily long. We outline the main ideas of \cite{L} assuming instead the setting in this paper.}. It was noticed that both in carrying out the Ricci coefficient estimates and the energy estimates for the curvature components, the structure of the Einstein equations allows us to prove that whenever a smallness constant is absent, the estimate is in fact linear.

To achieve the bounds of the Ricci coefficients, the following structure of the null structure equations was used. Let
$$\Gamma_1\in \{\trch,\chih,\trchb,\chibh,\eta,\omega\},\quad\Gamma_2=\etab,\quad\Gamma_3=\omega.$$
They satisfy the following transport equations:
\begin{equation}\label{Gammasch1}
\begin{split}
\nab_4\Gamma_1=&\Psi+(\Gamma_1+\Gamma_2+\Gamma_3)(\Gamma_1+\Gamma_2+\Gamma_3),\\
\nab_3\Gamma_2=&\Psi+(\Gamma_1+\Gamma_2)\Gamma_1,\\
\nab_3\Gamma_3=&\Psi+(\Gamma_1+\Gamma_2+\Gamma_3)(\Gamma_1+\Gamma_2).
\end{split}
\end{equation}

We prove the bounds for $\Gamma_1$, $\Gamma_2$, $\Gamma_3$ in the setting of a bootstrap argument in which the control for the curvature components $\Psi$ is assumed. The estimates for $\Gamma_1$ can easily be obtained since integrating in the $e_4$ (i.e., $\ub$) direction gives a smallness constant. For $\Gamma_2$, the integration is in the $e_3$ (i.e., $u$) direction and does not have a smallness constant. Nevertheless, using the bounds for $\Gamma_1$ that have already been obtained, the error term is linear in $\Gamma_2$! This can thus be dealt with using Gronwall's inequality. Finally, the equation for $\Gamma_3$ is also linear in $\Gamma_3$. Therefore, using the the estimates already derived for $\Gamma_1$ and $\Gamma_2$ together with Gronwall's inequality, the equation for $\Gamma_3$ can be applied to get the desired control for $\Gamma_3$.

In the energy estimates for the curvature components, there is likewise a term without a smallness constant. Nevertheless, it was noted in \cite{L} that the only term not accompanied by a smallness constant is also linear. Thus, as in the case in controlling the Ricci coefficient, the energy estimates can be closed using Gronwall's inequality.

Returning to the setting of this paper, this challenge of having an arbitrarily long $u$ interval is coupled to the difficulty that the curvature components $\alpha$ and $\alphab$ are singular and that the Ricci coefficients $\chih$, $\chibh$, $\omega$, $\omegab$ can only be estimated in appropriate mixed norms. As a result, unlike in \cite{L}, we cannot use the $\nab_4$ equations for $\chih$ and $\trch$ to gain a smallness constant. The $\nab_4\chih$ equation is unavailable because $\alpha$ appears as the source of this equation, and in this paper, due to the singularity of $\alpha$, one of our goals is to prove all estimates without any information on $\alpha$. The $\nab_4\trch$ equation, while can be used, has $|\chih|^2$ as a source term. Since $\chih$ can only be estimated in $L^2_{\ub}$ using the mixed norm bounds, the integration in the $\ub$ direction does not give a smallness constant.

Nevertheless, a different structure can be exploited to overcome this challenge. We group the Ricci coefficients into $\Gamma_1$, $\Gamma_2$, $\Gamma_3$ and $\Gamma_4$ according to the equations and estimates that they satisfy. Let
$$\Gamma_1\in \{\trchb,\chibh,\eta,\omega\},\quad\Gamma_2=\etab,\quad\Gamma_3\in\{\chih,\omega\},\quad\Gamma_4=\trch.$$
They satisfy the following transport equations:
\begin{equation*}
\begin{split}
\nab_4\Gamma_1&=\Psic+(\Gamma_1+\Gamma_2+\Gamma_3+\Gamma_4)(\Gamma_1+\Gamma_2+\Gamma_3+\Gamma_4),\\
\nab_3\Gamma_2&=\Psic+(\Gamma_1+\Gamma_2)\Gamma_1,\\
\nab_3\Gamma_3&=\Psic+(\Gamma_1+\Gamma_2+\Gamma_3+\Gamma_4)(\Gamma_1+\Gamma_2),\\
\nab_4\Gamma_4&=(\Gamma_3+\Gamma_4)(\Gamma_3+\Gamma_4).
\end{split}
\end{equation*}
Notice that $\Gamma_3$ corresponds to the Ricci coefficients $\psi_H$ and can only be estimated in $L^2_{\ub}L^\infty_uL^\infty(S)$.

As before, the control of $\Psic$ is assumed in a bootstrap setting. The equations for $\Gamma_1$ and $\Gamma_2$ have similar structures as (\ref{Gammasch1}). Thus, we first estimate $\Gamma_1$, using the smallness constant provided by the integration in the $e_4$ (i.e., $\ub$) direction. We then control $\Gamma_2$ noting that with the bounds already obtained for $\Gamma_1$, the error term is linear in $\Gamma_2$. The equation for $\Gamma_3$ is similar to (\ref{Gammasch1}), except for an extra term containing $\Gamma_4$, which has not been estimated. Nevertheless, $\Gamma_3$ are the terms $\chih$ and $\omega$ which are only estimated in $L^2_{\ub}L^\infty_u L^\infty(S)$. Thus the error term containing $\Gamma_4$ only has to be controlled after taking the $L^2_{\ub}$ norm. This provides an extra smallness constant. Finally, while $\Gamma_4$ satisfies an equation in the $e_4$ (i.e., $\ub$) direction, $\Gamma_3\Gamma_3$ appears as a source. Recall that since $\Gamma_3$ can only be controlled in $L^2_{\ub}L^\infty_u L^\infty(S)$, this error term is only bounded in $L^1_{\ub}L^\infty_u L^\infty(S)$. In other words, integrating this equation does not give a smallness constant. Nevertheless, we can use the control for $\Gamma_3$ derived in the previous step! Thus we obtain the desired bounds for all the Ricci coefficients.

In a similar fashion, the energy estimates also have to be carried out in two steps. Recall from \eqref{EEsch} that in establishing the energy estimates, we need to control the error terms
$$||\Gamma\Psic\Psic||_{L^1_uL^1_{\ub}L^1(S)},$$
where $\Psic$ are the renormalized curvature components. The most difficult error terms are those containing $\beta$. This is because $\beta$ can only be controlled in $L^2(H)$. In order to control the error terms, the $L^2(H)$ norm of $\beta$ has to be integrated over the long $u$-interval and the estimates do not have a smallness constant. To deal with this problem, we first control $\betab$ in $L^2(\Hb)$ and $(\rhoc,\sigmac)$ in $L^2(H)$. While deriving these bounds, all the error terms are accompanied by a smallness constant $\epsilon^{\frac 12}$. We estimate $\beta$ after we obtain these bounds. The error terms that contain $\beta$ are\footnote{To be more precise, the term that actually appears is $||\chi\beta\nab\chib||_{L^1_uL^1_{\ub}L^1(S)}$ instead of $||\chi\beta\betab||_{L^1_uL^1_{\ub}L^1(S)}$. We note that using elliptic estimates, the control for $\nab\chib$ can be retrieved from the bound for $\betab$. We omit the technical details here and refer the readers to the content of the paper for details.}
$$||\chi\beta\betab||_{L^1_uL^1_{\ub}L^1(S)}$$
and
$$||\chib\beta\beta||_{L^1_uL^1_{\ub}L^1(S)}.$$
Since the $\betab$ has been controlled first, the first error term is sublinear. For the second term, it can be shown that the estimates for $\chib$ are independent of the bounds on the curvature and this term is therefore a linear term. It can thus be dealt with using Gronwall's inequality.

\subsubsection{Signature}\label{secsign}

In the proof of the a priori estimates, the structure of the Einstein equations plays a crucial role. It is thus useful to understand the structure of the equations in a more systematic fashion. Here, inspired by the work of Klainerman-Rodnianski \cite{KlRo} on the formation of trapped surfaces, we introduce a notion of signature that allows us to explain and tract that certain undesirable terms do not appear in a particular equation. Such a notion of signature is intimately tied to the scaling properties of the Einstein equations.

\subsubsection{Nonlinear Interaction of Impulsive Gravitational Waves}

As mentioned above, Theorem \ref{rdthmv1} implies the existence and uniqueness of solutions to the vacuum Einstein equations with characteristic initial data as in Theorem \ref{cgiwthm}. In the setting of the nonlinear interaction of impulsive gravitational waves in Theorem \ref{cgiwthm}, however, the initial data are more regular than the general initial data allowed in the assumptions of Theorem \ref{rdthmv1}. In particular, on each of the initial null hypersurfaces, the initial data are only singular on an embedded 2-sphere. This allows us to prove that the spacetime is smooth away from the null hypersurfaces emanating from the initial singularities. Moreover, $\alpha$ and $\alphab$ can be defined as measures with singular atoms supported on these null hypersurfaces.

We first note that standard local well-posedness theory and the results of \cite{LR} imply that the spacetime is smooth in $\{0\leq u < u_s\}\cup\{0\leq \ub < \ub_s\}$. Thus in order to show that the spacetime is smooth away from the null hypersurfaces $\{u= u_s\}$ and $\{\ub=\ub_s\}$, we only need to demonstrate the regularity of the spacetime in $\{u > u_s\}\cap\{\ub > \ub_s\}$. 

It turns out that using the a priori estimates derived in the proof of Theorem \ref{rdthmv1}, this can be shown by directly integrating the null structure equations. For example, while $\nab_4\chih$ has a delta singularity across $\ub=\ub_s$, we can prove that it is bounded for $\ub>\ub_s$. To this end, we consider
$$\nab_3\chih+\frac 1 2 \trchb \chih=\nab\widehat{\otimes} \eta+2\omegab \chih-\frac 12 \trch \chibh +\eta\widehat{\otimes} \eta .$$
Commute the equation with the $\nab_4$ derivative and substituting appropriate null structure equations, we get
$$\nab_3\nab_4\chih+\frac 1 2 \trchb \nab_4\chih-2\omegab \nab_4\chih=...$$
where ... denotes terms that have already been estimated in the proof the Theorem \ref{rdthmv1}. Thus by integrating this equation, we conclude that $\nab_4\chih$ inherits the regularity of the initial data and is bounded as long as $\ub\neq \ub_s$. This procedure can be carried out for all higher derivatives to show that the spacetime is smooth in the region $\{u>u_s\}\cap\{\ub>\ub_s\}$.

A surprising feature of this proof of smoothness of the resulting spacetime is that it does not require $\alpha$ and $\alphab$ to have delta singularities supported on the corresponding 2-spheres. In fact, if the initial data satisfy the assumptions of Theorem \ref{rdthmv1} and are more regular for $\ub> \tilde{\ub}$ on $H_0$ and $u>\tilde{u}$ on $\Hb_0$, then the spacetime can be proved to be more regular in $\{\ub> \tilde{\ub}\}\cap\{u>\tilde{u}\}$!

Returning to the interacting impulsive gravitational waves, we show that $\alpha$ and $\alphab$ can be defined as measures with delta singularities supported on $\Hb_{\ub_s}$ and $H_{u_s}$ respectively. To see this, consider the equations
$$\alpha=-\nab_4\chih-\trch\chih-2\omega\chih,$$
and
$$\alphab=-\nab_3\chibh-\trchb\,\chibh-2\omegab\chibh.$$
We can prove that $\chih$ (resp. $\chibh$) is smooth except across $\ub=\ub_s$ (resp. $u=u_s$) where it has a jump discontinuity. This implies that $\alpha$ and $\alphab$ are well-defined as measures and they have delta singularities supported on $\Hb_{\ub_s}$ and $H_{u_s}$ respectively.

\subsubsection{Formation of Trapped Surfaces}

Using the existence and uniqueness result in Theorem \ref{rdthmv1}, we construct a large class of spacetimes such that the initial data do not contain a trapped surface, and a trapped surface is formed in evolution. In particular, unlike in \cite{Chr}, \cite{KlRo} and \cite{KlRo1}, our construction does not require the initial data on $\Hb_0$ to be close to that of Minkowski space.

The challenge in this problem lies in the fact that in order to have a trapped surface, certain geometric quantities are necessarily large. Recall that in the setting of Christodoulou \cite{Chr} (see Figure 7), characteristic initial data were prescribed on $\Hb_0$ and a short region of $H_0$, where $0\leq \ub\leq \epsilon$.

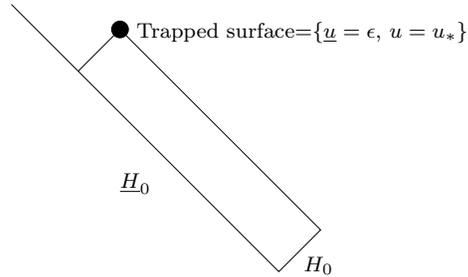
\begin{figure}[htbp]
\begin{center}
 
\input{TrappedSurfaceNew.pdf_t}
 
\caption{Formation of a Trapped Surface}
\end{center}
\end{figure}
In view of the equation
\begin{equation}\label{trapped0}
\nab_4\trch=-\frac 12(\trch)^2-|\chih|^2,
\end{equation}
in order that for some $u$, $\trch$ becomes negative after integrating in a $\ub$ length of $\epsilon$, $\chih$ has to be of size $\sim \epsilon^{-\frac 12}$ and consequently $\alpha$ has to be of size $\sim \epsilon^{-\frac 32}$. In the work of Christodoulou \cite{Chr}, and the later extensions of Klainerman-Rodnianski \cite{KlRo}, \cite{KlRo1}, this largeness of the geometric quantities is compensated by requiring smallness of initial data on $\Hb_0$.

To go beyond the requirement of Minkowski data on $\Hb_0$, we notice that while the $L^\infty_{\ub}L^\infty(S)$ norm of $\chih$ is large in terms of $\epsilon$, its $L^2_{\ub}L^\infty(S)$ is merely of size $\sim 1$ with respect to $\epsilon$. Therefore, Theorem \ref{rdthmv1} implies the existence and uniqueness of a spacetime solution for this type of initial data, even without any smallness assumptions on $\Hb_0$. Note in particular that the assumptions of Theorem \ref{rdthmv1} do not require any control of $\alpha$ for the initial data. It thus remains to show that one can find initial data which do not contain a trapped surface and such that a trapped surface is formed in evolution.

With the initial data that he imposed, Christodoulou identified a mechanism for the formation of a trapped surface \cite{Chr}. Recalling \eqref{trapped0}, for $\epsilon$ sufficiently small, if at $u=0$,
\begin{equation}\label{trapped1}
\trch(u=0,\ub=0,\vartheta)>\int_0^\epsilon |\chih|^2(u=0,\vartheta) d\ub,
\end{equation}
and at $u=u_*$,
\begin{equation}\label{trapped2}
\trch(u=u_*,\ub=0,\vartheta)<\int_0^\epsilon |\chih|^2(u=u_*,\vartheta) d\ub,
\end{equation}
then the initial data are free of trapped surfaces and the $2$-sphere given by $\{\ub=\epsilon,\,u=u_*\}$ is a trapped surface, i.e., a trapped surface forms in evolution.

To achieve (\ref{trapped1}) and (\ref{trapped2}), consider the equations
$$\nab_3\chih+\frac 1 2 \trchb \chih=\nab\widehat{\otimes} \eta+2\omegab \chih-\frac 12 \trch \chibh +\eta\widehat{\otimes} \eta,$$
and
$$\nab_3 \trch+\frac1 2 \trchb \trch =2\omegab \trch+2\rho- \chih\cdot\chibh+2\div \eta+2|\eta|^2.$$
Assuming the right hand side of these equations to be error terms, we get
\begin{equation}\label{trappedapprox1}
\nab_3|\chih|^2+ \trchb |\chih|^2 \approx 0
\end{equation}
and
\begin{equation}\label{trappedapprox2}
\nab_3 \trch+\frac1 2 \trchb \trch \approx 0,
\end{equation}
which imply
$$|\chih|^2(u,\ub,\vartheta)\approx|\chih|^2(u=0,\ub,\vartheta)\exp(-\int_0^u \trchb (u',\ub,\vartheta) du')$$
and
\begin{equation}\label{trchgr}
\trch(u,\ub=0,\vartheta)\approx \trch(u=0,\ub=0,\vartheta)\exp(-\frac 12 \int_0^u \trchb (u',\ub=0,\vartheta) du').
\end{equation}
Christodoulou showed that in the setting of \cite{Chr},
\begin{equation}\label{trappedapprox3}
\trchb(u,\ub,\vartheta)\approx \trchb(u,\ub=0,\vartheta),
\end{equation}
which implies that
\begin{equation}\label{chihgr}
|\chih|^2(u,\ub,\vartheta)\approx|\chih|^2(u=0,\ub,\vartheta)\exp(-\int_0^u \trchb (u',\ub=0,\vartheta) du').
\end{equation}
Comparing (\ref{trchgr}) and (\ref{chihgr}), since $\trchb<0$, $|\chih|^2$ has a larger amplification factor than $\trch$. Therefore, there is an open set of initial data such that a trapped surface is formed in evolution.

In our setting where we remove the smallness assumptions on the data on $\Hb_0$, the estimates derived in Theorem \ref{aprioriestimates} imply that (\ref{trappedapprox1}) and (\ref{trappedapprox3}) hold. Nevertheless, the approximation (\ref{trappedapprox2}) is not necessarily valid. Instead, we impose a condition (\ref{trappedsurfaceineq}) on $\Hb_0$ in Theorem \ref{trappedsurface} in order to guarantee that a trapped surface is formed in evolution. This condition guarantees that there is a choice of initial data on $H_0$ such that (\ref{trapped1}) and (\ref{trapped2}) hold in the resulting spacetime.\\

\noindent{\bf Acknowledgments:} The authors would like to thank Mihalis Dafermos for valuable discussions. J. Luk is supported by the NSF Postdoctoral Fellowship DMS-1204493. I. Rodnianski is supported by the NSF grant DMS-1001500 and the FRG grant DMS-1065710.

\section{Setting, Equations and Notations}
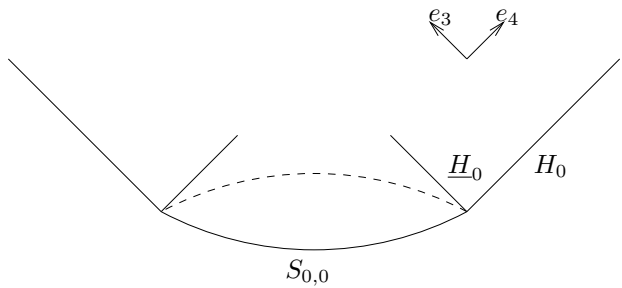
\begin{figure}[htbp]
\begin{center}
 
\input{frame.pdf_t}
 
\caption{The Basic Setup}
\end{center}
\end{figure}

Our setting is the characteristic initial value problem with data given on the two characteristic hypersurfaces $H_0$ and $\Hb_0$ intersecting at the sphere $S_{0,0}$. The spacetime will be a solution to the Einstein equations constructed in a neighborhood of $H_0$ and $\Hb_0$ containing $S_{0,0}$.

While we consider spacetimes with Riemann curvature tensors that are merely measures, it suffices to obtain a priori estimates for \emph{smooth} approximations of them. Once the a priori estimates are obtained, we can follow the limiting argument as in the case of one propagating impulsive gravitational wave \cite{LR} to obtain existence, uniqueness and regularity of the solutions. We refer the readers to \cite{LR} for details. In this paper, we will therefore focus on the proof of a priori estimates (see Theorem \ref{aprioriestimates}). To that end, we assume that we are given a smooth solution to the Einstein equations in a neighborhood of $H_0$ and $\Hb_0$. In particular, the double null foliation and the coordinate system introduced below are well-defined.

\subsection{Double Null Foliation}\label{secdnf}
For a spacetime in a neighborhood of $S_{0,0}$, we define a double null foliation as follows: Let $u$ and $\ub$ be solutions to the eikonal equation
$$(g^{-1})^{\mu\nu}\partial_\mu u\partial_\nu u=0,\quad (g^{-1})^{\mu\nu}\partial_\mu\ub\partial_\nu \ub=0,$$
satisfying the initial conditions $u=0$ on $H_0$ and $\ub=0$ on $\Hb_0$.
Let
$$L'^\mu=-2(g^{-1})^{\mu\nu}\partial_\nu u,\quad \Lb'^\mu=-2(g^{-1})^{\mu\nu}\partial_\nu \ub.$$ 
These are null and geodesic vector fields. Let
$$2\Omega^{-2}=-g(L',\Lb').$$
Define
$$e_3=\Omega\Lb'\mbox{, }e_4=\Omega L'$$
to be the normalized null pair such that 
$$g(e_3,e_4)=-2$$
and
$$\Lb=\Omega^2\Lb'\mbox{, }L=\Omega^2 L'$$
to be the so-called equivariant vector fields.

In the sequel, we will consider spacetime solutions to the vacuum Einstein equations in the gauge such that 
$$\Omega=1,\quad\mbox{on $H_0$ and $\Hb_0$}.$$

We denote the level sets of $u$ as $H_u$ and the level sets of $\ub$ and $\Hb_{\ub}$. By virtue of the eikonal equations, $H_u$ and $\Hb_{\ub}$ are null hypersurfaces. The sets defined by the intersections of the hypersurfaces $H_u$ and $\Hb_{\ub}$ are topologically 2-spheres, which we denote by $S_{u,\ub}$. Notice that the integral flows of $L$ and $\Lb$ respect the foliation $S_{u,\ub}$.

\subsection{The Coordinate System}\label{coordinates}
On a spacetime in a neighborhood of $S_{0,0}$, we define a coordinate system $(u,\ub,\th^1,\th^2)$ as follows:
On the sphere $S_{0,0}$, define a coordinate system $(\th^1,\th^2)$ such that on each coordinate patch the metric $\gamma$ is smooth, bounded and positive definite. Then we define the coordinates on the initial hypersurfaces $H_0$ and $\Hb_0$ by requiring $\th^A$ to be constant along the integral curves of $L$ and $\Lb$ respectively.
We now define the coordinate system in the spacetime in a neighborhood of $S_{0,0}$ by letting $u$ and $\ub$ to be solutions to the eikonal equations:
$$(g^{-1})^{\mu\nu}\partial_\mu u\partial_\nu u=0,\quad (g^{-1})^{\mu\nu}\partial_\mu\ub\partial_\nu \ub=0,$$
and define $\th^1, \th^2$ by
$$\Ls_L \th^A=0,$$ 
where $\Ls_L$ denotes the restriction of the Lie derivative to $TS_{u,\ub}$ (See \cite{Chr}, Chapter 1).
Relative to the coordinate system $(u,\ub,\th^1,\th^2)$, the null pair $e_3$ and $e_4$ can be expressed as
$$e_3=\Omega^{-1}\left(\frac{\partial}{\partial u}+b^A\frac{\partial}{\partial \th^A}\right), e_4=\Omega^{-1}\frac{\partial}{\partial \ub},$$
for some $b^A$ such that $b^A=0$ on $\Hb_0$, while the metric $g$ takes the form
$$g=-2\Omega^2(du\otimes d\ub+d\ub\otimes du)+\gamma_{AB}(d\th^A-b^Adu)\otimes (d\th^B-b^Bdu).$$ 

\subsection{Equations}\label{seceqn}
As indicated in the introduction, we will recast the Einstein equations as a system for Ricci coefficients and curvature components associated to a null frame $e_3$, $e_4$ defined above and an orthonormal frame ${e_1,e_2}$ tangent to the 2-spheres $S_{u,\ub}$. Using the indices $A,B$ to denote $1,2$, we recall the definition of the Ricci coefficients relative to the null fame:
 \begin{equation}
\begin{split}
&\chi_{AB}=g(D_A e_4,e_B),\, \,\, \quad \chib_{AB}=g(D_A e_3,e_B),\\
&\eta_A=-\frac 12 g(D_3 e_A,e_4),\quad \etab_A=-\frac 12 g(D_4 e_A,e_3),\\
&\omega=-\frac 14 g(D_4 e_3,e_4),\quad\,\,\, \omegab=-\frac 14 g(D_3 e_4,e_3),\\
&\zeta_A=\frac 1 2 g(D_A e_4,e_3)
\end{split}
\end{equation}
where $D_A=D_{e_{(A)}}$. We also recall the definition of the null curvature components,
 \begin{equation}
\begin{split}
\a_{AB}&=R(e_A, e_4, e_B, e_4),\quad \, \,\,   \ab_{AB}=R(e_A, e_3, e_B, e_3),\\
\b_A&= \frac 1 2 R(e_A,  e_4, e_3, e_4) ,\quad \bb_A =\frac 1 2 R(e_A,  e_3,  e_3, e_4),\\
\rho&=\frac 1 4 R(e_4,e_3, e_4,  e_3),\quad \sigma=\frac 1 4  \,^*R(e_4,e_3, e_4,  e_3).
\end{split}
\end{equation}
Here $\, ^*R$ denotes the Hodge dual of $R$.  We denote by $\nab$ the 
induced covariant derivative operator on $S_{u,\ub}$ and by $\nab_3$, $\nab_4$
the projections to $S_{u,\ub}$ of the covariant derivatives $D_3$, $D_4$ (see
precise definitions in \cite{KN}). 

Observe that,
\begin{equation}
\begin{split}
&\omega=-\frac 12 \nab_4 (\log\Omega),\qquad \omegab=-\frac 12 \nab_3 (\log\Omega),\\
&\eta_A=\zeta_A +\nab_A (\log\Omega),\quad \etab_A=-\zeta_A+\nab_A (\log\Omega).
\end{split}
\end{equation}

Define the following contractions of the tensor product $\phi^{(1)}$ and $\phi^{(2)}$ with respect to the metric $\gamma$:
$$\phi^{(1)}\cdot\phi^{(2)}:=(\gamma^{-1})^{AC}(\gamma^{-1})^{BD}\phi^{(1)}_{AB}\phi^{(2)}_{CD} \quad\mbox{for symmetric $2$-tensors $\phi^{(1)}_{AB}$, $\phi^{(2)}_{AB}$,}$$
$$\phi^{(1)}\cdot\phi^{(2)}:=(\gamma^{-1})^{AB}\phi^{(1)}_{A}\phi^{(2)}_{B} \quad\mbox{for $1$-forms $\phi^{(1)}_{A}$, $\phi^{(2)}_{A}$,}$$
$$(\phi^{(1)}\cdot\phi^{(2)})_A:=(\gamma^{-1})^{BC}\phi^{(1)}_{AB}\phi^{(2)}_{C} \quad\mbox{for a symmetric $2$-tensor $\phi^{(1)}_{AB}$ and a $1$-form $\phi^{(2)}_{A}$,}$$
$$(\phi^{(1)}\hot\phi^{(2)})_{AB}:=\phi^{(1)}_A\phi^{(2)}_B+\phi^{(1)}_B\phi^{(2)}_A-\gamma_{AB}(\phi^{(1)}\cdot\phi^{(2)}) \quad\mbox{for one forms $\phi^{(1)}_A$, $\phi^{(2)}_A$,}$$
$$\phi^{(1)}\wedge\phi^{(2)}:=\eps^{AB}(\gamma^{-1})^{CD}\phi^{(1)}_{AC}\phi^{(2)}_{BD}\quad\mbox{for symmetric two tensors $\phi^{(1)}_{AB}$, $\phi^{(2)}_{AB}$},$$
where $\eps$ is the volume form associated to the metric $\gamma$.
Define $^*$ of $1$-forms and symmetric $2$-tensors respectively as follows (note that on $1$-forms this is the Hodge dual on $S_{u,\ub}$):
\begin{align*}
^*\phi_A := & \gamma_{AC} \eps^{CB} \phi_B, \quad ^*\phi_{AB} := \gamma_{BD} \eps^{DC} \phi_{AC}.
\end{align*}
Define the operator $\nab\widehat{\otimes}$ on a $1$-form $\phi_{A}$ by
$$(\nab\widehat{\otimes}\phi)_{AB} :=  \nab_A \phi_B + \nab_B \phi_A - \gamma_{AB} \div \phi.$$
For totally symmetric tensors, the $\div$ and $\curl$ operators are defined by the formulas
$$(\div\phi)_{A_1...A_r}:=(\gamma^{-1})^{BC}\nabla_C\phi_{BA_1...A_r},$$
$$(\curl\phi)_{A_1...A_r}:=\eps^{BC}\nabla_B\phi_{CA_1...A_r}.$$
Define also the trace of totally symmetric tensors to be
$$(\mbox{tr}\phi)_{A_1...A_{r-1}}:=(\gamma^{-1})^{BC}\phi_{BCA_1...A_{r-1}}.$$

We separate the trace and traceless part of $\chi$ and $\chib$. Let $\chih$ and $\chibh$ be the traceless parts of $\chi$ and $\chib$ respectively. Then $\chi$ and $\chib$ satisfy the following null structure equations:
\begin{equation}
\label{null.str1}
\begin{split}
\nab_4 \trch+\frac 12 (\trch)^2&=-|\chih|^2-2\omega \trch,\\
\nab_4\chih+\trch \chih&=-2 \omega \chih-\alpha,\\
\nab_3 \trchb+\frac 12 (\trchb)^2&=-2\omegab \trchb-|\chibh|^2,\\
\nab_3\chibh + \trchb\,  \chibh&= -2\omegab \chibh -\alphab,\\
\nab_4 \trchb+\frac1 2 \trch \trchb &=2\omega \trchb +2\rho- \chih\cdot\chibh +2\div \etab +2|\etab|^2,\\
\nab_4\chibh +\frac 1 2 \trch \chibh&=\nab\widehat{\otimes} \etab+2\omega \chibh-\frac 12 \trchb \chih +\etab\widehat{\otimes} \etab,\\
\nab_3 \trch+\frac1 2 \trchb \trch &=2\omegab \trch+2\rho- \chih\cdot\chibh+2\div \eta+2|\eta|^2,\\
\nab_3\chih+\frac 1 2 \trchb \chih&=\nab\widehat{\otimes} \eta+2\omegab \chih-\frac 12 \trch \chibh +\eta\widehat{\otimes} \eta.
\end{split}
\end{equation}
The other Ricci coefficients satisfy the following null structure equations:
\begin{equation}
\label{null.str2}
\begin{split}
\nabla_4\eta&=-\chi\cdot(\eta-\etab)-\b,\\
\nabla_3\etab &=-\chib\cdot (\etab-\eta)+\bb,\\
\nabla_4\omegab&=2\omega\omegab-\eta\cdot\etab+\f 12|\eta|^2+\frac 12 \rho,\\
\nabla_3\omega&=2\omega\omegab-\eta\cdot\etab+\f 12|\etab|^2+\frac 12 \rho.
\end{split}
\end{equation}
The Ricci coefficients also satisfy the following constraint equations
\begin{equation}
\label{null.str3}
\begin{split}
\div\chih&=\frac 12 \nabla \trch - \frac 12 (\eta-\etab)\cdot (\chih -\frac 1 2 \trch) -\beta,\\
\div\chibh&=\frac 12 \nabla \trchb + \frac 12 (\eta-\etab)\cdot (\chibh-\frac 1 2   \trchb) +\betab,\\
\curl\eta &=-\curl\etab=\sigma +\frac 1 2\chibh \wedge\chih,\\
K&=-\rho+\frac 1 2 \chih\cdot\chibh-\frac 1 4 \trch \trchb.
\end{split}
\end{equation}
with $K$ the Gauss curvature of the spheres $S_{u,\ub}$.
The null curvature components satisfy the following null Bianchi equations:
\begin{equation}
\label{eq:null.Bianchi}
\begin{split}
&\nab_3\alpha+\frac 12 \trchb \alpha=\nabla\hot \beta+ 4\omegab\alpha-3(\chih\rho+^*\chih\sigma)+
(\zeta+4\eta)\hot\beta,\\
&\nab_4\beta+2\trch\beta = \div\alpha - 2\omega\beta +  (2\zeta+\etab)\cdot \alpha,\\
&\nab_3\beta+\trchb\beta=\nabla\rho + 2\omegab \beta +^*\nabla\sigma +2\chih\cdot\betab+3(\eta\rho+^*\eta\sigma),\\
&\nab_4\sigma+\frac 32\trch\sigma=-\div^*\beta+\frac 12\chibh\wedge\alpha-\zeta\wedge\beta-2\etab\wedge\beta,\\
&\nab_3\sigma+\frac 32\trchb\sigma=-\div ^*\betab-\frac 12\chih\wedge\alphab+\zeta\wedge\betab-2\eta\wedge\betab,\\
&\nab_4\rho+\frac 32\trch\rho=\div\beta-\frac 12\chibh\cdot\alpha+\zeta\cdot\beta+2\etab\cdot\beta,\\
&\nab_3\rho+\frac 32\trchb\rho=-\div\betab- \frac 12\chih\cdot\alphab+\zeta\cdot\betab-2\eta\cdot\betab,\\
&\nab_4\betab+\trch\betab=-\nabla\rho +^*\nabla\sigma+ 2\omega\betab +2\chibh\cdot\beta-3(\etab\rho-^*\etab\sigma),\\
&\nab_3\betab+2\trchb\betab=-\div\alphab-2\omegab\betab-(-2\zeta+\eta) \cdot\alphab,\\
&\nab_4\alphab+\frac 12 \trch\alphab=-\nabla\hot \betab+ 4\omega\alphab-3(\chibh\rho-^*\chibh\sigma)+
(\zeta-4\etab)\hot \betab.
\end{split}
\end{equation}

We now define the renormalized curvature components and rewrite the Bianchi equations in terms of them. Let
$$\rhoc=\rho-\frac 12 \chih \cdot\chibh,\quad \sigmac=\sigma+\frac 12 \chibh\wedge\chih.$$
The Bianchi equations expressed in terms of $\rhoc$ and $\sigmac$ instead of $\rho$ and $\sigma$ are as follows:
\begin{equation}
\label{eq:null.Bianchi2}
\begin{split}
\nab_3\beta+\trchb\beta=&\nabla\rhoc  +^*\nabla\sigmac + 2\omegab \beta+2\chih\cdot\betab+3(\eta\rhoc+^*\eta\sigmac)+\frac 1 2(\nabla(\chih\cdot\chibh)+^*\nabla(\chih\wedge\chibh))\\
&+\f 32(\eta\chih\cdot\chibh+^*\eta\chih\wedge\chibh),\\
\nab_4\sigmac+\frac 32\trch\sigmac=&-\div^*\beta-\zeta\wedge\beta-2\etab\wedge
\beta-\frac 12 \chih\wedge(\nab\widehat{\otimes}\etab)-\frac 12 \chih\wedge(\etab\widehat{\otimes}\etab),\\
\nab_4\rhoc+\frac 32\trch\rhoc=&\div\beta+\zeta\cdot\beta+2\etab\cdot\beta-\frac 12 \chih\cdot\nab\widehat{\otimes}\etab-\frac 12 \chih\cdot(\etab\widehat{\otimes}\etab)+\frac 14\trchb|\chih|^2,\\
\nab_3\sigmac+\frac 32\trchb\sigmac=&-\div ^*\betab+\zeta\wedge\betab-2\eta\wedge
\betab+\frac 12 \chibh\wedge(\nab\widehat{\otimes}\eta)+\frac 12 \chibh\wedge(\eta\widehat{\otimes}\eta),\\
\nab_3\rhoc+\frac 32\trchb\rhoc=&-\div\betab+\zeta\cdot\betab-2\eta\cdot\betab-\frac 12 \chibh\cdot\nab\widehat{\otimes}\eta-\frac 12 \chibh\cdot(\eta\widehat{\otimes}\eta)+\frac 14\trch|\chibh|^2,\\
\nab_4\betab+\trch\betab=&-\nabla\rhoc +^*\nabla\sigmac+ 2\omega\betab +2\chibh\cdot\beta-3(\etab\rhoc-^*\etab\sigmac)-\frac 1 2(\nabla(\chih\cdot\chibh)-^*\nabla(\chih\wedge\chibh))\\
&-\f 32(\eta\chih\cdot\chibh-^*\eta\chih\wedge\chibh).
\end{split}
\end{equation}
Notice that we have obtained a system for the renormalized curvature components in which the singular curvature components $\alpha$ and $\alphab$ do not appear.

In the sequel, we will use capital Latin letters $A\in \{1,2\}$ for indices on the spheres $S_{u,\ub}$ and Greek letters $\mu\in\{1,2,3,4\}$ for indices in the whole spacetime.

\subsection{Signature}\label{sec.signature}

In this subsection, we introduce the concept of signature. This will allow us to easily show that some undesirable terms are absent in various equations.

To every null curvature component $\alpha,\beta,\rho,\sigma,\betab,\alphab$, null Ricci coefficients $\chi,\zeta,\eta,\etab,\omega,\omegab$, and the metric components $\gamma,\Omega$, we assign a signature according to the following rule:
$$sgn(\phi)=1\cdot N_4(\phi)+(-1)\cdot N_3(\phi),$$
where $N_4(\phi), N_3(\phi)$ denote the number of times $e_4$, respectively $e_3$, which appears in the definition of $\phi$.
 Thus,
\begin{eqnarray*}
 sgn (\beta)=1, \quad  sgn(\rho, \sigma)=0,\quad  sgn(\betab)=-1.
\end{eqnarray*}
Also,
\begin{eqnarray*}
sgn(\chi)=sgn( \omega)=1,\quad sgn(\zeta, \eta, \etab)=sgn(\gamma,\Omega)=0,\quad sgn(\chib)=sgn(\omegab)= -1.
\end{eqnarray*}

We use the notation $\Psi^{(s)}$ and $\Gamma^{(s)}$ to denote the renormalized curvature component and Ricci coefficient respectively with signature $s$. Then all the equations conserve signature in the following sense:
The null structure equations are all in the form
$$\nabla_4\Gamma^{(s)}=\Psi^{(s+1)}+\sum_{s_1+s_2=s+1}\Gamma^{(s_1)}\cdot\Gamma^{(s_2)},$$
$$\nabla_3\Gamma^{(s)}=\Psi^{(s-1)}+\sum_{s_1+s_2=s-1}\Gamma^{(s_1)}\cdot\Gamma^{(s_2)}.$$
and the null Bianchi equations are of the form
$$\nabla_4\Psi^{(s)}=\nabla\Psi^{(s+1)}+\sum_{s_1+s_2=s+1}(\Gamma^{(s_1)}\cdot\Psi^{(s_2)}+\Gamma^{(s_1)}\cdot\nabla\Gamma^{(s_2)}),$$
$$\nabla_3\Psi^{(s)}=\nabla\Psi^{(s-1)}+\sum_{s_1+s_2=s-1}(\Gamma^{(s_1)}\cdot\Psi^{(s_2)}+\Gamma^{(s_1)}\cdot\nabla\Gamma^{(s_2)}).$$

\subsection{Schematic Notation}\label{secsche}
We introduce a schematic notation as follow: Let $\phi$ denote an arbitrary tensorfield. For the Ricci coefficients, we use the notation
\begin{equation}\label{schpsi}
\psi\in\{\trch,\trchb,\eta,\etab\},\quad\psi_H\in\{\chih,\omega\},\quad\psi_{\Hb}\in\{\chibh,\omegab\}.
\end{equation}
Notice that $\psi_H$ has signature $1$ and $\psi_{\Hb}$ has signature $-1$. Unless otherwise stated, we will not use the schematic notation for the renormalized curvature components but will write them explicitly.

We will simply write $\psi\psi$ (or $\psi\psi_H$, $\psi\beta$, etc.) to denote arbitrary contractions with respect to the metric $\gamma$. $\nab$ will be used to denote an arbitrary angular covariant derivative. The use of the schematic notation is reserved for the cases when the precise nature of the contraction is not important to the argument. Moreover, when using this schematic notation, we will neglect all constant factors.

We will use brackets to denote terms with any one of the components in the brackets. For example, $\psi(\rhoc,\sigmac)$ is used to denote either $\psi\rhoc$ or $\psi\sigmac$.

The expression $\nab^i\psi^j$ will be used to denote angular derivatives of products of Ricci coefficients. More precisely, $\nab^i\psi^j$ denotes the sum of all terms which are products of $j$ factors, with each factor being $\nab^{i_k}\psi$ and that the sum of all $i_k$'s being $i$, i.e., 
$$\nab^i\psi^j=\displaystyle\sum_{i_1+i_2+...+i_j}\underbrace{\nab^{i_1}\psi\nab^{i_2}\psi...\nab^{i_j}\psi}_\text{j factors}.$$

Using these notations, we write all the equations from Section \ref{seceqn} in the schematic form. The structure of the equations can be read off directly from Section \ref{seceqn}. On the other hand, we notice that the structure for most of the equations also follows from signature considerations as indicated in Section \ref{secsign}. We will later point out places where we need to use an additional structure of the equations that goes beyond signature considerations.

We first write down the null structure equations (\ref{null.str1}) and (\ref{null.str2}) in schematic form. Here, we do not write down the two equations that involve the singular curvature components $\alpha$ or $\alphab$.
\begin{equation}
\label{null.str1.sch}
\begin{split}
\nab_4 \trch&=\chih\chih+\psi(\psi+\psi_H),\\
\nab_3 \trchb&=\chibh\,\chibh+\psi(\psi+\psi_{\Hb}),\\
\nab_4 \trchb&=\rhoc +\nab \etab +\psi(\psi+\psi_H),\\
\nab_3 \trch&=\rhoc+\nab\eta+\psi(\psi+\psi_{\Hb}),\\
\nabla_4\eta&=\b+\psi(\psi+\psi_H),\\
\nabla_3\etab &=\bb+(\eta+\etab)(\trchb+\psi_{\Hb}),\\
\nab_4\chibh&=\rhoc+\nab\etab+\psi(\psi+\psi_{H})+\psi_{\Hb}(\trch+\psi_H),\\
\nab_3\chih&=\rhoc+\nab\eta+\psi(\psi+\psi_{\Hb})+\psi_H(\trchb+\psi_{\Hb}).
\end{split}
\end{equation}
Except for the equation $\nab_4\trchb$ and $\nab_3\trch$, the structure of the nonlinear terms in the other equations follow from signature considerations.\footnote{Notice that we have written a more precise version of schematic equation for $\nab_3\etab$ compared to $\nab_4\eta$. This will be useful in the proof since when integrating in the $u$ direction using the $\nab_3$ equation, we will not have a smallness in the length scale and we need to use the extra structure of the equation.} We now write the constraint equations (\ref{null.str3}) in schematic form:
\begin{equation}
\label{null.str3.sch}
\begin{split}
\div\chih&=\frac 12 \nabla \trch + \psi(\trch+\chih) -\beta,\\
\div\chibh&=\frac 12 \nabla \trchb + \psi(\trchb+\chibh) +\betab,\\
\curl\eta &=-\curl\etab=\sigmac, \\
K&=-\rhoc+\psi\psi.
\end{split}
\end{equation}
We now write down the Bianchi equations (\ref{eq:null.Bianchi2}) in schematic form, substituting the Codazzi equations in (\ref{null.str3.sch}) for some $\beta$ and $\betab$. In these equations, the left hand side is written with exact constants while the right hand side is written only schematically.
\begin{equation*}\label{eq:null.Bianchi2.sch}
\begin{split}
\nab_3\beta-\nabla\rhoc -^*\nabla\sigmac=& \psi(\rhoc,\sigmac)+\psi^{i_1}\nab^{i_2}(\psi_H+\trch)\nab^{i_3}(\psi_{\Hb}+\trchb),\\
\nab_4\sigmac+\div ^*\beta=&\psi\sigmac+\sum_{i_1+i_2+i_3\leq 1}\psi^{i_1}\nab^{i_2}\psi\nab^{i_3}\psi_{H}+\psi\chih\chih,\\
\nab_4\rhoc+\div\beta=&\psi\rhoc+\sum_{i_1+i_2+i_3\leq 1}\psi^{i_1}\nab^{i_2}\psi\nab^{i_3}\psi_{H}+\psi\chih\chih,\\
\nab_3\sigmac+\div ^*\betab=&\psi\sigmac+\sum_{i_1+i_2+i_3\leq 1}\psi^{i_1}\nab^{i_2}\psi\nab^{i_3}\psi_{\Hb}+\psi\chibh\,\chibh,\\
\nab_3\rhoc+\div\betab=&\psi\rhoc+\sum_{i_1+i_2+i_3\leq 1}\psi^{i_1}\nab^{i_2}\psi\nab^{i_3}\psi_{\Hb}+\psi\chibh\,\chibh,\\
\nab_4\betab+\nabla\rhoc -^*\nabla\sigmac=& \psi(\rhoc,\sigmac)+\psi^{i_1}\nab^{i_2}(\psi_H+\trch)\nab^{i_3}(\psi_{\Hb}+\trchb).\\
\end{split}
\end{equation*}
It is important in the sequel that in the equations for $\nab_4(\rhoc,\sigmac)$ (resp. $\nab_3(\rhoc,\sigmac)$), $\psi_{\Hb}$ (resp. $\psi_H$) does not appear. This does not follow from signature considerations alone since in principle the conservation of signature would allow a term $\psi_{\Hb}\psi_H\psi_H$ (resp. $\psi_H\psi_{\Hb}\psi_{\Hb}$). The fact that these terms do not appear can be observed directly in the equation (\ref{eq:null.Bianchi2}). 

\subsection{Integration}
Let $U$ be a coordinate patch on $S_{0,0}$ and define $U_{u,0}$ to be a coordinate patch on $S_{u,0}$ given by the one-parameter diffeomorphism generated by $\Lb$. Define $U_{u,\ub}$ to be the image of $U_{u,0}$ under the one-parameter diffeomorphism generated by $L$. Define also $D_U=\bigcup_{0\leq u\leq I ,0\leq \ub\leq \epsilon} U_{u,\ub}$. Let $\{p_U\}$ be a partition of unity such that $p_U$ is supported in $D_U$. Given a function $\phi$, the integration on $S_{u,\ub}$ is given by the formula:
$$\int_{S_{u,\ub}} \phi :=\sum_U \int_{-\infty}^{\infty}\int_{-\infty}^{\infty}\phi p_U\sqrt{\det\gamma}d\th^1 d\th^2.$$
Let $D_{u',\ub'}$ by the region $0\leq u\leq u'$, $0\leq \ub\leq \ub'$. The integration on $D_{u,\ub}$ is given by the formula
\begin{equation*}
\begin{split}
\int_{D_{u,\ub}} \phi :=&\sum_U \int_0^u\int_0^{\ub}\int_{-\infty}^{\infty}\int_{-\infty}^{\infty}\phi p_U\sqrt{-\det g}d\th^1 d\th^2d\ub du\\
=&2\sum_U \int_0^u\int_0^{\ub}\int_{-\infty}^{\infty}\int_{-\infty}^{\infty}\phi p_U\Omega^2\sqrt{\det \gamma}d\th^1 d\th^2d\ub du.
\end{split}
\end{equation*}
Since there are no canonical volume forms on $H_u$ and $\Hb_{\ub}$, we define integration by
$$\int_{H_{u}} \phi :=\sum_U \int_0^{\epsilon}\int_{-\infty}^{\infty}\int_{-\infty}^{\infty}\phi2 p_U\Omega\sqrt{\det\gamma}d\th^1 d\th^2d\ub,$$
and
$$\int_{H_{\ub}} \phi :=\sum_U \int_0^\epsilon\int_{-\infty}^{\infty}\int_{-\infty}^{\infty}\phi2p_U\Omega\sqrt{\det\gamma}d\th^1 d\th^2du.$$

With these notions of integration, we can define the norms that we will use. Let $\phi$ be an arbitrary tensorfield. For $1\leq p<\infty$, define
$$||\phi||_{L^p(S_{u,\ub})}^p:=\int_{S_{u,\ub}} <\phi,\phi>_\gamma^{p/2},$$
$$||\phi||_{L^p(H_u)}^p:=\int_{H_{u}} <\phi,\phi>_\gamma^{p/2},$$
$$||\phi||_{L^p(\Hb_{\ub})}^p:=\int_{\Hb_{\ub}} <\phi,\phi>_\gamma^{p/2}.$$
Define also the $L^\infty$ norm by
$$||\phi||_{L^\infty(S_{u,\ub})}:=\sup_{\th\in S_{u,\ub}} <\phi,\phi>_\gamma^{1/2}(\th).$$
We will also use mixed norms defined by
$$||\phi||_{L^2_{\ub}L^\infty_u L^p(S)}=\left(\int_0^{\ub_*}(\sup_{u\in [0,u_*]}||\phi||_{L^p(S_{u,\ub})})^2d\ub\right)^{\frac{1}{2}},$$
$$||\phi||_{L^2_{u}L^\infty_{\ub} L^p(S)}=\left(\int_0^{u_*}(\sup_{\ub\in [0,\epsilon]}||\nabla^i\phi||_{L^p(S_{u,\ub})})^2du\right)^{\frac{1}{2}}.$$
Note that $L^\infty L^p$ is taken before taking $L^2$. In the sequel, we will frequently use
$$||\cdot||_{L^\infty_{\ub}L^2_{u} L^p(S)}\leq||\cdot||_{L^2_{u}L^\infty_{\ub} L^p(S)}.$$  

With the above definition, $||\phi||_{L^2_uL^2(S_{u,\ub})}$ and $||\phi||_{L^2(\Hb_{\ub})}$ differ by a factor of $\Omega$.  Nevertheless, in view of Proposition \ref{Omega}, these norms are equivalent up to a factor of $2$.

\subsection{Norms}
We now define the norms that we will work with. Let
$$\mathcal R=\sum_{i\leq 2}\left(\sum_{\Psi\in\{\beta,\rhoc,\sigmac\}}\sup_{u}||\nabla^i\Psi||_{L^2(H_u)}+ \sum_{\Psi\in\{\rhoc,\sigmac,\betab\}}\sup_{\ub}||\nabla^i\Psi||_{L^2(\Hb_{\ub})}\right),$$
$$\mathcal R(S)=\sum_{i\leq 1}(\sup_{u,\ub}||\nab^i(\rhoc,\sigmac,K)||_{L^2(S_{u,\ub})}+||\nab^i\betab||_{L^2_{u}L^\infty_{\ub}L^3(S)}),$$
$$\mathcal O_{i,p}=\sup_{u,\ub}||\nabla^i(\trch,\eta,\etab,\trchb)||_{L^p(S_{u,\ub})}+||\nabla^i(\chih,\omega)||_{L^2_{\ub}L^\infty_u L^p(S)}+||\nabla^i(\chibh,\omegab)||_{L^2_{u}L^\infty_{\ub} L^p(S)},$$
\begin{equation*}
\begin{split}
\tilde{\mathcal O}_{3,2}=&||\nabla^3(\trch,\trchb)||_{L^\infty_uL^\infty_{\ub}L^2(S)}+||\nabla^3(\eta,\etab)||_{L^\infty_u L^2_{\ub}L^2(S)}+||\nabla^3(\eta,\etab)||_{L^\infty_{\ub} L^2_uL^2(S)} \\
&+||\nabla^3(\chih,\omega,\omega^{\dagger})||_{L^\infty_uL^2_{\ub} L^2(S)}+||\nabla^3(\chibh,\omegab,\omegab^{\dagger})||_{L^\infty_{\ub}L^2_{u} L^2(S)},
\end{split}
\end{equation*}
where $\omega^{\dagger}$ and $\omegab^{\dagger}$ are defined to be the solutions to
$$\nab_3\omega^{\dagger}=\frac 12 \sigmac,\quad\nab_4\omegab^{\dagger}=\frac 12 \sigmac$$
with zero data\footnote{i.e., $\omega^{\dagger}=0$ on $H_0$ and $\omegab^{\dagger}=0$ on $\Hb_0$.} and $\mu$, $\mub$, $\kappa$, $\kappab$ are defined by
$$\mu:=-\div\eta-\rhoc,\quad\mub:=-\div\etab-\rhoc,\quad\kappa:=\nab\omega+^*\nab\omega^{\dagger}-\frac 12\beta,\quad\kappab:=-\nab\omegab+^*\nab\omegab^{\dagger}-\frac 12\betab.$$

Moreover, we will use the notation $\mathcal O_{i,p}[\trch]$ (and $\mathcal R(S)[\beta]$, etc) to denote the part of the $\mathcal O$ norm that depends on $\trch$, i.e., $\sup_{u,\ub}\|\nab^i\trch\|_{L^p(S_{u,\ub})}$.

Recall from (\ref{schpsi}) that we use the schematic notation $\psi\in\{\trch,\eta,\etab,\trchb\}$, $\psi_H\in\{\chih,\omega\}$ and $\psi_{\Hb}\in\{\chibh,\omegab\}$. The choice of this notation is due to the fact that they obey different estimates.

For the norms of the third derivatives of the Ricci coefficients, i.e., the $\tilde{\mathcal O}_{3,2}$ norms, notice that $\nabla^{3}\trch$ and $\nabla^{3}\trchb$ obey the same type of estimates as for lower order derivatives. $\nabla^{3}(\eta,\etab)$ can no longer be controlled on a 2-sphere, but it obeys estimates on \emph{either} null hypersurface. $\nabla^{3}\psi_H$ (resp. $\nabla^{3}\psi_{\Hb}$) satisfies similar estimates as before, but at this level of derivatives, we have to take $L^2_{\ub}$ (resp. $L^2_u$) before $L^\infty_u$ (resp. $L^\infty_{\ub}$). 

\noindent We write
$$\mathcal O:=\mathcal O_{0,\infty}+\displaystyle\sum_{i\leq 1}\mathcal O_{i,4}+\sum_{i\leq 2}\mathcal O_{i,2}.$$

\section{Statement of Main Theorem}\label{maintheorem}
With the notations introduced in the previous section, we formulate a more precise version of Theorem \ref{rdthmv1}, which we call Theorem \ref{rdthmv2}. As noted before, since the proof for the existence and uniqueness of solutions in $\{0\leq u\leq \epsilon\}\cap\{0\leq \ub\leq I_1\}$ is the same as that in $\{0\leq \ub\leq \epsilon\}\cap\{0\leq u\leq I_2\}$, we will focus on the latter case.
\begin{theorem}\label{rdthmv2}
Suppose the initial data set for the characteristic initial value problem is given on $H_0$ for $0\leq \ub\leq \ub_*$ and on $\Hb_0$ for $0\leq u\leq u_*\leq I$ such that
$$c\leq |\det\gamma \restriction_{S_{u,0}} |, |\det\gamma \restriction_{S_{0,\ub}} |\leq C,$$
$$\sum_{i\leq 3}\left(|(\frac{\partial}{\partial\th})^i\gamma \restriction_{S_{u,0}}|+|(\frac{\partial}{\partial\th})^i\gamma \restriction_{S_{0,\ub}}|\right)\leq C,$$
\begin{equation*}
\begin{split}
\mathcal O_0:=& \sum_{i\leq 3} \left(||\nab^i\psi||_{L^\infty_uL^2(S_{u,0})}+||\nab^i\psi||_{L^\infty_{\ub}L^2(S_{0,\ub})}+||\nabla^i\psi_H||_{L^2(H_{0})}+||\nabla^i\psi_{\Hb}||_{L^2(\Hb_{0})}\right)\\
\leq &C,\\
\mathcal R_0:=&\sum_{i\leq 2}\left(||\nab^i\beta||_{L^2(H_0)}+||\nab^i\betab||_{L^2(\Hb_0)}+\sum_{\Psi\in\{\rhoc,\sigmac\}}(||\nab^i\Psi||_{L^\infty_uL^2(S_{u,0})}+||\nab^i\Psi||_{L^\infty_{\ub}L^2(S_{0,\ub})})\right)\\
\leq &C.
\end{split}
\end{equation*}
Then there exists $\epsilon>0$ sufficiently small depending only on $C$, $c$ and $I$ such that if $\ub_*\leq \epsilon$, there exists a spacetime $(\mathcal M,g)$ that solves the characteristic initial value problem to the vacuum Einstein equations in the region $0\leq u\leq u_*$, $0\leq \ub\leq \ub_*$. Geometrically, this is the region to the future of the initial hypersurfaces $H_0$ and $\Hb_0$ which is bounded in the future by the incoming null hypersurface emanating from $S_{0,\ub_*}$ and the outgoing null hypersurface emanating from $S_{u_*,0}$. Associated to the spacetime $(\mathcal M,g)$, there exists a system of null coordinates $(u,\ub, \th^1, \th^2)$ in which the metric is continuous and takes the form
$$g=-2\Omega^2(du\otimes d\ub+d\ub\otimes du)+ \gamma_{AB}(d\th^A-b^Adu)\otimes(d\th^B-b^Bdu).$$
In addition, given a sequence of smooth initial data sets such that the metrics $\gamma_n$ approaches $\gamma$ in $L^\infty_uW^{3,\infty}(S_{u,0})\cap L^\infty_{\ub}W^{3,\infty}(S_{0,\ub})$, the Ricci coefficients $(\psi,\psi_H,\psi_{\Hb})_n$ approaches $(\psi,\psi_H,\psi_{\Hb})$ in the norm\footnote{Here, we take the norms and the connection coefficients on the spheres $S_{0,\ub}$ and $S_{u,0}$ to be defined with respect to $\gamma$.} given by $\mathcal O_0$ and the renormalized curvature components $(\beta,\rhoc,\sigmac,\betab)_n$ approaches $(\beta,\rhoc,\sigmac,\betab)$ in the norm $\mathcal R_0$, this sequence of initial data gives rise to a sequence of smooth spacetimes which approaches $(\mathcal M, g)$ in $C^0$. $(\mathcal M, g)$ is also the unique spacetime solving the characteristic initial value problem among all such $C^0$ limits of smooth solutions. Moreover\footnote{Here, we use $g$ to denote any components of the metric in double null coordinates, i.e., the components, $b^A$, $\gamma_{AB}$ and $\Omega$.}, 
$$\frac{\partial}{\partial \th}g \in C^0_u C^0_{\ub} L^4(S),$$
$$\frac{\partial^2}{\partial \th^2}g\in C^0_u C^0_{\ub} L^2(S),$$
$$\frac{\partial^2}{\partial \th \partial u}g,\frac{\partial^2}{\partial u^2}b^A\in L^2_{u} L^\infty_{\ub}  L^4(S).$$
$$\frac{\partial}{\partial u}g \in L^2_u L^\infty_{\ub}  L^\infty(S),$$
$$ \frac{\partial}{\partial u}((\gamma^{-1})^{AB}\frac{\partial}{\partial u}(\gamma)_{AB}) \in L^1_u L^\infty_{\ub}  L^\infty(S),$$
$$\frac{\partial^2}{\partial \th \partial \ub}g,\frac{\partial^2}{\partial\ub^2}b^A \in L^2_{\ub} L^\infty_u  L^4(S).$$
$$\frac{\partial}{\partial \ub}g \in L^2_{\ub} L^\infty_u  L^\infty(S),$$
$$ \frac{\partial}{\partial\ub}((\gamma^{-1})^{AB}\frac{\partial}{\partial\ub}(\gamma)_{AB}) \in L^1_{\ub} L^\infty_u  L^\infty(S),$$
$$\frac{\partial^2}{\partial u\partial\ub}g\in L^2_{u} L^2_{\ub}  L^4(S).$$
In the $(u,\ub,\th^1,\th^2)$ coordinates, the Einstein equations are satisfied in $L^1_uL^1_{\ub}L^1(S)$. Furthermore, the higher angular\footnote{i.e., in the $\frac{\partial}{\partial \th^A}$ directions.} differentiability in the data results in higher angular differentiability of $(\mathcal M,g)$.
\end{theorem}
In the remainder of this paper, we will prove the a priori estimates needed to establish Theorem \ref{rdthmv2} (see Theorem \ref{aprioriestimates}). The existence, uniqueness and regularity statements in Theorem \ref{rdthmv2} follow from the a priori estimates and an approximation argument as in \cite{LR}. Moreover, as in \cite{LR}, it suffices to prove a priori estimates for smooth solutions. We refer the readers to \cite{LR} for details. In the subsequent sections, we will prove the following theorem on the a priori estimates: 
\begin{theorem}\label{aprioriestimates}
Suppose a smooth initial data set for the characteristic initial value problem is given on $H_0$ for $0\leq \ub\leq \ub_*$ and on $\Hb_0$ for $0\leq u\leq u_*$ such that
$$c\leq |\det\gamma \restriction_{S_{u,0}} |\leq C,\quad \sum_{i\leq 3}|(\frac{\partial}{\partial\th})^i\gamma \restriction_{S_{u,0}}|\leq C,$$
\begin{equation*}
\begin{split}
\mathcal O_0:=& \sum_{i\leq 3} \left(||\nab^i\psi||_{L^\infty_uL^2(S_{u,0})}+||\nab^i\psi||_{L^\infty_{\ub}L^2(S_{0,\ub})}+||\nabla^i\psi_H||_{L^2(H_{0})}+||\nabla^i\psi_{\Hb}||_{L^2(\Hb_{0})}\right)\\
\leq &C,\\
\mathcal R_0:=&\sum_{i\leq 2}\left(||\nab^i\beta||_{L^2(H_0)}+||\nab^i\betab||_{L^2(\Hb_0)}+\sum_{\Psi\in\{\rhoc,\sigmac\}}(||\nab^i\Psi||_{L^\infty_uL^2(S_{u,0})}+||\nab^i\Psi||_{L^\infty_{\ub}L^2(S_{0,\ub})})\right)\\
\leq &C.
\end{split}
\end{equation*}
Then, there exists $\epsilon$ depending only on $C$, $c$ and $I$ such that if $u_*\leq I$ and $\ub_*\leq \epsilon$, a smooth solution to the vacuum Einstein equations in the region $0\leq u\leq u_*$, $0\leq \ub\leq \ub_*$ has the following norms bounded above by a constant $C'$ depending only on $C$, $c$ and $I$:
$$\mathcal O, \tilde{\mathcal O}_{3,2}, \mathcal R < C'.$$
\end{theorem}

\subsection{Structure of the Proof}

We briefly outline the proof of Theorem \ref{aprioriestimates}:\\

\noindent {\bf STEP 0}: Assuming that $\mathcal O_{0,\infty}$ and $\mathcal O_{1,4}$ are controlled, we prove the bounds on the metric components, from which we derive preliminary estimates such as the Sobolev embedding theorem and the estimates for transport equations. (Section \ref{secbasic}).\\

\noindent {\bf STEP 1}: Assuming $\mathcal R<\infty$, $\mathcal R(S)<\infty$ and $\tilde{\mathcal O}_{3,2}<\infty$, we prove that $\mathcal O\leq C(\mathcal O_0,\mathcal R(S))$. (Sections \ref{secRicciL4}, \ref{secRicciL2})\\

\noindent {\bf STEP 2}: Assuming $\mathcal R<\infty$ and $\tilde{\mathcal O}_{3,2}<\infty$, we show that $\mathcal R(S)\leq C(\mathcal R_0)$. (Section \ref{secRS}) Together with Step 1 this implies $\mathcal O\leq C(\mathcal O_0,\mathcal R_0)$.\\

\noindent {\bf STEP 3}: Assuming $\mathcal R<\infty$, we establish that $\tilde{\mathcal O}_{3,2}\leq C(\mathcal O_0)(1+\mathcal R)$, i.e., $\tilde{\mathcal O}_{3,2}$ grows at most linearly with $\mathcal R$, with a constant depending only on the initial data. (Section \ref{secRicci32})\\

\noindent {\bf STEP 4}: Using the previous steps, we obtain the estimate $\mathcal R\leq C(\mathcal O_0,\mathcal R_0)$, thus finishing the proof of Theorem \ref{aprioriestimates}. (Section \ref{seccurv})

\section{The Preliminary Estimates}\label{secbasic}

All estimates in this section will be proved under the following bootstrap assumption:

\begin{equation}\tag{A1}\label{BA1}
\mathcal O_{0,\infty}+\sum_{i\leq 1}\mathcal O_{i,4}\leq \Delta_0
\end{equation}
where $\Delta_0$ is a positive constant to be chosen later.

\subsection{Estimates for Metric Components}\label{metric}
We first show that we can control $\Omega$ under the bootstrap assumption (\ref{BA1}):
\begin{proposition}\label{Omega}
There exists $\epsilon_0=\epsilon_0(\Delta_0)$ such that for every $\epsilon\leq\epsilon_0$,
$$\frac 12\leq \Omega\leq 2.$$
\end{proposition}
\begin{proof}
Consider the equation
\begin{equation}\label{Omegatransport}
 \omega=-\frac{1}{2}\nabla_4\log\Omega=\frac{1}{2}\Omega\nabla_4\Omega^{-1}=\frac{1}{2}\frac{\partial}{\partial \ub}\Omega^{-1}.
\end{equation}
Notice that both $\omega$ and $\Omega$ are scalars and therefore the $L^\infty$ norm is independent of the metric. We can integrate equation (\ref{Omegatransport}) using the fact that $\Omega^{-1}=1$ on $\Hb_0$ to obtain
$$||\Omega^{-1}-1||_{L^\infty(S_{u,\ub})}\leq C\int_0^{\ub}||\omega||_{L^\infty(S_{u,\ub'})}d\ub'\leq C\epsilon^{\frac 12}||\omega||_{L^\infty_uL^2_{\ub}L^\infty(S)}\leq C\Delta_0\epsilon^{\frac 12}.$$
This implies both the upper and lower bounds for $\Omega$ for sufficiently small $\epsilon$.
\end{proof}

We then show that we can control $\gamma$ under the bootstrap assumption (\ref{BA1}):
\begin{proposition}\label{gamma}
Consider a coordinate patch $U$ on $S_{0,0}$. Recall that $U_{u,0}$ is defined to be a coordinate patch on $S_{u,0}$ given by the one-parameter diffeomorphism generated by $\Lb$ and $U_{u,\ub}$ is defined to be to be the image of $U_{u,0}$ under the one-parameter diffeomorphism generated by $L$. Recall also that $D_U=\bigcup_{0\leq u\leq I ,0\leq \ub\leq \epsilon} U_{u,\ub}$. For $\epsilon$ small enough depending on initial data and $\Delta_0$, there exists $C$ and $c$ depending only on initial data such that the following pointwise bounds for $\gamma$ hold  in $D_U$:
$$c\leq \det\gamma\leq C. $$
Moreover, in $D_U$,
$$|\gamma_{AB}|,|(\gamma^{-1})^{AB}|\leq C.$$
\end{proposition}
\begin{proof}
The first variation formula states that
$$\Ls_L\gamma=2\Omega\chi.$$
In coordinates, this means
$$\frac{\partial}{\partial \ub}\gamma_{AB}=2\Omega\chi_{AB}.$$
From this we derive that 
$$\frac{\partial}{\partial \ub}\log(\det\gamma)=\Omega\trch.$$
Define $\gamma_0(u,\ub,\th^1,\th^2)=\gamma(u,0,\th^1,\th^2)$. Then
\begin{equation}\label{detgaper}
|\det\gamma-\det(\gamma_0)|\leq C\int_0^{\ub}|\trch|d\ub'\leq C\Delta_0\epsilon.
\end{equation}
This implies that the $\det \gamma$ is bounded above and below. Let $\Lambda$ be the larger eigenvalue of $\gamma$. Clearly,
\begin{equation}\label{La}
\Lambda\leq C\sup_{A,B=1,2}\gamma_{AB},
\end{equation}
and 
$$\sum_{A,B=1,2}|\chi_{AB}|\leq C\Lambda ||\chi||_{L^\infty(S_{u,\ub})}.$$
Then
$$|\gamma_{AB}-(\gamma_0)_{AB}|\leq C\int_0^{\ub}|\chi_{AB}|d\ub'\leq C\Lambda\Delta_0\epsilon^{\frac 12}.$$
Using the upper bound (\ref{La}), we thus obtain the upper bound for $|\gamma_{AB}|$. The upper bound for $|(\gamma^{-1})^{AB}|$ follows from the upper bound for $|\gamma_{AB}|$ and the lower bound for $\det\gamma$.
\end{proof}

A consequence of the previous proposition is an estimate on the surface area of the two sphere $S_{u,\ub}$.
\begin{proposition}\label{area}
$$\sup_{u,\ub}|\mbox{Area}(S_{u,\ub})-\mbox{Area}(S_{u,0})|\leq C\Delta_0\epsilon.$$
\end{proposition}
\begin{proof}
This follows from \eqref{detgaper}.
\end{proof}
With the estimate on the volume form, we can now show that the $L^p$ norms defined with respect to the metric and the $L^p$ norms defined with respect to the coordinate system are equivalent.
\begin{proposition}\label{eqnorm}
Given a covariant tensor $\phi_{A_1...A_r}$ on $S_{u,\ub}$, we have
$$\int_{S_{u,\ub}} <\phi,\phi>_{\gamma}^{p/2} \sim \sum_{i=1}^r\sum_{A_i=1,2}\iint |\phi_{A_1...A_r}|^p \sqrt{\det\gamma} d\th^1 d\th^2.$$
\end{proposition}
We can also bound $b$ under the bootstrap assumption, thus controlling the full spacetime metric: 
\begin{proposition}\label{b}
In the coordinate system $(u,\ub,\th^1,\th^2)$,
$$|b^A|\leq C\Delta_0\epsilon.$$
\end{proposition}
\begin{proof}
$b^A$ satisfies the equation
\begin{equation}\label{btrans}
\frac{\partial b^A}{\partial \ub}=-4\Omega^2\zeta^A.
\end{equation}
This can be derived from 
$$[L,\Lb]=\frac{\partial b^A}{\partial \ub}\frac{\partial}{\partial \th^A}.$$
Now, integrating \eqref{btrans} and using Proposition \ref{eqnorm} gives the result.
\end{proof}

\subsection{Estimates for Transport Equations}\label{transportsec}
The estimates for the Ricci coefficients and the null curvature components are derived from the null structure equations and the null Bianchi equations respectively. In order to use the equations, we need a way to obtain estimates from the covariant null transport equations. Such estimates require the boundedness of $\trch$ and $\trchb$, which is consistent with our bootstrap assumption (\ref{BA1}). Below, we state two Propositions which provide $L^p$ estimates for general quantities satisfying transport equations either in the $e_3$ or $e_4$ direction.
\begin{proposition}\label{transport}
There exists $\epsilon_0=\epsilon_0(\Delta_0)$ such that for all $\epsilon \leq \epsilon_0$ and for every $2\leq p<\infty$, we have
\[
 ||\phi||_{L^p(S_{u,\ub})}\leq C(||\phi||_{L^p(S_{u,\ub'})}+\int_{\ub'}^{\ub} ||\nabla_4\phi||_{L^p(S_{u,\ub''})}d{\ub''}),
\]
\[
 ||\phi||_{L^p(S_{u,\ub})}\leq C(||\phi||_{L^p(S_{u',\ub})}+\int_{u'}^{u} ||\nabla_3\phi||_{L^p(S_{u'',\ub})}d{u''}).
\]
\end{proposition}

\begin{proof}

The following identity holds for any scalar $f$:
\[
 \frac{d}{d\ub}\int_{\mathcal S_{u,\ub}} f=\int_{\mathcal S_{u,\ub}} \left(\frac{df}{d\ub}+\Omega \trch f\right)=\int_{\mathcal S_{u,\ub}} \Omega\left(e_4(f)+ \trch f\right).
\]
Similarly, we have
\[
 \frac{d}{du}\int_{\mathcal S_{u,\ub}} f=\int_{\mathcal S_{u,\ub}} \Omega\left(e_3(f)+ \trchb f\right).
\]
Hence, taking $f=|\phi|_{\gamma}^p$, we have
\begin{equation}\label{Lptransport}
\begin{split}
 ||\phi||^p_{L^p(S_{u,\ub})}=&||\phi||^p_{L^p(S_{u,\ub'})}+\int_{\ub'}^{\ub}\int_{S_{u,\ub''}} p|\phi|^{p-2}\Omega\left(<\phi,\nabla_4\phi>_\gamma+ \frac{1}{p}\trch |\phi|^2_{\gamma}\right)d{\ub''},\\
 ||\phi||^p_{L^p(S_{u,\ub})}=&||\phi||^p_{L^p(S_{u',\ub})}+\int_{u'}^{u}\int_{S_{u'',\ub}} p|\phi|^{p-2}\Omega\left(<\phi,\nabla_3\phi>_\gamma+ \frac{1}{p}\trchb |\phi|^2_{\gamma}\right)d{u''}.
\end{split}
\end{equation}
By the $L^\infty$ bounds for $\Omega$ and $\trch$ ($\trchb$) which are provided by Proposition \ref{Omega} and the bootstrap assumption (\ref{BA1}) respectively, we can control the last term in each of these equations using Gronwall's inequality to get
\begin{equation}\label{Lptransport.2}
\begin{split}
 ||\phi||^p_{L^p(S_{u,\ub})}\leq &C\left(||\phi||^p_{L^p(S_{u,\ub'})}+\int_{\ub'}^{\ub}\int_{S_{u,\ub''}} |\phi|^{p-1}|\nabla_4\phi|d{\ub''}\right),\\
 ||\phi||^p_{L^p(S_{u,\ub})}\leq &C\left(||\phi||^p_{L^p(S_{u',\ub})}+\int_{u'}^{u}\int_{S_{u'',\ub}} |\phi|^{p-1}|\nabla_3\phi|d{u''}\right).
\end{split}
\end{equation}
Notice that \eqref{Lptransport.2} allows us to in fact control $\sup_{\ub'\leq \ub''\leq \ub}||\phi||^p_{L^p(S_{u,\ub''})}$ and $\sup_{u'\leq u''\leq u}||\phi||^p_{L^p(S_{u'',\ub})}$ respectively. Therefore, using H\"older's inequality on the $2$-spheres, we get
\begin{equation*}
\begin{split}
\sup_{\ub'\leq \ub''\leq \ub}||\phi||^p_{L^p(S_{u,\ub''})}\leq &C\sup_{\ub'\leq \ub''\leq \ub}||\phi||^{p-1}_{L^p(S_{u,\ub''})}\left(||\phi||_{L^p(S_{u,\ub'})}+\int_{\ub'}^{\ub}\|\nabla_4\phi\|_{L^p(S_{u,\ub''})}d{\ub''}\right),\\
 \sup_{u'\leq u''\leq u}||\phi||^p_{L^p(S_{u'',\ub})}\leq &C\sup_{u'\leq u''\leq u}||\phi||^{p-1}_{L^p(S_{u'',\ub})}\left(||\phi||_{L^p(S_{u',\ub})}+\int_{u'}^{u}\int_{S_{u'',\ub}} \|\nabla_3\phi\|_{L^p(S_{u'',\ub})}d{u''}\right).
\end{split}
\end{equation*}
Dividing by $\sup_{\ub'\leq \ub''\leq \ub}||\phi||^{p-1}_{L^p(S_{u,\ub''})}$ and $\sup_{u'\leq u''\leq u}||\phi||^{p-1}_{L^p(S_{u'',\ub})}$ respectively gives the desired conclusion.
\end{proof}
The above estimates also hold for $p=\infty$:
\begin{proposition}\label{transportinfty}
There exists $\epsilon_0=\epsilon_0(\Delta_0)$ such that for all $\epsilon \leq \epsilon_0$, we have
\[
 ||\phi||_{L^\infty(S_{u,\ub})}\leq C\left(||\phi||_{L^\infty(S_{u,\ub'})}+\int_{\ub'}^{\ub} ||\nabla_4\phi||_{L^\infty(S_{u,\ub''})}d{\ub''}\right),
\]
\[
 ||\phi||_{L^\infty(S_{u,\ub})}\leq C\left(||\phi||_{L^\infty(S_{u',\ub})}+\int_{u'}^{u} ||\nabla_3\phi||_{L^\infty(S_{u'',\ub})}d{u''}\right).
\]
\end{proposition}
\begin{proof}
This follows simply from integrating along the integral curves of $L$ and $\Lb$, and the estimate on $\Omega$ in Proposition \ref{Omega}.
\end{proof}

\subsection{Sobolev Embedding}\label{Embedding}
Using the estimates for the metric $\gamma$ in Proposition \ref{gamma}, Sobolev embedding theorems in our setting follows from the standard Sobolev embedding theorems (see \cite{LR}): 
\begin{proposition}\label{L4}
There exists $\epsilon_0=\epsilon_0(\Delta_0)$ such that as long as $\epsilon\leq \epsilon_0$, we have
$$||\phi||_{L^4(S_{u,\ub})}\leq C\sum_{i=0}^1||\nabla^i\phi||_{L^2(S_{u,\ub})}. $$
\end{proposition}
Similarly, we can also prove the Sobolev embedding theorem for the $L^\infty$ norm: 
\begin{proposition}\label{Linfty}
There exists $\epsilon_0=\epsilon_0(\Delta_0)$ such that as long as $\epsilon\leq \epsilon_0$, we have
$$||\phi||_{L^\infty(S_{u,\ub})}\leq C\left(||\phi||_{L^2(S_{u,\ub})}+||\nabla\phi||_{L^3(S_{u,\ub})}\right). $$
As a consequence, since the area of $S_{u,\ub}$ is uniformly bounded, we have
$$||\phi||_{L^\infty(S_{u,\ub})}\leq C\left(||\phi||_{L^2(S_{u,\ub})}+||\nabla\phi||_{L^4(S_{u,\ub})}\right) $$
and
$$||\phi||_{L^\infty(S_{u,\ub})}\leq C\sum_{i=0}^2||\nabla^i\phi||_{L^2(S_{u,\ub})}. $$
\end{proposition}
Besides the Sobolev embedding theorem on the 2-spheres, we also have a co-dimensional 1 trace estimate that controls the $L^3(S)$ norm by the $L^2(H)$ norm with a small constant.
\begin{proposition}\label{L3trace}
\begin{equation*}
\begin{split}
  ||\phi||_{L^3(S_{u,\ub})}
\leq &C\left(||\phi||_{L^3(S_{u,\ub'})}+\epsilon^{\frac 14}||\nab\phi||_{L^2_{\ub}L^2(S)}+\epsilon^{\frac 18}||\nab_4\phi||_{L^2_{\ub}L^2(S)}\right).
\end{split}
\end{equation*}
\end{proposition}
\begin{proof}
It follows from the standard Sobolev embedding theorem and the lower and upper bounds of the volume form that
\begin{equation}\label{GNL4}
||\phi||_{L^4(S)}\leq C(||\phi||_{L^3(S)}^{\frac 34}||\nabla\phi||_{L^2(S)}^{\frac 14} +||\phi||_{L^3(S)}).
\end{equation}
Using (\ref{Lptransport}) and (\ref{GNL4}), we have
\begin{equation*}
\begin{split}
 ||\phi||^3_{L^3(S_{u,\ub})}=&||\phi||^3_{L^3(S_{u,\ub'})}+\int_{\ub'}^{\ub}\int_{S_{u,\ub''}} 3\Omega|\phi|_\gamma\left(<\phi,\nabla_4\phi>_\gamma+ \frac{1}{3}\trch |\phi|^2_{\gamma}\right)d{\ub''}\\
\leq &||\phi||^3_{L^3(S_{u,\ub'})}+C||\phi||^2_{L^4(H)}||\nab_4\phi||_{L^2(H)}+\int_{0}^{\ub} C\Delta_0||\phi||_{L^3(S_{u,\ub'})}^3 d{\ub'}\\
\leq &||\phi||^3_{L^3(S_{u,\ub'})}+C(||\phi||_{L^\infty_{\ub} L^3(S)}^{\frac 32}||\nab\phi||_{L^1_{\ub}L^2(S)}^{\frac 12}+||\phi||_{L^4_{\ub}L^3(S)}^2)||\nab_4\phi||_{L^2(H)}\\
&+\int_0^{\ub} C\Delta_0||\phi||_{L^3(S_{u,\ub'})}^3 d\ub'\\
\leq &||\phi||^3_{L^3(S_{u,\ub'})}+C(||\phi||_{L^\infty_{\ub} L^3(S)}^{\frac 32}||\nab\phi||_{L^1_{\ub}L^2(S)}^{\frac 12}+\epsilon^{\frac 12}||\phi||_{L^\infty_{\ub}L^3(S)}^2)||\nab_4\phi||_{L^2(H)}\\
&+\int_0^{\ub} C\Delta_0||\phi||_{L^3(S_{u,\ub'})}^3 d\ub'.
\end{split}
\end{equation*}
Using H\"older's inequality and absorbing the term $||\phi||_{L^\infty_{\ub} L^3(S)}^3$ to the left hand side, we have
\begin{equation*}
\begin{split}
& ||\phi||^3_{L^3(S_{u,\ub})}\\
\leq &C\left(||\phi||^3_{L^3(S_{u,\ub'})}+||\nab\phi||_{L^1_{\ub}L^2(S)}||\nab_4\phi||^2_{L^2_{\ub}L^2(S)}+\epsilon^{\frac 32}||\nab_4\phi||^3_{L^2_{\ub}L^2(S)}\right.\\
&\left.+\int_0^{\ub} C\Delta_0||\phi||_{L^3(S_{u,\ub'})}^3 d\ub'\right)\\
\leq &C\left(||\phi||^3_{L^3(S_{u,\ub'})}+\epsilon^{\frac 34}||\nab\phi||_{L^2_{\ub}L^2(S)}^3+\epsilon^{\frac 38}||\nab_4\phi||^3_{L^2_{\ub}L^2(S)}+\epsilon^{\frac 32}||\nab_4\phi||^3_{L^2_{\ub}L^2(S)}\right.\\
&\left.+\int_0^{\ub} C\Delta_0||\phi||_{L^3(S_{u,\ub'})}^3 d\ub'\right),
\end{split}
\end{equation*}
where we have gained a smallness constant by changing $L^1_{\ub}$ to $L^2_{\ub}$ for $\nab\phi$ in the last line.
By Gronwall's inequality, and using the fact that $\ub\leq \epsilon$, we have
\begin{equation*}
\begin{split}
  ||\phi||^3_{L^3(S_{u,\ub})}
\leq &C\left(||\phi||^3_{L^3(S_{u,\ub'})}+\epsilon^{\frac 34}||\nab\phi||_{L^2_{\ub}L^2(S)}^3+\epsilon^{\frac 38}||\nab_4\phi||^3_{L^2_{\ub}L^2(S)}\right).
\end{split}
\end{equation*}
\end{proof}

\subsection{Commutation Formulae}\label{commutation}
We have the following formulae from \cite{KN}:
\begin{proposition}\label{commutation.prop}
The commutator $[\nabla_4,\nabla]$ acting on a $(0,r)$ S-tensor is given by
\begin{equation*}
 \begin{split}
[\nabla_4,\nabla_B]\phi_{A_1...A_r}=&[D_4,D_B]\phi_{A_1...A_r}+(\nabla_B\log\Omega)\nabla_4\phi_{A_1...A_r}-(\gamma^{-1})^{CD}\chi_{BD}\nabla_C\phi_{A_1...A_r} \\
&-\sum_{i=1}^r (\gamma^{-1})^{CD}\chi_{BD}\etab_{A_i}\phi_{A_1...\hat{A_i}C...A_r}+\sum_{i=1}^r (\gamma^{-1})^{CD}\chi_{A_iB}\etab_{D}\phi_{A_1...\hat{A_i}C...A_r}.
 \end{split}
\end{equation*}
Similarly, the commutator $[\nabla_3,\nabla]$ acting on a $(0,r)$ S-tensor is given by
\begin{equation*}
 \begin{split}
[\nabla_3,\nabla_B]\phi_{A_1...A_r}=&[D_3,D_B]\phi_{A_1...A_r}+(\nabla_B\log\Omega)\nabla_3\phi_{A_1...A_r}-(\gamma^{-1})^{CD}\chib_{BD}\nabla_C\phi_{A_1...A_r} \\
&-\sum_{i=1}^r (\gamma^{-1})^{CD}\chib_{BD}\eta_{A_i}\phi_{A_1...\hat{A_i}C...A_r}+\sum_{i=1}^r (\gamma^{-1})^{CD}\chib_{A_iB}\eta_{D}\phi_{A_1...\hat{A_i}C...A_r}.
 \end{split}
\end{equation*}
\end{proposition}
By induction, we get the following schematic formula for repeated commutations (see \cite{LR}):
\begin{proposition}\label{commuteeqn}
Suppose $\nabla_4\phi=F_0$ where $\phi$ and $F_0$ are $(0,r)$ S-tensors. Let $\nabla_4\nabla^i\phi=F_i$ where $F_i$ is a $(0,r+i)$ S-tensor.
Then $F_i$ is given schematically by
\begin{equation*}
\begin{split}
F_i\sim &\sum_{i_1+i_2+i_3=i}\nabla^{i_1}(\eta+\underline{\eta})^{i_2}\nabla^{i_3} F_0+\sum_{i_1+i_2+i_3+i_4=i}\nabla^{i_1}(\eta+\underline{\eta})^{i_2}\nabla^{i_3}\chi\nabla^{i_4} \phi\\
&+\sum_{i_1+i_2+i_3+i_4=i-1} \nabla^{i_1}(\eta+\underline{\eta})^{i_2}\nabla^{i_3}\beta\nabla^{i_4} \phi.
\end{split}
\end{equation*}
where by $\nabla^{i_1}(\eta+\underline{\eta})^{i_2}$ we mean the sum of all terms which is a product of $i_2$ factors, each factor being $\nabla^j (\eta+\underline{\eta})$ for some $j$ and that the sum of all $j$'s is $i_1$, i.e., $\nabla^{i_1}(\eta+\underline{\eta})^{i_2}=\displaystyle\sum_{j_1+...+j_{i_2}=i_1}\nabla^{j_1}(\eta+\underline{\eta})...\nabla^{j_{i_2}}(\eta+\underline{\eta})$. Similarly, suppose $\nabla_3\phi=G_{0}$ where $\phi$ and $G_0$ are $(0,r)$ S-tensors. Let $\nabla_3\nabla^i\phi=G_{i}$ where $G_i$ is a $(0,r_i)$ S-tensor.
Then $G_i$ is given schematically by
\begin{equation*}
\begin{split}
G_{i}\sim &\sum_{i_1+i_2+i_3=i}\nabla^{i_1}(\eta+\underline{\eta})^{i_2}\nabla^{i_3} G_{0}+\sum_{i_1+i_2+i_3+i_4=i}\nabla^{i_1}(\eta+\underline{\eta})^{i_2}\nabla^{i_3}\underline{\chi}\nabla^{i_4} \phi\\
&+\sum_{i_1+i_2+i_3+i_4=i-1} \nabla^{i_1}(\eta+\underline{\eta})^{i_2}\nabla^{i_3}\underline{\beta}\nabla^{i_4} \phi.
\end{split}
\end{equation*}

\end{proposition}

The following further simplified version is useful for our estimates in the next section:
\begin{proposition}\label{commute.prop}
Suppose $\nabla_4\phi=F_0$ where $\phi$ and $F_0$ are $(0,r)$ S-tensors. Let $\nabla_4\nabla^i\phi=F_i$ where $F_i$ is a $(0,r+i)$ S-tensor.
Then $F_i$ is given schematically by
\begin{equation*}
\begin{split}
F_i\sim &\sum_{i_1+i_2+i_3=i}\nabla^{i_1}\psi^{i_2}\nabla^{i_3} F_0+\sum_{i_1+i_2+i_3+i_4=i}\nabla^{i_1}\psi^{i_2}\nabla^{i_3}\chi\nabla^{i_4} \phi.\\
\end{split}
\end{equation*}
Similarly, suppose $\nabla_3\phi=G_{0}$ where $\phi$ and $G_0$ are $(0,r)$ S-tensors. Let $\nabla_3\nabla^i\phi=G_{i}$ where $G_i$ is a $(0,r_i)$ S-tensor.
Then $G_i$ is given schematically by
\begin{equation*}
\begin{split}
G_{i}\sim &\sum_{i_1+i_2+i_3=i}\nabla^{i_1}\psi^{i_2}\nabla^{i_3} G_{0}+\sum_{i_1+i_2+i_3+i_4=i}\nabla^{i_1}\psi^{i_2}\nabla^{i_3}\underline{\chi}\nabla^{i_4} \phi.
\end{split}
\end{equation*}
\end{proposition}
\begin{proof}
We replace $\beta$ and $\betab$ using the Codazzi equations, which schematically looks like
$$\beta=\nabla\chi+\psi\chi,$$
$$\betab=\nabla\chib+\psi\chib.$$
\end{proof}

\subsection{General Elliptic Estimates for Hodge Systems}\label{elliptic}
We recall the definition of the divergence and curl of a symmetric covariant tensor of an arbitrary rank:
$$(\div\phi)_{A_1...A_r}=\nabla^B\phi_{BA_1...A_r},$$
$$(\curl\phi)_{A_1...A_r}=\eps^{BC}\nabla_B\phi_{CA_1...A_r},$$
where $\eps$ is the volume form associated to the metric $\gamma$.
Recall also that the trace is defined to be
$$(tr\phi)_{A_1...A_{r-1}}=(\gamma^{-1})^{BC}\phi_{BCA_1...A_{r-1}}.$$
The following elliptic estimate is standard (See for example \cite{CK} or \cite{Chr}):
\begin{proposition}\label{ellipticthm}
Let $\phi$ be a totally symmetric $r+1$ covariant tensorfield on a 2-sphere $(\mathbb S^2,\gamma)$ satisfying
$$\div\phi=f,\quad \curl\phi=g,\quad \mbox{tr}\phi=h.$$
Suppose also that
$$\sum_{i\leq 1}||\nabla^i K||_{L^2(S)}< \infty.$$
Then for $i\leq 3$,
$$||\nabla^{i}\phi||_{L^2(S)}\leq C(\sum_{k\leq 1}||\nabla^k K||_{L^2(S)})(\sum_{j=0}^{i-1}(||\nabla^{j}f||_{L^2(S)}+||\nabla^{j}g||_{L^2(S)}+||\nabla^{j}h||_{L^2(S)}+||\phi||_{L^2(S)})).$$
\end{proposition}
For the special case that $\phi$ a symmetric traceless 2-tensor, we only need to know its divergence:
\begin{proposition}\label{elliptictraceless}
Suppose $\phi$ is a symmetric traceless 2-tensor satisfying
$$\div\phi=f.$$
Suppose moreover that $$\sum_{i\leq 1}||\nabla^i K||_{L^2(S)}< \infty.$$
Then, for $i\leq 3$,
$$||\nabla^{i}\phi||_{L^2(S)}\leq C(\sum_{k\leq 1}||\nabla^k K||_{L^2(S)})(\sum_{j=0}^{i-1}(||\nabla^{j}f||_{L^2(S)}+||\phi||_{L^2(S)})).$$
\end{proposition}
\begin{proof}
In view of Proposition \ref{ellipticthm}, this Proposition follows from
$$\curl\phi=^*f.$$
This is a direct computation using the fact that $\phi$ is both symmetric and traceless.
\end{proof}

\section{Estimates for the Ricci Coefficients}
We continue to work under the bootstrap assumptions (\ref{BA1}). In this section, we show that assuming the curvature norm $\mathcal R$ is bounded, then so are the Ricci coefficient norms $\mathcal O$, $\tilde{\mathcal O}_{3,2}$ and the curvature norm $\mathcal R(S)$ on the spheres. In particular, our bootstrap assumption (\ref{BA1}) and all the estimates in the last section are verified as long as $\mathcal R$ is controlled.

\subsection{$L^4(S)$ Estimates for First Derivatives of Ricci Coefficients}\label{secRicciL4}

\begin{proposition}\label{L4Ricci1}
Assume
$$\mathcal R <\infty,\quad\tilde{\mathcal O}_{3,2}<\infty,\quad\mathcal O_{2,2}<\infty.$$
Then there exists $\epsilon_0=\epsilon_0(\mathcal R, \tilde{\mathcal O}_{3,2}, \mathcal O_{2,2},\Delta_0)$ such that whenever $\epsilon\leq \epsilon_0$, 
\[
 \sum_{i\leq 1}\mathcal O_{i,4}[\trchb,\eta]\leq C(\mathcal O_0).
\]
In particular, $C(\mathcal O_0)$ is independent of $\Delta_0$.
\end{proposition}
\begin{proof}
Using the null structure equations, we have a schematic equation of the type
$$\nabla_{4}(\trchb,\eta)=\beta+\rhoc+\nabla\etab+\psi\psi+\psi_H\psi.$$
It is important to note that $\betab$, $\psi_{\Hb}$ do not appear in the source terms. In other words, only the terms that can be controlled on the outgoing hypersurface $H_u$ enter the equation. By Proposition \ref{commute.prop}, we have the following null structure equations commuted with angular derivatives:
$$\nabla_{4}\nabla^i(\trchb,\eta)=\sum_{i_1+i_2+i_3=i}\nabla^{i_1}\psi^{i_2}\nabla^{i_3}(\beta+\rhoc+\nabla\etab)+\sum_{i_1+i_2+i_3+i_4=i}\nabla^{i_1}\psi^{i_2}\nabla^{i_3}\psi\nabla^{i_4}(\psi+\psi_H).$$
By Proposition \ref{transport}, in order to estimate $||\nabla^i(\trchb,\eta)||_{L^\infty_uL^\infty_{\ub}L^4(S)}$, it suffices to estimate the initial data and the $||\cdot||_{L^\infty_uL^1_{\ub}L^4(S)}$ norm of the right hand side. We now estimate each of the terms in the equations. For the curvature terms, we have
\begin{equation*}
\begin{split}
&||\sum_{i_1+i_2\leq 1}\psi^{i_1}\nabla^{i_2}(\beta,\rhoc)||_{L^\infty_{u}L^1_{\ub}L^4(S)} \\
\leq& C(\sum_{i_1\leq 1}||\psi||^{i_1}_{L^\infty_{u}L^\infty_{\ub}L^\infty(S)})(\sum_{i_2\leq 1}||\nabla^{i_2}(\beta,\rhoc)||_{L^\infty_{u}L^1_{\ub}L^4(S)})\\
\leq &C(1+\Delta_0)\epsilon^{\frac{1}{2}}\sum_{i\leq 2}||\nabla^{i}(\beta,\rhoc)||_{L^\infty_{u}L^2_{\ub}L^2(S)} \\
\leq &C(1+\Delta_0)\epsilon^{\frac{1}{2}}\mathcal R.
\end{split}
\end{equation*}
The term with $\nabla\etab$ instead of $(\beta,\rhoc)$ can be bounded analogously, except for using the $\mathcal O$ and $\tilde{\mathcal O}_{3,2}$ norms together instead of the $\mathcal R$ norm:
\begin{equation*}
\begin{split}
&||\sum_{i_1+i_2\leq 1}\psi^{i_1}\nabla^{i_2+1}\etab||_{L^\infty_{u}L^1_{\ub}L^4(S)} \\
\leq& C(\sum_{i_1\leq 1}||\psi||^{i_1}_{L^\infty_{u}L^\infty_{\ub}L^\infty(S)})(\sum_{i_2\leq 2}||\nabla^{i_2}\etab||_{L^\infty_uL^1_{\ub}L^4(S)})\\
\leq &C(1+\Delta_0)\epsilon^{\frac{1}{2}}\left(\sum_{i\leq 1}||\nabla^{i}\eta||_{L^\infty_{u}L^\infty_{\ub}L^4(S)}+\sum_{2\leq i\leq 3}||\nabla^i\etab||_{L^\infty_{u}L^2_{\ub}L^2(S)}\right) \\
\leq &C(1+\Delta_0)\epsilon^{\frac{1}{2}}(\sum_{i\leq 1}\mathcal O_{i,4}+\mathcal O_{2,2}+\tilde{\mathcal O}_{3,2})\\
\leq &C(1+\Delta_0)\epsilon^{\frac{1}{2}}(\Delta_0+\mathcal O_{2,2}+\tilde{\mathcal O}_{3,2}).
\end{split}
\end{equation*}
We now move on to the lower order terms:
\begin{equation*}
\begin{split}
&||\sum_{i_1+i_2+i_3\leq 1}\psi^{i_1}\nabla^{i_2}\psi\nabla^{i_3}\psi||_{L^\infty_{u}L^1_{\ub}L^4(S)} \\
\leq& C(\sum_{i_1\leq 2}||\psi||^{i_1}_{L^\infty_{u}L^\infty_{\ub}L^\infty(S)})(\sum_{i_2\leq 1}||\nabla^{i_2}\psi||_{L^\infty_{u}L^1_{\ub}L^4(S)})\\
\leq &C\Delta_0(1+\Delta_0)^{2}\epsilon.
\end{split}
\end{equation*}
Finally, we bound the lower order terms that contain $\psi_{H}$: 
\begin{equation*}
\begin{split}
&||\sum_{i_1+i_2+i_3\leq 1}\psi^{i_1}\nabla^{i_2}\psi\nabla^{i_3}\psi_{H}||_{L^\infty_{u}L^1_{\ub}L^4(S)} \\
\leq& C(\sum_{i_1\leq 2}||\psi||^{i_1}_{L^\infty_{u}L^\infty_{\ub}L^\infty(S)})(\sum_{i_2\leq 1}||\nabla^{i_2}\psi_H||_{L^\infty_{u}L^1_{\ub}L^4(S)}) \\
&+C(\sum_{i_1\leq 1}||\psi||^{i_1}_{L^\infty_{u}L^\infty_{\ub}L^\infty(S)})(\sum_{i_2\leq 1}||\nabla^{i_2}\psi||_{L^\infty_{u}L^2_{\ub}L^4(S)})(||\psi_H||_{L^\infty_{u}L^2_{\ub}L^\infty(S)})\\
\leq &C(1+\Delta_0)^{2}\epsilon^\frac{1}{2}(\sum_{i\leq 1}||\nabla^{i}\psi_H||_{L^\infty_{u}L^2_{\ub}L^4(S)}+ \sum_{i\leq 1}||\nabla^{i}\psi||_{L^\infty_{u}L^\infty_{\ub}L^4(S)})\\
\leq &C\Delta_0(1+\Delta_0)^{2}\epsilon^{\frac{1}{2}}.
\end{split}
\end{equation*}
Hence, by Proposition \ref{transport}, we have
\begin{equation*}
\begin{split}
\sum_{i\leq 1}\mathcal O_{i,4}[\trchb,\eta]\leq &\mathcal O_0+C(1+\Delta_0)^2\epsilon^{\frac{1}{2}}(\mathcal R+\mathcal O_{2,2}+\tilde{\mathcal O}_{3,2}+\Delta_0). 
\end{split}
\end{equation*}
The proposition follows from choosing $\epsilon$ to be sufficiently small, depending on $\mathcal R, \tilde{\mathcal O}_{3,2}, \mathcal O_{2,2},\Delta_0$.
\end{proof}

We now estimate the terms that we denote by $\psi_{\Hb}$, i.e., $\chibh$ and $\omegab$. Both of them obey a $\nabla_4$ equation. However, a new difficulty compared Proposition \ref{L4Ricci1} arises since the initial data for $\chibh$ and $\omegab$ are not in $L^\infty_u$. Thus they can only be estimated after taking the $L^2_u$ norm.
\begin{proposition}\label{L4Ricci3}
Assume
$$\mathcal R <\infty,\quad\tilde{\mathcal O}_{3,2}<\infty,\quad\mathcal O_{2,2}<\infty.$$
Then there exists $\epsilon_0=\epsilon_0(\mathcal R, \tilde{\mathcal O}_{3,2}, \mathcal O_{2,2},\Delta_0)$ such that whenever $\epsilon\leq \epsilon_0$, 
$$\sum_{i\leq 1}\mathcal O_{i,4}[\chibh,\omegab]\leq C(\mathcal O_0) .$$
In particular, this estimate is independent of $\Delta_0$.
\end{proposition}
\begin{proof}
Using the null structure equations, for each $\psi_{\Hb}\in\{\chibh,\omegab\}$, we have an equation of the type
$$\nabla_{4}\psi_{\Hb}=\rhoc+\nabla\etab+(\psi+\psi_H)(\psi+\psi_{\Hb}).$$
We also use the null structure equations commuted with angular derivatives:
$$\nabla_{4}\nabla^i\psi_{\Hb}=\sum_{i_1+i_2+i_3=i}\nabla^{i_1}\psi^{i_2}\nabla^{i_3}(\rhoc+\nabla\etab)+\sum_{i_1+i_2+i_3+i_4=i}\nabla^{i_1}\psi^{i_2}\nabla^{i_3}(\psi+\psi_H)\nabla^{i_4}(\psi+\psi_{\Hb}).$$
From the proof of Proposition \ref{L4Ricci1}, we have
\begin{equation*}
\begin{split}
||\sum_{i_1+i_2+i_3\leq 1}\nabla^{i_1}\psi^{i_2}\nabla^{i_3}\rhoc||_{L^\infty_{u}L^1_{\ub}L^4(S)} \leq &C(1+\Delta_0)\epsilon^{\frac{1}{2}}\mathcal R,
\end{split}
\end{equation*}
and
\begin{equation*}
\begin{split}
||\sum_{i_1+i_2\leq 1}\psi^{i_1}\nabla^{i_2+1}\etab||_{L^\infty_{u}L^1_{\ub}L^4(S)} 
\leq &C(1+\Delta_0)\epsilon^{\frac{1}{2}}(\Delta_0+\mathcal O_{2,2}+\tilde{\mathcal O}_{3,2}),
\end{split}
\end{equation*}
and
\begin{equation*}
\begin{split}
||\sum_{i_1+i_2+i_3\leq 1}\psi^{i_1}\nabla^{i_2}\psi\nabla^{i_3}\psi||_{L^\infty_{u}L^1_{\ub}L^4(S)} 
\leq &C\Delta_0(1+\Delta_0)^2\epsilon,
\end{split}
\end{equation*}
and
\begin{equation*}
\begin{split}
||\sum_{i_1+i_2+i_3\leq 1}\psi^{i_1}\nabla^{i_2}\psi\nabla^{i_3}\psi_{H}||_{L^\infty_{u}L^1_{\ub}L^4(S)} 
\leq &C\Delta_0(1+\Delta_0)^{2}\epsilon^{\frac{1}{2}}.
\end{split}
\end{equation*}
The two new terms that did not appear in the proof of Proposition \ref{L4Ricci1} are
$$\sum_{i_1+i_2+i_3+i_4\leq 1}\nabla^{i_1}\psi^{i_2}\nabla^{i_3}\psi_H\nabla^{i_4}\psi_{\Hb},\quad\mbox{and}\quad \sum_{i_1+i_2+i_3+i_4\leq 1}\nabla^{i_1}\psi^{i_2}\nabla^{i_3}\psi\nabla^{i_4}\psi_{\Hb}.$$
Both of these terms cannot be controlled in the $L^\infty_u L^1_{\ub}L^4(S)$ norm. Instead, for each fixed $u$, we bound the first term in the $L^1_{\ub}L^4(S)$ norm:
\begin{equation*}
\begin{split}
&||\sum_{i_1+i_2+i_3+i_4\leq 1}\nabla^{i_1}\psi^{i_2}\nabla^{i_3}\psi_H\nabla^{i_4}\psi_{\Hb}||_{L^1_{\ub}L^4(S)} \\
\leq& C(1+||\psi||_{L^\infty_{\ub}L^\infty(S)})(||\psi_{\Hb}||_{L^\infty_{\ub}L^\infty(S)})(\sum_{i\leq 1}||\nabla^{i}\psi_{H}||_{L^1_{\ub}L^4(S)}) \\
&+C(1+||\psi||_{L^\infty_{\ub}L^\infty(S)})(\sum_{i\leq 1}||\nabla^{i}\psi_{\Hb}||_{L^\infty_{\ub}L^4(S)})(||\psi_{H}||_{L^1_{\ub}L^\infty(S)}) \\
\leq &C(1+\Delta_0)^{2}\epsilon^\frac{1}{2}(||\psi_{H}||_{L^2_{\ub}L^\infty(S)}+\sum_{i\leq 1}||\nabla^{i}\psi_{H}||_{L^2_{\ub}L^4(S)})( \sum_{i\leq 1}||\nabla^{i}\psi_{\Hb}||_{L^\infty_{\ub}L^4(S)})\\
\leq &C(1+\Delta_0)^{2}\epsilon^\frac{1}{2}(\sum_{i\leq 1}||\nabla^{i}\psi_{\Hb}||_{L^2_{\ub}L^4(S)}+ \sum_{i\leq 1}||\nabla^{i}\psi_{\Hb}||_{L^\infty_{\ub}L^4(S)}).\\
&\quad\mbox{by Sobolev embedding in Proposition }\ref{L4} \\
\leq &C(1+\Delta_0)^{2}\epsilon^{\frac{1}{2}}(\Delta_0+ \sum_{i\leq 1}||\nabla^{i}\psi_{\Hb}||_{L^\infty_{\ub}L^4(S)}).
\end{split}
\end{equation*}
According to the definition of the $\mathcal O_{i,4}$ norm, $\psi$ obeys stronger estimates than $\psi_H$. Therefore, we can control the remaining term in the same manner:
$$||\sum_{i_1+i_2+i_3+i_4\leq 1}\nabla^{i_1}\psi^{i_2}\nabla^{i_3}\psi\nabla^{i_4}\psi_{\Hb}||_{L^1_{\ub}L^4(S)} \leq C(1+\Delta_0)^{2}\epsilon^{\frac{1}{2}}(\Delta_0+ \sum_{i\leq 1}||\nabla^{i}\psi_{\Hb}||_{L^\infty_{\ub}L^4(S)}).$$
Therefore, by Proposition \ref{transport}, for all $u\in [0,u_*]$,
\begin{equation*}
\begin{split}
&\sum_{i\leq 1}||\nabla^i\psi_{\Hb}||_{L^4(S_{u,\ub})}\\
\leq& C(\sum_{i\leq 1}||\nabla^i\psi_{\Hb}||_{L^4(S_{u,0})}+(1+\Delta_0)^{2}\epsilon^{\frac{1}{2}}(\mathcal R+\mathcal O_{2,2}+\tilde{\mathcal O}_{3,2}+\Delta_0+ \sum_{i\leq 1}||\nabla^{i}\psi_{\Hb}||_{L^\infty_{\ub}L^4(S)})).
\end{split}
\end{equation*}
Clearly the right hand side is independent of $\ub$. Thus we can take supremum in $\ub$ on the left hand side. Then, we take the $L^2_u$ norm to obtain
\begin{equation*}
\begin{split}
&\sum_{i\leq 1}||\nabla^i\psi_{\Hb}||_{L^2_{u}L^\infty_{\ub}L^4(S)} \\
\leq & C(\sum_{i\leq 1}||\nabla^i\psi_{\Hb}||_{L^2_{u}L^4(S_{u,0})}+(1+\Delta_0)^{2}\epsilon^{\frac{1}{2}}(\mathcal R+\tilde{\mathcal O}_{3,2}+\Delta_0+\sum_{i\leq 1}||\nabla^{i}\psi_{\Hb}||_{L^2_{u}L^\infty_{\ub}L^4(S)})) \\
\leq &C(\mathcal O_0+(1+\Delta_0)^{2}\epsilon^{\frac{1}{2}}(\mathcal R+\mathcal O_{2,2}+\tilde{\mathcal O}_{3,2}+\Delta_0))
\end{split}
\end{equation*}
since by \eqref{BA1} $\sum_{i\leq 1}||\nabla^{i}\psi_{\Hb}||_{L^2_{u}L^\infty_{\ub}L^4(S)}$ is controlled by $\Delta_0$. The left hand side is precisely what we need to control for the $\displaystyle\sum_{i\leq 1}\mathcal O_{i,4}[\chibh,\omegab]$ norm. Thus
$$\sum_{i\leq 1}\mathcal O_{i,4}[\chibh,\omegab]\leq C(\mathcal O_0+(1+\Delta_0)^{2}\epsilon^{\frac{1}{2}}(\mathcal R+\mathcal O_{2,2}+\tilde{\mathcal O}_{3,2}+\Delta_0)).$$
We conclude the proof by choosing $\epsilon$ to be sufficiently small.
\end{proof}

We now turn to the Ricci coefficients $\etab$, $\chih$, $\omega$. To estimate these Ricci coefficients, we use the $\nab_3$ equations. Unlike in the proofs of Propositions \ref{L4Ricci1} and \ref{L4Ricci3} where a smallness constant can be gained from the shortness of the $\ub$ interval, when integrating the $\nab_3$ equation, the $u$ interval is arbitrarily long. Instead, we show that the inhomogeneous terms are at worst linear in the unknown and the desired bounds can be obtained via Gronwall's inequality. Notice that $\chih$ satisfies a $\nab_4$ equation with $\alpha$ as a source term. We avoid this equation because $\alpha$ is singular. We begin with the estimates for $\etab$. As we will see below, we cannot directly estimate the $L^4(S)$ norms of $\etab$ and its derivatives, but have to first estimate the $L^\infty(S)$ norm of $\etab$:
\begin{proposition}\label{LinftyRicci1}
Assume
$$\mathcal R <\infty,\quad\tilde{\mathcal O}_{3,2}<\infty,\quad\mathcal O_{2,2}<\infty.$$
Then there exists $\epsilon_0=\epsilon_0(\mathcal R, \tilde{\mathcal O}_{3,2}, \mathcal O_{2,2},\Delta_0)$ such that whenever $\epsilon\leq \epsilon_0$, 
\[
 \mathcal O_{0,\infty}[\etab]\leq C(\mathcal O_0,\mathcal R(S)[\betab]).
\]
In particular, this estimate is independent of $\Delta_0$.
\end{proposition}
\begin{proof}
$\etab$ satisfies a $\nabla_3$ equation. As remarked above, integrating in the $u$ direction does not give a small constant as in integrating in the $\ub$ direction. We therefore need to exploit the structure of the equation. We have, schematically
$$\nabla_3\etab=(\psi_{\Hb}+\trchb)(\eta+\etab)+\betab.$$
We notice that the quadratic term $\etab^2$ does not appear. Moreover, $\trch$, $\chih$ and $\omega$ do not enter the equation. In other words, all Ricci coefficients except $\etab$ in this equation have been estimated in the previous propositions by $C(\mathcal O_0)$. We now bound each of the terms. Firstly, the term with curvature can be controlled using H\"older's inequality and the Sobolev embedding theorem in Proposition \ref{Linfty} by $\mathcal R(S)$:
\begin{equation*}
\begin{split}
||\betab||_{L^\infty_{\ub}L^1_{u}L^\infty(S)}\leq C\sum_{i\leq 1}||\nabla^i\betab||_{L^\infty_{\ub}L^1_{u}L^3(S)} \leq CI^{\frac 12}R(S)[\betab].
\end{split}
\end{equation*}
Here, and below, we will simplify the notation by absorbing powers of $I$ into the constant $C$. We therefore simply write
\begin{equation*}
\begin{split}
||\betab||_{L^\infty_{\ub}L^1_{u}L^\infty(S)}\leq CI^{\frac 12}R(S)[\betab].
\end{split}
\end{equation*}
Then, we estimate the terms quadratic in the Ricci coefficients, which do not involve $\etab$:
\begin{equation*}
\begin{split}
&||(\psi_{\Hb}+\trchb)\eta||_{L^\infty_{\ub}L^1_{u}L^\infty(S)} \\
\leq& C ||\psi_{\Hb}+\trchb||_{L^\infty_{\ub}L^2_uL^\infty(S)}||\eta||_{L^\infty_{\ub}L^\infty_uL^\infty(S)}\\
\leq&C(\mathcal O_0),
\end{split}
\end{equation*}
by Propositions \ref{L4Ricci1} and \ref{L4Ricci3} and the Sobolev embedding theorem in Proposition \ref{Linfty}.
Finally, we estimate the term $(\psi_{\Hb}+\trchb)\etab$. Fix $\ub$. Then
\begin{equation*}
\begin{split}
&||(\psi_{\Hb}+\trchb)\etab||_{L^1_{u}L^\infty(S)} \\
\leq& C\int_0^u ||\psi_{\Hb}+\trchb||_{L^\infty(S_{u',\ub})}||\etab||_{L^\infty(S_{u',\ub})} du'.\\
\end{split}
\end{equation*}
Therefore, by Proposition \ref{transport}, we have, for every $\ub$,
$$||\etab||_{L^\infty_uL^\infty(S)}\leq C(\mathcal O_0)+ CR(S)[\betab]+C\int_0^u ||\psi_{\Hb}+\trchb||_{L^\infty(S_{u',\ub})}||\etab||_{L^\infty(S_{u',\ub})} du'.$$
 By Gronwall's inequality, we have, for every $\ub$,
$$||\etab||_{L^\infty_uL^\infty(S)}\leq C(\mathcal O_0,\mathcal R(S)[\betab])\exp(C\int_0^I ||\psi_{\Hb}+\trchb||_{L^\infty(S_{u',\ub})} du').$$
Using the Cauchy-Schwarz inequality, the Sobolev embedding theorem in Proposition \ref{Linfty}, as well as the estimates for the Ricci coefficients $\trchb$ and $\psi_{\Hb}$ derived in Propositions \ref{L4Ricci1} and \ref{L4Ricci3}, we have
$$||\etab||_{L^\infty_uL^\infty(S)}\leq C(\mathcal O_0,\mathcal R(S)[\betab]),$$
as desired.
\end{proof}
Using the $L^\infty$ estimate of $\etab$, we now control $\nabla\etab$ in $L^2$:
\begin{proposition}\label{L4Ricci40}
Assume
$$\mathcal R <\infty,\quad\tilde{\mathcal O}_{3,2}<\infty,\quad\mathcal O_{2,2}<\infty.$$
Then there exists $\epsilon_0=\epsilon_0(\mathcal R, \tilde{\mathcal O}_{3,2}, \mathcal O_{2,2},\Delta_0)$ such that whenever $\epsilon\leq \epsilon_0$, 
\[
 \mathcal O_{1,2}[\etab]\leq C(\mathcal O_0,\mathcal R(S)[\betab]).
\]
In particular, this estimate is independent of $\Delta_0$.
\end{proposition}
\begin{proof}
Recall that we have, schematically,
$$\nabla_3\etab=(\psi_{\Hb}+\trchb)(\eta+\etab)+\betab.$$
Commuting with angular derivatives, we get
\begin{equation}\label{etabcommuted}
\nabla_{3}\nabla\etab=\sum_{i_1+i_2=1}(\eta+\etab)^{i_1}\nabla^{i_2}\betab+\sum_{i_1+i_2+i_3=1}(\eta+\etab)^{i_1}\nabla^{i_2}(\eta+\etab)\nabla^{i_3}(\psi_{\Hb}+\trchb).
\end{equation}
We notice that in (\ref{etabcommuted}), when two $\etab$'s appear in a term, neither of them has a derivative. Fix $\ub$. We now estimate each of the terms. Firstly, the term with curvature:
\begin{equation*}
\begin{split}
&||\sum_{i_1+i_2\leq 1}(\eta+\etab)^{i_1}\nabla^{i_2}\betab||_{L^1_{u}L^2(S)} \\
\leq& C(1+||(\eta,\etab)||_{L^\infty_{u}L^\infty(S)})(\sum_{i\leq 1}||\nabla^{i}\betab||_{L^1_{u}L^2(S)})\\
\leq &C(\mathcal O_0,\mathcal R(S)[\betab])\sum_{i\leq 1}||\nabla^{i}\betab||_{L^2_{u}L^2(S)} \\
\leq &C(\mathcal O_0,\mathcal R(S)[\betab]),
\end{split}
\end{equation*}
since $\nabla\betab$ can be controlled in $L^2(\Hb_{\ub})$ by $\mathcal R(S)[\betab]$.
We then estimate the nonlinear term in the Ricci coefficients:
\begin{equation*}
\begin{split}
&||\sum_{i_1+i_2+i_3\leq 1}(\eta+\etab)^{i_1}\nabla^{i_2}(\eta+\etab)\nabla^{i_3}(\psi_{\Hb}+\trchb)||_{L^1_{u}L^2(S)} \\
\leq& C(1+||(\eta,\etab)||_{L^\infty_{u}L^\infty(S)})^2(\sum_{i\leq 1}||\nabla^{i}(\psi_{\Hb},\trchb)||_{L^1_{u}L^2(S)})\\
&+C\int_0^u||\nabla(\eta,\etab)||_{L^2(S_{u',\ub})}||(\psi_{\Hb},\trchb)||_{L^\infty(S_{u',\ub})}du'\\
\leq &C(\mathcal O_0,\mathcal R(S)[\betab])+ C\int_0^u||\nabla\etab||_{L^2(S_{u',\ub})}||(\psi_{\Hb},\trchb)||_{L^\infty(S_{u',\ub})}du',
\end{split}
\end{equation*}
where the first term is bounded using Propositions \ref{L4Ricci1}, \ref{L4Ricci3} and \ref{LinftyRicci1}. Therefore, by Proposition \ref{transport}, we have, for every $\ub$,
\begin{equation*}
\begin{split}
&\sum_{i\leq 1}||\nabla^i\etab||_{L^\infty_u L^2(S_{u,\ub})}\\
\leq &C(\mathcal O_0,\mathcal R(S)[\betab])+ C\int_0^u||\nabla\etab||_{L^2(S_{u',\ub})}||(\psi_{\Hb},\trchb)||_{L^\infty(S_{u',\ub})}du'.
\end{split}
\end{equation*}
By Gronwall's inequality, we have
\begin{equation*}
\begin{split}
&\sum_{i\leq 1}||\nabla^i\etab||_{L^\infty_u L^2(S_{u,\ub})}\\
\leq &C(\mathcal O_0,\mathcal R(S)[\betab])\exp(\int_0^u||(\psi_{\Hb},\trchb)||_{L^\infty(S_{u',\ub})}du').
\end{split}
\end{equation*}
The right hand side satisfies the desired bound by Propositions \ref{L4Ricci1} and \ref{L4Ricci3}.
\end{proof}
Recall that by Proposition \ref{LinftyRicci1} we now have a bound on $\mathcal O_{0,\infty}[\etab]$ independent of $\Delta_0$. This allows us to prove the $\mathcal O_{1,4}[\etab]$ estimates. However, unlike the $L^\infty(S)$ estimates for $\etab$ and $L^2(S)$ estimates for $\nabla\etab$, the $L^4(S)$ control that we prove at this point for $\nabla\etab$ grows linearly in $\mathcal R$. This bound will be improved in the next subsection.
\begin{proposition}\label{L4Ricci4}
Assume
$$\mathcal R <\infty,\quad\tilde{\mathcal O}_{3,2}<\infty,\quad\mathcal O_{2,2}<\infty.$$
Then there exists $\epsilon_0=\epsilon_0(\mathcal R, \tilde{\mathcal O}_{3,2}, \mathcal O_{2,2},\Delta_0)$ such that whenever $\epsilon\leq \epsilon_0$, 
\[
 \mathcal O_{1,4}[\etab]\leq C(\mathcal O_0,\mathcal R(S))(1+\mathcal R).
\]
This estimate is linear in the $\mathcal R$ norm and is independent of $\Delta_0$.
\end{proposition}
\begin{proof}
Recall that we have,
$$\nabla_{3}\nabla\etab=\sum_{i_1+i_2=1}(\eta+\etab)^{i_1}\nabla^{i_2}\betab+\sum_{i_1+i_2+i_3=1}(\eta+\etab)^{i_1}\nabla^{i_2}(\eta+\etab)\nabla^{i_3}(\psi_{\Hb}+\trchb).$$
As in the proof of Proposition \ref{L4Ricci40}, we notice that in this equation, when two $\etab$'s appear in a term, neither of them has a derivative. Fix $\ub$. Now, we estimate each of the terms. Firstly, the term with curvature:
\begin{equation*}
\begin{split}
&||\sum_{i_1+i_2\leq 1}(\eta+\etab)^{i_1}\nabla^{i_2}\betab||_{L^1_{u}L^4(S)} \\
\leq& C(1+||(\eta,\etab)||_{L^\infty_{u}L^\infty(S)})(\sum_{i\leq 1}||\nabla^{i}\betab||_{L^1_{u}L^4(S)})\\
\leq &C(\mathcal O_0,\mathcal R(S))\sum_{i\leq 2}||\nabla^{i}\betab||_{L^2_{u}L^2(S)} \\
\leq &C(\mathcal O_0,\mathcal R(S))\mathcal R.
\end{split}
\end{equation*}
We then control the nonlinear term in the Ricci coefficients:
\begin{equation*}
\begin{split}
&||\sum_{i_1+i_2+i_3\leq 1}(\eta+\etab)^{i_1}\nabla^{i_2}(\eta+\etab)\nabla^{i_3}(\psi_{\Hb}+\trchb)||_{L^1_{u}L^4(S)} \\
\leq& C(1+||(\eta,\etab)||_{L^\infty_{u}L^\infty(S)})^2(\sum_{i\leq 1}||\nabla^{i}(\psi_{\Hb},\trchb)||_{L^1_{u}L^4(S)})\\
&+C\int_0^u||\nabla(\eta,\etab)||_{L^4(S_{u',\ub})}||(\psi_{\Hb},\trchb)||_{L^\infty(S_{u',\ub})}du'\\
\leq &C(\mathcal O_0,\mathcal R(S))+ C\int_0^u||\nabla\etab||_{L^4(S_{u',\ub})}||(\psi_{\Hb},\trchb)||_{L^\infty(S_{u',\ub})}du'.
\end{split}
\end{equation*}
Therefore, by Proposition \ref{transport}, we have, for every $\ub$,
\begin{equation*}
\begin{split}
&||\nabla\etab||_{L^\infty_u L^4(S)}\\
\leq &C(\mathcal O_0,\mathcal R(S))(1+\mathcal R)+ C\int_0^u||\nabla\etab||_{L^4(S_{u',\ub})}||(\psi_{\Hb},\trchb)||_{L^\infty(S_{u',\ub})}du'.
\end{split}
\end{equation*}
By Gronwall's inequality, we have
\begin{equation*}
\begin{split}
&||\nabla\etab||_{L^\infty_u L^4(S)}\\
\leq &C(\mathcal O_0,\mathcal R(S))(1+\mathcal R)\exp(\int_0^u||(\psi_{\Hb},\trchb)||_{L^\infty(S_{u',\ub})}du').
\end{split}
\end{equation*}
The right hand side satisfies the desired bound by Propositions \ref{L4Ricci1} and \ref{L4Ricci3}.
\end{proof}
We now estimate the $\displaystyle\sum_{i\leq 1}\mathcal O_{i,4}$ norm of $\psi_H$.
\begin{proposition}\label{L4Ricci5}
Assume
$$\mathcal R <\infty,\quad\tilde{\mathcal O}_{3,2}<\infty,\quad\mathcal O_{2,2}<\infty.$$
Then there exists $\epsilon_0=\epsilon_0(\mathcal R, \tilde{\mathcal O}_{3,2}, \mathcal O_{2,2},\Delta_0)$ such that whenever $\epsilon\leq \epsilon_0$, 
$$\sum_{i\leq 1}\mathcal O_{i,4}[\chih,\omega]\leq C(\mathcal O_0,\mathcal R(S)[\betab]).$$
In particular, this estimate is independent of $\Delta_0$.
\end{proposition}

\begin{proof}
Consider the following equations for $\psi_H\in\{\chih,\omega\}$:
$$\nabla_3\psi_H=\nabla\eta+\rhoc+\psi_H(\trchb+\psi_{\Hb})+\psi(\psi+\psi_{\Hb}).$$
As before, we commute the equations with angular derivatives:
\begin{equation*}
\begin{split}
\nabla_{3}\nabla\psi_H
=&\sum_{i_1+i_2=1}(\eta+\etab)^{i_1}\nabla^{i_2}(\rhoc+\nabla\eta)\\
&+\sum_{i_1+i_2+i_3+i_4=1}(\eta+\etab)^{i_1}(\nabla^{i_2}\psi_H\nabla^{i_3}(\trchb+\psi_{\Hb})+\nabla^{i_2}\psi\nabla^{i_3}(\psi+\psi_{\Hb})).
\end{split}
\end{equation*}
We bound each of the terms in $L^1_u L^4(S)$.
First, we look at the curvature term:
\begin{equation*}
\begin{split}
&||\sum_{i_1+i_2\leq 1}(\eta+\etab)^{i_1}\nabla^{i_2}\rhoc||_{L^1_{u}L^4(S)} \\
\leq& C(1+||(\eta,\etab)||_{L^\infty_{u}L^\infty(S)})(\sum_{i\leq 1}||\nabla^{i}\rhoc||_{L^\infty_{\ub}L^1_{u}L^4(S)})\\
\leq &C(\Delta_0)\sum_{i\leq 2}||\nabla^{i}\rhoc||_{L^2_{u}L^2(S)} \\
\leq &C(\Delta_0,\mathcal R).
\end{split}
\end{equation*}
The term containing $\nabla^2\eta$ can be estimated analogously:
\begin{equation*}
\begin{split}
&||\sum_{i_1+i_2\leq 1}(\eta+\etab)^{i_1}\nabla^{i_2+1}\eta||_{L^1_{u}L^4(S)} \\
\leq& C(1+||(\eta,\etab)||_{L^\infty_{u}L^\infty(S)})(\sum_{i\leq 1}||\nabla^{i+1}\eta||_{L^\infty_{\ub}L^1_{u}L^4(S)})\\
\leq &C(\Delta_0)\sum_{i\leq 2}||\nabla^{i+1}\eta||_{L^2_{u}L^2(S)} \\
\leq &C(\Delta_0,\tilde{\mathcal O}_{3,2}[\eta],\mathcal O_{2,2}[\eta]).
\end{split}
\end{equation*}
Then, we control the terms containing $\psi_H$:
\begin{equation*}
\begin{split}
&||\sum_{i_1+i_2+i_3\leq 1}(\eta+\etab)^{i_1}\nabla^{i_2}\psi_H\nabla^{i_3}(\psi_{\Hb}+\trchb)||_{L^1_{u}L^4(S)} \\
\leq& C||(\eta,\etab)||_{L^\infty_{u}L^\infty(S)}\int_0^u ||\psi_{H}||_{L^\infty(S_{u',\ub})}\sum_{i\leq 1}||\nabla^{i}(\psi_{\Hb},\trchb)||_{L^4(S_{u',\ub})} du' \\
&+C||(\eta,\etab)||_{L^\infty_{u}L^\infty(S)}\int_0^u \sum_{i\leq 1}||\nabla^{i}\psi_{H}||_{L^4(S_{u',\ub})}||(\psi_{\Hb},\trchb)||_{L^\infty(S_{u',\ub})} du'\\
\leq &C(\mathcal O_0,\mathcal R(S)[\betab])\int_0^u \sum_{i_1\leq 1}||\nabla^{i_1}\psi_{H}||_{L^4(S_{u',\ub})}\sum_{i_2\leq 1}||\nabla^{i_2}(\psi_{\Hb},\trchb)||_{L^4(S_{u',\ub})} du'.
\end{split}
\end{equation*}
In the above, we noticed that $\eta$ and $\etab$ obey estimates from Propositions \ref{L4Ricci1} and \ref{L4Ricci40} that depend only on $\mathcal O_0$ and $\mathcal R(S)[\betab]$. For the terms not containing $\psi_H$, we can bound directly using the bootstrap assumption (\ref{BA1}),
\begin{equation*}
\begin{split}
&||\sum_{i_1+i_2+i_3\leq 1}(\eta+\etab)^{i_1}\nabla^{i_2}\psi\nabla^{i_3}(\psi+\psi_{\Hb})||_{L^1_{u}L^4(S)} \\
\leq &C(\sum_{i_1\leq 1} ||(\eta,\etab)||^{i_1}_{L^\infty_u L^\infty(S)})\sum_{i_2\leq 1}||\nabla^{i_2}\psi||_{L^\infty_uL^4(S)}\sum_{i_3\leq 1}||\nabla^{i_3}(\psi+\psi_{\Hb})||_{L^1_uL^4(S)} \\
\leq &C(\Delta_0).
\end{split}
\end{equation*}
Therefore, by Proposition \ref{transport},
\begin{equation*}
\begin{split}
&\sum_{i\leq 1}||\nabla^i\psi_H||_{L^4(S_{u,\ub})}\\
\leq &C\sum_{i\leq 1}||\nabla^i\psi_H||_{L^4(S_{0,\ub})}+C(\mathcal R,\tilde{\mathcal O}_{3,2}[\eta],\mathcal O_{2,2}[\eta],\Delta_0)\\
&+C(\mathcal O_0,\mathcal R(S)[\betab])\int_0^u \sum_{i_1\leq 1}||\nabla^{i_1}\psi_{H}||_{L^4(S_{u',\ub})}\sum_{i_2\leq 1}||\nabla^{i_2}(\psi_{\Hb},\trchb)||_{L^4(S_{u',\ub})} du'.
\end{split}
\end{equation*}
By Gronwall's inequality, we have
\begin{equation*}
\begin{split}
&\sum_{i \leq 1}||\nabla^i\psi_H||_{L^4(S_{u,\ub})}\\
\leq &C(\sum_{i\leq 1}||\nabla^i\psi_H||_{L^4(S_{0,\ub})}+C(\mathcal R,\tilde{\mathcal O}_{3,2}[\eta],\mathcal O_{2,2}[\eta],\Delta_0))\\
&\times\exp(\int_0^u C(\mathcal O_0,\mathcal R(S)[\betab])\sum_{i\leq 1}||\nabla^{i}(\psi_{\Hb},\trchb)||_{L^4(S_{u',\ub})} du').
\end{split}
\end{equation*}
By Propositions \ref{L4Ricci1} and \ref{L4Ricci3}, we have
$$\exp(\int_0^u C(\mathcal O_0,\mathcal R(S)[\betab])\sum_{i\leq 1}||\nabla^{i}(\psi_{\Hb},\trchb)||_{L^4(S_{u',\ub})} du'\leq C(\mathcal O_0,\mathcal R(S)[\betab]).$$
Therefore, we have, for any $u,\ub$,
\begin{equation*}
\begin{split}
&\sum_{i\leq 1}||\nabla^i\psi_H||_{L^4(S_{u,\ub})}\\
\leq &C(\mathcal O_0,\mathcal R(S)[\betab])(\sum_{i\leq 1}||\nabla^i\psi_H||_{L^4(S_{0,\ub})}+C(\mathcal R,\tilde{\mathcal O}_{3,2}[\eta],\mathcal O_{2,2}[\eta],\Delta_0)).
\end{split}
\end{equation*}
Clearly the right hand side is independent of $u$. We first take supremum in $u$ and then take the $L^2_{\ub}$ norm to obtain
\begin{equation*}
\begin{split}
&\sum_{i\leq 1}||\nabla^i\psi_H||_{L^2_{\ub}L^\infty_uL^4(S)} \\
\leq &C(\mathcal O_0,\mathcal R(S)[\betab])(\sum_{i=0}^1||\nabla^i\psi_H||_{L^2_{\ub}L^4(S_{0,\ub})}+\epsilon^{\frac 12}C(\mathcal R,\tilde{\mathcal O}_{3,2}[\eta],\mathcal O_{2,2}[\eta],\Delta_0)).
\end{split}
\end{equation*}
By choosing $\epsilon$ sufficiently small depending on $\mathcal R,\tilde{\mathcal O}_{3,2},\mathcal O_{2,2},\Delta_0$, we have
$$\sum_{i\leq 1}||\nabla^i\psi_H||_{L^2_{\ub}L^\infty_uL^4(S)}\leq C(\mathcal O_0,\mathcal R(S)[\betab]).$$
\end{proof}
We then estimate the $\displaystyle\sum_{i\leq 1}\mathcal O_{i,4}$ norm of $\trch$. Although $\trch$ satisfies a $\nab_4$ equation, the term $\chih\chih$ appears on the right hand side and each of the $\chih$ factor has to be estimated in $L^2_{\ub}$. Therefore, the bound for this term does not have a smallness constant.
\begin{proposition}\label{L4Ricci2}
Assume
$$\mathcal R <\infty,\quad\tilde{\mathcal O}_{3,2}<\infty,\quad\mathcal O_{2,2}<\infty.$$
Then there exists $\epsilon_0=\epsilon_0(\mathcal R, \tilde{\mathcal O}_{3,2}, \mathcal O_{2,2},\Delta_0)$ such that whenever $\epsilon\leq \epsilon_0$, 
\[
 \sum_{i\leq 1}\mathcal O_{i,4}[\trch]\leq C(\mathcal O_0,\mathcal R(S)[\betab]).
\]
In particular, this estimate is independent of $\Delta_0$.
\end{proposition}
\begin{proof}
Using the null structure equations, we have an equation of the type
$$\nabla_{4}\trch=\psi\psi+\psi_H\psi+\chih\chih.$$
We also have the null structure equations commuted with angular derivatives:
$$\nabla_{4}\nabla^i\trch=\sum_{i_1+i_2+i_3+i_4=i}\nabla^{i_1}\psi^{i_2}\nabla^{i_3}\psi\nabla^{i_4}(\psi+\psi_H)+\sum_{i_1+i_2+i_3+i_4=i}\nabla^{i_1}(\eta,\etab)^{i_2}\nabla^{i_3}\chih\nabla^{i_4}\chih.$$
By Proposition \ref{transport}, in order to estimate $||\nabla^i\psi||_{L^\infty_uL^\infty_{\ub}L^4(S)}$, it suffices to estimate the initial data and the $||\cdot||_{L^\infty_uL^1_{\ub}L^4(S)}$ norm of the right hand side. Notice that all terms except the one with $\chih\chih$ have appeared in the Proposition \ref{L4Ricci1}. We estimate those terms in the same manner. Hence,
\begin{equation*}
\begin{split}
||\sum_{i_1+i_2+i_3\leq 1}\psi^{i_1}\nabla^{i_2}\psi\nabla^{i_3}\psi||_{L^\infty_{u}L^1_{\ub}L^4(S)} 
\leq &C\Delta_0(1+\Delta_0)^2\epsilon,
\end{split}
\end{equation*}
and
\begin{equation*}
\begin{split}
||\sum_{i_1+i_2+i_3\leq 1}\psi^{i_1}\nabla^{i_2}\psi\nabla^{i_3}\psi_{H}||_{L^\infty_{u}L^1_{\ub}L^4(S)} 
\leq &C\Delta_0(1+\Delta_0)^{2}\epsilon^{\frac{1}{2}}.
\end{split}
\end{equation*}
For the term with $\chih\chih$, using the estimates obtained in Propositions \ref{L4Ricci40} and \ref{L4Ricci5}, we have
\begin{equation*}
\begin{split}
&||\sum_{i_1+i_2+i_3\leq 1}(\eta,\etab)^{i_1}\nabla^{i_2}\chih\nabla^{i_3}\chih||_{L^\infty_{u}L^1_{\ub}L^4(S)} \\
\leq& C(\sum_{i_1\leq 2}||(\eta,\etab)||^{i_1}_{L^\infty_{u}L^\infty_{\ub}L^\infty(S)})(||\chih||_{L^\infty_{u}L^2_{\ub}L^\infty(S)})(\sum_{i_2\leq 1}||\nabla^{i_2}\chih||_{L^\infty_{u}L^2_{\ub}L^4(S)}) \\
\leq &C(\mathcal O_0,\mathcal R(S)[\betab]).
\end{split}
\end{equation*}
Hence,
\begin{equation*}
\begin{split}
\sum_{i\leq 1}\mathcal O_{i,4}[\trch]\leq &C(\mathcal O_0,\mathcal R(S)[\betab])+C\Delta_0(1+\Delta_0)^2\epsilon^{\frac{1}{2}}.
\end{split}
\end{equation*}
The proposition follows by choosing $\epsilon$ to be sufficiently small depending on $\Delta_0$.
\end{proof}

Clearly Propositions \ref{L4Ricci1}, \ref{L4Ricci3}, \ref{L4Ricci4}, \ref{L4Ricci5}, \ref{L4Ricci2} imply the following estimate for the $L^4$ norms of the Ricci coefficients:
\begin{proposition}\label{L4Ricci}
$$\mathcal R <\infty,\quad\mathcal R(S)<\infty,\quad\tilde{\mathcal O}_{3,2}<\infty,\quad\mathcal O_{2,2}<\infty.$$
There exists $\epsilon_0=\epsilon_0(\mathcal O_0, \mathcal R,\mathcal R(S), \tilde{\mathcal O}_{3,2},\mathcal O_{2,2})$ such that for all $\epsilon\leq \epsilon_0$, we have
$$\sum_{i\leq 1}\mathcal O_{i,4}[\trchb,\eta,\chibh,\omegab]\leq C(\mathcal O_0),$$
$$\sum_{i\leq 1}\mathcal O_{i,4}[\trch,\chih,\omega]\leq C(\mathcal O_0,\mathcal R(S)[\betab]),$$
and
$$\sum_{i\leq 1}\mathcal O_{i,4}[\etab]\leq C(\mathcal O_0,\mathcal R(S),\mathcal R).$$
Together with Sobolev embedding in Proposition \ref{Linfty}, the bootstrap assumptions (\ref{BA1}) can be improved under the assumptions on $\mathcal R,\mathcal R(S), \tilde{\mathcal O}_{3,2}$ and $\mathcal O_{2,2}$.
\end{proposition}
\begin{proof}
Let $\Delta_0\gg \max\{C(\mathcal O_0), C(\mathcal O_0,\mathcal R(S)[\betab]), C(\mathcal O_0,\mathcal R(S),\mathcal R)\}$, where the right hand side is the maximum of the constants in Propositions \ref{L4Ricci1}, \ref{L4Ricci3}, \ref{L4Ricci4}, \ref{L4Ricci5}, \ref{L4Ricci2}. Then, take $\epsilon_0$ sufficiently small so that the conclusions of Propositions \ref{L4Ricci1}, \ref{L4Ricci3}, \ref{L4Ricci4}, \ref{L4Ricci5}, \ref{L4Ricci2} hold. Then by the Sobolev embedding theorem from Proposition \ref{Linfty}, we have improved (\ref{BA1}). Since the choice of $\Delta_0$ depends only on $\mathcal O_0, \mathcal R,\mathcal R(S)$, the choice of $\epsilon_0$ depends only on $\mathcal O_0, \mathcal R,\mathcal R(S), \tilde{\mathcal O}_{3,2},\mathcal O_{2,2}$.
\end{proof}

\subsection{$L^2(S)$ Estimates for Second Derivatives of Ricci Coefficients}\label{secRicciL2}

We now estimate the $\mathcal O_{2,2}$ norm. We make the bootstrap assumption:
\begin{equation}\tag{A2}\label{BA2}
\mathcal O_{2,2}\leq \Delta_1,
\end{equation}
where $\Delta_1$ is a positive constant to be chosen later.

The proof of the estimates for the $\mathcal O_{2,2}$ norm is very similar to that for the $\displaystyle\sum_{i\leq 1}\mathcal O_{i,4}$ norms, except that we now need to use the $L^4$ control that was obtained in the previous subsection. From now on, we will assume $\epsilon\leq \epsilon_0$ as in Proposition \ref{L4Ricci}, where $\epsilon_0$ depends on $\mathcal O_0, \mathcal R, \tilde{\mathcal O}_{3,2}$, $\mathcal R(S)$ and also on $\Delta_1$. 

We first prove the estimates for $\nab^2\trchb$ and $\nab^2\eta$:

\begin{proposition}\label{L2Ricci1}
Assume
$$\mathcal R <\infty,\quad\tilde{\mathcal O}_{3,2}<\infty,\quad\mathcal R(S)<\infty.$$
Then there exists $\epsilon_0=\epsilon_0(\mathcal O_0, \mathcal R, \tilde{\mathcal O}_{3,2},\mathcal R(S), \Delta_1)$ such that whenever $\epsilon\leq \epsilon_0$, 
\[
 \mathcal O_{2,2}[\trchb,\eta]\leq C(\mathcal O_0).
\]
In particular, this estimate is independent of $\Delta_1$.
\end{proposition}
\begin{proof}
Using the null structure equations, we have
$$\nabla_{4}(\trchb,\eta)=\beta+\rhoc+\nabla\etab+\psi\psi+\psi_H\psi.$$
We use the null structure equations commuted with angular derivatives:
$$\nabla_{4}\nabla^i(\trchb,\eta)=\sum_{i_1+i_2+i_3=i}\nabla^{i_1}\psi^{i_2}\nabla^{i_3}(\beta+\rhoc+\nabla\etab)+\sum_{i_1+i_2+i_3+i_4=i}\nabla^{i_1}\psi^{i_2}\nabla^{i_3}\psi\nabla^{i_4}(\psi+\psi_H).$$
By Proposition \ref{transport}, in order to estimate $||\nabla^i(\trchb,\etab)||_{L^\infty_uL^\infty_{\ub}L^2(S)}$, it suffices to bound the initial data and the $||\cdot||_{L^\infty_{u}L^1_{\ub}L^2(S)}$ of the right hand side. We now estimate each of the terms in the equation. We first control the curvature term. As mentioned in the beginning of this subsection, the bounds are derived similarly as that for the $L^4$ norms, except we now need to use the $L^4(S)$ estimates proved above for $\nabla\psi$.
\begin{equation*}
\begin{split}
&||\sum_{i_1+i_2+i_3\leq 2}\nabla^{i_1}\psi^{i_2}\nabla^{i_3}(\beta,\rhoc)||_{L^\infty_{u}L^1_{\ub}L^2(S)} \\
\leq& C(\sum_{i_1\leq 3}||\psi||^{i_1}_{L^\infty_{u}L^\infty_{\ub}L^\infty(S)})(\sum_{i_2\leq 2}||\nabla^{i_2}(\beta,\rhoc)||_{L^\infty_{u}L^1_{\ub}L^2(S)}) \\
&+C(\sum_{i_1\leq 2}||\psi||_{L^\infty_{u}L^\infty_{\ub}L^\infty(S)}^{i_1})(\sum_{i_2\leq 2}||\nabla^{i_2}\psi||_{L^\infty_{u}L^\infty_{\ub}L^2(S)})(\sum_{i_3\leq 1}||\nabla^{i_3}(\beta,\rhoc)||_{L^\infty_{u}L^1_{\ub}L^4(S)})\\
\leq &C\epsilon^{\frac{1}{2}}(\sum_{i_1\leq 3}\sum_{i_2\leq 2}||\nabla^{i_2}\psi||_{L^\infty_{u}L^\infty_{\ub}L^2(S)}^{i_1})(\sum_{i_3\leq 2}||\nabla^{i_3}(\beta,\rhoc)||_{L^\infty_{u}L^2_{\ub}L^2(S)}) \\
\leq &C(\mathcal O_0,\mathcal R(S),\Delta_1,\mathcal R)\epsilon^{\frac{1}{2}}.
\end{split}
\end{equation*}
The term with $\nabla\etab$ instead of curvature can be estimated analogously, except for using the $\mathcal O_{2,2}$ and $\tilde{\mathcal O}_{3,2}$ norms instead of the $\mathcal R$ norm. Moreover, recall that the $\displaystyle\sum_{i\leq 1}\mathcal O_{i,4}[\etab]$ bounds that we have derived depend on $\mathcal R$. Hence the estimate below also depends on $\mathcal R$.
\begin{equation*}
\begin{split}
&||\sum_{i_1+i_2+i_3\leq 3}\nabla^{i_1}\psi^{i_2}\nabla^{i_3+1}\eta||_{L^\infty_{u}L^1_{\ub}L^2(S)} \\
\leq& C(\sum_{i_1\leq 3}||\psi||^{i_1}_{L^\infty_{u}L^\infty_{\ub}L^\infty(S)})(\sum_{i_2\leq 2}||\nabla^{i_2+1}\eta||_{L^\infty_{u}L^1_{\ub}L^2(S)}) \\
&+C(\sum_{i_1\leq 2}||\psi||_{L^\infty_{u}L^\infty_{\ub}L^\infty(S)}^{i_1})(\sum_{i_2\leq 2}||\nabla^{i_2}\psi||_{L^\infty_{u}L^\infty_{\ub}L^2(S)})(\sum_{i_3\leq 1}||\nabla^{i_3+1}\eta||_{L^\infty_{u}L^1_{\ub}L^4(S)})\\
\leq &C(\sum_{i_1\leq 3}\sum_{i_2\leq 2}||\nabla^{i_2}\psi||_{L^\infty_{u}L^\infty_{\ub}L^2(S)}^{i_1})(\sum_{i\leq 3}||\nabla^{i}\eta||_{L^\infty_{u}L^2_{\ub}L^2(S)}) \\
\leq &C(\mathcal O_0,\mathcal R(S),\Delta_1,\tilde{\mathcal O}_{3,2},\mathcal R)\epsilon^{\frac{1}{2}}.
\end{split}
\end{equation*}
We now move on to the lower order terms. We first control the lower order terms that contain $\psi_{H}$: 
\begin{equation*}
\begin{split}
&||\sum_{i_1+i_2+i_3+i_4\leq 2}\nabla^{i_1}\psi^{i_2}\nabla^{i_3}\psi\nabla^{i_4}\psi_H||_{L^\infty_{u}L^1_{\ub}L^2(S)} \\
\leq& C\epsilon^{\frac 12}(\sum_{i_1\leq 3}||\psi||^{i_1}_{L^\infty_{u}L^\infty_{\ub}L^\infty(S)})(\sum_{i_2\leq 2}||\nabla^{i_2}\psi_H||_{L^\infty_{u}L^2_{\ub}L^2(S)}) \\
&+C\epsilon^{\frac 12}(\sum_{i_1\leq 1}||\nabla^{i_1}\psi||_{L^\infty_{u}L^\infty_{\ub}L^4(S)})(\sum_{i_2\leq 1}||\nabla^{i_2}\psi_H||_{L^\infty_{u}L^2_{\ub}L^4(S)})\\
&+C\epsilon^{\frac 12}(\sum_{i_1\leq 1}||\psi||^{i_1}_{L^\infty_{u}L^\infty_{\ub}L^\infty(S)})(\sum_{i_2\leq 2}||\nabla^{i_2}\psi||_{L^\infty_{\ub}L^2_{u}L^2(S)})(||\psi_H||_{L^\infty_{u}L^2_{\ub}L^\infty(S)})\\
\leq &C(\mathcal O_0,\mathcal R(S),\mathcal R)\epsilon^\frac{1}{2}(\sum_{i\leq 2}||\nabla^{i}\psi_H||_{L^\infty_uL^2_{\ub}L^2(S)}+ \sum_{i\leq 2}||\nabla^{i}\psi||_{L^\infty_{\ub}L^\infty_{u}L^2(S)})\\
\leq &C(\mathcal O_0,\mathcal R(S),\mathcal R,\Delta_1)\epsilon^{\frac{1}{2}}.
\end{split}
\end{equation*}
The remaining lower order terms can be estimated in the same way since the norms for $\psi$ are stronger than those for $\psi_H$
\begin{equation*}
\begin{split}
||\sum_{i_1+i_2+i_3+i_4\leq 2}\nabla^{i_1}\psi^{i_2}\nabla^{i_3}\psi\nabla^{i_4}\psi||_{L^\infty_{u}L^1_{\ub}L^2(S)} 
\leq &C(\mathcal O_0,\mathcal R(S),\mathcal R,\Delta_1)\epsilon^{\frac{1}{2}}.
\end{split}
\end{equation*}
The conclusion thus follows from the above estimates and Proposition \ref{transport}, after choosing $\epsilon$ to be sufficiently small depending on $\mathcal O_0,\mathcal R(S),\Delta_1,\tilde{\mathcal O}_{3,2},\mathcal R$.
\end{proof}
We then estimate $\nab^2\psi_{\Hb}$. We again recall the notation that $\psi_{\Hb}\in\{\chibh, \omegab\}$.
\begin{proposition}\label{L2Ricci3}
Assume
$$\mathcal R <\infty,\quad\tilde{\mathcal O}_{3,2}<\infty,\quad\mathcal R(S)<\infty.$$
Then there exists $\epsilon_0=\epsilon_0(\mathcal O_0, \mathcal R, \tilde{\mathcal O}_{3,2},\mathcal R(S), \Delta_1)$ such that whenever $\epsilon\leq \epsilon_0$, 
\[
 \mathcal O_{2,2}[\chibh,\omegab]\leq C(\mathcal O_0).
\]
In particular, this estimate is independent of $\Delta_1$.
\end{proposition}
\begin{proof}
Using the null structure equations, for each $\psi_{\Hb}\in\{\chibh,\omegab\}$, we have an equation of the type
$$\nabla_{4}\psi_{\Hb}=\rhoc+\nabla\etab+(\psi+\psi_H)(\psi+\psi_{\Hb}).$$
We also use the null structure equations commuted with angular derivatives:
$$\nabla_{4}\nabla^i\psi_{\Hb}=\sum_{i_1+i_2+i_3=i}\nabla^{i_1}\psi^{i_2}\nabla^{i_3}(\rhoc+\nabla\etab)+\sum_{i_1+i_2+i_3+i_4=i}\nabla^{i_1}\psi^{i_2}\nabla^{i_3}(\psi+\psi_H)\nabla^{i_4}(\psi+\psi_{\Hb}).$$
From the proof of Proposition \ref{L2Ricci1}, we have
\begin{equation*}
||\sum_{i_1+i_2+i_3\leq 2}\nabla^{i_1}\psi^{i_2}\nabla^{i_3}\rhoc||_{L^\infty_{u}L^1_{\ub}L^2(S)} \leq C(\mathcal O_0,\mathcal R(S),\Delta_1,\mathcal R)\epsilon^{\frac{1}{2}},
\end{equation*}
and
\begin{equation*}
||\sum_{i_1+i_2+i_3\leq 2}\nabla^{i_1}\psi^{i_2}\nabla^{i_3+1}\etab||_{L^\infty_{u}L^1_{\ub}L^2(S)} \leq C(\mathcal O_0,\mathcal R(S),\Delta_1,\tilde{\mathcal O}_{3,2},\mathcal R)\epsilon^{\frac{1}{2}},
\end{equation*}
and
\begin{equation*}
\begin{split}
&||\sum_{i_1+i_2+i_3+i_4\leq 2}\nabla^{i_1}\psi^{i_2}\nabla^{i_3}\psi\nabla^{i_4}\psi_H||_{L^\infty_{u}L^1_{\ub}L^2(S)}
\leq C(\mathcal O_0,\mathcal R(S),\mathcal R,\Delta_1)\epsilon^{\frac{1}{2}},
\end{split}
\end{equation*}
and
\begin{equation*}
\begin{split}
||\sum_{i_1+i_2+i_3+i_4\leq 2}\nabla^{i_1}\psi^{i_2}\nabla^{i_3}\psi\nabla^{i_4}\psi||_{L^\infty_{u}L^1_{\ub}L^2(S)} 
\leq &C(\mathcal O_0,\mathcal R(S),\mathcal R,\Delta_1)\epsilon^{\frac{1}{2}}.
\end{split}
\end{equation*}
It remains to estimate 
$$\sum_{i_1+i_2+i_3+i_4\leq 2}\nabla^{i_1}\psi^{i_2}\nabla^{i_3}\psi_H\nabla^{i_4}\psi_{\Hb},\quad \sum_{i_1+i_2+i_3+i_4\leq 2}\nabla^{i_1}\psi^{i_2}\nabla^{i_3}\psi_{\Hb}\nabla^{i_4}\psi.$$
For the first term, as in the proof of Proposition \ref{L4Ricci3}, we first fix $u$ and bound the $L^1_{\ub}L^2(S)$ norm for each fixed $u$.
\begin{equation*}
\begin{split}
&||\sum_{i_1+i_2+i_3+i_4\leq 2}\nabla^{i_1}\psi^{i_2}\nabla^{i_3}\psi_H\nabla^{i_4}\psi_{\Hb}||_{L^1_{\ub}L^2(S)} \\
\leq& C\epsilon^{\frac 12}(1+||\psi||_{L^\infty_{\ub}L^\infty(S)})^2(||\psi_{\Hb}||_{L^\infty_{\ub}L^\infty(S)})(\sum_{i\leq 2}||\nabla^{i}\psi_{H}||_{L^2_{\ub}L^2(S)}) \\
&+C\epsilon^{\frac 12}||\nabla\psi_{\Hb}||_{L^\infty_{\ub}L^4(S)}||\nabla\psi_{H}||_{L^2_{\ub}L^4(S)} \\
&+C\epsilon^{\frac 12}||\nab\psi||_{L^\infty_{\ub}L^2(S)}||\psi_{\Hb}||_{L^\infty_{\ub}L^\infty(S)}||\psi_{H}||_{L^2_{\ub}L^\infty(S)} \\
&+C\epsilon^{\frac 12}(1+||\psi||_{L^\infty_{\ub}L^\infty(S)})(\sum_{i\leq 2}||\nabla^{i}\psi_{\Hb}||_{L^\infty_{\ub}L^2(S)})(||\psi_{H}||_{L^2_{\ub}L^\infty(S)}) \\
\leq &C(\mathcal O_0,\mathcal R(S))\epsilon^\frac{1}{2}(\sum_{i\leq 2}||\nabla^{i}\psi_{\Hb}||_{L^\infty_{\ub}L^2(S)})(1+||\nabla^2\psi_{H}||_{L^2_{\ub}L^2(S)})\\
\leq &C(\mathcal O_0,\mathcal R(S))\epsilon^\frac{1}{2}(\sum_{i\leq 2}||\nabla^{i}\psi_{\Hb}||_{L^\infty_{\ub}L^2(S)})(1+\Delta_1).
\end{split}
\end{equation*}
Finally, we have the term $\displaystyle\sum_{i_1+i_2+i_3+i_4\leq 2}\nabla^{i_1}\psi^{i_2}\nabla^{i_3}\psi_{\Hb}\nabla^{i_4}\psi$. As before, we have, for each fixed $u$,
\begin{equation*}
\begin{split}
&||\sum_{i_1+i_2+i_3+i_4\leq 2}\nabla^{i_1}\psi^{i_2}\nabla^{i_3}\psi_{\Hb}\nabla^{i_4}\psi||_{L^1_{\ub}L^2(S)} \\
\leq &C(\mathcal O_0)\epsilon^\frac{1}{2}(\sum_{i\leq 2}||\nabla^{i}\psi_{\Hb}||_{L^\infty_{\ub}L^2(S)})(1+||\nabla^2\psi||_{L^2_{\ub}L^2(S)})\\
\leq &C(\mathcal O_0)\epsilon^\frac{1}{2}(\sum_{i\leq 2}||\nabla^{i}\psi_{\Hb}||_{L^\infty_{\ub}L^2(S)})(1+\epsilon^{\frac 12}\Delta_1).
\end{split}
\end{equation*}
Putting all these together, and using Proposition \ref{transport}, we have, for each $u$
\begin{equation*}
\begin{split}
&||\nabla^2(\chibh,\omegab)||_{L^\infty_{\ub}L^2(S)}\\
\leq &C(\mathcal O_0)+C(\mathcal O_0,\mathcal R(S),\Delta_1,\tilde{\mathcal O}_{3,2},\mathcal R)\epsilon^{\frac{1}{2}}+C(\mathcal O_0,\mathcal R(S),\Delta_1)\epsilon^\frac{1}{2}(\sum_{i\leq 2}||\nabla^{i}\psi_{\Hb}||_{L^\infty_{\ub}L^2(S)}).
\end{split}
\end{equation*}
Taking the $L^2$ norm in $u$ and using $\displaystyle \sum_{i\leq 2}||\nab^i\psi_{\Hb}||_{L^2_uL^\infty_{\ub}L^2(S)}\leq C(\mathcal O_0,\mathcal R(S),\Delta_1)$ from the bootstrap assumption \eqref{BA2}, we get
$$||\nabla^2(\chibh,\omegab)||_{L^2_uL^\infty_{\ub}L^2(S)}\leq C(\mathcal O_0)+C(\mathcal O_0,\mathcal R(S),\Delta_1,\tilde{\mathcal O}_{3,2},\mathcal R)\epsilon^{\frac{1}{2}}.$$
The conclusion follows from choosing $\epsilon$ sufficiently small depending on $\mathcal O_0,\mathcal R(S),\Delta_1,\tilde{\mathcal O}_{3,2},\mathcal R$.
\end{proof}

We now prove the estimates for $\nab^2\psi_H$. We recall our notation that $\psi_H\in\{\chih,\omega\}$.
\begin{proposition}\label{L2Ricci5}
Assume
$$\mathcal R <\infty,\quad\tilde{\mathcal O}_{3,2}<\infty,\quad\mathcal R(S)<\infty.$$
Then there exists $\epsilon_0=\epsilon_0(\mathcal O_0, \mathcal R, \tilde{\mathcal O}_{3,2},\mathcal R(S), \Delta_1)$ such that whenever $\epsilon\leq \epsilon_0$, 
$$\mathcal O_{2,2}[\chih,\omega]\leq C(\mathcal O_0,\mathcal R(S)).$$
In particular, this estimate is independent of $\Delta_1$.
\end{proposition}

\begin{proof}
Consider the following equations for $\psi_H\in\{\chih,\omega\}$:
$$\nabla_3\psi_H=\rhoc+\nabla\eta+\psi\psi+\psi_{\Hb}\psi+\psi_H(\trchb+\psi_{\Hb}).$$
As before, we commute the equations with angular derivatives:
\begin{equation*}
\begin{split}
\nabla_{3}\nabla^2\psi_H=&\sum_{i_1+i_2+i_3=2}\nabla^{i_1}\psi^{i_2}\nabla^{i_3}(\rhoc+\nabla\eta)\\
&+\sum_{i_1+i_2+i_3+i_4=2}\nabla^{i_1}\psi^{i_2}(\nabla^{i_3}\psi_H\nabla^{i_4}(\trchb+\psi_{\Hb})+\nabla^{i_3}\psi\nabla^{i_4}(\psi+\psi_{\Hb})).
\end{split}
\end{equation*}
We first consider the term involving the curvature component $\rhoc$:
\begin{equation*}
\begin{split}
&||\sum_{i_1+i_2+i_3\leq 2}\nabla^{i_1}\psi^{i_2}\nabla^{i_3}\rhoc||_{L^1_{u}L^2(S)} \\
\leq& C(1+\sum_{i_1\leq 1}\sum_{i_2\leq 2}||\nabla^{i_1}\psi||^{i_2}_{L^\infty_{u}L^4(S)})(\sum_{i_3\leq 2}||\nabla^{i_3}\rhoc||_{L^1_{u}L^2(S)})\\
\leq &C(\mathcal O_0,\mathcal R(S),\mathcal R)\sum_{i\leq 2}||\nabla^{i}\rhoc||_{L^2_{u}L^2(S)} \\
\leq &C(\mathcal O_0,\mathcal R(S),\mathcal R).
\end{split}
\end{equation*}
The term containing $\nabla^3\eta$ can be estimated in a similar fashion:
\begin{equation*}
\begin{split}
&||\sum_{i_1+i_2+i_3\leq 2}\nabla^{i_1}\psi^{i_2}\nabla^{i_3+1}\eta||_{L^1_{u}L^2(S)} \\
\leq& C(1+\sum_{i_1\leq 1}\sum_{i_2\leq 2}||\nabla^{i_1}\psi||^{i_2}_{L^\infty_{u}L^4(S)})(\sum_{i_3\leq 2}||\nabla^{i_3+1}\eta||_{L^1_{u}L^2(S)})\\
\leq &C(\mathcal O_0,\mathcal R(S),\mathcal R)\sum_{i\leq 3}||\nabla^{i}\eta||_{L^2_{u}L^2(S)} \\
\leq &C(\mathcal O_0,\mathcal R(S),\mathcal R,\tilde{\mathcal O}_{3,2}[\eta],\Delta_1).
\end{split}
\end{equation*}
We now move to lower order terms. First, we control the terms in which both $\psi_H$ and $\psi_{\Hb}$ appear:
\begin{equation*}
\begin{split}
&||\sum_{i_1+i_2+i_3+i_4\leq 2}\nabla^{i_1}\psi^{i_2}\nabla^{i_3}\psi_H\nabla^{i_4}\psi_{\Hb}||_{L^1_{u}L^2(S)} \\
\leq& C(\sum_{i_1\leq 2}||\psi||^{i_1}_{L^\infty_{u}L^\infty(S)})(||\psi_{H}||_{L^\infty_{u}L^\infty(S)})(\sum_{i_2\leq 2}||\nabla^{i_2}\psi_{\Hb}||_{L^1_{u}L^2(S)}) \\
&+C(\sum_{i_1\leq 2}||\psi||^{i_1}_{L^\infty_{u}L^\infty(S)})(\sum_{i_2\leq 1}||\nabla^{i_2}\psi_{H}||_{L^\infty_{u}L^4(S)})(\sum_{i_3\leq 1}||\nabla^{i_3}\psi_{\Hb}||_{L^1_{u}L^4(S)}) \\
&+C\int_0^u ||\nabla^2\psi_{H}||_{L^2(S_{u',\ub})}||\psi_{\Hb}||_{L^\infty(S_{u',\ub})} du' \\
&+C||\nab\psi||_{L^\infty_{u}L^2(S)}||\psi_{H}||_{L^\infty_{u}L^\infty(S)}||\psi_{\Hb}||_{L^1_{u}L^\infty(S)} \\
\leq &C(\mathcal O_0,\mathcal R(S))(\sum_{i\leq 1}||\nabla^i\psi_{H}||_{L^\infty_{u}L^4(S)})
+C\int_0^u ||\nabla^2\psi_{H}||_{L^2(S_{u',\ub})}||\psi_{\Hb}||_{L^\infty(S_{u',\ub})} du',
\end{split}
\end{equation*}
by Propositions \ref{L4Ricci} and \ref{L2Ricci3}. The term with $\psi_H$ and $\trchb$ can be bounded in a similar fashion:
\begin{equation*}
\begin{split}
&||\sum_{i_1+i_2+i_3+i_4\leq 2}\nabla^{i_1}\psi^{i_2}\nabla^{i_3}\psi_H\nabla^{i_4}\trchb||_{L^1_{u}L^2(S)} \\
\leq &C(\mathcal O_0,\mathcal R(S))(\sum_{i\leq 1}||\nabla^i\psi_{H}||_{L^\infty_{u}L^4(S)})+C\int_0^u ||\nabla^2\psi_{H}||_{L^2(S_{u',\ub})}||\trchb||_{L^\infty(S_{u',\ub})} du',
\end{split}
\end{equation*}
by Propositions \ref{L4Ricci} and \ref{L2Ricci1}.
Then, we estimate the term with $\psi$ and $\psi_{\Hb}$. 
\begin{equation*}
\begin{split}
&||\sum_{i_1+i_2+i_3+i_4\leq 2}\nabla^{i_1}\psi^{i_2}\nabla^{i_3}\psi\nabla^{i_4}\psi_{\Hb}||_{L^1_{u}L^2(S)} \\
\leq& C(\sum_{i_1\leq 2}||\psi||^{i_1}_{L^\infty_{u}L^\infty(S)})(||\psi||_{L^\infty_{u}L^\infty(S)})(\sum_{i_2\leq 2}||\nabla^{i_2}\psi_{\Hb}||_{L^1_{u}L^2(S)}) \\
&+C(\sum_{i_1\leq 2}||\psi||^{i_1}_{L^\infty_{u}L^\infty(S)})(\sum_{i_2\leq 1}||\nabla^{i_2}\psi||_{L^\infty_{u}L^4(S)})(\sum_{i_3\leq 1}||\nabla^{i_3}\psi_{\Hb}||_{L^1_{u}L^4(S)}) \\
&+C||\nabla^2\psi||_{L^\infty_u L^2(S)}||\psi_{\Hb}||_{L^1_uL^\infty(S)} \\
\leq &C(\mathcal O_0,\mathcal R(S),\mathcal R,\Delta_1),
\end{split}
\end{equation*}
by Propositions \ref{L4Ricci}, \ref{L2Ricci1}, \ref{L2Ricci3} and the bootstrap assumption (\ref{BA2}). The term with only $\psi$ can also be controlled similarly:
\begin{equation*}
\begin{split}
&||\sum_{i_1+i_2+i_3+i_4\leq 2}\nabla^{i_1}\psi^{i_2}\nabla^{i_3}\psi\nabla^{i_4}\psi||_{L^1_{u}L^2(S)} \\
\leq& C(\sum_{i_1\leq 2}||\psi||^{i_1}_{L^\infty_{u}L^\infty(S)})(||\psi||_{L^\infty_{u}L^\infty(S)})(\sum_{i_2\leq 2}||\nabla^{i_2}\psi||_{L^\infty_{u}L^2(S)}) \\
&+C(\sum_{i_1\leq 2}||\psi||^{i_1}_{L^\infty_{u}L^\infty(S)})(\sum_{i_2\leq 1}||\nabla^{i_2}\psi||_{L^\infty_{u}L^4(S)})(\sum_{i_3\leq 1}||\nabla^{i_3}\psi||_{L^\infty_{u}L^4(S)}) \\
\leq &C(\mathcal O_0,\mathcal R(S),\mathcal R,\Delta_1),
\end{split}
\end{equation*}
by Propositions \ref{L4Ricci}, \ref{L2Ricci1} and the bootstrap assumption (\ref{BA2}). Therefore, by Proposition \ref{transport}, for fixed $u,\ub$,
\begin{equation*}
\begin{split}
&||\nabla^2\psi_H||_{L^2(S_{u,\ub})}\\
\leq &C||\nabla^2\psi_H||_{L^2(S_{0,\ub})}+C(\mathcal O_0,\mathcal R(S),\mathcal R,\tilde{\mathcal O}_{3,2}[\eta],\Delta_1)+C(\mathcal O_0,\mathcal R(S))(\sum_{i\leq 1}||\nabla^i\psi_{H}||_{L^\infty_{u}L^4(S)})\\
&+C\int_0^u ||\nabla^2\psi_{H}||_{L^2(S_{u',\ub})}||(\trchb,\psi_{\Hb})||_{L^\infty(S_{u',\ub})} du'.
\end{split}
\end{equation*}
By Gronwall's inequality,
\begin{equation*}
\begin{split}
&||\nabla^2\psi_H||_{L^2(S_{u,\ub})}\\
\leq &\left(C||\nabla^2\psi_H||_{L^2(S_{0,\ub})}+C(\mathcal O_0,\mathcal R(S),\mathcal R,\tilde{\mathcal O}_{3,2}[\eta],\Delta_1)+C(\mathcal O_0,\mathcal R(S))(\sum_{i\leq 1}||\nabla^i\psi_{H}||_{L^\infty_{u}L^4(S)})\right)\\
&\times\exp(||(\trchb,\psi_{\Hb})||_{L^1_uL^\infty(S)}).
\end{split}
\end{equation*}
By Propositions \ref{L2Ricci1} and \ref{L2Ricci3}, the norm inside the exponential function is bounded by $C(\mathcal O_0,\mathcal R(S))$. Thus, we have, for each fixed $\ub$
\begin{equation*}
\begin{split}
&||\nabla^2\psi_H||_{L^\infty_uL^2(S)}\\
\leq &C(\mathcal O_0,\mathcal R(S))\left(||\nabla^2\psi_H||_{L^2(S_{0,\ub})}+C(\mathcal O_0,\mathcal R(S),\mathcal R,\tilde{\mathcal O}_{3,2}[\eta],\Delta_1)\right.\\
&\left.\qquad\qquad\qquad\quad+C(\mathcal O_0,\mathcal R(S))(\sum_{i\leq 1}||\nabla^i\psi_{H}||_{L^\infty_{u}L^4(S)})\right).
\end{split}
\end{equation*}
We can now take the $L^2$ norm in $\ub$ to get
\begin{equation*}
\begin{split}
&||\nabla^2\psi_H||_{L^2_{\ub}L^\infty_uL^2(S)}\\
\leq &C(\mathcal O_0,\mathcal R(S))\left(||\nabla^2\psi_H||_{L^2_{\ub}L^2(S_{0,\ub})}+C(\mathcal O_0,\mathcal R(S),\mathcal R,\tilde{\mathcal O}_{3,2}[\eta],\Delta_1)\epsilon^{\frac 12}\right.\\
&\left.\quad\quad\quad\quad\quad\quad\quad\quad\quad\quad\quad\quad\quad+C(\mathcal O_0,\mathcal R(S))(\sum_{i\leq 1}||\nabla^i\psi_{H}||_{L^2_{\ub}L^\infty_{u}L^4(S)})\right)\\
\leq&C(\mathcal O_0,\mathcal R(S))(1+C(\mathcal O_0,\mathcal R(S),\mathcal R,\tilde{\mathcal O}_{3,2},\Delta_1)\epsilon^{\frac 12}),
\end{split}
\end{equation*}
using Proposition \ref{L4Ricci}. Choosing $\epsilon$ sufficiently small, we have
\begin{equation*}
\begin{split}
||\nabla^2\psi_H||_{L^2_{\ub}L^\infty_uL^2(S)}
\leq C(\mathcal O_0,\mathcal R(S)).
\end{split}
\end{equation*}
\end{proof}

We now estimate $\nab^2\trch$:
\begin{proposition}\label{L2Ricci2}
Assume
$$\mathcal R <\infty,\quad\tilde{\mathcal O}_{3,2}<\infty,\quad\mathcal R(S)<\infty.$$
Then there exists $\epsilon_0=\epsilon_0(\mathcal O_0, \mathcal R, \tilde{\mathcal O}_{3,2},\mathcal R(S), \Delta_1)$ such that whenever $\epsilon\leq \epsilon_0$, 
\[
 \mathcal O_{2,2}[\trch]\leq C(\mathcal O_0,\mathcal R(S)).
\]
In particular, this estimate is independent of $\Delta_1$.
\end{proposition}
\begin{proof}
Using the null structure equations, we have
$$\nabla_{4}\trch=\psi\psi+\psi_H\psi+\chih\chih.$$
We also have the null structure equations commuted with angular derivatives:
$$\nabla_{4}\nabla^i\trch=\sum_{i_1+i_2+i_3+i_4=i}\nabla^{i_1}\psi^{i_2}\nabla^{i_3}\psi\nabla^{i_4}(\psi+\psi_H)+\sum_{i_1+i_2+i_3+i_4=i}\nabla^{i_1}\psi^{i_2}\nabla^{i_3}\chih\nabla^{i_4}\chih.$$
By Proposition \ref{transport}, in order to estimate $||\nabla^i\trch||_{L^\infty_uL^\infty_{\ub}L^2(S)}$, it suffices to estimate the initial data and the $||\cdot||_{L^\infty_{u}L^1_{\ub}L^2(S)}$ of the right hand side. From the proof of Proposition \ref{L2Ricci1}, we have
\begin{equation*}
\begin{split}
&||\sum_{i_1+i_2+i_3+i_4\leq 2}\nabla^{i_1}\psi^{i_2}\nabla^{i_3}\psi\nabla^{i_4}\psi_H||_{L^\infty_{u}L^1_{\ub}L^2(S)}
\leq C(\mathcal O_0,\mathcal R(S),\mathcal R,\Delta_1)\epsilon^{\frac{1}{2}}.
\end{split}
\end{equation*}
and
\begin{equation*}
\begin{split}
||\sum_{i_1+i_2+i_3+i_4\leq 2}\nabla^{i_1}\psi^{i_2}\nabla^{i_3}\psi\nabla^{i_4}\psi||_{L^\infty_{u}L^1_{\ub}L^2(S)} 
\leq C(\mathcal O_0,\mathcal R(S),\mathcal R,\Delta_1)\epsilon^{\frac{1}{2}}.
\end{split}
\end{equation*}
The only new term compared which did not appear in the proof of Proposition \ref{L2Ricci1} is the term involving $\chih\chih$:
\begin{equation*}
\begin{split}
&||\sum_{i_1+i_2+i_3+i_4\leq 2}\nabla^{i_1}\psi^{i_2}\nabla^{i_3}\chih\nabla^{i_4}\chih||_{L^\infty_{u}L^1_{\ub}L^2(S)} \\
\leq& C(\sum_{i_1\leq 2}||\psi||^{i_1}_{L^\infty_{u}L^\infty_{\ub}L^\infty(S)})(\sum_{i_2\leq 2}||\nabla^{i_2}\chih||_{L^\infty_{u}L^2_{\ub}L^2(S)})(||\chih||_{L^\infty_{u}L^2_{\ub}L^\infty(S)}) \\
&+ C(\sum_{i_1\leq 2}||\psi||^{i_1}_{L^\infty_{u}L^\infty_{\ub}L^\infty(S)})(\sum_{i_2\leq 1}||\nabla^{i_2}\chih||_{L^\infty_{u}L^2_{\ub}L^4(S)})(||\nabla\chih||_{L^\infty_{u}L^2_{\ub}L^4(S)}) \\
&+C(\sum_{i_1\leq 1}||\nabla^{i_1}\psi||_{L^\infty_{u}L^\infty_{\ub}L^2(S)})(||\chih||^2_{L^2_{\ub}L^\infty_{u}L^\infty(S)})\\
\leq &C(\mathcal O_0,\mathcal R(S)),
\end{split}
\end{equation*}
using Propositions \ref{LinftyRicci1}, \ref{L4Ricci40}, \ref{L4Ricci} and \ref{L2Ricci5}. The conclusion thus follows from the above bounds and Proposition \ref{transport} by choosing $\epsilon$ appropriately small.
\end{proof}
We now prove the $L^2(S)$ control for $\nab^2\etab$, thus obtaining all the $\mathcal O_{2,2}$ estimates. As in Proposition \ref{L4Ricci4} where the $L^4(S)$ estimates for $\nab\etab$ were derived, we need to integrate in the $\nab_3$ direction and will not be able to gain a smallness constant. In order to get a bound independent of $\mathcal R$, instead of controlling $\nab^2\etab$ directly, we first estimate $\nab\mub$, where $\mub=-\div\etab-\rhoc$. We then obtain the desired bounds for $\nab^2\etab$ by elliptic estimates. This allows us to take only one derivative of the curvature components which can be controlled using the $\mathcal R(S)$ norm instead of the $\mathcal R$ norm.
\begin{proposition}\label{L2Ricci4}
Assume
$$\mathcal R <\infty,\quad\tilde{\mathcal O}_{3,2}<\infty,\quad\mathcal R(S)<\infty.$$
Then there exists $\epsilon_0=\epsilon_0(\mathcal O_0, \mathcal R, \tilde{\mathcal O}_{3,2},\mathcal R(S), \Delta_1)$ such that whenever $\epsilon\leq \epsilon_0$, 
\[
 \mathcal O_{2,2}[\etab]\leq C(\mathcal O_0,\mathcal R(S)).
\]
In particular, this estimate is independent of $\Delta_1$.
\end{proposition}

\begin{proof}
Recall that
$$\mub=-\div\etab-\rhoc.$$
We have the following equation:
$$\nabla_3 \mub=\psi(\rhoc,\sigmac,\betab)+\psi\nabla(\psi+\psi_{\Hb})+\psi_{\Hb}\nabla\psi+\psi\psi(\trchb+\psi_{\Hb})+\psi\chibh\chibh,$$
The mass aspect function $\mub$ is constructed so that there is no first derivative of curvature components in the equation. Moreover, the equation does not contain $\psi_H$. This cannot be derived from signature considerations alone, but requires the exact form of the $\nab_3\rhoc$ equation as shown in Section \ref{secsche}.

After commuting with angular derivatives and substituting the Codazzi equation
$$\betab=\sum_{i_1+i_2=1}\psi^{i_1}\nab^{i_2}(\trchb+\psi_{\Hb}),$$ 
we get
\begin{equation*}
\begin{split}
\nabla_{3}\nabla\mub=&\sum_{i_1+i_2+i_3=1}\psi^{i_1}\nabla^{i_2}\psi\nabla^{i_3}(\rhoc,\sigmac)+\sum_{i_1+i_2+i_3+i_4=2}\nabla^{i_1}\psi^{i_2}\nabla^{i_3}\psi\nabla^{i_4}(\trchb+\psi_{\Hb})\\
&+\sum_{i_1+i_2+i_3+i_4=2}\nabla^{i_1}\psi^{i_2}\nabla^{i_3}\chibh\nabla^{i_4}\chibh.
\end{split}
\end{equation*}
Fix $\ub$. We now estimate each of the terms. Firstly, the term with curvature:
\begin{equation*}
\begin{split}
&||\sum_{i_1+i_2+i_3\leq 1}\psi^{i_1}\nabla^{i_2}\psi\nabla^{i_3}(\rhoc,\sigmac)||_{L^1_{u}L^2(S)} \\
\leq& C(1+\sum_{i_1\leq 1}\sum_{i_2\leq 2}||\nabla^{i_1}\psi||^{i_2}_{L^\infty_{u}L^2(S)})(\sum_{i_3\leq 1}||\nabla^{i_3}(\rhoc,\sigmac)||_{L^\infty_{\ub}L^1_{u}L^2(S)})\\
\leq &C(\mathcal O_0,\mathcal R(S))
\end{split}
\end{equation*}
by Propositions \ref{L4Ricci40}, \ref{L4Ricci} and the definition of $\mathcal R(S)$.
We then estimate the nonlinear Ricci coefficient terms with at most one $\psi_H$:
\begin{equation*}
\begin{split}
&||\sum_{i_1+i_2+i_3+i_4\leq 2}\nabla^{i_1}\psi^{i_2}\nabla^{i_3}\psi\nabla^{i_4}(\trchb+\psi_{\Hb})||_{L^1_{u}L^2(S)} \\
\leq& C(1+\sum_{i_1\leq 1}\sum_{i_2\leq 2}||\nabla^{i_1}\psi||^{i_2}_{L^\infty_{u}L^2(S)})(\sum_{i\leq 1}||\nabla^{i}(\psi_{\Hb},\trchb)||_{L^1_{u}L^4(S)})\\
&+C||\psi||_{L^\infty_{u}L^\infty(S)}||\nabla^2(\psi_{\Hb},\trchb)||_{L^1_{u}L^2(S)}\\
&+C||\nab^2(\eta,\trch,\trchb)||_{L^\infty_{u}L^2(S)}||(\psi_{\Hb},\trchb)||_{L^1_{u}L^\infty(S)}\\
&+C\int_0^u||\nabla^2\etab||_{L^4(S_{u',\ub})}||(\psi_{\Hb},\trchb)||_{L^\infty(S_{u',\ub})}du'\\
\leq &C(\mathcal O_0,\mathcal R(S))+ C\int_0^u||\nabla^2\etab||_{L^2(S_{u',\ub})}||(\psi_{\Hb},\trchb)||_{L^\infty(S_{u',\ub})}du'
\end{split}
\end{equation*}
by Propositions \ref{LinftyRicci1}, \ref{L4Ricci40}, \ref{L4Ricci}, \ref{L2Ricci1} and \ref{L2Ricci3}. We control the nonlinear term with two $\chibh$:
\begin{equation*}
\begin{split}
&||\sum_{i_1+i_2+i_3+i_4\leq 1}\psi^{i_1}\nabla^{i_2}\psi\nabla^{i_3}\chibh\nabla^{i_4}\chibh||_{L^1_{u}L^2(S)} \\
\leq& C(1+\sum_{i_1\leq 1}\sum_{i_2\leq 2}||\nabla^{i_1}\psi||^{i_2}_{L^\infty_{u}L^2(S)})(\sum_{i_3\leq 1}||\nabla^{i_3}\chibh||_{L^2_{u}L^2(S)})(||\chibh||_{L^2_{u}L^\infty(S)})\\
\leq &C(\mathcal O_0,\mathcal R(S))
\end{split}
\end{equation*}
by Propositions \ref{LinftyRicci1}, \ref{L4Ricci40} and \ref{L2Ricci2}.
Therefore, by Proposition \ref{transport}, we have, for every $\ub$,
\begin{equation}\label{mub1est}
\begin{split}
&||\nabla\mub||_{L^\infty_u L^2(S)}\\
\leq &C(\mathcal O_0,\mathcal R(S))+ C\int_0^u||\nabla^2\etab||_{L^2(S_{u',\ub})}||(\psi_{\Hb},\trchb)||_{L^\infty(S_{u',\ub})}du'.
\end{split}
\end{equation}
We now use the div-curl system
\begin{equation*}
\begin{split}
\div\etab=&-\mub-\rhoc,\\
\curl\etab=&-\sigmac
\end{split}
\end{equation*}
together with the elliptic estimates in Proposition \ref{ellipticthm} to get the bound
\begin{equation*}
\begin{split}
&||\nabla^2\etab||_{L^\infty_{\ub}L^\infty_uL^2(S)}\\
\leq& C(\sum_{i_1\leq 1}||\nab^{i_1}K||_{L^\infty_{\ub}L^\infty_uL^2(S)})\\
&\qquad\times(\sum_{i_2\leq 1}||\nabla^{i_2}\mub||_{L^\infty_{\ub}L^\infty_uL^2(S)}+\sum_{i_3\leq 1}||\nabla^{i_3}(\rhoc,\sigmac)||_{L^\infty_{\ub}L^\infty_uL^2(S)}+||\etab||_{L^\infty_{\ub}L^\infty_uL^2(S)}).
\end{split}
\end{equation*}
Since the $\mathcal R(S)$ norm controls $\nab^i K$ and $\nab^i(\rhoc,\sigmac)$ in $L^\infty_{\ub}L^\infty_uL^2(S)$ for $i\leq 1$, we have
\begin{equation}\label{ellipticetab0}
\begin{split}
&||\nabla^2\etab||_{L^\infty_{\ub}L^\infty_uL^2(S)}\\
\leq& C(\mathcal O_0,\mathcal R(S))(1+\sum_{i_2\leq 1}||\nabla^{i_2}\mub||_{L^\infty_{\ub}L^\infty_uL^2(S)}+||\etab||_{L^\infty_{\ub}L^\infty_uL^2(S)}).
\end{split}
\end{equation}
This, together with \eqref{mub1est} and Propositions \ref{L4Ricci40}, \ref{L4Ricci}, implies that 
\begin{equation*}
\begin{split}
&||\nabla\mub||_{L^\infty_u L^2(S)}\\
\leq &C(\mathcal O_0,\mathcal R(S))+ C(\mathcal O_0,\mathcal R(S))\int_0^u(1+||\nabla\mub||_{L^2(S_{u',\ub})})||(\psi_{\Hb},\trchb)||_{L^\infty(S_{u',\ub})}du'.
\end{split}
\end{equation*}
By Gronwall's inequality and Proposition \ref{L4Ricci}, we have
\begin{equation*}
\begin{split}
&||\nabla\mub||_{L^\infty_u L^2(S)}\\
\leq &C(\mathcal O_0,\mathcal R(S))\exp(\int_0^u||(\psi_{\Hb},\trchb)||_{L^\infty(S_{u',\ub})}du').
\end{split}
\end{equation*}
By Proposition \ref{L4Ricci},
$$\exp(\int_0^u||(\psi_{\Hb},\trchb)||_{L^\infty(S_{u',\ub})}du')\leq C(\mathcal O_0,\mathcal R(S)).$$
Thus
$$||\nabla\mub||_{L^\infty_u L^2(S)}\leq C(\mathcal O_0,\mathcal R(S)).$$
By (\ref{ellipticetab0}) and Proposition \ref{L4Ricci}, this implies
$$||\nabla^2\etab||_{L^\infty_{\ub}L^\infty_uL^2(S)}\leq C(\mathcal O_0,\mathcal R(S)),$$
as claimed.
\end{proof}
Using the Sobolev embedding theorem given by Theorem \ref{L4}, we improve the estimates in Proposition \ref{L4Ricci4}:
\begin{proposition}\label{L4Ricci4b}
\[
 \sum_{i\leq 1}\mathcal O_{i,4}[\etab]\leq C(\mathcal O_0,\mathcal R(S)).
\]
\end{proposition}
Putting all the estimates in this subsection together, we obtain
\begin{proposition}\label{L2Ricci}
Assume
$$\mathcal R <\infty,\quad\tilde{\mathcal O}_{3,2}<\infty,\quad\mathcal R(S)<\infty.$$
Then there exists $\epsilon_0=\epsilon_0(\mathcal O_0, \mathcal R, \tilde{\mathcal O}_{3,2},\mathcal R(S))$ such that whenever $\epsilon\leq \epsilon_0$, 
$$\mathcal O_{2,2}\leq C(\mathcal O_0,\mathcal R(S)).$$ 
\end{proposition}
\begin{proof}
Let 
$$\Delta_1 \gg C(\mathcal O_0,\mathcal R(S)),$$
where $C(\mathcal O_0,\mathcal R(S))$ is taken to be the maximum of the bounds in Propositions \ref{L2Ricci1}, \ref{L2Ricci3}, \ref{L2Ricci5}, \ref{L2Ricci2} and \ref{L2Ricci4}. Then Propositions \ref{L2Ricci1}, \ref{L2Ricci3}, \ref{L2Ricci5}, \ref{L2Ricci2} and \ref{L2Ricci4} together show that the bootstrap assumption (\ref{BA2}) can be improved under appropriate choice of $\epsilon$. Since the choice of $\Delta_1$ depends only on $\mathcal O_0$ and $\mathcal R(S)$, we conclude that $\epsilon$ can be chosen to depend only on $\mathcal O_0, \mathcal R, \tilde{\mathcal O}_{3,2},\mathcal R(S)$.
\end{proof}

\subsection{$L^p(S)$ Estimates for Curvature Components}\label{secRS}
In this subsection, we estimate the $\mathcal R(S)$ norm. For this purpose, we make the bootstrap assumption
\begin{equation}\tag{A3}\label{BA3}
\mathcal R(S)\leq \Delta_2,
\end{equation}
where $\Delta_2$ is a positive constant to be chosen later.

We first prove the bounds on $\betab$.
\begin{proposition}\label{RSbetab}
Assume
$$\mathcal R <\infty,\quad\tilde{\mathcal O}_{3,2}<\infty.$$
Then there exists $\epsilon_0=\epsilon_0(\mathcal O_0,\mathcal R, \tilde{\mathcal O}_{3,2},\Delta_2)$ such that whenever $\epsilon\leq \epsilon_0$, 
$$\sum_{i\leq 1}||\nabla^i\betab||_{L^2_u L^\infty_{\ub} L^3(S)}\leq C(\mathcal R_0).$$
\end{proposition}
\begin{proof}
Recall the $\nabla_4$ Bianchi equation for $\betab$
$$\nab_4\betab=\sum_{i_1+i_2=1}\psi^{i_1}\nabla^{i_2}(\rhoc,\sigmac)+\sum_{i_1+i_2+i_3=1}\psi^{i_1}\nabla^{i_2}(\psi+\psi_H)\nabla^{i_3}(\psi+\psi_{\Hb}).$$
From this we get the estimates for $\betab$ in  $L^2_u L^\infty_{\ub} L^3(S)$. To see this, by Proposition \ref{transport}, we need to estimate
$$||\sum_{i_1+i_2=1}\psi^{i_1}\nabla^{i_2}(\rhoc,\sigmac)+\sum_{i_1+i_2+i_3=1}\psi^{i_1}\nabla^{i_2}(\psi+\psi_H)\nabla^{i_3}(\psi+\psi_{\Hb})||_{L^2_u L^1_{\ub}L^3(S)}.$$
We have, by Propositions \ref{L3trace} and \ref{L2Ricci},
\begin{equation*}
\begin{split}
&||\sum_{i_1+i_2=1}\psi^{i_1}\nabla^{i_2}(\rhoc,\sigmac)||_{L^2_u L^1_{\ub}L^3(S)}\\
\leq &\epsilon^{\frac 12}(1+||\psi||_{L^\infty_uL^\infty_{\ub}L^\infty(S)})(\sum_{i\leq 2}||\nab^i(\rhoc,\sigmac)||_{L^2_u L^2_{\ub}L^2(S)})\\
\leq& C(\mathcal O_0,\Delta_2)\epsilon^{\frac 12}\mathcal R
\end{split}
\end{equation*}
and by Propositions \ref{LinftyRicci1}, \ref{L4Ricci} and \ref{L4Ricci4b},
\begin{equation*}
\begin{split}
&||\sum_{i_1+i_2+i_3=1}\psi^{i_1}\nabla^{i_2}(\psi+\psi_H)\nabla^{i_3}(\psi+\psi_{\Hb})||_{L^2_u L^1_{\ub}L^3(S)}\\
\leq &C\epsilon^{\frac 12}(1+||\psi||_{L^\infty_uL^\infty_{\ub}L^\infty(S)})(\sum_{i_1\leq 1}||\nab^{i_1}(\psi+\psi_H)||_{L^2_{\ub}L^\infty_uL^4(S)})(\sum_{i_2\leq 1}||\nab^{i_2}(\psi+\psi_H)||_{L^2_{u}L^\infty_{\ub}L^4(S)})\\
\leq& C(\mathcal O_0,\Delta_2)\epsilon^{\frac 12}.
\end{split}
\end{equation*}
Therefore,
$$||\betab||_{L^2_u L^\infty_{\ub} L^3(S)}\leq C(\mathcal R_0)+C(\mathcal O_0,\Delta_2,\mathcal R)\epsilon^{\frac 12},$$
which implies
$$||\betab||_{L^2_u L^\infty_{\ub} L^3(S)}\leq C(\mathcal R_0)$$
for $\epsilon$ sufficiently small depending on $\mathcal O_0$, $\Delta_2$ and $\mathcal R$.
We now estimate $\nab\betab$. Commuting the $\nab_4\betab$ equation with angular derivatives, we have
$$\nab_4\nabla\betab=\sum_{i_1+i_2+i_3=2}\nab^{i_1}\psi^{i_2}\nabla^{i_3}(\rhoc,\sigmac)+\sum_{i_1+i_2+i_3+i_4=2}\nabla^{i_1}\psi^{i_2}\nabla^{i_3}(\psi+\psi_H)\nabla^{i_4}(\psi+\psi_{\Hb}).$$
After taking the $L^2_u$ norm, Proposition \ref{L3trace} implies that
\begin{equation}\label{betabL3}
\begin{split}
  ||\nabla\betab||_{L^2_u L^\infty_{\ub}L^3(S)}
\leq &C\left(||\nabla\betab||_{L^2_u L^3(S_{u,0})}+\epsilon^{\frac 14}||\nab^2\betab||_{L^2_uL^2_{\ub}L^2(S)}+\epsilon^{\frac 18}||\nab_4\nab\betab||_{L^2_uL^2_{\ub}L^2(S)}\right).
\end{split}
\end{equation}
The initial data is bounded by the initial data norm
$$||\nabla\betab||_{L^2_u L^3(S_{u,0})}\leq C\mathcal R_0.$$
Then, we note that by the definition of the norm $\mathcal R$, 
$$||\nabla^2\betab||_{L^2_u L^2_{\ub}L^2(S)}\leq C\mathcal R.$$
We estimate each term in the right of side of the equation for $\nab_4\nab\betab$ in $L^2_u L^2_{\ub}L^2(S)$. First, we control the curvature term:
\begin{equation*}
\begin{split}
&||\sum_{i_1+i_2+i_3\leq 2}\nab^{i_1}\psi^{i_2}\nabla^{i_3}(\rhoc,\sigmac)||_{L^2_u L^2_{\ub}L^2(S)}\\
\leq &C(\sum_{i_1\leq 2}||\psi||^{i_1}_{L^\infty_uL^\infty_{\ub}L^\infty(S)}) (\sum_{i_2\leq 2}||\nabla^{i_2}(\rhoc,\sigmac)||_{L^2_uL^2_{\ub}L^2(S)})\\
&+C||\nabla\psi||_{L^\infty_uL^\infty_{\ub}L^4(S)}||(\rhoc,\sigmac)||_{L^2_uL^2_{\ub}L^4(S)}\\
\leq &C(\mathcal O_0,\Delta_2)\mathcal R,
\end{split}
\end{equation*}
using Propositions \ref{L4Ricci} and \ref{L4Ricci4b}.
Then we bound the term with $\psi_H$ and $\psi_{\Hb}$. Using Propositions \ref{L4Ricci}, \ref{L4Ricci4b} and \ref{L2Ricci},
\begin{equation*}
\begin{split}
&||\sum_{i_1+i_2+i_3\leq 2}\nabla^{i_1}\psi^{i_2}\nabla^{i_3}\psi_{\Hb}\nabla^{i_4}\psi_H||_{L^2_u L^2_{\ub}L^2(S)}\\
\leq &C(\sum_{i_1\leq 2}||\psi||^{i_1}_{L^\infty_u L^\infty_{\ub}L^\infty(S)}) (\sum_{i_2\leq 2}||\nabla^{i_2}\psi_{\Hb}||_{L^\infty_{\ub}L^2_{u}L^2(S)})(\sum_{i_3\leq 1}||\nabla^{i_3}\psi_H||_{L^2_{\ub}L^\infty_u L^4(S)})\\
&+C(\sum_{i_1\leq 2}||\psi||^{i_1}_{L^\infty_u L^\infty_{\ub}L^\infty(S)}) (\sum_{i_2\leq 1}||\nabla^{i_2}\psi_{\Hb}||_{L^2_{u}L^\infty_{\ub}L^4(S)})(\sum_{i_3\leq 2}||\nabla^{i_3}\psi_H||_{L^\infty_uL^2_{\ub}L^2(S)})\\
&+C ||\nabla\psi||_{L^\infty_uL^\infty_{\ub}L^2(S)}||\psi_{\Hb}||_{L^2_uL^\infty_{\ub}L^\infty(S)}||\psi_H||_{L^2_{\ub}L^\infty_uL^\infty(S)}\\
\leq &C(\mathcal O_0,\Delta_2,\mathcal R).
\end{split}
\end{equation*}
Since $\psi$ satisfies stronger estimates than either $\psi_H$ or $\psi_{\Hb}$, we have
\begin{equation*}
\begin{split}
&||\sum_{i_1+i_2+i_3\leq 2}\nabla^{i_1}\psi^{i_2}\nabla^{i_3}(\psi+\psi_{\Hb})\nabla^{i_4}(\psi+\psi_H)||_{L^2_u L^2_{\ub}L^2(S)}\\
\leq &C(\mathcal O_0,\Delta_2,\mathcal R).
\end{split}
\end{equation*}
Putting the bounds together, we have
$$||\nab_4\nab\betab||_{L^2_uL^2_{\ub}L^2(S)}\leq C(\mathcal O_0,\Delta_2,\mathcal R).$$
Thus, (\ref{betabL3}) implies that
\begin{equation*}
\begin{split}
  ||\nabla\betab||_{L^2_u L^\infty_{\ub}L^3(S)}
\leq &C\left(\mathcal R_0+\epsilon^{\frac 14}\mathcal R+\epsilon^{\frac 18}C(\mathcal O_0,\Delta_2,\mathcal R)\right).
\end{split}
\end{equation*}
The proposition follows from choosing $\epsilon$ sufficiently small depending on $\mathcal O_0,\mathcal R, \tilde{\mathcal O}_{3,2},\Delta_2$.
\end{proof}
Since we have proved the estimate of $\mathcal R(S)[\betab]$ independent of the $\mathcal R$ norm, we get the following improved bounds on the Ricci coefficients:
\begin{proposition}\label{L4Ricciimproved}
Assume
$$\mathcal R <\infty,\quad\tilde{\mathcal O}_{3,2}<\infty.$$
Then there exists $\epsilon_0=\epsilon_0(\mathcal O_0,\mathcal R_0,\mathcal R, \tilde{\mathcal O}_{3,2},\Delta_2)$ such that whenever $\epsilon\leq \epsilon_0$, 
$$\mathcal O_{0,\infty}[\etab]\leq C(\mathcal O_0,\mathcal R_0),$$
$$\mathcal O_{1,2}[\etab]\leq C(\mathcal O_0,\mathcal R_0),$$
$$\sum_{i\leq 1}\mathcal O_{i,4}[\chih,\omega,\trch]\leq C(\mathcal O_0,\mathcal R_0).$$
\end{proposition}
\begin{proof}
This follows from substituting the bound for $\mathcal R(S)[\betab]$ from Proposition \ref{RSbetab} into the estimates from Propositions \ref{LinftyRicci1}, \ref{L4Ricci40}, \ref{L4Ricci5} and \ref{L4Ricci2}.
\end{proof}

Using this improvement, we prove the $\mathcal R(S)$ estimates for $\rhoc$ and $\sigmac$.
\begin{proposition}\label{RSrhocsigmac}
Assume
$$\mathcal R <\infty,\quad\tilde{\mathcal O}_{3,2}<\infty.$$
Then there exists $\epsilon_0=\epsilon_0(\mathcal O_0,\mathcal R_0,\mathcal R, \tilde{\mathcal O}_{3,2},\Delta_2)$ such that whenever $\epsilon\leq \epsilon_0$, 
$$\sum_{i\leq 1}||\nabla^i(\rhoc,\sigmac)||_{L^\infty_u L^\infty_{\ub} L^2(S)}\leq C(\mathcal O_0,\mathcal R_0).$$
\end{proposition}
\begin{proof}
Consider the $\nabla_4$ equations for $\rhoc$ and $\sigmac$:
\begin{equation*}
\begin{split}
\nabla_4(\rhoc,\sigmac)
=&\nabla\beta+\psi(\beta,\rhoc,\sigmac)+\psi\nabla(\psi+\psi_H)+\psi\psi(\psi+\psi_H)+\psi\chih\chih.
\end{split}
\end{equation*}
After commuting with angular derivatives, we get
\begin{equation*}
\begin{split}
\nabla_{4}\nabla(\rhoc,\sigmac)
=&\sum_{i_1+i_2=2}\nab^{i_1}\psi^{i_2}\nabla^{i_3}(\beta,\rhoc,\sigmac)+\sum_{i_1+i_2+i_3=1}\psi^{i_1}\nabla^{i_2}\psi_H\nabla^{i_3}(\rhoc,\sigmac)\\
&+\sum_{i_1+i_2+i_3+i_4=2}\nabla^{i_1}\psi^{i_2}\nabla^{i_3}\psi\nabla^{i_4}(\psi+\psi_H)+\sum_{i_1+i_2+i_3+i_4=1}\psi^{i_1}\nabla^{i_2}\psi\nabla^{i_3}\chih\nabla^{i_4}\chih.
\end{split}
\end{equation*}

By Proposition \ref{transport}, in order to estimate $\nabla^i(\rhoc,\sigmac)$ in $L^\infty_u L^\infty_{\ub} L^2(S)$, we need to estimate $\nab_4\nab^i(\rhoc,\sigmac)$ in $L^\infty_u L^1_{\ub}L^2(S)$. The first term with curvature can be estimated by
\begin{equation*}
\begin{split}
&||\sum_{i_1+i_2\leq 2}\nab^{i_1}\psi^{i_2}\nabla^{i_3}(\beta,\rhoc,\sigmac)||_{L^1_{\ub}L^2(S)}\\
\leq &C\epsilon^{\frac 12}(\sum_{i_1\leq 2}||\psi||^{i_1}_{L^\infty_{\ub}L^\infty(S)}) (\sum_{i_2\leq 2}||\nabla^{i_2}(\beta,\rhoc,\sigmac)||_{L^2_{\ub}L^2(S)})\\
&+C\epsilon^{\frac 12}||\nabla\psi||_{L^\infty_{\ub}L^4(S)} ||(\beta,\rhoc,\sigmac)||_{L^2_{\ub}L^4(S)}\\
\leq &C(\mathcal O_0,\mathcal R(S),\mathcal R)\epsilon^{\frac 12}\mathcal R
\end{split}
\end{equation*}
by Proposition \ref{L4Ricci}.
The second term with curvature can be estimated by
\begin{equation*}
\begin{split}
&||\sum_{i_1+i_2+i_3=1}\psi^{i_1}\nabla^{i_2}\psi_H\nabla^{i_3}(\rhoc,\sigmac)||_{L^1_{\ub}L^2(S)}\\
\leq &C\epsilon^{\frac 12}(\sum_{i_1\leq 1}||\psi||^{i_1}_{L^\infty_{\ub}L^2(S)})(\sum_{i_2\leq 2}||\nab^{i_2}\psi_H||_{L^2_{\ub}L^2(S)}) (\sum_{i_3\leq 1}||\nabla^{i_3}(\rhoc,\sigmac)||_{L^\infty_{\ub}L^2(S)})\\
\leq &C(\mathcal O_0,\mathcal R(S),\mathcal R)\epsilon^{\frac 12}
\end{split}
\end{equation*}
by Proposition \ref{L2Ricci}.
The nonlinear Ricci coefficient term with at most one $\psi_H$ can be controlled by
\begin{equation*}
\begin{split}
&||\sum_{i_1+i_2+i_3\leq 2}\nabla^{i_1}\psi^{i_2}\nabla^{i_3}\psi\nabla^{i_4}(\psi+\psi_H)||_{L^1_{\ub}L^2(S)}\\
\leq &\epsilon^{\frac 12}||\sum_{i_1+i_2+i_3\leq 2}\nabla^{i_1}\psi^{i_2}\nabla^{i_3}\psi\nabla^{i_4}(\psi+\psi_H)||_{L^2_{\ub}L^2(S)}\\
\leq &C\epsilon^{\frac 12}(\sum_{i_1\leq 3}||\psi||^{i_1}_{L^\infty_{\ub}L^\infty(S)}) (\sum_{i_2\leq 2}||\nabla^{i_2}(\psi,\psi_H)||_{L^2_{\ub}L^2(S)})\\
&+C\epsilon^{\frac 12}||\nabla^2 \psi||_{L^\infty_{\ub}L^2(S)}||(\psi,\psi_H)||_{L^2_{\ub}L^2(S)}\\
&+C\epsilon^{\frac 12} ||\nabla\psi||_{L^\infty_{\ub}L^4(S)} ||\nabla(\psi,\psi_H)||_{L^2_{\ub}L^4(S)}\\
\leq &C(\mathcal O_0,\mathcal R(S),\mathcal R)\epsilon^{\frac 12}
\end{split}
\end{equation*}
by Propositions \ref{L4Ricci} and \ref{L2Ricci}.
The remaining term containing $\chih\chih$ can be estimated using Proposition \ref{L4Ricciimproved}:
\begin{equation*}
\begin{split}
&||\sum_{i_1+i_2+i_3\leq 1}\psi^{i_1}\nabla^{i_2}\psi\nabla^{i_3}\chih\nabla^{i_4}\chih||_{L^1_{\ub}L^2(S)}\\
\leq &C(\sum_{i_1\leq 2}||\psi||^{i_1}_{L^\infty_{\ub}L^\infty(S)})(||\chih||_{L^2_{\ub}L^\infty(S)}) (\sum_{i_2\leq 1}||\nabla^{i_2}\psi_H||_{L^2_{\ub}L^2(S)})\\
&+C||\psi||_{L^\infty_{\ub}L^\infty(S)}||\nabla\psi||_{L^\infty_{\ub}L^2(S)} (||\psi_H||_{L^2_{\ub}L^\infty(S)})^2\\
\leq &C(\mathcal O_0,\mathcal R_0).
\end{split}
\end{equation*}
Therefore, by Proposition \ref{transport}
\begin{equation*}
\begin{split}
\sum_{i\leq 1}||\nabla^i(\rhoc,\sigmac)||_{L^2(S_{u,\ub})}
\leq &C(\mathcal O_0,\mathcal R_0)+\epsilon^{\frac 12}C(\mathcal O_0,\mathcal R_0,\mathcal R(S),\mathcal R).
\end{split}
\end{equation*}
By the bootstrap assumption (\ref{BA3}) on $\mathcal R(S)$, we can choose $\epsilon$ small depending on $\mathcal O_0$, $\mathcal R_0$, $\mathcal R$, $\tilde{\mathcal O}_{3,2}$ and $\Delta_2$ to conclude the proposition.
\end{proof}
Finally, we prove the bounds for the Gauss curvature. This will be used in the next subsection to carry out elliptic estimates.
\begin{proposition}\label{Kest}
Assume
$$\mathcal R <\infty,\quad\tilde{\mathcal O}_{3,2}<\infty.$$
Then there exists $\epsilon_0=\epsilon_0(\mathcal O_0,\mathcal R_0,\mathcal R, \tilde{\mathcal O}_{3,2},\Delta_2)$ such that whenever $\epsilon\leq \epsilon_0$, 
$$\sum_{i\leq 1}||\nabla^iK||_{L^\infty_u L^\infty_{\ub} L^2(S)}\leq C(\mathcal O_0,\mathcal R_0).$$
\end{proposition}
\begin{proof}
Consider the Gauss equation:
$$K=-\rhoc+\psi\psi.$$
We estimate each term on the right hand side. By Proposition \ref{RSrhocsigmac},
$$||\rhoc||_{L^\infty_u L^\infty_{\ub} L^2(S)}\leq C(\mathcal O_0,\mathcal R_0).$$
By Propositions \ref{L4Ricci} and \ref{L4Ricciimproved},
$$||\psi\psi||_{L^\infty_u L^\infty_{\ub} L^2(S)}\leq ||\psi||_{L^\infty_u L^\infty_{\ub} L^4(S)}^2\leq C(\mathcal O_0,\mathcal R_0).$$
\end{proof}
We can thus close the bootstrap assumption (\ref{BA3}) to prove the following estimates for $\mathcal R(S)$, under the assumption that $\mathcal R$ and $\tilde{\mathcal O}_{3,2}$ are bounded.
\begin{proposition}\label{RS}
Assume
$$\mathcal R<\infty,\quad\tilde{\mathcal O}_{3,2}<\infty.$$
Then there exists $\epsilon_0=\epsilon_0(\mathcal O_0,\mathcal R_0,\mathcal R,\tilde{\mathcal O}_{3,2})$ such that whenever $\epsilon\leq \epsilon_0$,
$$\mathcal R(S)\leq C(\mathcal O_0,\mathcal R_0).$$
\end{proposition}
\begin{proof}
Let 
$$\Delta_2 \gg C(\mathcal O_0,\mathcal R_0),$$
where $C(\mathcal O_0,\mathcal R_0)$ is taken to be the maximum of the bounds in Propositions \ref{RSbetab}, \ref{RSrhocsigmac} and \ref{Kest}. Hence, the choice of $\Delta_2$ depends only on $\mathcal O_0$ and $\mathcal R_0$. Thus, by Propositions \ref{RSbetab}, \ref{RSrhocsigmac} and \ref{Kest}, the bootstrap assumption (\ref{BA3}) can be improved by choosing $\epsilon$ sufficiently small depending on $\mathcal O_0,\mathcal R_0,\mathcal R$ and $\tilde{\mathcal O}_{3,2}$.
\end{proof}

Using Proposition \ref{RS}, we improve our estimates on the Ricci coefficients in Propositions \ref{L4Ricci} and \ref{L2Ricci} to get the following:
\begin{proposition}\label{Ricci}
Assume
$$\mathcal R<\infty,\quad\tilde{\mathcal O}_{3,2}<\infty.$$
Then there exists $\epsilon_0=\epsilon_0(\mathcal O_0,\mathcal R_0,\mathcal R,\tilde{\mathcal O}_{3,2})$ such that whenever $\epsilon\leq \epsilon_0$,
$$\sum_{i\leq 2} \mathcal O_{i,2}\leq C(\mathcal O_0,\mathcal R_0).$$
\end{proposition}

\subsection{Elliptic Estimates for Third Derivatives of the Ricci Coefficients}\label{secRicci32}

We now estimate the third angular derivatives of the Ricci coefficients. Introduce the bootstrap assumption:
\begin{equation}\tag{A4}\label{BA4}
\tilde{\mathcal O}_{3,2}\leq \Delta_3.
\end{equation}

The bounds for the third derivative of the Ricci coefficients cannot be achieved by the transport equations alone since there will be a loss of derivatives. We can however combine the transport equation bounds with the estimates derived from the Hodge systems as in \cite{KN}, \cite{Chr}, \cite{KlRo}. We first derive the control for some chosen combination of $\nab^3(\psi,\psi_H,\psi_{\Hb})+\nab^2(\beta,\rhoc,\sigmac)$ by the transport equations. Then we show that the estimates for the third derivatives of all the Ricci coefficients can be proved via elliptic estimates. We begin with the bounds for $\nab^3\trch$ and $\nab^3\chih$:

\begin{proposition}\label{trch3}
Assume 
$$\mathcal R<\infty.$$
Then there exists $\epsilon_0=\epsilon_0(\mathcal O_0,\mathcal R_0,\mathcal R,\Delta_3)$ such that whenever $\epsilon\leq \epsilon_0$,
$$\tilde{\mathcal O}_{3,2}[\trch,\chih]\leq C(\mathcal O_0)(1+\mathcal R[\beta]).$$
\end{proposition}
\begin{proof}
Consider the following equation:
$$\nabla_4 \trch=\chih\chih+\trch(\psi+\psi_H),$$
After commuting with angular derivatives, we have
$$\nabla_{4}\nabla^{3}\trch=\sum_{i_1+i_2+i_3+i_4=3}\nabla^{i_1}\psi^{i_2}\nabla^{i_3}\trch\nabla^{i_4}(\psi+\psi_H)+\sum_{i_1+i_2+i_3+i_4=3}\nabla^{i_1}\psi^{i_2}\nabla^{i_3}\chih\nabla^{i_4}\chih.$$
We estimate term by term. First, we bound the term with $\trch$ and $\psi_H$. Integrating in the $\ub$ direction, applying the Sobolev embedding Theorem in Propositions \ref{L4} and \ref{Linfty} and using Proposition \ref{Ricci}, we get
\begin{equation*}
\begin{split}
&||\sum_{i_1+i_2+i_3+i_4\leq 3}\nabla^{i_1}\psi^{i_2}\nabla^{i_3}\trch\nabla^{i_4}\psi_H||_{L^1_{\ub}L^2(S)} \\
\leq& C\epsilon^{\frac 12}(\sum_{i_1\leq 2}\sum_{i_2\leq 3}||\nab^{i_1}\psi||^{i_2}_{L^\infty_{\ub}L^2(S)})(\sum_{i_3\leq 1}||\nabla^{i_3}\trch||_{L^\infty_{\ub}L^4(S)})(\sum_{i_4\leq 3}||\nabla^{i_4}\psi_H||_{L^2_{\ub}L^2(S)})\\
&+C\epsilon^{\frac 12}(\sum_{i_1\leq 2}\sum_{i_2\leq 3}||\nab^{i_1}\psi||^{i_2}_{L^\infty_{\ub}L^2(S)})(\sum_{i_3\leq 3}||\nabla^{i_3}\trch||_{L^\infty_{\ub}L^2(S)})(\sum_{i_4\leq 1}||\nabla^{i_4}\psi_H||_{L^2_{\ub}L^4(S)})\\
\leq &C(\mathcal O_0,\mathcal R_0)\epsilon^{\frac{1}{2}}(1+\Delta_3).
\end{split}
\end{equation*}
Since $\psi$ satisfies stronger estimates than $\psi_H$, we have the same bounds for the term with $\trch$ and $\psi$: 
\begin{equation*}
\begin{split}
&||\sum_{i_1+i_2+i_3+i_4\leq 3}\nabla^{i_1}\psi^{i_2}\nabla^{i_3}\trch\nabla^{i_4}\psi||_{L^1_{\ub}L^2(S)} \\
\leq &C(\mathcal O_0,\mathcal R_0)\epsilon^{\frac{1}{2}}(1+\Delta_3).
\end{split}
\end{equation*}
Finally, we consider the term with $\chih\chih$:
\begin{equation}\label{chichi}
\begin{split}
&||\sum_{i_1+i_2+i_3+i_4\leq 3}\nabla^{i_1}\psi^{i_2}\nabla^{i_3}\chih\nabla^{i_4}\chih||_{L^1_{\ub}L^2(S)} \\
\leq& C(\sum_{i_1\leq 2}\sum_{i_2\leq 3}||\nab^{i_1}\psi||^{i_2}_{L^\infty_{\ub}L^2(S)})\int_0^{\ub}(\sum_{i_3\leq 2}||\nabla^{i_3}\chih||_{L^2(S_{u,\ub'})})(\sum_{i_4\leq 3}||\nabla^{i_4}\chih||_{L^2(S_{u,\ub'})})d\ub'\\
\leq &C(\mathcal O_0,\mathcal R_0)(1+\int_0^{\ub}(\sum_{i\leq 2}||\nabla^{i}\chih||_{L^2(S_{u,\ub'})})(||\nabla^3\chih||_{L^2(S_{u,\ub'})})d\ub').
\end{split}
\end{equation}
We now use the Codazzi equation
$$\div\chih=\frac 12\nabla\trch-\beta+\psi(\psi+\psi_H)$$
and apply elliptic estimates in Proposition \ref{elliptictraceless} to get
\begin{equation}\label{chihelliptic}
\begin{split}
&||\nabla^3\chih||_{L^2(S)}\\
\leq &C(\mathcal O_0,\mathcal R_0)(\sum_{i\leq 3}||\nabla^i\trch||_{L^2(S)}+\sum_{i\leq 2}||\nabla^i\beta||_{L^2(S)}\\
&\qquad\qquad\quad+\sum_{i_1+i_2\leq 2}||\nabla^{i_1}\psi\nabla^{i_2}(\psi+\psi_H)||_{L^2(S)} +||\chih||_{L^2(S)}).
\end{split}
\end{equation}
Notice that we can apply elliptic estimates using Proposition \ref{elliptictraceless} since we have bounds for the Gauss curvature by Proposition \ref{Kest}.
Therefore, 
\begin{equation*}
\begin{split}
&\int_0^{\ub}(\sum_{i\leq 2}||\nabla^{i}\chih||_{L^2(S_{u,\ub'})})(||\nabla^3\chih||_{L^2(S_{u,\ub'})})d\ub'\\
\leq &C(\mathcal O_0,\mathcal R_0)(1+ \int_0^{\ub}(\sum_{i\leq 2}||\nabla^{i}\chih||_{L^2(S_{u,\ub'})})(||\nabla^3\trch||_{L^2(S_{u,\ub'})})d\ub'+\mathcal R[\beta]).
\end{split}
\end{equation*}
Gathering all the estimates, we get
\begin{equation*}
\begin{split}
&||\nabla^3\trch||_{L^2(S_{u,\ub})}\\
\leq &C(\mathcal O_0,\mathcal R_0)(1+\epsilon^{\frac{1}{2}}\Delta_3+ \int_0^{\ub}(\sum_{i_2\leq 2}||\nabla^{i_2}\chih||_{L^2(S_{u,\ub'})})(||\nabla^3\trch||_{L^2(S_{u,\ub'})})d\ub'+\mathcal R[\beta]).
\end{split}
\end{equation*}
Gronwall's inequality implies that
\begin{equation*}
\begin{split}
&||\nabla^3\trch||_{L^2(S_{u,\ub})}\\
\leq &C(\mathcal O_0,\mathcal R_0)(1+\epsilon^{\frac{1}{2}}\Delta_3+\mathcal R[\beta])\exp( \int_0^{\ub}(\sum_{i\leq 2}||\nabla^{i}\chih||_{L^2(S_{u,\ub'})})d\ub')\\
\leq &C(\mathcal O_0,\mathcal R_0)(1+\epsilon^{\frac{1}{2}}\Delta_3+\mathcal R[\beta])\exp( \epsilon^{\frac 12}\sum_{i\leq 2}||\nabla^{i}\chih||_{L^\infty_uL^2_{\ub}L^2(S)}))\\
\leq &C(\mathcal O_0,\mathcal R_0)(1+\epsilon^{\frac 12}\Delta_3+\mathcal R[\beta]).
\end{split}
\end{equation*}
By choosing $\epsilon$ sufficiently small depending on $\mathcal O_0, \mathcal R_0$ and $\Delta_3$, 
$$||\nabla^3\trch||_{L^2(S_{u,\ub})}\leq C(\mathcal O_0,\mathcal R_0)(1+\mathcal R[\beta]).$$
This, together with (\ref{chihelliptic}), implies that
$$||\nabla^3\chih||_{L^\infty_uL^2_{\ub}L^2(S)}\leq C(\mathcal O_0)(1+\mathcal R[\beta]).$$
\end{proof}
We now prove estimates for $\nab^3\eta$. To do so, we first prove bounds for second derivatives of $\mu=-\div\eta-\rhoc$ and recover the control for $\nab^3\eta$ via elliptic estimates.
\begin{proposition}\label{mu3}
Assume 
$$\mathcal R<\infty.$$
Then there exists $\epsilon_0=\epsilon_0(\mathcal O_0,\mathcal R_0,\mathcal R,\Delta_3)$ such that whenever $\epsilon\leq \epsilon_0$,
$$\tilde{\mathcal O}_{3,2}[\mu,\eta]\leq C(\mathcal O_0)(1+\epsilon^{\frac{1}{2}}\tilde{\mathcal O}_{3,2}+\mathcal R).$$
\end{proposition}
\begin{proof}
Recall that
$$\mu=-\div\eta-\rhoc$$
and $\mu$ satisfies the following equations:
$$\nabla_4 \mu=\psi(\beta,\rhoc,\sigmac)+\psi\nabla(\psi+\psi_H)+\psi_H\nabla\psi+\psi\psi(\psi+\psi_H)+\psi\chih\chih.$$
It is important to note that $\betab$, $\psi_{\Hb}$ are absent in this equation. This cannot be derived from signature considerations alone, but requires an exact cancellation in the equation for $\nab_4\rhoc$ as indicated in Section \ref{secsche}.

After commuting with angular derivatives, and substituting the Codazzi equation
$$\beta=\sum_{i_1+i_2= 1}\psi^{i_1}\nab^{i_2}(\psi+\psi_H),$$
we get
\begin{equation*}
\begin{split}
\nabla_{4}\nabla^2\mu=&\sum_{i_1+i_2+i_3+i_4=2}\nabla^{i_1}\psi^{i_2}\nabla^{i_3}\psi\nabla^{i_4}(\rhoc,\sigmac)+\sum_{i_1+i_2+i_3+i_4=3}\nabla^{i_1}\psi^{i_2}\nabla^{i_3}\psi\nabla^{i_4}(\psi+\psi_H) \\
&+\sum_{i_1+i_2+i_3+i_4=3}\nabla^{i_1}\psi^{i_2}\nabla^{i_3}\chih\nabla^{i_4}\chih.
\end{split}
\end{equation*}
The term with curvature can be estimated using Proposition \ref{Ricci} by
\begin{equation*}
\begin{split}
&||\sum_{i_1+i_2+i_3+i_4\leq 2}\nabla^{i_1}\psi^{i_2}\nabla^{i_3}\psi\nabla^{i_4}(\rhoc,\sigmac)||_{L^\infty_uL^1_{\ub}L^2(S)} \\
\leq& C(\sum_{i_1\leq 2}\sum_{i_2\leq 3}||\nabla^{i_1}\psi||^{i_2}_{L^\infty_{\ub}L^\infty_{u}L^2(S)})(\sum_{i_3\leq 2}||\nabla^{i_3}(\rhoc,\sigmac)||_{L^\infty_{u}L^1_{\ub}L^2(S)}) \\
\leq &C(\mathcal O_0,\mathcal R_0)\epsilon^{\frac{1}{2}}\sum_{i\leq 2}||\nabla^{i}(\rhoc,\sigmac)||_{L^\infty_uL^2_{\ub}L^2(S)} \\
\leq &C(\mathcal O_0,\mathcal R_0)\epsilon^{\frac{1}{2}}\mathcal R.
\end{split}
\end{equation*}
We next consider the term with two $\chih$'s. By (\ref{chichi}) in the proof of Proposition \ref{trch3}, we have
\begin{equation*}
\begin{split}
||\sum_{i_1+i_2+i_3+i_4\leq 3}\nabla^{i_1}\psi^{i_2}\nabla^{i_3}\chih\nabla^{i_4}\chih||_{L^\infty_{u}L^1_{\ub}L^2(S)} \leq &C(\mathcal O_0,\mathcal R_0)(1+\tilde{\mathcal O}_{3,2}[\chih]).
\end{split}
\end{equation*}
Applying Proposition \ref{trch3}, we get
$$||\sum_{i_1+i_2+i_3+i_4\leq 3}\nabla^{i_1}\psi^{i_2}\nabla^{i_3}\chih\nabla^{i_4}\chih||_{L^\infty_{u}L^1_{\ub}L^2(S)}\leq C(\mathcal O_0,\mathcal R_0)(1+\mathcal R).$$
We then move to the remaining terms with at most one $\psi_H$. First, we look at the terms that do not contain $\psi_H\nab^3\psi$. These are the terms
$$\sum_{i_1+i_2+i_3+i_4=3,i_3\leq 2}\nabla^{i_1}\psi^{i_2}\nabla^{i_3}\psi\nabla^{i_4}\psi_{H},$$
and
$$\sum_{i_1+i_2+i_3+i_4=3}\nabla^{i_1}\psi^{i_2}\nabla^{i_3}\psi\nabla^{i_4}\psi.$$
The first term can be estimated using Proposition \ref{Ricci} by
\begin{equation}\label{mu.aux}
\begin{split}
&||\sum_{i_1+i_2+i_3+i_4=3,i_3\leq 2}\nabla^{i_1}\psi^{i_2}\nabla^{i_3}\psi\nabla^{i_4}\psi_{H}||_{L^\infty_{u}L^1_{\ub}L^2(S)} \\
\leq& C\epsilon^{\frac 12}(\sum_{i_1\leq 1}\sum_{i_2\leq 3}||\nab^{i_1}\psi||^{i_2}_{L^\infty_{u}L^\infty_{\ub}L^4(S)})(\sum_{i_3\leq 1}||\nabla^{i_3}\psi||_{L^\infty_{u}L^\infty_{\ub}L^4(S)})(\sum_{i_4\leq 3}||\nabla^{i_4}(\psi,\psi_H)||_{L^\infty_{u}L^2_{\ub}L^2(S)})\\
&+C\epsilon^{\frac 12}(\sum_{i_1\leq 3}||\psi||^{i_1}_{L^\infty_{u}L^\infty_{\ub}L^\infty(S)})(\sum_{i_2\leq 2}||\nabla^{i_2}\psi||_{L^\infty_{u}L^\infty_{\ub}L^2(S)})||(\psi,\psi_H)||_{L^\infty_{u}L^2_{\ub}L^\infty(S)}\\
\leq &C(\mathcal O_0,\mathcal R_0)\epsilon^{\frac{1}{2}}(1+\Delta_3).
\end{split}
\end{equation}
The second term can be controlled using Proposition \ref{Ricci} by
\begin{equation*}
\begin{split}
&||\sum_{i_1+i_2+i_3+i_4=3}\nabla^{i_1}\psi^{i_2}\nabla^{i_3}\psi\nabla^{i_4}\psi||_{L^\infty_{u}L^1_{\ub}L^2(S)} \\
\leq& C\epsilon^{\frac 12}(\sum_{i_1\leq 1}\sum_{i_2\leq 3}||\nab^{i_1}\psi||^{i_2}_{L^\infty_{u}L^\infty_{\ub}L^4(S)})(\sum_{i_3\leq 1}||\nabla^{i_3}\psi||_{L^\infty_{u}L^\infty_{\ub}L^4(S)})(\sum_{i_4\leq 3}||\nabla^{i_4}\psi||_{L^\infty_{u}L^2_{\ub}L^2(S)})\\
\leq &C(\mathcal O_0,\mathcal R_0)\epsilon^{\frac{1}{2}}(1+\Delta_3).
\end{split}
\end{equation*}
We now bound the terms $\psi_H\nab^3\psi$. If $\psi\in\{\trch,\trchb\}$, we can estimate in a similar fashion as \eqref{mu.aux}, since we have $L^\infty_u L^\infty_{\ub} L^2(S)$ estimates for $\nabla^3(\trch,\trchb)$:
\begin{equation*}
\begin{split}
||\psi_H\nabla^3(\trch,\trchb)||_{L^\infty_{u}L^1_{\ub}L^2(S)} 
\leq &C(\mathcal O_0,\mathcal R_0)\epsilon^{\frac{1}{2}}\Delta_3.
\end{split}
\end{equation*}
The remaining terms are of the form $\psi_H\nab^3(\eta,\etab)$. The difficulty in estimating these terms is the fact that using the $\tilde{\mathcal O}_{3,2}$ norm, $\nab^3\eta$ and $\nab^3\etab$ can only be estimated in $L^2(H)$ but not $L^2(S)$. Thus we need to estimate both $\nab^3(\eta,\etab)$ and $\psi_H$ in $L^2_{\ub}$ and will not have an extra smallness constant $\epsilon^{\frac 12}$. Therefore, instead of bounding $\nab^3(\eta,\etab)$ with the $\tilde{\mathcal O}_{3,2}$ norm, we apply elliptic estimates and control $\nab^2\mu$ in $L^\infty_u L^\infty_{\ub}L^2(S)$.

By the div-curl systems
$$\div\eta=-\mu-\rhoc,\quad \curl\eta=\sigmac,$$
$$\div\etab=-\mub-\rhoc,\quad \curl\etab=-\sigmac,$$
and the elliptic estimates given by Propositions \ref{ellipticthm} and \ref{Kest}, we have
\begin{equation}\label{etaelliptic}
||\nabla^3\eta||_{L^2(S)}\leq C(\sum_{i\leq 2}||\nabla^i\mu||_{L^2(S)}+\sum_{i\leq 2}||\nabla^i(\rhoc,\sigmac)||_{L^2(S)}+||\eta||_{L^2(S)}),
\end{equation}
$$||\nabla^3\etab||_{L^2(S)}\leq C(\sum_{i\leq 2}||\nabla^i\mub||_{L^2(S)}+\sum_{i\leq 2}||\nabla^i(\rhoc,\sigmac)||_{L^2(S)}+||\etab||_{L^2(S)}).$$
This implies
\begin{equation*}
\begin{split}
&||\psi_H\nabla^3(\eta,\etab)||_{L^1_{\ub}L^2(S)}\\
\leq &C ||\psi_H||_{L^2_{\ub}L^\infty(S)}||\nabla^3(\eta,\etab)||_{L^2_{\ub}L^2(S)}\\
\leq &C(\mathcal O_0,\mathcal R_0)(1+||\psi_H||_{L^2_{\ub}L^\infty(S)}(\sum_{i\leq 2}\epsilon^{\frac 12}||\nabla^i(\mu,\mub)||_{L^\infty L^2(S)}+\sum_{i\leq 2}||\nabla^i(\rhoc,\sigmac)||_{L^2_{\ub}L^2(S)}))\\
\leq &C(\mathcal O_0,\mathcal R_0)(1+\epsilon^{\frac 12}\Delta_3+\mathcal R).
\end{split}
\end{equation*}
Hence, gathering all the above estimates, by Proposition \ref{transport}, we have
$$||\nabla^3\mu||_{L^\infty_u L^\infty_{\ub}L^2(S)}\leq C(\mathcal O_0,\mathcal R_0)(1+\epsilon^{\frac{1}{2}}\Delta_3+\mathcal R).$$
By choosing $\epsilon$ sufficiently small depending on $\Delta_3$, we have
$$||\nabla^3\mu||_{L^\infty_u L^\infty_{\ub}L^2(S)}\leq C(\mathcal O_0,\mathcal R_0)(1+\mathcal R).$$
Therefore, by (\ref{etaelliptic}), we have
$$||\nabla^3\eta||_{L^\infty_uL^2_{\ub}L^2(S)}\leq C(\mathcal O_0,\mathcal R_0)(1+\mathcal R),$$
and
$$||\nabla^3\eta||_{L^\infty_{\ub}L^2_{u}L^2(S)}\leq C(\mathcal O_0,\mathcal R_0)(1+\mathcal R).$$
\end{proof}
We now estimate $\nab^3\omegab$:
\begin{proposition}\label{omegab3}
Assume 
$$\mathcal R<\infty.$$
Then there exists $\epsilon_0=\epsilon_0(\mathcal O_0,\mathcal R_0,\mathcal R,\Delta_3)$ such that whenever $\epsilon\leq \epsilon_0$,
$$\tilde{\mathcal O}_{3,2}[\kappab,\omegab,\omegab^{\dagger}]\leq C(\mathcal O_0,\mathcal R_0)(1+\mathcal R[\betab]).$$
\end{proposition}
\begin{proof}
Recall that $\omegab^{\dagger}$ is defined to be the solution to
$$\nab_4\omegab^{\dagger}=\frac 12 \sigmac$$
with zero initial data, i.e., $\omegab^{\dagger}=0$ on $\Hb_0$ and $\kappab$ is defined by
$$\kappab:=-\nab\omegab+^*\nab\omegab^{\dagger}-\frac 12\betab.$$
By the definition of $\omegab^{\dagger}$, it is easy to see that using Propositions \ref{transport} and \ref{Ricci},
$$\sum_{i\leq 2}||\nab^i\omegab^{\dagger}||_{L^2_uL^\infty_{\ub}L^2(S)}\leq C(\mathcal O_0,\mathcal R_0).$$
In other words, $\omegab^{\dagger}$ satisfies the same bounds as $\psi_{\Hb}$. In the remainder of the proof of this Proposition, we therefore also use $\psi_{\Hb}$ to denote $\omegab^{\dagger}$.

Consider the following equation for $\kappab$:
$$\nabla_4 \kappab=\psi(\rhoc+\sigmac)+\sum_{i_1+i_2+i_3=1}\psi^{i_1}\nabla^{i_2}(\psi+\psi_H)\nabla^{i_3}(\psi+\psi_{\Hb}).$$
After commuting with angular derivatives, we get
\begin{equation*}
\begin{split}
\nabla_{4}\nabla^2\kappab=&\sum_{i_1+i_2+i_3+i_4=2}\nabla^{i_1}\psi^{i_2}\nabla^{i_3}\psi\nabla^{i_4}(\rhoc+\sigmac)
+\sum_{i_1+i_2+i_3+i_4=3}\nabla^{i_1}\psi^{i_2}\nabla^{i_3}(\psi+\psi_H)\nabla^{i_4}(\psi+\psi_{\Hb}).\\
\end{split}
\end{equation*}
We estimate $\kappab$ in $L^2_uL^\infty_{\ub}L^2(S)$. By Proposition \ref{transport}, for each $u$, to bound $\nab^2\kappab$ in $L^\infty_{\ub}L^2(S)$, we need to estimate the right hand side in $L^1_{\ub}L^2(S)$. After taking the $L^2$ norm in $u$, we thus need to control the right hand side in $L^2_uL^1_{\ub}L^2(S)$. The term involving curvature has already been estimated in Proposition \ref{mu3} and can be controlled by
$$||\sum_{i_1+i_2+i_3+i_4=2}\nabla^{i_1}\psi^{i_2}\nabla^{i_3}\psi\nabla^{i_4}(\rhoc+\sigmac)||_{L^1_{\ub}L^2(S)}\leq C(\mathcal O_0,\mathcal R_0)\epsilon^{\frac{1}{2}}\mathcal R.$$
Thus,
$$||\sum_{i_1+i_2+i_3+i_4=2}\nabla^{i_1}\psi^{i_2}\nabla^{i_3}\psi\nabla^{i_4}(\rhoc+\sigmac)||_{L^2_u L^1_{\ub}L^2(S)}\leq C(\mathcal O_0,\mathcal R_0)\epsilon^{\frac{1}{2}}\mathcal R.$$
For the other terms, it suffices to consider 
$$\sum_{i_1+i_2+i_3+i_4=3}\nabla^{i_1}\psi^{i_2}\nabla^{i_3}\psi_H\nabla^{i_4}\psi_{\Hb}$$
since $\psi$ satisfies stronger estimates that either $\psi_H$ or $\psi_{\Hb}$. To this end, we have
\begin{equation*}
\begin{split}
&||\sum_{i_1+i_2+i_3+i_4 \leq 3}\nabla^{i_1}\psi^{i_2}\nabla^{i_3}\psi_H\nabla^{i_4}\psi_{\Hb}||_{L^2_{u}L^1_{\ub}L^2(S)} \\
\leq&C\epsilon^{\frac 12}||\sum_{i_1+i_2+i_3+i_4 \leq 3}\nabla^{i_1}\psi^{i_2}\nabla^{i_3}\psi_H\nabla^{i_4}\psi_{\Hb}||_{L^2_{u}L^2_{\ub}L^2(S)} \\
\leq& C\epsilon^{\frac 12}(\sum_{i_1\leq 2}\sum_{i_2\leq 3}||\nab^{i_1}\psi||^{i_2}_{L^\infty_{u}L^\infty_{\ub}L^2(S)})(\sum_{i_3\leq 3}||\nabla^{i_3}\psi_H||_{L^\infty_uL^2_{\ub}L^2(S)})(\sum_{i_4\leq 1}||\nabla^{i_4}\psi_{\Hb}||_{L^2_{u}L^\infty_{\ub}L^4(S)})\\
&+ C\epsilon^{\frac 12}(\sum_{i_1\leq 2}\sum_{i_2\leq 3}||\nab^{i_1}\psi||^{i_2}_{L^\infty_{u}L^\infty_{\ub}L^2(S)})(\sum_{i_3\leq 1}||\nabla^{i_3}\psi_H||_{L^2_{\ub}L^\infty_uL^4(S)})(\sum_{i_4\leq 3}||\nabla^{i_4}\psi_{\Hb}||_{L^\infty_{\ub}L^2_{u}L^2(S)})\\
\leq &C(\mathcal O_0,\mathcal R_0)\epsilon^{\frac 12}(1+\Delta_3)
\end{split}
\end{equation*}
by Propositions \ref{Ricci}.
Therefore, by Proposition \ref{transport},
\begin{equation*}
\begin{split}
||\nabla^2\kappab||_{L^2_uL^\infty_{\ub} L^2(S)}\leq &C(\mathcal O_0,\mathcal R_0)(1+\epsilon^{\frac 12}(\mathcal R+\Delta_3)).
\end{split}
\end{equation*}
By choosing $\epsilon$ sufficiently small depending on $\mathcal R$ and $\Delta_3$, we have
$$||\nabla^2\kappab||_{L^2_uL^\infty_{\ub} L^2(S)}\leq C(\mathcal O_0,\mathcal R_0).$$
Consider the div-curl system
$$\div\nabla\omegab=-\div\kappab-\frac 12\div\betab,$$
$$\curl\nabla\omegab=0,$$
$$\div\nabla\omegab^{\dagger}=-\curl\kappab-\frac 12\curl\betab,$$
$$\curl\nabla\omegab^{\dagger}=0.$$
By elliptic estimates given by Propositions \ref{ellipticthm} and \ref{Kest}, we have
$$||\nabla^3(\omegab,\omegab^{\dagger})||_{L^\infty_{\ub}L^2_uL^2(S)}\leq C(\mathcal O_0,\mathcal R_0)(1+\mathcal R[\betab]).$$
\end{proof}
In the remainder of this subsection, we consider the third derivatives of the Ricci coefficients that are estimated by integrating in the $u$ direction. We need to use the fact that the estimates derived in Propositions \ref{trch3}, \ref{mu3}, \ref{omegab3} are independent of $\Delta_3$. We now estimate $\nabla^3\trchb$ and $\nabla^3\chibh$:
\begin{proposition}\label{trchb3}
Assume 
$$\mathcal R<\infty.$$
Then there exists $\epsilon_0=\epsilon_0(\mathcal O_0,\mathcal R_0,\mathcal R,\Delta_3)$ such that whenever $\epsilon\leq \epsilon_0$,
$$\tilde{\mathcal O}_{3,2}[\trchb,\chibh]\leq C(\mathcal O_0,\mathcal R_0)(1+\mathcal R[\betab]).$$
\end{proposition}
\begin{proof}
Consider the following equation:
$$\nabla_3 \trchb=\chibh\chibh+\trchb(\trchb+\psi_{\Hb}).$$
After commuting with angular derivatives, we have
$$\nabla_{3}\nabla^3\trchb=\sum_{i_1+i_2+i_3+i_4=3}\nabla^{i_1}\psi^{i_2}\nabla^{i_3}\trchb\nabla^{i_4}(\trchb+\psi_{\Hb})+\sum_{i_1+i_2+i_3+i_4=3}\nabla^{i_1}\psi^{i_2}\nabla^{i_3}\chibh\nabla^{i_4}\chibh.$$
Fix $\ub$. We estimate term by term. First, by Proposition \ref{Ricci},
\begin{equation*}
\begin{split}
&||\sum_{i_1+i_2+i_3+i_4\leq 3}\nabla^{i_1}\psi^{i_2}\nabla^{i_3}\trchb\nabla^{i_4}\trchb||_{L^1_{u}L^2(S)} \\
\leq& C(\sum_{i_1\leq 2}\sum_{i_2\leq 3}||\nab^{i_1}\psi||^{i_2}_{L^\infty_{\ub}L^\infty_{u}L^2(S)})\int_0^u(\sum_{i_3\leq 1}||\nabla^{i_3}\trchb||_{L^4(S_{u',\ub})})(\sum_{i_4\leq 3}||\nabla^{i_4}\trchb||_{L^2(S_{u',\ub})})du'\\
\leq &C(\mathcal O_0,\mathcal R_0)(1+\int_0^u(\sum_{i\leq 1}||\nabla^{i}\trchb||_{L^4(S_{u',\ub})})||\nabla^{3}\trchb||_{L^2(S_{u',\ub})}du').
\end{split}
\end{equation*}
Then we bound the terms with one $\psi_{\Hb}$. We separate the cases where $\psi_{\Hb}=\omegab$ and $\psi_{\Hb}=\chibh$. First, for $\psi_{\Hb}=\omegab$:
\begin{equation*}
\begin{split}
&||\sum_{i_1+i_2+i_3+i_4\leq 3}\nabla^{i_1}\psi^{i_2}\nabla^{i_3}\trchb\nabla^{i_4}\omegab||_{L^1_{u}L^2(S)} \\
\leq& C(\sum_{i_1\leq 2}\sum_{i_2\leq 3}||\nab^{i_1}\psi||^{i_2}_{L^\infty_{\ub}L^\infty_{u}L^2(S)})(\sum_{i_3\leq 1}||\nabla^{i_3}\trchb||_{L^\infty_{\ub}L^\infty_{u}L^4(S)})(\sum_{i_4\leq 3}||\nabla^{i_4}\omegab||_{L^\infty_{\ub}L^2_uL^2(S)})\\
&+ C(\sum_{i_1\leq 2}\sum_{i_2\leq 3}||\nab^{i_1}\psi||^{i_2}_{L^\infty_{\ub}L^\infty_{u}L^2(S)})\int_0^u(\sum_{i_3\leq 1}||\nabla^{i_3}\omegab||_{L^4(S_{u',\ub})})(\sum_{i_4\leq 3}||\nabla^{i_4}\trchb||_{L^2(S_{u',\ub})})du'\\
\leq &C(\mathcal O_0,\mathcal R_0)(1+\mathcal R[\betab]+\int_0^u(\sum_{i\leq 1}||\nabla^{i}\omegab||_{L^4(S_{u',\ub})})||\nabla^{3}\trchb||_{L^2(S_{u',\ub})}du'),
\end{split}
\end{equation*}
where we have used Propositions \ref{Ricci} and \ref{omegab3}. Then, we consider the term with one $\psi_{\Hb}$, where $\psi_{\Hb}=\chibh$:
\begin{equation}\label{trchbchibh}
\begin{split}
&||\sum_{i_1+i_2+i_3+i_4\leq 3}\nabla^{i_1}\psi^{i_2}\nabla^{i_3}\trchb\nabla^{i_4}\chibh||_{L^1_{u}L^2(S)} \\
\leq& C(\sum_{i_1\leq 2}\sum_{i_2\leq 3}||\nab^{i_1}\psi||^{i_2}_{L^\infty_{\ub}L^\infty_{u}L^2(S)})\int_0^u(\sum_{i_3\leq 1}||\nabla^{i_3}\chibh||_{L^4(S_{u',\ub})})(\sum_{i_4\leq 3}||\nabla^{i_4}\trchb||_{L^2(S_{u',\ub})})du'\\
&+ C(\sum_{i_1\leq 2}\sum_{i_2\leq 3}||\nab^{i_1}\psi||^{i_2}_{L^\infty_{\ub}L^\infty_{u}L^2(S)})\int_0^u(\sum_{i_3\leq 1}||\nabla^{i_3}\trchb||_{L^4(S_{u',\ub})})(\sum_{i_4\leq 3}||\nabla^{i_4}\chibh||_{L^2(S_{u',\ub})})du'\\
\leq& C(\mathcal O_0,\mathcal R_0)(1+\int_0^u(\sum_{i\leq 1}||\nabla^{i}\chibh||_{L^4(S_{u',\ub})})(||\nabla^3\trchb||_{L^2(S_{u',\ub})})du')\\
&+ C(\mathcal O_0,\mathcal R_0)\int_0^u(\sum_{i\leq 1}||\nabla^{i}\trchb||_{L^4(S_{u',\ub})})(||\nabla^3\chibh||_{L^2(S_{u',\ub})})du',\\
\end{split}
\end{equation}
using Proposition \ref{Ricci}.
In order to control this, we need to use the Codazzi equation
$$\div\chibh=\frac 12\nabla\trchb+\betab+\psi(\psi+\psi_{\Hb})$$
and apply elliptic estimates using Propositions \ref{elliptictraceless} and \ref{Kest} to get
\begin{equation}\label{chibhelliptic}
\begin{split}
&||\nabla^3\chibh||_{L^2(S)}\\
\leq& C(\mathcal O_0,\mathcal R_0)(\sum_{i\leq 3}||\nabla^i\trchb||_{L^2(S)}+\sum_{i\leq 2}||\nabla^i\betab||_{L^2(S)}\\
&\qquad\qquad\quad+\sum_{i_1+i_2\leq 2}||\nabla^{i_1}\psi\nabla^{i_2}(\psi+\psi_{\Hb})||_{L^2(S)} +||\chibh||_{L^2(S)}).
\end{split}
\end{equation}
Using (\ref{chibhelliptic}), we can bound the second term in (\ref{trchbchibh}):
\begin{equation*}
\begin{split}
&\int_0^u(\sum_{i\leq 1}||\nabla^{i}\trchb||_{L^4(S_{u',\ub})})(||\nabla^3\chibh||_{L^2(S_{u',\ub})})du'\\
\leq &C\int_0^u(\sum_{i\leq 1}||\nabla^{i}\trchb||_{L^4(S_{u',\ub})})(||\nabla^3\trchb||_{L^2(S_{u',\ub})})du'\\
&+C(\sum_{i_1\leq 1}||\nabla^{i_1}\trchb||_{L^\infty_uL^4(S)})(\sum_{i_2\leq 2}||\nabla^{i_2}\betab||_{L^2_uL^2(S)})\\
&+C(\sum_{i_1\leq 1}||\nabla^{i_1}\trchb||_{L^\infty_uL^4(S)})(\sum_{i_2+i_3\leq 2}||\nabla^{i_2}\psi\nabla^{i_3}(\psi+\psi_H)||_{L^2_uL^2(S)}+||\chibh||_{L^2_uL^2(S)})\\
\leq &C(\mathcal O_0,\mathcal R_0)(1+\mathcal R[\betab]+\int_0^u(\sum_{i_1\leq 1}||\nabla^{i_1}\trchb||_{L^4(S_{u',\ub})})(||\nabla^3\trchb||_{L^2(S_{u',\ub})})du').
\end{split}
\end{equation*}
This, together with (\ref{trchbchibh}), implies that
\begin{equation*}
\begin{split}
&||\sum_{i_1+i_2+i_3+i_4\leq 3}\nabla^{i_1}\psi^{i_2}\nabla^{i_3}\trchb\nabla^{i_4}\chibh||_{L^1_{u}L^2(S)} \\
\leq &C(\mathcal O_0,\mathcal R_0)(1+\mathcal R[\betab]+\int_0^u(\sum_{i\leq 1}||\nabla^{i}(\trchb,\chibh)||_{L^4(S_{u',\ub})})(||\nabla^3\trchb||_{L^2(S_{u',\ub})})du').
\end{split}
\end{equation*}
Finally, we estimate the term with two $\psi_{\Hb}$'s. We note that the only such term is of the form $\chibh\chibh$. We control this term using (\ref{chibhelliptic}):
\begin{equation*}
\begin{split}
&||\sum_{i_1+i_2+i_3+i_4\leq 3}\nabla^{i_1}\psi^{i_2}\nabla^{i_3}\chibh\nabla^{i_4}\chibh||_{L^1_{u}L^2(S)} \\
\leq& C(\sum_{i_1\leq 2}\sum_{i_2\leq 3}||\nab^{i_1}\psi||^{i_2}_{L^\infty_{\ub}L^\infty_{u}L^2(S)})\int_0^u(\sum_{i_3\leq 1}||\nabla^{i_3}\chibh||_{L^4(S_{u',\ub})})(\sum_{i_4\leq 3}||\nabla^{i_4}\chibh||_{L^2(S_{u',\ub})})du'\\
\leq &C(\mathcal O_0,\mathcal R_0)+C\int_0^u(\sum_{i\leq 1}||\nabla^{i}\chibh||_{L^4(S_{u',\ub})})(||\nabla^3\trchb||_{L^2(S_{u',\ub})})du'\\
&+C(\sum_{i_1\leq 1}||\nabla^{i_1}\chibh||_{L^2_uL^4(S)})(\sum_{i_2\leq 2}||\nabla^{i_2}\betab||_{L^2_uL^2(S)})\\
&+C(\sum_{i_1\leq 1}||\nabla^{i_1}\chibh||_{L^2_uL^4(S)})(\sum_{i_2+i_3\leq 2}||\nabla^{i_2}\psi\nabla^{i_3}(\psi+\psi_H)||_{L^2_uL^2(S)}+||\chibh||_{L^2_uL^2(S)})\\
\leq &C(\mathcal O_0,\mathcal R_0)(1+\mathcal R[\betab]+\int_0^u(\sum_{i\leq 1}||\nabla^{i}\chibh||_{L^4(S_{u',\ub})})(||\nabla^3\trchb||_{L^2(S_{u',\ub})})du').
\end{split}
\end{equation*}
Therefore, by Proposition \ref{transport}, we have
\begin{equation*}
\begin{split}
&||\nabla^3\trchb||_{L^\infty_{u}L^2(S)} \\
\leq &C(\mathcal O_0,\mathcal R_0)(1+\mathcal R[\betab]+\int_0^u(\sum_{i\leq 1}||\nabla^{i}(\psi,\psi_{\Hb})||_{L^4(S_{u',\ub})})(||\nabla^3\trchb||_{L^2(S_{u',\ub})})du').
\end{split}
\end{equation*}
By Gronwall's inequality, we have
\begin{equation*}
\begin{split}
&||\nabla^3\trchb||_{L^\infty_{u}L^2(S)} \\
\leq &C(\mathcal O_0,\mathcal R_0)(1+\mathcal R[\betab])\exp(\int_0^u(\sum_{i\leq 1}||\nabla^{i}(\psi,\psi_{\Hb})||_{L^4(S_{u',\ub})})du')\\
\leq &C(\mathcal O_0,\mathcal R_0)(1+\mathcal R[\betab])\exp(C(\mathcal O_0,\mathcal R_0))\\
\leq &C(\mathcal O_0,\mathcal R_0)(1+\mathcal R[\betab]).
\end{split}
\end{equation*}
By (\ref{chibhelliptic}), this implies
$$||\nab^3\chibh||_{L^\infty_{\ub}L^2_uL^2(S)}\leq C(\mathcal O_0,\mathcal R_0)(1+\mathcal R[\betab]).$$
\end{proof}
We now control $\nabla^3\etab$.
\begin{proposition}\label{mub3}
Assume 
$$\mathcal R<\infty.$$
Then there exists $\epsilon_0=\epsilon_0(\mathcal O_0,\mathcal R_0,\mathcal R,\Delta_3)$ such that whenever $\epsilon\leq \epsilon_0$,
$$\tilde{\mathcal O}_{3,2}[\mub,\etab]\leq C(\mathcal O_0)(1+\mathcal R).$$
\end{proposition}
\begin{proof}
Recall that
$$\mub=-\div\etab-\rhoc.$$
We have the following equation:
$$\nabla_3 \mub=\psi(\rhoc,\sigmac,\betab)+\psi\nabla(\psi+\psi_{\Hb})+\psi_{\Hb}\nabla\psi+\psi\psi(\psi+\psi_{\Hb})+\psi\chibh\chibh.$$
After commuting with angular derivatives, we get
\begin{equation*}
\begin{split}
\nabla_{3}\nabla^2\mub=&\sum_{i_1+i_2+i_3+i_4=2}\nabla^{i_1}\psi^{i_2}\nabla^{i_3}\psi\nabla^{i_4}(\rhoc,\sigmac)+\sum_{i_1+i_2+i_3+i_4=3}\nabla^{i_1}\psi^{i_2}\nabla^{i_3}\psi\nabla^{i_4}(\psi+\psi_{\Hb})\\
&+\sum_{i_1+i_2+i_3+i_4=3}\nabla^{i_1}\psi^{i_2}\nabla^{i_3}\chibh\nabla^{i_4}\chibh.
\end{split}
\end{equation*}
We estimate every term in the above expression. First, we bound the term with curvature:
\begin{equation*}
\begin{split}
&||\sum_{i_1+i_2+i_3+i_4= 2}\nabla^{i_1}\psi^{i_2}\nabla^{i_3}\psi\nabla^{i_4}(\rhoc,\sigmac)||_{L^\infty_{\ub}L^1_{u}L^2(S)} \\
\leq& C(\sum_{i_1\leq 2}||\psi||^{i_1}_{L^\infty_{\ub}L^\infty_{u}L^\infty(S)})(\sum_{i_2\leq 2}||\nabla^{i_2}(\rhoc,\sigmac)||_{L^\infty_{\ub}L^1_{u}L^2(S)}) \\
&+C(\sum_{i_1\leq 2}||\psi||^{i_1}_{L^\infty_{\ub}L^\infty_{u}L^\infty(S)})(\sum_{i_2\leq 1}||\nabla^{i_2}\psi||_{L^\infty_{\ub}L^\infty_{u}L^4(S)})(\sum_{i_3\leq 1}||\nabla^{i_3}(\rhoc,\sigmac)||_{L^\infty_{\ub}L^1_{u}L^4(S)})\\
\leq &C(\mathcal O_0,\mathcal R_0)\sum_{i\leq 2}||\nabla^{i}(\rhoc,\sigmac)||_{L^\infty_{\ub}L^2_{u}L^2(S)} \\
\leq &C(\mathcal O_0,\mathcal R_0)\mathcal R.
\end{split}
\end{equation*}
We now move on to the terms with the Ricci coefficients. Notice that by Propositions \ref{trch3}, \ref{mu3} \ref{trchb3}, all the terms of the form $\nab^3\psi$ except $\nab^3\etab$ have been estimated. Thus, by Propositions \ref{Ricci}, \ref{trch3}, \ref{mu3} and \ref{trchb3},
\begin{equation*}
\begin{split}
&||\sum_{i_1+i_2+i_3+i_4=3}\nabla^{i_1}\psi^{i_2}\nabla^{i_3}(\psi+\psi_{\Hb})\nabla^{i_4}(\psi+\psi_{\Hb})||_{L^1_{u}L^2(S)} \\
\leq& C(\sum_{i_1\leq 2}\sum_{i_2\leq 3}||\nab^{i_1}\psi||^{i_2}_{L^\infty_{\ub}L^\infty_{u}L^2(S)})(\sum_{i_3\leq 2}||\nabla^{i_3}(\psi,\psi_{\Hb})||_{L^\infty_{\ub}L^2_{u}L^2(S)})\\
&\qquad\times(\sum_{i_4\leq 3}||\nabla^{i_4}(\trch,\trchb,\eta,\psi_{\Hb})||_{L^\infty_{\ub}L^2_uL^2(S)})\\
&+C\int_0^u (\sum_{i\leq 1}||\nab^i(\psi,\psi_{\Hb})||_{L^4(S_{u',\ub})}||\nab^3\etab||_{L^2(S_{u',\ub})}du'\\
\leq &C(\mathcal O_0,\mathcal R_0)(1+\mathcal R)+C\int_0^u (\sum_{i\leq 1}||\nab^i(\psi,\psi_{\Hb})||_{L^4(S_{u',\ub})}||\nab^3\etab||_{L^2(S_{u',\ub})}du'.
\end{split}
\end{equation*}
We control the last term using the div-curl system
$$\div \etab=-\mub-\rhoc,\quad\curl\etab=-\sigmac.$$
Applying elliptic estimates using Propositions \ref{ellipticthm} and \ref{Kest}, we get
\begin{equation}\label{ellipticetab}
||\nabla^3\etab||_{L^2(S)}\leq C(\mathcal O_0,\mathcal R_0)(\sum_{i\leq 2}||\nabla^i\mub||_{L^2(S)}+\sum_{i\leq 2}||\nabla^i(\rhoc,\sigmac)||_{L^2(S)}+||\etab||_{L^2(S)}).
\end{equation}
Thus, we have
\begin{equation*}
\begin{split}
&\int_0^u (\sum_{i\leq 1}||\nab^i(\psi,\psi_{\Hb})||_{L^4(S_{u',\ub})}||\nab^3\etab||_{L^2(S_{u',\ub})}du'\\
\leq& C(\mathcal O_0,\mathcal R_0)(1+\mathcal R+\int_0^u(\sum_{i\leq 1}||\nab^i(\psi,\psi_{\Hb})||_{L^4(S_{u',\ub})})||\nabla^2\mub||_{L^2(S_{u',\ub})}du').
\end{split}
\end{equation*}
Hence, by Proposition \ref{transport}, we have
$$||\nabla^2\mub||_{L^\infty_{\ub}L^2(S)}\leq C(\mathcal O_0,\mathcal R_0)(1+\mathcal R+\int_0^u(\sum_{i\leq 1}||\nab^i(\psi,\psi_{\Hb})||_{L^4(S_{u',\ub})})||\nabla^2\mub||_{L^2(S_{u',\ub})}du').$$
By Gronwall's inequality, we have
$$||\nabla^2\mub||_{L^\infty_{\ub}L^2(S)}\leq C(\mathcal O_0,\mathcal R_0)(1+\mathcal R)\exp(\int_0^u \sum_{i\leq 1}||\nab^i(\psi,\psi_{\Hb})||_{L^4(S_{u',\ub})}du').$$
By Proposition \ref{Ricci}, $\displaystyle \sum_{i\leq 1}||\nab^i(\psi,\psi_{\Hb})||_{L^4(S_{u,\ub})}$ is controlled by $C(\mathcal O_0,\mathcal R_0)$.
Therefore, 
$$||\nabla^2\mub||_{L^\infty_{\ub}L^2(S)}\leq C(\mathcal O_0,\mathcal R_0)(1+\mathcal R).$$
The desired estimates for $\nabla^3\etab$ thus follow from (\ref{ellipticetab}) and taking the $L^2$ norm in either the $u$ or the $\ub$ direction.
\end{proof}
We finally prove estimates for $\nabla^3\omega$.
\begin{proposition}\label{omega3}
Assume 
$$\mathcal R<\infty.$$
Then there exists $\epsilon_0=\epsilon_0(\mathcal O_0,\mathcal R_0,\mathcal R,\Delta_3)$ such that whenever $\epsilon\leq \epsilon_0$,
$$\tilde{\mathcal O}_{3,2}[\kappa,\omega,\omega^{\dagger}]
\leq C(\mathcal O_0)(1+\mathcal R[\beta]).$$
\end{proposition}
\begin{proof}
Recall that $\omega^{\dagger}$ is defined to be the solution to 
$$\nab_3\omega^{\dagger}=\frac 12 \sigmac,$$
with zero initial data, i.e., $\omega^{\dagger}=0$ on $H_0$ and $\kappa$ is defined to be
$$\kappa:=\nab\omega+^*\nab\omega^{\dagger}-\frac 12 \beta.$$
By the definition of $\omega^{\dagger}$, it is easy to see that using Propositions \ref{transport} and \ref{Ricci},
$$\sum_{i\leq 2}||\nab^i\omega^{\dagger}||_{L^2_{\ub}L^\infty_uL^2(S)}\leq C(\mathcal O_0,\mathcal R_0).$$
In other words, $\omega^{\dagger}$ satisfies the same estimates as $\psi_H$. In the remainder of the proof of this Proposition, we therefore also use $\psi_H$ to denote $\omega^{\dagger}$.
Consider the following equations:
$$\nabla_3 \kappa=\psi(\rhoc+\sigmac)+\sum_{i_1+i_2+i_3=1}\psi^{i_1}\nabla^{i_2}(\psi+\psi_H)\nabla^{i_3}(\psi+\psi_{\Hb}).$$
Commuting with angular derivatives, we get
\begin{equation*}
\begin{split}
\nabla_{3}\nabla^2\kappa=&\sum_{i_1+i_2+i_3+i_4=2}\nabla^{i_1}\psi^{i_2}\nabla^{i_3}\psi\nabla^{i_4}(\rhoc+\sigmac)
+\sum_{i_1+i_2+i_3+i_4=3}\nabla^{i_1}\psi^{i_2}\nabla^{i_3}(\psi+\psi_H)\nabla^{i_4}(\psi+\psi_{\Hb}).
\end{split}
\end{equation*}
Fix $\ub$. The term involving curvature has already been bounded in Proposition \ref{mub3} and can be controlled by
$$||\sum_{i_1+i_2+i_3+i_4\leq 2}\nabla^{i_1}\psi^{i_2}\nabla^{i_3}\psi\nabla^{i_4}(\rhoc+\sigmac)||_{L^1_{u}L^2(S)}\leq C(\mathcal O_0,\mathcal R_0)\mathcal R.$$
For the terms with only Ricci coefficients, notice that the third derivatives of all the Ricci coefficients except $\omega$ and $\omega^{\dagger}$ have been estimated. Thus using Propositions \ref{Ricci}, \ref{trch3}, \ref{mu3}, \ref{omegab3}, \ref{trchb3} and \ref{mub3}, we have
\begin{equation*}
\begin{split}
&||\sum_{i_1+i_2+i_3+i_4 \leq 3}\nabla^{i_1}\psi^{i_2}\nabla^{i_3}(\psi,\psi_H)\nabla^{i_4}(\psi,\psi_{\Hb})||_{L^1_{u}L^2(S)} \\
\leq& C(\sum_{i_1\leq 2}\sum_{i_2\leq 3}||\nab^{i_1}\psi||^{i_2}_{L^\infty_{u}L^2(S)})(\sum_{i_3\leq 2}||\nabla^{i_3}(\psi,\psi_H)||_{L^2_{u}L^2(S)})(\sum_{i_4\leq 3}||\nabla^{i_4}(\psi,\psi_{\Hb})||_{L^2_{u}L^2(S)})\\
&+ C(||\nab^3\chih||_{L^2_{u}L^2(S)})(\sum_{i\leq 2}||\nabla^{i}(\psi,\psi_{\Hb})||_{L^2_{u}L^2(S)})\\
&+ C\int_0^u ||\nabla^3(\omega,\omega^{\dagger})||_{L^4(S_{u',\ub})}(\sum_{i\leq 1}||\nab^{i}(\psi,\psi_{\Hb})||_{L^\infty(S_{u',\ub})})du'\\
\leq &C(\mathcal O_0,\mathcal R_0)(1+\sum_{i_1\leq 2}||\nabla^{i_1}\psi_H||_{L^2_{u}L^2(S)}+||\nab^3\chih||_{L^2_{u}L^2(S)}\\
&\quad\quad\quad\quad\quad\quad+\int_0^u ||\nabla^3(\omega,\omega^{\dagger})||_{L^2(S_{u',\ub})}(\sum_{i_2\leq 1}||\nab^{i_2}(\psi,\psi_{\Hb})||_{L^4(S_{u',\ub})}du').
\end{split}
\end{equation*}
Therefore, by Proposition \ref{transport},
\begin{equation}\label{kappa.est}
\begin{split}
&||\nabla^2\kappa||_{L^\infty_u L^2(S)}\\
\leq &C(\mathcal O_0,\mathcal R_0)(1+\mathcal R+\sum_{i_1\leq 2}||\nabla^{i_1}\psi_H||_{L^2_{u}L^2(S)}+||\nab^3\chih||_{L^2_{u}L^2(S)}\\
&\quad\quad\quad\quad\quad\quad+\int_0^u ||\nabla^3(\omega,\omega^{\dagger})||_{L^2(S_{u',\ub})}(\sum_{i_2\leq 1}||\nab^{i_2}(\psi,\psi_{\Hb})||_{L^4(S_{u',\ub})}du').
\end{split}
\end{equation}
By the following div-curl system:
$$\div\nabla\omega=\div\kappa+\frac 12\div\beta,$$
$$\curl\nabla\omega=0,$$
$$\div\nabla\omega^\dagger=\curl\kappa+\frac 12\curl\beta,$$
$$\curl\nabla\omega^\dagger=0,$$
we have, using Propositions \ref{ellipticthm} and \ref{Kest},
\begin{equation}\label{ellipticomega}
||\nabla^3(\omega,\omega^{\dagger})||_{L^2(S)}\leq C(||\nabla^2\kappa||_{L^2(S)}+||\nabla^2\beta||_{L^2(S)}+||\nabla(\omega,\omega^{\dagger})||_{L^2(S)}).
\end{equation}
Applying this to the estimates for $\nab^3\kappa$ in \eqref{kappa.est} and using Proposition \ref{Ricci}, we get
\begin{equation*}
\begin{split}
||\nabla^2\kappa||_{L^\infty_u L^2(S)}
\leq &C(\mathcal O_0,\mathcal R_0)(1+\mathcal R+\sum_{i_1\leq 2}||\nabla^{i_1}\psi_H||_{L^2_{u}L^2(S)}+||\nab^3\chih||_{L^2_{u}L^2(S)}\\
&\quad\quad\quad\quad\quad\quad+\int_0^u ||\nabla^2\kappa||_{L^2(S_{u',\ub})}(\sum_{i_2\leq 1}||\nab^{i_2}(\psi,\psi_{\Hb})||_{L^4(S_{u',\ub})}du').
\end{split}
\end{equation*}
By Gronwall's inequality, and using Proposition \ref{Ricci},
\begin{equation*}
\begin{split}
&||\nabla^2\kappa||_{L^\infty_u L^2(S)}\\
\leq &C(\mathcal O_0,\mathcal R_0)(1+\mathcal R+\sum_{i\leq 2}||\nabla^{i}\psi_H||_{L^2_{u}L^2(S)}+||\nab^3\chih||_{L^2_{u}L^2(S)})\exp(\int_0^u ||\psi_{\Hb}||_{L^\infty(S_{u',\ub})}du')\\
\leq &C(\mathcal O_0,\mathcal R_0)(1+\mathcal R+\sum_{i\leq 2}||\nabla^{i}\psi_H||_{L^2_{u}L^2(S)}+||\nab^3\chih||_{L^2_{u}L^2(S)}).
\end{split}
\end{equation*}
Now, taking the $L^2$ norm in $\ub$, we get
\begin{equation*}
\begin{split}
&||\nabla^2\kappa||_{L^2_{\ub}L^\infty_u L^2(S)}\\
\leq &C(\mathcal O_0)(1+\epsilon^{\frac 12}\mathcal R+\sum_{i\leq 2}||\nabla^{i}\psi_H||_{L^2_{\ub}L^2_{u}L^2(S)}+||\nab^3\chih||_{L^2_{\ub}L^2_{u}L^2(S)})\\
\leq &C(\mathcal O_0,\mathcal R_0)(1+\epsilon^{\frac 12}\mathcal R+\mathcal R[\beta]),
\end{split}
\end{equation*}
where in the last line we have used Propositions \ref{Ricci} and \ref{trch3}. By choosing $\epsilon$ sufficiently small, we have
\begin{equation*}
\begin{split}
||\nabla^2\kappa||_{L^2_{\ub}L^\infty_u L^2(S)}
\leq &C(\mathcal O_0,\mathcal R_0)(1+\mathcal R[\beta]).
\end{split}
\end{equation*}
Therefore, by (\ref{ellipticomega}), we have
$$||\nabla^3(\omega,\omega^{\dagger})||_{L^\infty_uL^2_{\ub}L^2(S)}\leq C(\mathcal O_0,\mathcal R_0)(1+\mathcal R[\beta]).$$
\end{proof}

Putting these all together gives
\begin{proposition}\label{Ricci32}
Assume 
$$\mathcal R<\infty.$$
Then there exists $\epsilon_0=\epsilon_0(\mathcal O_0,\mathcal R_0,\mathcal R)$ such that whenever $\epsilon\leq \epsilon_0$,
\[
\tilde{\mathcal O}_{3,2}[\trch,\chih,\kappa,\omega,\omega^{\dagger}]\leq C(\mathcal O_0,\mathcal R_0)(1+\mathcal R[\beta]),
\]
\[
\tilde{\mathcal O}_{3,2}[\trchb,\chibh,\kappab,\omegab,\omegab^{\dagger}]\leq C(\mathcal O_0,\mathcal R_0)(1+\mathcal R[\betab]),
\]
and
\[
\tilde{\mathcal O}_{3,2}[\mu,\mub,\eta,\etab]\leq C(\mathcal O_0,\mathcal R_0)(1+\mathcal R).
\]
\end{proposition}
\begin{proof}
Let 
$$\Delta_3 \gg C(\mathcal O_0,\mathcal R_0)(1+\mathcal R),$$
where $C(\mathcal O_0,\mathcal R_0)(1+\mathcal R)$ is taken to be the maximum of the bounds in Propositions \ref{trch3}, \ref{mu3}, \ref{omegab3}, \ref{trchb3}, \ref{mub3}, \ref{omega3}. Hence, the choice of $\Delta_3$ depends only on $\mathcal O_0$, $\mathcal R_0$ and $\mathcal R$. Thus, by Propositions \ref{trch3}, \ref{mu3}, \ref{omegab3}, \ref{trchb3}, \ref{mub3}, \ref{omega3}, the bootstrap assumption (\ref{BA4}) can be improved by choosing $\epsilon$ sufficiently small depending on $\mathcal O_0,\mathcal R_0$ and $\mathcal R$.
\end{proof}

\section{Estimates for Curvature}\label{seccurv}

In this section, we derive and prove the energy estimates for the curvature components and their first two derivatives and conclude that $\mathcal R$ is controlled by a constant depending only on the size of the initial data. By the propositions in the previous sections, this shows that all the $\mathcal O$ norms can be bounded by a constant depending only on the size of the initial data, thus proving Theorem \ref{aprioriestimates}. In order to derive the energy estimates, we need the following integration by parts formula, which can be proved by a direct computation:
\begin{proposition}\label{intbyparts34}
Suppose $\phi_1$ and $\phi_2$ are $r$ tensorfields, then
$$\int_{D_{u,\ub}} \phi_1 \nabla_4\phi_2+\int_{D_{u,\ub}}\phi_2\nabla_4\phi_1= \int_{\Hb_{\ub}(0,u)} \phi_1\phi_2-\int_{\Hb_0(0,u)} \phi_1\phi_2+\int_{D_{u,\ub}}(2\omega-\trch)\phi_1\phi_2,$$
$$\int_{D_{u,\ub}} \phi_1 \nabla_3\phi_2+\int_{D_{u,\ub}}\phi_2\nabla_3\phi_1= \int_{H_{u}(0,\ub)} \phi_1\phi_2-\int_{H_0(0,\ub)} \phi_1\phi_2+\int_{D_{u,\ub}}(2\omegab-\trchb)\phi_1\phi_2.$$
\end{proposition}
\begin{proposition}\label{intbypartssph}
Suppose we have an $r$ tensorfield $^{(1)}\phi$ and an $r-1$ tensorfield $^{(2)}\phi$.
\begin{equation*}
\begin{split}
&\int_{D_{u,\ub}}{ }^{(1)}\phi^{A_1A_2...A_r}\nabla_{A_r}{ }^{(2)}\phi_{A_1...A_{r-1}}+\int_{D_{u,\ub}}\nabla^{A_r}{ }^{(1)}\phi_{A_1A_2...A_r}{ }^{(2)}\phi^{A_1...A_{r-1}}\\
=& -\int_{D_{u,\ub}}(\eta+\etab){ }^{(1)}\phi{ }^{(2)}\phi.
\end{split}
\end{equation*}
\end{proposition}
Using these we derive energy estimates for $\rhoc,\sigmac$ in $L^2(H_u)$ and for $\betab$ in $L^2(\Hb_{\ub})$.
\begin{proposition}\label{ee1}
The following $L^2$ estimates for the curvature components hold:
\begin{equation*}
\begin{split}
&\sum_{i\leq 2}(||\nab^i(\rhoc,\sigmac)||^2_{L^\infty_u L^2_{\ub}L^2(S)}+||\nab^i\betab||^2_{L^\infty_{\ub} L^2_{u}L^2(S)})  \\
\leq &\sum_{i\leq 2}(||\nab^i(\rhoc,\sigmac)||^2_{L^2_{\ub}L^2(S_{0,\ub})}+||\nab^i\betab||^2_{L^2_{u}L^2(S_{u,0})}) \\
&+||(\sum_{i\leq 2}\nabla^i(\rhoc,\sigmac))(\sum_{i_1+i_2+i_3\leq 2}\psi^{i_1}\nabla^{i_2}(\psi+\psi_{\Hb})\nabla^{i_3}(\rhoc,\sigmac))||_{L^1_uL^1_{\ub}L^1(S)}\\
&+||(\sum_{i\leq 2}\nabla^i(\rhoc,\sigmac))(\sum_{i_1+i_2+i_3+i_4\leq 3}\nabla^{i_1}\psi^{i_2}\nabla^{i_3}\psi\nabla^{i_4}(\psi+\psi_{\Hb}))||_{L^1_uL^1_{\ub}L^1(S)}\\
&+||(\sum_{i\leq 2}\nabla^i(\rhoc,\sigmac))(\sum_{i_1+i_2+i_3+i_4\leq 2}\psi^{i_1}\nabla^{i_2}\psi\nabla^{i_3}\chibh\nabla^{i_4}\chibh)||_{L^1_uL^1_{\ub}L^1(S)}\\
&+||(\sum_{i\leq 2}\nabla^i\betab)(\sum_{i_1+i_2+i_3\leq 2}\psi^{i_1}\nabla^{i_2}\psi\nabla^{i_3}(\rhoc,\sigmac))||_{L^1_uL^1_{\ub}L^1(S)}\\
&+||(\sum_{i\leq 2}\nabla^i\betab)(\sum_{i_1+i_2+i_3\leq 1}\psi^{i_1}\nabla^{i_2}\rhoc\nabla^{i_3}(\rhoc,\sigmac))||_{L^1_uL^1_{\ub}L^1(S)}\\
&+||(\sum_{i\leq 2}\nabla^i\betab)(\sum_{i_1+i_2+i_3+i_4\leq 3}\nabla^{i_1}\psi^{i_2}\nabla^{i_3}(\trchb+\psi_{\Hb})\nabla^{i_4}(\trch+\psi_H))||_{L^1_uL^1_{\ub}L^1(S)}.
\end{split}
\end{equation*}
\end{proposition}
\begin{proof}
Consider the following schematic Bianchi equations:
\begin{equation*}
\begin{split}
\nab_3\sigmac+\div ^*\betab=&\psi\sigmac+\sum_{i_1+i_2+i_3\leq 1}\psi^{i_1}\nab^{i_2}\psi\nab^{i_3}\psi_{\Hb}+\psi\chibh\chibh,\\
\nab_3\rhoc+\div\betab=&\psi\rhoc+\sum_{i_1+i_2+i_3\leq 1}\psi^{i_1}\nab^{i_2}\psi\nab^{i_3}\psi_{\Hb}+\psi\chibh\chibh,\\
\nab_4\betab+\nabla\rhoc -^*\nabla\sigmac=& \psi(\rhoc,\sigmac)+\sum_{i_1+i_2+i_3\leq 1}\psi^{i_1}\nab^{i_2}(\psi_H+\trch)\nab^{i_3}(\psi_{\Hb}+\trchb),\\
\end{split}
\end{equation*}
Commuting these equations with angular derivatives for $i\leq 2$, we get the equation for $\nab_3\nab^i\sigmac$,
\begin{equation*}
\begin{split}
&\nab_3\nab^i\sigmac+\div ^*\nab^i\betab=F_1,
\end{split}
\end{equation*}
where $F_1$ denotes the terms
\begin{equation*}
\begin{split}
F_1:=&\sum_{i_1+i_2+i_3\leq 2}\psi^{i_1}\nab^{i_2}(\psi+\psi_{\Hb})\nab^{i_3}(\rhoc,\sigmac)+\sum_{i_1+i_2+i_3+i_4\leq 3}\nab^{i_1}\psi^{i_2}\nab^{i_3}\psi\nab^{i_4}(\psi+\psi_{\Hb})\\
&+\sum_{i_1+i_2+i_3+i_4\leq 2}\psi^{i_1}\nab^{i_2}\psi\nab^{i_3}\chibh\nab^{i_4}\chibh.
\end{split}
\end{equation*}
Notice that in the derivation of the $\nab_3\nab^i\sigmac$ equation, there are terms arising from the commutator $[\nab^i,\div]$. These can be expressed in terms of the Gauss curvature, which can be substituted by $-\rhoc-\frac 12\trch\trchb$ and rewritten as the terms in the above expression. The equation for $\nab_3\nab^i\rhoc$ has a similar structure:
\begin{equation*}
\begin{split}
\nab_3\nab^i\rhoc+\div\nab^i\betab=F_1.
\end{split}
\end{equation*}
We have the following equation for $\nab_4\nab^i\betab$:
\begin{equation*}
\begin{split}
\nab_4\nab^i\betab+\nabla\nab^i\rhoc -^*\nabla\nab^i\sigmac=F_2,
\end{split}
\end{equation*}
where $F_2$ denotes the terms of the form
\begin{equation*}
\begin{split}
F_2:=&\sum_{i_1+i_2+i_3\leq 2}\psi^{i_1}\nab^{i_2}\psi\nab^{i_3}(\rhoc,\sigmac)+\sum_{i_1+i_2+i_3\leq 1}\psi^{i_1}\nab^{i_2}\rhoc\nab^{i_3}(\rhoc,\sigmac)\\
&+\sum_{i_1+i_2+i_3+i_4\leq 3}\nab^{i_1}\psi^{i_2}\nab^{i_3}(\psi_H+\trch)\nab^{i_4}(\psi_{\Hb}+\trchb).
\end{split}
\end{equation*}
Applying Proposition \ref{intbypartssph} yields the following identity on the derivatives of the curvature.
\begin{equation*}
\begin{split}
&\int \langle\nab^i\betab,\nabla_4\nab^i\betab \rangle_\gamma\\ =&\int \langle\nab^i\betab,-\nabla\nab^i\rho+^*\nabla\nab^i\sigma\rangle_\gamma+\langle\nab^i\betab,F_2\rangle_\gamma \\
=&\int \langle\div\nab^i\betab,\nab^i\rhoc\rangle_\gamma+\langle\div ^*\nab^i\betab,\nab^i\sigmac\rangle_\gamma +\langle\nab^i\betab,F_2\rangle_\gamma\\
=&\int -\langle\nabla_3\nab^i\rhoc,\nab^i\rhoc\rangle_\gamma-\langle\nabla_3\nab^i\sigmac,\nab^i\sigmac\rangle_\gamma +\langle\nab^i\betab,F_2\rangle_\gamma+\langle\nab^i(\rhoc,\sigmac),F_1\rangle_\gamma.
\end{split}
\end{equation*}
Using Proposition \ref{intbyparts34}, we have
\begin{equation*}
\begin{split}
\int \langle\nab^i\betab,\nabla_4\nab^i\betab\rangle_\gamma=\frac 12 (\int_{\Hb_{\ub}}|\nab^i\betab|^2 - \int_{\Hb_0}|\nab^i\betab|^2)+||(\omega-\frac 12 \trch)|\nab^i\betab|^2||_{L^1_uL^1_{\ub}L^1(S)}.
\end{split}
\end{equation*}
Substituting the Codazzi equation 
$$\betab=\sum_{i_1+i_2=1}\psi^{i_1}\nab^{i_2}(\psi+\psi_{\Hb})$$
for one of the $\betab$'s, we note that the last term
$$||(\omega-\frac 12 \trch)|\nab^i\betab|^2||_{L^1_uL^1_{\ub}L^1(S)}$$
is of the form of one of the terms stated in the Proposition. We call such terms acceptable. Also by using Proposition \ref{intbyparts34}, we have
\begin{equation*}
\begin{split}
&\int \langle\nab^i(\rhoc,\sigmac),\nabla_3\nab^i(\rhoc,\sigmac)\rangle_\gamma\\
=&\frac 12 (\int_{H_{u}}|\nab^i(\rhoc,\sigmac)|^2 - \int_{H_0}|\nab^i(\rhoc,\sigmac)|^2)+||(\omegab-\frac 12 \trchb)|\nab^i(\rhoc,\sigmac)|^2||_{L^1_uL^1_{\ub}L^1(S)}.
\end{split}
\end{equation*}
The last term
$$||(\omegab-\frac 12 \trchb)|\nab^i(\rhoc,\sigmac)|^2||_{L^1_uL^1_{\ub}L^1(S)}$$
is also acceptable. We thus have
\begin{equation*}
\begin{split}
&\int_{\Hb_{\ub}} |\nabla^i\betab|^2_\gamma+\int_{H_{u}} |\nabla^i(\rhoc,\sigmac)|^2_\gamma  \\
\leq &\int_{\Hb_{\ub'}} |\nabla^i\betab|^2_\gamma+\int_{H_{u'}} |\nabla^i(\rhoc,\sigmac)|^2_\gamma+|<\nab^i\betab,F_2>_\gamma|+|<\nab^i(\rhoc,\sigmac),F_1>_\gamma|\\
&+\mbox{acceptable terms}.
\end{split}
\end{equation*}
We conclude the proposition by noting that the structure for $F_1$ and $F_2$ implies that
$$|\langle\nab^i\betab,F_2\rangle_\gamma|$$
and 
$$|\langle\nab^i(\rhoc,\sigmac),F_1\rangle_\gamma|$$
are also acceptable.
\end{proof}
To close the energy estimates, we also need to control $\beta$ in $L^2(H)$ and $(\rhoc,\sigmac)$ in $L^2(\Hb)$. It is not difficult to see that due to the structure of the Einstein equations, Proposition \ref{ee1} also holds when all the barred and unbarred quantities are exchanged. The proof is exactly analogous to that of Proposition \ref{ee1} and will be omitted.
\begin{proposition}\label{ee2}
The following $L^2$ estimates for the curvature components hold:
\begin{equation*}
\begin{split}
&\sum_{i\leq 2}(||\nab^i(\rhoc,\sigmac)||^2_{L^\infty_{\ub} L^2_{u}L^2(S)}+||\nab^i\beta||^2_{L^\infty_{u} L^2_{\ub}L^2(S)})  \\
\leq &\sum_{i\leq 2}(||\nab^i(\rhoc,\sigmac)||^2_{L^2_{u}L^2(S_{u,0})}+||\nab^i\beta||^2_{L^2_{\ub}L^2(S_{0,\ub})}) \\
&+||(\sum_{i\leq 2}\nabla^i(\rhoc,\sigmac))(\sum_{i_1+i_2+i_3\leq 2}\psi^{i_1}\nabla^{i_2}(\psi+\psi_H)\nabla^{i_3}(\rhoc,\sigmac))||_{L^1_uL^1_{\ub}L^1(S)}\\
&+||(\sum_{i\leq 2}\nabla^i(\rhoc,\sigmac))(\sum_{i_1+i_2+i_3+i_4\leq 3}\nabla^{i_1}\psi^{i_2}\nabla^{i_3}\psi\nabla^{i_4}(\psi+\psi_H))||_{L^1_uL^1_{\ub}L^1(S)}\\
&+||(\sum_{i\leq 2}\nabla^i(\rhoc,\sigmac))(\sum_{i_1+i_2+i_3+i_4\leq 2}\psi^{i_1}\nabla^{i_2}\psi\nabla^{i_3}\chih\nabla^{i_4}\chih)||_{L^1_uL^1_{\ub}L^1(S)}\\
&+||(\sum_{i\leq 2}\nabla^i\beta)(\sum_{i_1+i_2+i_3\leq 2}\psi^{i_1}\nabla^{i_2}\psi\nabla^{i_3}(\rhoc,\sigmac))||_{L^1_uL^1_{\ub}L^1(S)}\\
&+||(\sum_{i\leq 2}\nabla^i\beta)(\sum_{i_1+i_2+i_3\leq 1}\psi^{i_1}\nabla^{i_2}\rhoc\nabla^{i_3}(\rhoc,\sigmac))||_{L^1_uL^1_{\ub}L^1(S)}\\
&+||(\sum_{i\leq 2}\nabla^i\beta)(\sum_{i_1+i_2+i_3+i_4\leq 3}\nabla^{i_1}\psi^{i_2}\nabla^{i_3}(\trchb+\psi_{\Hb})\nabla^{i_4}(\trch+\psi_H))||_{L^1_uL^1_{\ub}L^1(S)}.
\end{split}
\end{equation*}
\end{proposition}

We now control all the error terms in the energy estimates. Introduce the bootstrap assumption:
\begin{equation}\tag{A5}\label{BA5}
\mathcal R\leq \Delta_4,
\end{equation}
where $\Delta_4$ is a positive constant to be chosen later.
First we estimate $\rhoc$ and $\sigmac$ in $L^2(H_u)$ and $\betab$ in $L^2(\Hb_{\ub})$.
\begin{proposition}\label{R1}
There exist $\epsilon_0=\epsilon_0(\mathcal O_0,\mathcal R_0,\Delta_4)$ such that whenever $\epsilon\leq\epsilon_0$,
$$\sum_{i\leq 2}(||\nab^i(\rhoc,\sigmac)||_{L^\infty_u L^2_{\ub}L^2(S)}+||\nab^i\betab||_{L^\infty_{\ub} L^2_{u}L^2(S)}) \leq C(\mathcal O_0,\mathcal R_0).$$
\end{proposition}
\begin{proof}
We control the six terms in Proposition \ref{ee1}. We first estimate the term $\betab\psi_H\psi_{\Hb}$, i.e., the last term in the expression in Proposition \ref{ee1}. As we will see, this is the most difficult term because all three factors can only be controlled after taking the $L^2$ norm along one of the null variables. 
\begin{equation*}
\begin{split}
&||(\sum_{i\leq 2}\nabla^{i}\betab)(\sum_{i_1+i_2+i_3+i_4\leq 3}\nabla^{i_1}\psi^{i_2}\nabla^{i_3}\psi_H\nabla^{i_4}\psi_{\Hb})||_{L^1_uL^1_{\ub}L^1(S)} \\
\leq &\epsilon^{\frac{1}{2}}(\sum_{i\leq 2}||\nabla^{i}\betab||_{L^\infty_{\ub} L^2_{u}L^2(S)})(\sum_{i_1+i_2+i_3+i_4\leq 3}||\nabla^{i_1}\psi^{i_2}\nabla^{i_3}\psi_H\nabla^{i_4}\psi_{\Hb}||_{L^2_{\ub}L^2_{u}L^2(S)}).
\end{split}
\end{equation*}
Since we have a small constant $\epsilon^{\frac 12}$, we only need to bound the remaining contribution by a constant depending on $\mathcal O_0$, $\mathcal R_0$ and $\Delta_4$. The first factor is bounded by $\Delta_4$ by the definition of the norm $\mathcal R$ and the bootstrap assumption (\ref{BA5}). We now look at the second factor.
\begin{equation*}
\begin{split}
&||\sum_{i_1+i_2+i_3+i_4\leq 3}\nabla^{i_1}\psi^{i_2}\nabla^{i_3}\psi_H\nabla^{i_4}\psi_{\Hb}||_{L^2_uL^2_{\ub}L^2(S)} \\
\leq &C(\sum_{i_1\leq 3}||\psi||^{i_1}_{L^\infty_uL^\infty_{\ub}L^\infty(S)})(\sum_{i_2\leq 3}||\nabla^{i_2}\psi_H||_{L^\infty_uL^2_{\ub}L^2(S)}) ||\psi_{\Hb}||_{L^2_uL^\infty_{\ub}L^\infty(S)} \\
&+C(\sum_{i_1\leq 3}||\psi||_{L^\infty_uL^\infty_{\ub}L^\infty(S)}^{i_1})||\psi_H||_{L^2_{\ub}L^\infty_uL^\infty(S)} (\sum_{i_2\leq 3}||\nabla^{i_2}\psi_{\Hb}||_{L^\infty_{\ub}L^2_uL^2(S)}) \\
&+C(\sum_{i_1\leq 3}||\psi||_{L^\infty_uL^\infty_{\ub}L^\infty(S)}^{i_1})(\sum_{i_2\leq 2}||\nabla^{i_2}\psi_H||_{L^\infty_uL^2_{\ub}L^4(S)})( \sum_{i_3\leq 2}||\nabla^{i_3}\psi_{\Hb}||_{L^2_uL^\infty_{\ub}L^4(S)}) \\
&+C(\sum_{i_1\leq 1}||\psi||_{L^\infty_uL^\infty_{\ub}L^\infty(S)}^{i_1})||\psi_H||_{L^\infty_uL^2_{\ub}L^\infty(S)} ||\psi_{\Hb}||_{L^2_uL^\infty_{\ub}L^\infty(S)} (\sum_{i_2\leq 2}||\nabla^{i_2}\psi||_{L^\infty_u L^\infty_{\ub}L^2(S)})\\
&+C(\sum_{i_1\leq 1}||\nabla^{i_1}\psi_H||_{L^\infty_uL^2_{\ub}L^4(S)}) ||\psi_{\Hb}||_{L^2_uL^\infty_{\ub}L^\infty(S)} (\sum_{i_2\leq 1}||\nabla^{i_2}\psi||_{L^\infty_u L^\infty_{\ub}L^4(S)})\\
&+C||\psi_H||_{L^\infty_uL^2_{\ub}L^\infty(S)} (\sum_{i_1\leq 1}||\nabla^{i_4}\psi_{\Hb}||_{L^2_uL^\infty_{\ub}L^4(S)}) (\sum_{i_2\leq 1}||\nabla^{i_2}\psi||_{L^\infty_u L^\infty_{\ub}L^4(S)}).
\end{split}
\end{equation*}
By Propositions \ref{Ricci} and \ref{Ricci32}, we have
$$||\sum_{i_1+i_2+i_3+i_4\leq 3}\nabla^{i_1}\psi^{i_2}\nabla^{i_3}\psi_H\nabla^{i_4}\psi_{\Hb}||_{L^2_uL^2_{\ub}L^2(S)}\leq C(\mathcal O_0,\mathcal R_0,\Delta_4).$$
Thus
$$||(\sum_{i\leq 2}\nabla^{i}\betab)(\sum_{i_1+i_2+i_3+i_4\leq 3}\nabla^{i_1}\psi^{i_2}\nabla^{i_3}\psi_H\nabla^{i_4}\psi_{\Hb})||_{L^1_uL^1_{\ub}L^1(S)}\leq C(\mathcal O_0,\mathcal R_0,\Delta_4)\epsilon^{\frac 12}.$$
We next consider the following four terms from Proposition \ref{ee1}:
$$||(\sum_{i\leq 2}\nabla^i(\rhoc,\sigmac))(\sum_{i_1+i_2+i_3\leq 2}\psi^{i_1}\nabla^{i_2}(\psi+\psi_{\Hb})\nabla^{i_3}(\rhoc,\sigmac))||_{L^1_uL^1_{\ub}L^1(S)},$$
$$||(\sum_{i\leq 2}\nabla^i(\rhoc,\sigmac))(\sum_{i_1+i_2+i_3+i_4\leq 3}\nabla^{i_1}\psi^{i_2}\nabla^{i_3}\psi\nabla^{i_4}(\psi+\psi_{\Hb}))||_{L^1_uL^1_{\ub}L^1(S)},$$
$$||(\sum_{i\leq 2}\nabla^i\betab)(\sum_{i_1+i_2+i_3\leq 2}\psi^{i_1}\nabla^{i_2}\psi\nabla^{i_3}(\rhoc,\sigmac))||_{L^1_uL^1_{\ub}L^1(S)},$$
$$||(\sum_{i\leq 2}\nabla^i\betab)(\sum_{i_1+i_2+i_3\leq 1}\psi^{i_1}\nabla^{i_2}\rhoc\nabla^{i_3}(\rhoc,\sigmac))||_{L^1_uL^1_{\ub}L^1(S)}.$$
Since $\psi$ satisfies stronger estimates than either $\psi_H$ or $\psi_{\Hb}$; and $\rhoc,\sigmac$ satisfy strong estimates than either $\nab\psi_H$ or $\nab\psi_{\Hb}$, we can bound these terms exactly the way as above by $C(\mathcal O_0,\mathcal R_0,\Delta_4)\epsilon^{\frac 12}$.

We are thus left with the term
$$||(\sum_{i\leq 2}\nabla^i(\rhoc,\sigmac))(\sum_{i_1+i_2+i_3+i_4\leq 2}\psi^{i_1}\nabla^{i_2}\psi\nabla^{i_3}\chibh\nabla^{i_4}\chibh)||_{L^1_uL^1_{\ub}L^1(S)}.$$
Since $\chibh$ can only be controlled after taking the $L^2$ norm in $u$, we must bound the curvature term $\nabla^i(\rhoc,\sigmac)$ in $L^2(H)$. Nevertheless, we get a smallness constant in this estimate:
\begin{equation*}
\begin{split}
&||(\sum_{i\leq 2}\nabla^i(\rhoc,\sigmac))(\sum_{i_1+i_2+i_3+i_4\leq 2}\psi^{i_1}\nabla^{i_2}\psi\nabla^{i_3}\chibh\nabla^{i_4}\chibh)||_{L^1_uL^1_{\ub}L^1(S)}\\
\leq &(\sum_{i\leq 2}||\nabla^i(\rhoc,\sigmac)||_{L^\infty_uL^2_{\ub}L^2(S)})(\sum_{i_1+i_2+i_3+i_4\leq 2}||\psi^{i_1}\nabla^{i_2}\psi\nabla^{i_3}\chibh\nabla^{i_4}\chibh||_{L^1_uL^2_{\ub}L^2(S)})\\
\leq &C\epsilon^{\frac 12}\mathcal R(\sum_{i_1\leq 2}\sum_{i_2\leq 3}||\nabla^{i_1}\psi||_{L^\infty_uL^\infty_{\ub}L^2(S)}^{i_2})(\sum_{i_3\leq 2}||\nabla^{i_3}\chibh||_{L^2_uL^\infty_{\ub}L^2(S)}^2)\\
\leq &C(\mathcal O_0,\mathcal R_0)\epsilon^{\frac 12}\Delta_4.
\end{split}
\end{equation*}
Therefore,
$$\sum_{i\leq 2}(||\nab^i(\rhoc,\sigmac)||_{L^\infty_u L^2_{\ub}L^2(S)}^2+||\nab^i\betab||_{L^\infty_{\ub} L^2_{u}L^2(S)}^2)\leq \mathcal R_0^2+\epsilon^{\frac 12}C(\mathcal O_0,\mathcal R_0,\Delta_4).$$
Thus the conclusion follows by choosing $\epsilon$ to be sufficiently small depending on $\mathcal O_0, \mathcal R_0$ and $\Delta_4$.
\end{proof}

We now estimate the remaining components of curvature:
\begin{proposition}\label{R2}
There exist $\epsilon_0=\epsilon_0(\mathcal O_0,\mathcal R_0,\Delta_4)$ such that whenever $\epsilon\leq\epsilon_0$,
$$\sum_{i\leq 2}(||\nab^i\beta||_{L^\infty_u L^2_{\ub}L^2(S)}+||\nab^i(\rhoc,\sigmac)||_{L^\infty_{\ub} L^2_{u}L^2(S)}) \leq C(\mathcal O_0,\mathcal R_0).$$
\end{proposition}
\begin{proof}
In order to prove this estimate, we heavily rely on the bounds that we have already derived in Proposition \ref{R1} for $\nab^i(\rhoc,\sigmac)$ and $\nab^i\betab$. In particular, we need to use the fact that those estimates are independent of $\Delta_4$. In order to effectively distinguish the norms for the different components of curvature, we introduce the following notation:
$$R_u[\beta]:=\sum_{i\leq 2}\sup_{0\leq u'\leq u}||\nab^i\beta||_{L^2_{\ub}L^2(S_{u',\ub})},$$
$$R_{\ub}[\rhoc,\sigmac]:=\sum_{i\leq 2}\sup_{0\leq \ub'\leq \ub}||\nab^i(\rhoc,\sigmac)||_{L^2_{u}L^2(S_{u,\ub'})},$$
$$R_u[\rhoc,\sigmac]:=\sum_{i\leq 2}\sup_{0\leq u'\leq u}||\nab^i(\rhoc,\sigmac)||_{L^2_{\ub}L^2(S_{u',\ub})},$$
$$R_{\ub}[\betab]:=\sum_{i\leq 2}\sup_{0\leq \ub'\leq \ub}||\nab^i\betab||_{L^2_{u}L^2(S_{u,\ub'})}.$$
We now proceed to proving the proposition by controlling the six error terms in Proposition \ref{ee2}. We start with the first, second, fourth and fifth terms:
$$||(\sum_{i\leq 2}\nabla^i(\rhoc,\sigmac))(\sum_{i_1+i_2+i_3\leq 2}\psi^{i_1}\nabla^{i_2}(\psi+\psi_H)\nabla^{i_3}(\rhoc,\sigmac))||_{L^1_uL^1_{\ub}L^1(S)}$$
and
$$||(\sum_{i\leq 2}\nabla^i(\rhoc,\sigmac))(\sum_{i_1+i_2+i_3+i_4\leq 3}\nabla^{i_1}\psi^{i_2}\nabla^{i_3}\psi\nabla^{i_4}(\psi+\psi_H))||_{L^1_uL^1_{\ub}L^1(S)}$$
and
$$||(\sum_{i\leq 2}\nabla^i\beta)(\sum_{i_1+i_2+i_3\leq 2}\psi^{i_1}\nabla^{i_2}\psi\nabla^{i_3}(\rhoc,\sigmac))||_{L^1_uL^1_{\ub}L^1(S)}$$
and
$$||(\sum_{i\leq 2}\nabla^i\beta)(\sum_{i_1+i_2+i_3\leq 1}\psi^{i_1}\nabla^{i_2}\rhoc\nabla^{i_3}(\rhoc,\sigmac))||_{L^1_uL^1_{\ub}L^1(S)}.$$
In these terms $\beta$ or $\psi_H$ appears at most once. Therefore, after applying Cauchy-Schwarz in $\ub$ and putting $\beta$ or $\psi_H$ in $L^2_{\ub}$, there is still an extra smallness constant $\epsilon^{\frac 12}$. These terms can be estimated in a similar fashion as in Proposition \ref{R1} by $C(\mathcal O_0,\mathcal R_0,\Delta_4)\epsilon^{\frac 12}$.
We then look at the last term in Proposition \ref{ee2}:
\begin{equation}\label{betaHH}
||(\sum_{i\leq 2}\nabla^i\beta)(\sum_{i_1+i_2+i_3+i_4\leq 3}\nabla^{i_1}\psi^{i_2}\nabla^{i_3}(\trchb+\psi_{\Hb})\nabla^{i_4}(\trch+\psi_H))||_{L^1_uL^1_{\ub}L^1(S)}.
\end{equation}
Among these terms, there are two possibilities: the case where $(\trch,\psi_H)$ has at least 2 derivatives and the case where $(\trch,\psi_H)$ has at most 1 derivative. For the term where $(\trch,\psi_H)$ has at least 2 derivatives, we have
\begin{equation*}
\begin{split}
&||(\sum_{i\leq 2}\nabla^i\beta)(\sum_{i_1+i_2+i_3+i_4\leq 3,2\leq i_3\leq 3}\nabla^{i_1}\psi^{i_2}\nabla^{i_3}(\trch,\psi_H)\nabla^{i_4}(\trchb,\psi_{\Hb}))||_{L^1_uL^1_{\ub}L^1(S)} \\
\leq &\int_0^u (\sum_{i\leq 2}||\nabla^i\beta||_{L^2_{\ub}L^2(S_{u',\ub})})(\sum_{i_1\leq 1}||\psi||^{i_1}_{L^\infty_uL^\infty_{\ub}L^\infty(S)})\\
&\quad\quad\quad\quad\quad\times(\sum_{2\leq i_2\leq 3}||\nabla^{i_2}(\trch,\psi_H)||_{L^2_{\ub}L^2(S_{u',\ub})})(\sum_{i_3\leq 1}||\nab^{i_3}(\trchb,\psi_{\Hb})||_{L^\infty_{\ub}L^4(S_{u',\ub})}) du'\\
\leq &\int_0^u C(\mathcal O_0,\mathcal R_0)(1+\mathcal R_{u'}[\beta]^2) (\sum_{i\leq 1}||\nab^i(\trchb,\psi_{\Hb})||_{L^\infty_{\ub}L^4(S_{u',\ub})}) du'\\
\end{split}
\end{equation*}
by Propositions \ref{Ricci} and \ref{Ricci32}. Notice that in the first inequality above, we have also used the Sobolev embedding theorems in Propositions \ref{L4} and \ref{Linfty}. For the term where $(\trch,\psi_H)$ has at most one derivative, notice that the estimate for $\nab^2(\trch,\psi_H)$ in $L^2$ in Proposition \ref{Ricci} depends only on initial data and the bound for $\nab^3(\trchb,\psi_{\Hb})$ in Proposition \ref{Ricci32} depends only on initial data and $\mathcal R_{\ub}[\betab]$. Thus,
\begin{equation*}
\begin{split}
&||(\sum_{i\leq 2}\nab^i\beta)(\sum_{i_1+i_2+i_3+i_4\leq 3,i_3\leq 1}\nabla^{i_1}\psi^{i_2}\nabla^{i_3}(\trch,\psi_H)\nabla^{i_4}(\trchb,\psi_{\Hb}))||_{L^1_uL^1_{\ub}L^1(S)} \\
\leq &C(\sum_{i_1\leq 2}||\nab^i\beta||_{L^2_uL^2_{\ub}L^2(S)})(\sum_{i_2\leq 2}\sum_{i_3\leq 3}||\nab^{i_3}\psi||_{L^\infty_uL^\infty_{\ub}L^2(S)}^{i_2})\\
&\quad\times(\sum_{i_4\leq 2}||\nabla^{i_4}(\trch,\psi_H)||_{L^2_{\ub}L^\infty_uL^2(S)})( \sum_{i_5\leq 3}||\nab^{i_5}(\trchb,\psi_{\Hb})||_{L^\infty_{\ub}L^2_uL^2(S)})\\
\leq &C(\mathcal O_0,\mathcal R_0)(1+\mathcal R_u[\beta])(1+\mathcal R_{\ub}[\betab]).
\end{split}
\end{equation*}
Therefore, (\ref{betaHH}) can be estimated by 
\begin{equation}\label{betaHH.est}
\begin{split}
&||(\sum_{i\leq 2}\nabla^i\beta)(\sum_{i_1+i_2+i_3+i_4\leq 3}\nabla^{i_1}\psi^{i_2}\nabla^{i_3}(\trchb+\psi_{\Hb})\nabla^{i_4}(\trch+\psi_H))||_{L^1_uL^1_{\ub}L^1(S)}\\
\leq &\int_0^u C(\mathcal O_0,\mathcal R_0)(1+\mathcal R_{u'}[\beta]^2) (\sum_{i\leq 1}||\nab^i(\trchb,\psi_{\Hb})||_{L^\infty_{\ub}L^4(S_{u',\ub})}) du'\\
&+C(\mathcal O_0,\mathcal R_0)(1+\mathcal R_u[\beta])(1+\mathcal R_{\ub}[\betab]).
\end{split}
\end{equation}
We note explicitly that it is important that we do not allow all terms of the type $\psi_H\psi_H\psi$ but only allow $\psi_H\psi_H\trchb$ since by Proposition \ref{Ricci32}, $\tilde{\mathcal O}_{3,2}[\trchb]$ can be controlled by a constant depending on initial data and $\mathcal R_{\ub}[\betab]$, but the bound for $\tilde{\mathcal O}_{3,2}[\eta,\etab]$ depends on $\mathcal R$. As we will see below, it is important that one of the factors in the last term in the estimate \eqref{betaHH} depends only on $\mathcal R_{\ub}[\betab]$ rather than $\mathcal R$, since $\mathcal R_{\ub}[\betab]$ has already been previously controlled in Proposition \ref{R1} by a constant depending only on the initial data.

Returning to estimating the error terms in Proposition \ref{ee2}, we are thus only left with the term
$$||(\sum_{i_1\leq 2}\nabla^{i_1}(\rhoc,\sigmac))(\sum_{i_2+i_3+i_4+i_5+i_6\leq 2}\nab^{i_2}\psi^{i_3}\nab^{i_4}\psi\nabla^{i_5}\psi_H\nab^{i_6}\psi_H)||_{L^1_uL^1_{\ub}L^1(S)},$$
i.e., the third of the six error terms in Proposition \ref{ee2}. Since $\psi_H$ does not enter with three derivatives, it can be estimated using Proposition \ref{Ricci} by
\begin{equation*}
\begin{split}
&||(\sum_{i_1\leq 2}\nabla^{i_1}(\rhoc,\sigmac))(\sum_{i_2+i_3+i_4+i_5+i_6\leq 2}\nab^{i_2}\psi^{i_3}\nab^{i_4}\psi\nabla^{i_5}\psi_H\nab^{i_6}\psi_H)||_{L^1_uL^1_{\ub}L^1(S)}\\
\leq &(\sum_{i_1\leq 2}||\nabla^{i_1}(\rhoc,\sigmac)||_{L^\infty_{\ub}L^2_uL^2(S)})(\sum_{i_2+i_3+i_4+i_5+i_6\leq 2}||\nab^{i_2}\psi^{i_3}\nab^{i_4}\psi\nabla^{i_5}\psi_H\nab^{i_6}\psi_H||_{L^1_{\ub}L^2_{u}L^2(S)})\\
\leq &(\sum_{i_1\leq 2}||\nabla^{i_1}(\rhoc,\sigmac)||_{L^\infty_{\ub}L^2_uL^2(S)})(\sum_{i_2\leq 2}\sum_{i_3\leq 3}||\nab^{i_2}\psi||_{L^\infty_u L^\infty_{\ub}L^2(S)}^{i_3})(\sum_{i_4\leq 2}||\nab^{i_4}\psi_H||_{L^2_{\ub}L^\infty_{u}L^2(S)}^2)\\
\leq &C(\mathcal O_0,\mathcal R_0)\mathcal R_{\ub}[\rhoc,\sigmac].
\end{split}
\end{equation*}
Therefore, we have
\begin{equation*}
\begin{split}
&\mathcal R_u[\beta]^2+\mathcal R_{\ub}[\rhoc,\sigmac]^2\\
\leq &C(\mathcal O_0,\mathcal R_0)(1+\epsilon^{\frac{1}{2}}C(\mathcal O_0,\mathcal R_0,\Delta_4)+\int_0^u (\mathcal R_{u'}[\beta])^2(\sum_{i\leq 1}||\nab^i\psi_{\Hb}||_{L^\infty_{\ub}L^4(S_{u',\ub})}) du')\\
&+C(\mathcal O_0,\mathcal R_0)(\mathcal R_{\ub}[\rhoc,\sigmac]+(1+\mathcal R_u[\beta])(1+\mathcal R_{\ub}[\betab])).
\end{split}
\end{equation*}
Applying Cauchy-Schwarz on the last two terms and absorbing $\frac 12(\mathcal R_u[\beta]^2+\mathcal R_{\ub}[\rhoc,\sigmac]^2)$ to the left hand side, we have
\begin{equation*}
\begin{split}
&\mathcal R_u[\beta]^2+\mathcal R_{\ub}[\rhoc,\sigmac]^2\\
\leq &C(\mathcal O_0,\mathcal R_0)(1+\epsilon^{\frac{1}{2}}C(\mathcal O_0,\mathcal R_0,\Delta_4)+\int_0^u (\mathcal R_{u'}[\beta])^2(\sum_{i\leq 1}||\nab^i\psi_{\Hb}||_{L^\infty_{\ub}L^4(S_{u',\ub})}) du'+\mathcal R_{\ub}[\betab]^2).
\end{split}
\end{equation*}
Gronwall's inequality implies
\begin{equation*}
\begin{split}
&\mathcal R_u[\beta]^2+\mathcal R_{\ub}[\rhoc,\sigmac]^2\\
\leq &C(\mathcal O_0,\mathcal R_0)(1+\epsilon^{\frac{1}{2}}C(\mathcal O_0,\mathcal R_0,\Delta_4)+\mathcal R_{\ub}[\betab]^2)\exp(\int_0^u(\sum_{i\leq 1}||\nab^i\psi_{\Hb}||_{L^\infty_{\ub}L^4(S_{u',\ub})})du').
\end{split}
\end{equation*}
By Proposition \ref{Ricci},
$$\exp(\int_0^u(\sum_{i\leq 1}||\nab^i\psi_{\Hb}||_{L^\infty_{\ub}L^4(S_{u',\ub})})du')\leq C(\mathcal O_0,\mathcal R_0).$$
By Proposition \ref{R1},
$$\mathcal R_{\ub}[\betab]^2\leq C(\mathcal O_0,\mathcal R_0).$$
Therefore,
\begin{equation*}
\begin{split}
\mathcal R_u[\beta]^2+\mathcal R_{\ub}[\rhoc,\sigmac]^2
\leq C(\mathcal O_0,\mathcal R_0)(1+\epsilon^{\frac{1}{2}}C(\mathcal O_0,\mathcal R_0,\Delta_4)).
\end{split}
\end{equation*}
Taking $\epsilon$ sufficiently small depending on $\mathcal O_0$, $\mathcal R_0$ and $\Delta_4$, we conclude that
$$\mathcal R_u[\beta]^2+\mathcal R_{\ub}[\rhoc,\sigmac]^2\leq C(\mathcal O_0,\mathcal R_0).$$
\end{proof}
Propositions \ref{R1} and \ref{R2} together imply
\begin{proposition}
There exists $\epsilon_0=(\mathcal O_0,\mathcal R_0)$ such that whenever $\epsilon\leq \epsilon_0$,
$$\mathcal R\leq C(\mathcal O_0,\mathcal R_0).$$
\end{proposition}
\begin{proof}
Let 
$$\Delta_4 \gg C(\mathcal O_0,\mathcal R_0),$$
where $C(\mathcal O_0,\mathcal R_0)$ is taken to be the maximum of the bounds in Propositions \ref{R1}, and \ref{R2}. Hence, the choice of $\Delta_4$ depends only on $\mathcal O_0$ and $\mathcal R_0$. Thus, by Propositions \ref{R1}, and \ref{R2}, the bootstrap assumption (\ref{BA5}) can be improved by choosing $\epsilon$ sufficiently small depending on $\mathcal O_0$ and $\mathcal R_0$.
\end{proof}
This concludes the proof of Theorem \ref{aprioriestimates}.

\section{Nonlinear Interaction of Impulsive Gravitational Waves}

In this section, we return to the special case of the nonlinear interaction of impulsive gravitational waves, thus proving Theorem \ref{cgiwthm}. Recall in that setting we prescribe characteristic initial data such that on $H_0(0,\ub_*)$ (resp. $\Hb_0(0,u_*)$), $\chih$ (resp. $\chibh$) is smooth except on a 2-sphere $S_{0,\ub_s}$ (resp. $S_{u_s,0}$) where it has a jump discontinuity. Thus the curvature in the data has delta singularities supported on $S_{0,\ub_s}$ and $S_{u_s,0}$.

Such an initial data set can be constructed by solving a system of ODEs in a way similar to the construction of the initial data with one gravitational impulsive wave in \cite{LR}. Moreover, one can find a sequence of smooth characteristic data that converges to the data for the colliding impulsive gravitational waves. We refer the readers to \cite{LR} for more details.

With the given initial data, Theorem \ref{rdthmv2} implies that a unique spacetime solution $(\mathcal M, g)$ to the Einstein equations exists in $0\leq u\leq u_*$ and $0\leq \ub\leq \ub_*$. Moreover, using the a priori estimates established in Theorem \ref{aprioriestimates}, we can show that the sequence of smooth data described above gives rise to a sequence of smooth spacetimes that converges to $(\mathcal M, g)$.

In this section, we prove that in addition to the a priori estimates proved in Theorem \ref{aprioriestimates}, the colliding impulsive gravitational spacetime $(\mathcal M,g)$ possesses extra regularity properties as described in parts $(b), (c)$ of Theorem \ref{cgiwthm}. We give an outline of the remainder of the section:

{\bf Section \ref{cigw1}}: We show the first part of Theorem \ref{cgiwthm}$(c)$, i.e., that $\beta,\rho,\sigma,\betab$ can be defined in $L^2_u L^2_{\ub}L^2(S)$. This follows directly from the estimates in the proof of Theorem \ref{rdthmv2}. 

{\bf Section \ref{cigw2}}: We prove the second part of Theorem \ref{cgiwthm}$(c)$, showing that the solution is smooth away from $\Hb_{\ub_s}\cup H_{u_s}$.

{\bf Section \ref{cigw3}}: We establish Theorem \ref{cgiwthm}$(b)$. We define $\alpha$ and $\alphab$ in the colliding impulsive gravitational spacetime and show that they are measures with singular atoms supported on $\Hb_{\ub_s}$ and $H_{u_s}$ respectively. This shows that the singularities indeed propagate along the null hypersurfaces $H_{u_s}$ and $\Hb_{\ub_s}$.

\subsection{Control of the Regular Curvature Components}\label{cigw1}

\begin{proposition}
All the curvature components except $\alpha$ and $\alphab$ are in $L^2_u L^2_{\ub} L^2(S)$.
\end{proposition}
\begin{proof}
It follows directly from the proof of Theorem \ref{rdthmv2} that $\beta,\rhoc,\sigmac,\betab\in L^2_u L^2_{\ub} L^2(S)$. It remains to show that $\rho, \sigma$ are in $L^2_u L^2_{\ub} L^2(S)$. Recalling the definition of $\rhoc$ and $\sigmac$, it suffices to show that $\chih\chibh$ is in $L^2_u L^2_{\ub} L^2(S)$. This follows from
$$||\chih\chibh||_{L^2_u L^2_{\ub} L^2(S)}\leq ||\chih||_{L^\infty_u L^2_{\ub} L^4(S)}||\chibh||_{L^2_u L^\infty_{\ub} L^4(S)}.$$
\end{proof}

\subsection{Smoothness of Spacetime away from the Two Singular Hypersurfaces}\label{cigw2}

In this subsection, we prove that in the case of two colliding impulsive gravitational waves, the spacetime is smooth away from the null hypersurfaces $H_{u_s}$ and $\Hb_{\ub_s}$. For $u<u_s$ or $\ub <\ub_s$, this follows from the result of \cite{LR}. We will therefore only prove the statement for $\{u> u_s\}\cap\{\ub >\ub_s\}$. 

\begin{proposition}\label{smoothness}
The unique solution to the vacuum Einstein equations for the initial data of colliding impulsive gravitational is smooth in $\{u> u_s\}\cap\{\ub >\ub_s\}$
\end{proposition}
\begin{proof}
We establish estimates for all derivatives of all the Ricci coefficients. We prove by induction on $j$, $k$ that $\nab^i\nab^j_3\nab^k_4(\psi,\psi_H,\psi_{\Hb})$ and $\nab^i\nab^j_3\nab^k_4\rho$ are in $L^\infty_u L^\infty_{\ub}L^2(S)$ for all $i,j,k$. This then implies that all the Ricci coefficients and curvature components\footnote{Notice that all curvature components except for $\rho$ can be expressed as a combination of the Ricci coefficients and their first derivatives by virtue of the null structure equations and elliptic equations \eqref{null.str1}, \eqref{null.str2} and \eqref{null.str3}.} are in $C^\infty$.

\noindent{\bf 1. Base case: $j=k=0$}

\noindent{\bf 1(a). Estimates for $\psi$ and $\rhoc$}

For $i\leq 2$, $\nab^i\psi$ is in $L^\infty_uL^\infty_{\ub}L^2(S)$ by Theorem \ref{aprioriestimates}. Using exactly the same arguments but allowing more angular derivatives in the initial data, it is easy to show that $\nab^i\psi$ is in $L^\infty_uL^\infty_{\ub}L^2(S)$ for all $i$. 

Similarly, an adaptation of the arguments in Theorem \ref{aprioriestimates} imply that $\nab^i\rhoc$ are in $L^\infty_uL^\infty_{\ub}L^2(S)$ for all $i$. 

\noindent{\bf 1(b). Estimates for $\psi_H$, $\psi_{\Hb}$ and $\rho$}

The a priori estimates given by Theorem \ref{aprioriestimates} only imply that for $i\leq 2$,
\begin{equation}\label{psiHspace}
\nab^i\psi_H\in L^2_{\ub}L^\infty_uL^2(S)\mbox{ and }\nab^i\psi_{\Hb}\in L^2_{u}L^\infty_{\ub}L^2(S).
\end{equation}
Applying a simple modification to the proof of Theorem \ref{aprioriestimates} with more angular derivatives in the initial data, it is easy to show that (\ref{psiHspace}) holds for all $i\geq 0$. In order to improve this to a bound in $L^\infty_{\ub}L^\infty_uL^2(S)$, we need to use the fact that we are away from the hypersurfaces $H_{u_s}$ and $\Hb_{\ub_s}$.

We first prove estimates for $\chih$. Consider the equation
\begin{equation}\label{chihtransport}
\nab_3\chih+\frac 1 2 \trchb \chih=\nab\widehat{\otimes} \eta+2\omegab \chih-\frac 12 \trch \chibh +\eta\widehat{\otimes} \eta.
\end{equation}
Since we know that the initial data on $H_0\cap\{\ub>\ub_s\}$ are smooth, $\nab^i\chih$ is in $L^\infty_{\ub}L^2(S_{0,\ub})$. Using the control that has already been obtained and Gronwall's inequality, we integrate (\ref{chihtransport}) to show that $\nab^i\chih$ is in $L^\infty_uL^\infty_{\ub}L^2(S)$ for $\ub>\ub_s$ for all $i$. 

Similarly, using 
$$\nab_4\chibh +\frac 1 2 \trch \chibh=\nab\widehat{\otimes} \etab+2\omega \chibh-\frac 12 \trchb \chih +\etab\widehat{\otimes} \etab,$$
we show that $\nab^i\chibh$ is in $L^\infty_uL^\infty_{\ub}L^2(S)$ for $u>u_s$ for all $i$.

The estimates for $\rhoc$ together with the bounds for $\chih$ and $\chibh$ imply that $\nab^i\rho$ is in $L^\infty_uL^\infty_{\ub}L^2(S)$. Now using
$$\nabla_3\omega =2\omega\omegab-\eta\cdot\etab+\f 12|\etab|^2+\frac 12 \rho,$$
and
$$\nabla_4\omegab=2\omega\omegab- \eta\cdot\etab+\f 12|\eta|^2+\frac 12 \rho,$$
we show that $\nab^i\omega$ and $\nab^i\omegab$ are in $L^\infty_uL^\infty_{\ub}L^2(S)$ for all $i$.

\noindent{\bf 2. Induction Step}

We now proceed to the induction step. Assume $\nab^i\nab^j_3\nab^k_4(\psi,\psi_H,\psi_{\Hb})$ and $\nab^i\nab^j_3\nab^k_4\rho$ are in $L^\infty_u L^\infty_{\ub}L^2(S)$ for all $i$, for all $j\leq J$ and for all $k\leq K$ in the region $\{u>u_s\}\cap\{\ub>\ub_s\}$. We will show below that $\nab^i\nab^{J+1}_3\nab^k_4(\psi,\psi_H,\psi_{\Hb})$ and $\nab^i\nab^{J+1}_3\nab^k_4\rho$ are in $L^\infty_u L^\infty_{\ub}L^2(S)$ for all $i$ and for all $k\leq K$ in $\{u>u_s\}\cap\{\ub>\ub_s\}$. A similar argument then shows that $\nab^i\nab^j_3\nab^{K+1}_4(\psi,\psi_H,\psi_{\Hb})$ and $\nab^i\nab^j_3\nab^{K+1}_4\rho$ are in $L^\infty_u L^\infty_{\ub}L^2(S)$ for all $i$ and for all $j\leq J$ in $\{u>u_s\}\cap\{\ub>\ub_s\}$. This completes the induction step and proves the proposition.

We first estimate $\nab^i\nab^{J+1}_3\nab^k_4\rhoc$. Notice that by signature considerations (see Section \ref{sec.signature}), we have the following schematic expression for the commutator $[\nab_3,\nab_4]$:
\begin{equation}\label{34.commute}
[\nab_3,\nab_4]\phi=(\psi,\psi_{\Hb})\nab_4\phi+(\psi,\psi_H)\nab_3\phi+\psi\nab\phi+(\rho,\sigma)\phi+(\psi,\psi_H)(\psi,\psi_{\Hb})\phi.
\end{equation}
Using the Bianchi equation for $\nab_3\rhoc$, commuting $k\leq K$ times with $\nab_4$ and differentiating $J$ times with $\nab_3$ and $i$ times with $\nab$, we obtain
$$\nab^i\nab^{J+1}_3\nab^k_4\rhoc=...,$$
where $...$ on the right hand side denotes terms that have at most $J$ $\nab_3$ derivatives on $\rhoc$ or $(\psi,\psi_H,\psi_{\Hb})$. They are therefore bounded\footnote{Note that the terms on the right hand side may have more than $i$ angular derivatives. Nevertheless, the induction hypothesis allows us to control an arbitrary number of angular derivatives.} in $L^\infty_uL^\infty_{\ub}L^2(S)$ by the induction hypothesis. Hence we obtain
\begin{equation}\label{HO.rhoc}
||\nab^i\nab^{J+1}_3\nab^k_4\rhoc||_{L^\infty_uL^\infty_{\ub}L^2(S)}\leq C_{i,k}
\end{equation}
for every $i$ and every $k\leq K$ in the region $\{u>u_s\}\cap\{\ub>\ub_s\}$.

To proceed, we will consider separately the cases where $\psi$ satisfies a $\nab_3\psi$ equation and where $\psi$ satisfies a $\nab_4\psi$ equation. We introduce a notation such that we denote the $\psi$'s in the first case by $\psi_3$ and those in the second case by $\psi_4$. More precisely, we use the notation
$$\psi_3\in\{\etab,\trch,\trchb\},\quad \psi_4\in\{\eta,\trch,\trchb\}.$$

For $\psi_H$ and $\psi_3$, we commute the equations $k\leq K$ times with $\nab_4$ and then differentiate the equation $J$ times by $\nab_3$ and $i$ times with $\nab$. 
As a consequence, we obtain
$$\nab^i\nab^{J+1}_3\nab^k_4(\psi_3,\psi_H)=...$$
where $...$ on the right hand side represents terms that have at most $J$ $\nab_3$'s. As in the estimates for $\nab^i\nab^{J+1}_3\nab^k_4\rhoc$, these terms can therefore be controlled in $L^\infty_uL^\infty_{\ub}L^2(S)$ by the induction hypothesis.  Thus we can estimate these terms directly to show that 
\begin{equation}\label{HO.3}
||\nab^i\nab^{J+1}_3\nab^k_4(\psi_3,\psi_H)||_{L^\infty_uL^\infty_{\ub}L^2(S)}\leq C_{i,k}
\end{equation}
for all $i$ and all $k\leq K$ in the region $\{u>u_s\}\cap\{\ub>\ub_s\}$.

For $\psi_{\Hb}$ and $\psi_4$, we commute the equations $i$ times with $\nab$, $J+1$ times with $\nab_3$ and $k\leq K$ times with $\nab_4$. Here, we use both the schematic commutation formula for $[\nab_4,\nab^i]$ in Proposition \ref{commute.prop} and also the schematic expression for $[\nab_3,\nab_4]$. Then we have
\begin{equation*}
\begin{split}
&\nab_4(\nab^i\nab^{J+1}_3\nab^k_4(\psi_4,\psi_{\Hb}))\\
=&\sum_{\substack{i_1+i_2+i_3=i\\k_1+k_2+k_3=k}}\nab^{i_1}\nab_4^{k_1}(\psi,\psi_H,\psi_{\Hb})^{i_2+k_2+1}\nab^{i_3}\nab^{J+1}_3\nab^{k_3}_4(\psi,\psi_{\Hb})\\
&+\sum_{\substack{i_1+i_2+i_3=i\\k_1+k_2+k_3=k}}\nab^{i_1}\nab_4^{k_1}(\psi,\psi_H,\psi_{\Hb})^{i_2+k_2}\nab^{i_3}\nab^{J+1}_3\nab^{k_3}_4\rhoc+...
\end{split}
\end{equation*}
where $...$ are again terms that can be bounded in $L^\infty_uL^\infty_{\ub}L^2(S)$ using the induction hypothesis. Notice that the second term on the right hand side can be estimated by \eqref{HO.rhoc}. Moreover, by assumption, the initial data on $\Hb_0$ for $\nab^i\nab^{J+1}_3\nab^k_4(\psi_4,\psi_{\Hb})$ are bounded in $L^\infty_uL^2(S)$ for $u>u_s$ Thus, by Gronwall's inequality, we obtain
\begin{equation}\label{HO.4}
||\nab^i\nab^{J+1}_3\nab^k_4(\psi_4,\psi_{\Hb})||_{L^\infty_u L^\infty_{\ub}L^2(S)}\leq C_{i,k}
\end{equation} 
for all $i$ and all $k\leq K$ in $\{u>u_s\}\cap\{\ub>\ub_s\}$. Finally, combining \eqref{HO.rhoc}, \eqref{HO.3} and \eqref{HO.4}, and using the formula $\rhoc=\rho-\frac 12 \chih\cdot\chibh$, we obtain
\begin{equation}\label{HO.rho}
||\nab^i\nab^{J+1}_3\nab^k_4\rho||_{L^\infty_uL^\infty_{\ub}L^2(S)}\leq C_{i,k}
\end{equation}
for every $i$ and every $k\leq K$ in $\{u>u_s\}\cap\{\ub>\ub_s\}$. \eqref{HO.3}, \eqref{HO.4} and \eqref{HO.rho} together imply that in the region $\{u>u_s\}\cap\{\ub>\ub_s\}$, $\nab^i\nab^{J+1}_3\nab^k_4(\psi,\psi_H,\psi_{\Hb})$ and $\nab^i\nab^{J+1}_3\nab^k_4\rho$ are in $L^\infty_u L^\infty_{\ub}L^2(S)$ for all $i$ and for all $k\leq K$, as desired.
\end{proof}

\subsection{Propagation of Singularities}\label{cigw3}

We first show that $\alpha$ and $\alphab$ can be defined as measures. We take the null structure equations
\begin{equation}\label{alphadef}
\alpha=-\nab_4\chih-\trch\chibh-2\omega\chih,
\end{equation}
\begin{equation*}
\alphab=-\nab_3\chibh-\trchb\chibh-2\omegab\chibh
\end{equation*}
as the definitions of $\alpha$ and $\alphab$. In view of the fact $\chih$ and $\chibh$ are not differentiable, $\alpha$ and $\alphab$ cannot be defined as functions. Nevertheless, we will show that they can be defined as measures. By (\ref{alphadef}), if $\alpha$ is smooth, for each component of $\alpha$ with respect to the coordinate vector fields, we have
$$\int_0^{\ub} \alpha(u,\ub',\vartheta) d\ub'=\int_0^{\ub} (\Omega^{-1}\frac{\partial}{\partial\ub}\chih+\trch\chih+2\omega\chih)(u,\ub',\vartheta) d\ub'.$$
Integrating by parts and using 
$$\frac{\partial}{\partial\ub}\Omega^{-1}=2\omega,$$
we derive
$$\int_0^{\ub} \alpha(u,\ub',\vartheta) d\ub'=(\Omega^{-1}\chih)(u,\ub,\vartheta)-(\Omega^{-1}\chih)(u,\ub=0,\vartheta)+\int_0^{\ub} (\trch\chih)(u,\ub',\vartheta) d\ub'.$$
Returning to the setting of colliding impulsive gravitational wave, for every $\ub\neq\ub_s$, the right hand side is well-defined. For each $u$, $\vartheta\in \mathbb S^2$, we define $\alpha$ as a measure such that 
$$\alpha([0,\ub))=(\Omega^{-1}\chih)(u,\ub,\vartheta)-(\Omega^{-1}\chih)(u,\ub=0,\vartheta)+\int_0^{\ub} (\trch\chih)(u,\ub',\vartheta) d\ub' \quad\mbox{for }\ub\neq\ub_s.$$
By continuity, we have
$$\alpha([0,\ub_s))=\lim_{\ub\to\ub_s^-}(\Omega^{-1}\chih)(u,\ub,\vartheta)-(\Omega^{-1}\chih)(u,\ub=0,\vartheta)+\int_0^{\ub_s} (\trch\chih)(u,\ub',\vartheta) d\ub' .$$
This defines $\alpha$ as a measure. 
Similarly, for each $u$, $\vartheta\in \mathbb S^2$, we define $\alphab$ to be a measure by
\begin{equation*}
\begin{split}
\alphab([0,u))
=(\Omega^{-1}\chibh)(u,\ub,\vartheta)-(\Omega^{-1}\chibh)(u=0,\ub,\vartheta)+\int_0^{u} (\Omega^{-1}b^A\frac{\partial}{\partial\th^A}\chibh+\trchb\chibh)(u',\ub,\vartheta) du',
\end{split}
\end{equation*}
for $u\neq u_s$. By continuity
\begin{equation*}
\begin{split}
&\alphab([0,u_s))\\
=&\lim_{u\to u_s^-}(\Omega^{-1}\chibh)(u,\ub,\vartheta)-(\Omega^{-1}\chibh)(u=0,\ub,\vartheta)+\int_0^{u_s} (\Omega^{-1}b^A\frac{\partial}{\partial\th^A}\chibh+\trchb\chibh)(u',\ub,\vartheta) du'.
\end{split}
\end{equation*}

\begin{remark}
If we take a sequence of smooth initial data converging to the data for nonlinearly interacting impulsive gravitational waves, it can be shown that in the spacetimes $(\mathcal M_n,g_n)$ arising from these data are smooth and $\alpha_n \to \alpha$, $\alphab_n \to\alphab$ weakly, where $\alpha$ and $\alphab$ are as defined above. We refer the readers to \cite{LR} for details in the case of one impulsive gravitational wave.
\end{remark}

\begin{proposition}\label{chihdiscont}
$\chih$ is discontinuous across $\ub=\ub_s$. Similarly, $\chibh$ is discontinuous across $u=u_s$.
\end{proposition}
\begin{proof}
We focus on the proof for $\chih$. The proof for $\chibh$ is similar. Consider the equation.
\begin{equation}\label{chihtranscoord}
\nab_3\chih+\frac 1 2 \trchb \chih-2\omegab \chih=\nab\widehat{\otimes} \eta-\frac 12 \trch \chibh +\eta\widehat{\otimes} \eta.
\end{equation}
For the initial data, $\chih(\tilde{u}_0,\ub,\theta)$ is smooth for $\ub\neq\ub_s$ and has a jump discontinuity for $\ub=\ub_s$. On the other hand, the right hand side is continuous by the bounds that we have obtained. Moreover, the vector field $e_3$, the connection $\nab_3$, as well as the connection coefficients $\trchb$ and $\omegab$ are also continuous. The conclusion thus follows from integrating (\ref{chihtranscoord}).
\end{proof}

Finally, we show that $\alpha$ (resp. $\alphab$) has a delta singularity on the incoming null hypersurface $\Hb_{\ub_s}$ (resp. outgoing hypersurface $H_{u_s}$).
\begin{proposition}
$\alpha$ can be decomposed into
$$\alpha=\delta(\ub_s)\alpha_s+\alpha_r,$$
where $\delta(\ub_s)$ is the scalar delta function supported on the null hypersurface $\Hb_{\ub_s}$, $\alpha_s=\alpha_s(u,\vartheta)\neq 0$ belongs to $L^2_uL^2(S)$ and $\alpha_r$ belongs to $L^\infty_uL^\infty_{\ub}L^2(S)$.

Similarly, $\alphab$ can be decomposed into
$$\alphab=\delta(u_s)\alphab_s+\alphab_r,$$
where $\delta(u_s)$ is the scalar delta function supported on the null hypersurface $H_{u_s}$, $\alphab_s=\alphab_s(\ub,\vartheta)\neq 0$ belongs to $L^2_{\ub}L^2(S)$ and $\alphab_r$ belongs to $L^\infty_uL^\infty_{\ub}L^2(S)$.
\end{proposition}
\begin{proof}
We prove the proposition for $\alpha$. The statement for $\alphab$ can be proved in a similar fashion. Define 
$$\alpha_s(u,\vartheta):=\lim_{\ub\to\ub_s^+} \Omega^{-1}\chih (u,\ub,\vartheta)-\lim_{\ub\to\ub_s^-} \Omega^{-1}\chih (u,\ub,\vartheta),$$
and 
$$\alpha_r:=\alpha-\delta(\ub_s)\alpha_s.$$
We now show that $\alpha_s$ and $\alpha_r$ have the desired property. By Theorem \ref{aprioriestimates}, $\alpha_s$ belongs to $L^2_uL^2(S)$. That $\alpha_s\neq 0$ follows from the fact that $\chih$ has a jump discontinuity across $\ub=\ub_s$, which is proved in Proposition \ref{chihdiscont}.

It remains to show that $\alpha_r$ belongs to $L^\infty_u L^\infty_{\ub} L^2(S)$. To show this, we consider the measure of the half open interval $[0,\ub)$ with respect to the measure $\alpha_r(u,\vartheta)$:
\begin{equation*}
\begin{split}
&(\alpha_r(u,\vartheta))([0,\ub))\\
=&(\Omega^{-1}\chih)(u,\ub,\vartheta)-\lim_{\tilde{\ub}\to\ub_s^+}(\Omega^{-1}\chih)(u,\tilde{\ub},\vartheta)+\lim_{\tilde{\ub}\to\ub_s^-}(\Omega^{-1}\chih)(u,\tilde{\ub},\vartheta)-(\Omega^{-1}\chih)(u,\ub=0,\vartheta)\\
&+\int_0^{\ub} (\trch\chih)(u,\tilde{\ub},\vartheta) d\tilde{\ub}\\
=&\lim_{\tilde{\ub}\to\ub_s^-}\int_0^{\tilde{\ub}} \frac{\partial}{\partial\ub}(\Omega^{-1}\chih)(u,\tilde{\ub}',\vartheta)d\tilde{\ub}'+\lim_{\tilde{\ub}\to\ub_s^+}\int_{\tilde{\ub}}^{\ub} \frac{\partial}{\partial\ub}(\Omega^{-1}\chih)(u,\tilde{\ub}',\vartheta)d\tilde{\ub}'\\
&+\int_0^{\ub} (\trch\chih)(u,\tilde{\ub},\vartheta) d\tilde{\ub}.
\end{split}
\end{equation*}
By Proposition \ref{smoothness}, $\frac{\partial}{\partial\ub}(\Omega^{-1}\chih)(u,\ub,\vartheta)$ is in $L^\infty_u L^\infty_{\ub}L^2(S)$ away from the the hypersurface $\Hb_{\ub_s}$. Thus $(\alpha_r(u,\vartheta))([0,\ub))$ can be expressed as an integral over $[0,\ub)$ whose integrand belongs to $L^\infty_u L^\infty_{\ub}L^2(S)$, as desired.
\end{proof}

\section{Formation of Trapped Surfaces}\label{sectrapped}

We also apply the existence and uniqueness result in Theorem \ref{rdthmv2} to the problem the formation of trapped surfaces. In \cite{Chr}, Christodoulou proved that trapped surfaces can form in evolution. This was later simplified and generalized by Klainerman and Rodnianski \cite{KlRo}, \cite{KlRo1}. These are also the first large data results for the long time dynamics of the Einstein equations without symmetry assumptions.

In all the previous works, the setting is a characteristic initial value problem such that the data on the incoming null hypersurface are that of Minkowski spacetime. The data on the outgoing null hypersurface, termed a ``short pulse'' by Christodoulou, are large, but are only prescribed on a region with a short characteristic length. The large data on the outgoing hypersurface and the small data on the incoming hypersurface together give rise to a hierarchy of large and small quantities, which was shown to be propagated by the evolution equations.

In particular, in order to guarantee the formation of a trapped surface, the initial norm of $\chih$ is large on $H_0$, and is of size
$$||\chih||_{L^\infty_{\ub}L^\infty(S)}\sim \epsilon^{-\frac 12},$$
where $\epsilon$ is the short characteristic length in the $\ub$ direction. Moreover, $\alpha$ has initial norm of size
$$||\alpha||_{L^\infty_{\ub}L^\infty(S)}\sim \epsilon^{-\frac 32}.$$
It was precisely to offset the largeness of $\chih$ and $\alpha$ (and their derivatives) that the data on $\Hb_0$ were required to be small.

However, when viewed in the weaker topology $L^2_{\ub}L^\infty(S)$, the initial size for $\chih$ in \cite{Chr} is bounded by a constant independent of $\epsilon$:
$$||\chih||_{L^2_{\ub}L^\infty(S)}\sim 1.$$
Our main existence result applies for initial data such that $\chih$ and its angular derivatives are only in $L^2_{\ub}L^\infty(S)$ without any requiring any smallness for the data on $\Hb_0$. In particular, no assumptions on $\alpha$ and its derivatives are imposed. Using this theorem, we obtain the following extension to the theorem in \cite{Chr}, \cite{KlRo}:
\begin{theorem}\label{trappedsurface}
Suppose the characteristic initial data are smooth on $\Hb_0$ for $0\leq u\leq u_*$ and satisfy the following two inequalities:
\begin{equation}\label{trappedsurfaceineq.0}
\trchb(u_*,\vartheta)<0
\end{equation}
and
\begin{equation}\label{trappedsurfaceineq}
\begin{split}
&\trch(u=0,\ub=0,\vartheta)\\
&\quad+\int_0^{u_*} \exp(\frac 12\int_0^{u'}\trchb(u'',\ub=0,\vartheta) du'') (-2K+2\div\zeta+2|\zeta|^2)(u',\ub=0,\vartheta)du'\\
<& \exp(-\frac 12\int_0^{u_*} \trchb(u',\ub=0,\vartheta) du')\trch(u=0,\ub=0,\vartheta)
\end{split}
\end{equation}
for every $\vartheta\in\mathbb S^2$. Then there exists an open set of smooth initial data on $H_0$ such that the initial data do not contain a trapped surface while a trapped surface is formed in evolution. 

More precisely, for every constant $C$, there exists $\epsilon>0$ sufficiently small such that if the characteristic initial data on $H_0\cap\{0\leq \ub\leq\epsilon\}$ are smooth and satisfy
\begin{equation}\label{trappedreg}
\sum_{i\leq 5}||\nab^i\chih||_{L^2_{\ub}L^2(S)}\leq C
\end{equation}
and the following two inequalities\footnote{Of course, the condition \eqref{trappedsurfaceineq} is necessary precisely so that \eqref{trappedsurfaceineq1} and \eqref{trappedsurfaceineq2} can be verified simultaneously.} are verified for every $\vartheta\in\mathbb S^2$,
\begin{equation}\label{trappedsurfaceineq1}
\begin{split}
&\int_0^{\epsilon} |\chih|^2(u=0,\ub,\vartheta) d\ub\\
> &\exp(\frac 12\int_0^{u_*}\trchb(u',\ub=0,\vartheta) du')\\
&\quad\times\big(\trch(u=0,\ub=0,\vartheta)\\
&\quad\quad\quad+\int_0^{u_*}\exp(\frac 12\int_0^{u'}\trchb(u'',\ub=0,\vartheta) du'')(-2K+2\div\zeta+2|\zeta|^2)(u',\ub=0,\vartheta)du'\big),
\end{split}
\end{equation}
and
\begin{equation}\label{trappedsurfaceineq2}
\int_0^{\epsilon} |\chih|^2(u=0,\ub,\vartheta) d\ub < \trch(u=0,\ub=0,\vartheta),
\end{equation}
then there exists a unique spacetime $(\mathcal M,g)$ that solves the characteristic initial value problem for the vacuum Einstein equations in the region $0\leq u\leq u_*$, $0\leq \ub\leq\epsilon$. Moreover, $H_0\cap\{0\leq \ub\leq\epsilon\}$ does not contain a trapped surface and $S_{u_*,\epsilon}$ is a trapped surface.
\end{theorem}

\begin{remark}
\eqref{trappedsurfaceineq.0} and \eqref{trappedsurfaceineq} hold in particular on a regular null cone with smooth Ricci coefficients such that 
$$||\trch-\frac 2r, \trchb+\frac 2r, \zeta, \nab\zeta, K||_{L^\infty_u L^\infty(S_{u,0})}\leq C,$$
where $r$ is a positive smooth function, $C^{-1}\leq |\frac{dr}{du}|\leq C$ and $r\to 0$ as $u\to u_0$. We will call $r=0$ the vertex of the cone. It is easy to see that \eqref{trappedsurfaceineq.0} and \eqref{trappedsurfaceineq} hold sufficiently close to the vertex, i.e., when $u_*$ is chosen to be sufficiently close to $u_0$. Notice in particular that we have
$$\trchb(u_*,\vartheta)\to -\infty$$
and
$$\int_0^{u_*}\trchb(u',\ub=0,\vartheta) du'\to -\infty$$
as $u_*\to u_0$.
\end{remark}

In particular, this implies the celebrated theorem of Christodoulou\footnote{The original theorem of Christodoulou in \cite{Chr} constructs a spacetime from past null infinity. Here, we retrieve only the theorem in a finite region. Nevertheless, the infinite problem can be treated as in \cite{Chr} once the finite problem is understood.}:
\begin{corollary}[Christodoulou]\label{Chrthm}
If the characteristic initial data on $\Hb_0$ is that of the truncated backward light cone\footnote{Here, we adapt the notation that $u=t-r-2$, $\ub=t+r$. Therefore, $0\leq u\leq 1$ corresponds to the $t$-range $-2\leq t\leq -1$.} 
$$\{\ub=t+r=0,\mbox{ }0\leq u\leq 1\}$$
in Minkowski space, then for $\epsilon$ sufficiently small, if the data on $H_0$ satisfy (\ref{trappedreg}), (\ref{trappedsurfaceineq1}) and (\ref{trappedsurfaceineq2}), then there exists a unique spacetime $(\mathcal M,g)$ endowed with a double null foliation $u$, $\ub$ and solves the characteristic initial value problem for the vacuum Einstein equations in the region $0\leq u\leq 1$, $0\leq \ub\leq\epsilon$. Moreover, $H_0\cap\{0\leq \ub\leq\epsilon\}$ does not contain a trapped surface and $S_{1,\epsilon}$ is a trapped surface.
\end{corollary}

We now begin the proof of Theorem \ref{trappedsurface}. We need the following series of propositions. First, it is easy to see using the null structure equations and Bianchi equations on $H_0$ that the assumptions for Theorem \ref{rdthmv2} are satisfied.
\begin{proposition}
Given the assumptions for Theorem \ref{trappedsurface}, the initial data satisfy the assumptions of Theorem \ref{rdthmv2}. Therefore, using the conclusion of Theorem \ref{rdthmv2}, there exists a unique spacetime $(\mathcal M,g)$that solves the characteristic initial value problem for the vacuum Einstein equations in the region $0\leq u\leq u_*$, $0\leq \ub\leq\epsilon$. Moreover, all the estimates in Theorem \ref{aprioriestimates} hold.
\end{proposition}
\begin{proof}
Since the initial data on $\Hb_0$ is smooth, there exists $c$ and $C$ such that 
$$c\leq |\det\gamma \restriction_{S_{u,0}} |\leq C,\quad \sum_{i\leq 3}|(\frac{\partial}{\partial\th})^i\gamma \restriction_{S_{u,0}}|\leq C,$$
\begin{equation*}
\begin{split}
& \sum_{i\leq 3} \left(||\nab^i\psi||_{L^\infty_uL^2(S_{u,0})}+||\nabla^i\psi_{\Hb}||_{L^2(\Hb_{0})}\right)\leq C,\\
&\sum_{i\leq 2}\left(||\nab^i\betab||_{L^2(\Hb_0)}+\sum_{\Psi\in\{\rhoc,\sigmac\}}||\nab^i\Psi||_{L^\infty_uL^2(S_{u,0})}\right)\leq C.
\end{split}
\end{equation*}
By (\ref{trappedreg}), $\chih$ satisfies the bounds in the assumptions of Theorem \ref{rdthmv2}. By the null structure equations and the Bianchi equations, for $\epsilon$ sufficiently, all the norms for the initial data on $H_0$ in the assumptions of Theorem \ref{rdthmv2} are controlled by a constant independent of $\epsilon$.
\end{proof}
We now use the a priori estimates derived in Theorem \ref{aprioriestimates} together with  (\ref{trappedsurfaceineq1}) and (\ref{trappedsurfaceineq2}) to show that the initial data do not contain a trapped surface and that a trapped surface is formed dynamically. We first show that there are no trapped surfaces on $H_0$:
\begin{proposition}
There exists $\epsilon$ sufficiently small such that for all $\vartheta$, 
$$\trch(u=0,\ub,\vartheta)>0\mbox{ for all $\ub\in[0,\epsilon]$}.$$
\end{proposition}
\begin{proof}
On $H_0$, since $\Omega=1$, $\trch$ satisfies the equation
$$\frac{\partial}{\partial\ub}\trch=-\frac 12(\trch)^2-|\chih|^2.$$
Integrating the equation for $\trch$, we have
$$\trch(u=0,\ub,\vartheta)=\trch(u=0,\ub=0,\vartheta)-\int_0^{\ub} (\frac 12 (\trch)^2+|\chih|^2)(\ub',\vartheta) d\ub'.$$
Hence
$$\trch(u=0,\ub,\vartheta)\geq \trch(u=0,\ub=0,\vartheta)-\int_0^{\epsilon}|\chih|^2(\ub',\vartheta) d\ub'-C\epsilon.$$
(\ref{trappedsurfaceineq2}) implies that for every $\vartheta$,
$$\trch(u=0,\ub=0,\vartheta)>\int_0^{\epsilon}|\chih|^2(\ub',\vartheta) d\ub'.$$
Therefore, for $\epsilon$ sufficiently small, 
$$\trch(u=0,\ub,\vartheta)>0,$$
for all $\ub\in[0,\epsilon]$
\end{proof}

We now prove that $S_{u_*,\epsilon}$ is a trapped surface. First, we show that $\trchb<0$ everywhere on $S_{u_*,\epsilon}$.
\begin{proposition}\label{trchb.est.trapped}
For $\epsilon$ sufficiently small, we have
$$\trchb(u=u_*,\ub=\epsilon,\vartheta)<0$$
for every $\vartheta$.
\end{proposition}
\begin{proof}
Consider the equation 
$$\nab_4\trchb=-\frac 12\trch\trchb+2\omega\trchb+2\rhoc+2\div\etab+2|\etab|^2.$$
Writing $\nab_4=\Omega^{-1}\frac{\partial}{\partial\ub}$ and integrating, it is easy to see that by the estimates in Theorem \ref{aprioriestimates}, we have
\begin{equation}\label{trchb.est.trapped.1}
|\trchb(u,\ub,\vartheta) du'- \trchb(u,\ub=0,\vartheta)|\leq C\epsilon^{\frac 12}\quad\mbox{for all $u$ for all $\vartheta\in\mathbb S^2$}.
\end{equation}
The conclusion of the proposition thus follows from \eqref{trappedsurfaceineq.0}.
\end{proof}
We then prove in the following sequence of propositions that we moreover have $\trch<0$ everywhere on $S_{u_*,\epsilon}$. As a first step, we solve for $\trch$ on $S_{u,0}$ on the initial hypersurface $\Hb_0$.
\begin{proposition}\label{trchampprop}
On the initial hypersurface $\Hb_0$, $\trch(u,\ub=0,\vartheta)$ is given by
\begin{equation*}
\begin{split}
&\trch(u,\ub=0,\vartheta)\\
=&\exp(-\frac 12\int_0^u\trchb du')\big(\trch(u=0,\ub=0,\vartheta)\\
&\qquad\qquad\qquad+\int_0^u\exp(\frac 12\int_0^{u'}\trchb du'')(-2K+2\div\zeta+2|\zeta|^2)du'\big).
\end{split}
\end{equation*}
\end{proposition}
\begin{proof}
On $\Hb_0$, since $\Omega=1$, we have 
$$\frac{\partial}{\partial u} \trch +\frac{1}{2}\trchb\trch=2\rhoc+2\div\zeta+2|\zeta|^2.$$
Substituting the Gauss equation
$$K=-\rhoc-\frac 14 \trch\trchb,$$
we have
$$\frac{\partial}{\partial u} \trch +\frac{1}{2}\trchb\trch=-2K+2\div\zeta+2|\zeta|^2.$$
The conclusion follows easily.
\end{proof}
We compare $\int_0^{u_*} \trchb(u',\ub,\vartheta) du'$ and $\int_0^{u_*} \trchb(u',\ub=0,\vartheta) du'$ in the following proposition:
\begin{proposition}\label{trchint}
For every $\ub\in[0,\epsilon]$, we have
$$|\int_0^{u_*} \trchb(u',\ub,\vartheta) du'-\int_0^{u_*} \trchb(u',\ub=0,\vartheta) du'|\leq C\epsilon^{\frac 12}. $$
\end{proposition}
\begin{proof}
The proposition follows directly from integrating in $u$ the equation \eqref{trchb.est.trapped.1} in the proof of Proposition \ref{trchb.est.trapped}.
\end{proof}
Using Proposition \ref{trchint}, we compute $\int_0^\epsilon |\chih|_{\gamma}^2 d\ub$ for every $u$ and every $\vartheta\in\mathbb S^2$:
\begin{proposition}\label{chihampprop}
For every $\vartheta$ in $S_{u,\epsilon}$, the integral of $|\chih|_{\gamma}^2$ along the integral curve of $L$ through $(u,\vartheta)$ satisfies
\begin{equation*}
\begin{split}
\int_0^{\epsilon}|\chih|^2_\gamma(u,\ub,\vartheta)d\ub \geq \exp(-\int_0^u \trchb(u',\ub=0,\vartheta) du')\int_0^{\epsilon}|\chih|^2_\gamma(u=0,\ub,\vartheta)d\ub-C\epsilon^{\frac 12}.
\end{split}
\end{equation*}
\end{proposition}
\begin{proof}
Fix $\vartheta$. Consider the null structure equation
$$\nabla_3\chih+\frac{1}{2}\trchb\chih=\nabla\hat{\otimes}\eta+2\omegab\chih-\frac{1}{2}\trch\chibh+\eta\hat{\otimes}\eta.$$
Contracting this two tensor with $\chih$ using the metric, we have
$$\frac{1}{2}\nabla_3|\chih|^2_\gamma+\frac{1}{2}\trchb|\chih|_\gamma^2-2\omegab|\chih|^2_\gamma= \chih(\nabla\hat{\otimes}\eta-\frac{1}{2}\trch\chibh+\eta\hat{\otimes}\eta).$$
In coordinates, we have
$$\frac{1}{2\Omega}(\frac{\partial}{\partial u}+b^A\frac{\partial}{\partial\th^A})|\chih|^2_\gamma+\frac{1}{2}\trchb|\chih|_\gamma^2-2\omegab|\chih|^2_\gamma= \chih(\nabla\hat{\otimes}\eta-\frac{1}{2}\trch\chibh+\eta\hat{\otimes}\eta).$$
Using
$$\omegab=-\frac 12\nab_3(\log\Omega),$$
we get
\begin{equation}\label{chihamp}
\begin{split}
&\Omega^2\exp(-\int_0^u \Omega\trchb du')\frac{\partial}{\partial u}\left(\exp(\int_0^u \Omega\trchb du')\Omega^{-2}|\chih|^2_\gamma\right)\\
= & -b^A\frac{\partial}{\partial\th^A}|\chih|^2_\gamma-b^A\frac{\partial\Omega}{\partial\th^A}+2\Omega\chih\cdot(\nabla\hat{\otimes}\eta-2\Omega\frac{1}{2}\trch\chibh+2\Omega\eta\hat{\otimes}\eta).
\end{split}
\end{equation}
Let
$$F=\Omega^{-2}\exp(\int_0^u \Omega\trchb du')(-b^A\frac{\partial}{\partial\th^A}|\chih|^2_\gamma-b^A\frac{\partial\Omega}{\partial\th^A}+2\Omega\chih\cdot(\nabla\hat{\otimes}\eta-2\Omega\frac{1}{2}\trch\chibh+2\Omega\eta\hat{\otimes}\eta)).$$
By (\ref{chihamp}), we have
\begin{equation*}
\begin{split}
&\exp(\int_0^u \Omega(u',\ub)\trchb(u',\ub) du')\Omega^{-2}(u,\ub)|\chih|^2_\gamma(u,\ub)\\
\geq &|\chih|^2_\gamma(u=0,\ub)-C||F(\ub)||_{L^1_uL^\infty(S)}.
\end{split}
\end{equation*}
Using the equation 
$$\frac{\partial b^A}{\partial \ub}= -4\Omega^2\zeta^A,$$
the estimates for $\Omega$ and $\zeta$ and the fact that $b^A=0$ on $\Hb_0$, we have a uniform upper bound for $b$:
$$\|b\|_{L^\infty_uL^\infty_{\ub}L^\infty(S)}\leq C\epsilon.$$
Thus, together with the estimates derived in Theorem \ref{aprioriestimates}, we have
\begin{equation}\label{trapped.error.bound.1}
||F||_{L^2_{\ub}L^2_u L^\infty(S)}\leq C\epsilon^{\frac 12}.
\end{equation}
On the other hand, the proof of Proposition \ref{Omega} implies that
$$||\Omega-1||_{L^\infty_uL^\infty_{\ub}L^\infty(S)}\leq C\epsilon^{\frac 12}.$$
This, together with Proposition \ref{trchint}, gives
$$|\frac{\Omega^{-2}(u,\ub)\exp(\int_0^u \Omega(u',\ub)\trchb(u',\ub) du')}{\exp(\int_0^u \Omega(u',\ub=0)\trchb(u',\ub=0) du')}-1|\leq C\epsilon^{\frac 12}.$$
Therefore,
\begin{equation*}
\begin{split}
&\exp(\int_0^u \trchb(u',\ub=0) du')|\chih|^2_\gamma(u,\ub)\\
\geq &|\chih|^2_\gamma(u=0,\ub)-C\epsilon^{\frac 12}|\chih|^2_\gamma(u,\ub)-C||F(\ub)||_{L^1_uL^\infty(S)}.
\end{split}
\end{equation*}
Taking the $L^2_{\ub}$ norm, we get
\begin{equation*}
\begin{split}
&\exp(\int_0^u \trchb(u',\ub=0) du')\int_0^{\epsilon}|\chih|^2_\gamma(u,\ub)d\ub\\
\geq &\int_0^{\epsilon}|\chih|^2_\gamma(u=0,\ub)d\ub-C\epsilon^{\frac 12}\int_0^{\epsilon}|\chih|^2_\gamma(u,\ub)du-C||F(\ub)||_{L^2_{\ub}L^1_uL^\infty(S)}\\
\geq &\int_0^{\epsilon}|\chih|^2_\gamma(u=0,\ub)d\ub-C\epsilon^{\frac 12},
\end{split}
\end{equation*}
where in the last step we have used \eqref{trapped.error.bound.1} and the bound for $\|\chih\|_{L^\infty_uL^2_{\ub}L^\infty(S)}$ derived in the proof of Theorem \ref{aprioriestimates}.
\end{proof}
This allows us to conclude the formation of trapped surfaces:
\begin{proposition}
Given the assumptions of Theorem \ref{trappedsurface}, for $\epsilon$ sufficiently small, $\trch<0$ pointwise on $S_{u_*,\epsilon}$. Together with Proposition \ref{trchb.est.trapped}, this implies that $S_{u_*,\epsilon}$ is a trapped surface.
\end{proposition}
\begin{proof}
By Proposition \ref{trchampprop}, we have
\begin{equation}\label{trchampeqn}
\begin{split}
&\trch(u_*,\ub=0,\vartheta)\\
=&\exp(-\frac 12\int_0^{u_*}\trchb du')\big(\trch(u=0,\ub=0,\vartheta)\\
&\qquad\qquad\qquad\qquad+\int_0^{u_*}\exp(\frac 12\int_0^{u'}\trchb du'')(-2K+2\div\zeta+2|\zeta|^2)du'\big).
\end{split}
\end{equation}
By Proposition \ref{chihampprop},
\begin{equation}\label{chihampeqn}
\begin{split}
\int_0^{\epsilon}|\chih|^2_\gamma(u_*,\ub,\vartheta)d\ub \geq \exp(-\int_0^{u_*} \trchb(u',\ub=0,\vartheta) du')\int_0^{\epsilon}|\chih|^2_\gamma(u=0,\ub,\vartheta)d\ub-C\epsilon^{\frac 12}.
\end{split}
\end{equation}
Using the equation
$$\nab_4\trch=-\frac 12(\trch)^2-|\chih|^2-2\omega\trch,$$
which can be written in coordinates as
$$\Omega^{-1}\frac{\partial}{\partial\ub}\trch=-\frac 12(\trch)^2-|\chih|^2-2\omega\trch,$$
we have
$$\trch(u_*,\ub=\epsilon,\vartheta)\leq \trch(u_*,\ub=0,\vartheta)-\int_0^{\epsilon} |\chih|^2(u_*,\ub,\vartheta)d\ub+C\epsilon^{\frac 12}.$$
Therefore, using (\ref{trchampeqn}) and (\ref{chihampeqn}), we have
\begin{equation*}
\begin{split}
&\trch(u_*,\ub=\epsilon,\vartheta)\\
\leq& \exp(-\frac 12\int_0^{u_*}\trchb(u',\ub=0,\vartheta) du')\\
&\quad\times\left(\trch(u=0,\ub=0,\vartheta)+\int_0^{u_*}\exp(\frac 12\int_0^{u'}\trchb du'')(-2K+2\div\zeta+2|\zeta|^2)du'\right)\\
&-\exp(-\int_0^{u_*} \trchb(u',\ub=0,\vartheta) du')\int_0^{\epsilon}|\chih|^2_\gamma(u=0,\ub,\vartheta)d\ub+C\epsilon^{\frac 12}.
\end{split}
\end{equation*}
Since by (\ref{trappedsurfaceineq1}), for all $\vartheta$,
\begin{equation*}
\begin{split}
&\left(\trch(u=0,\ub=0,\vartheta)+\int_0^{u_*}\exp(\frac 12\int_0^{u'}\trchb du'')(-2K+2\div\zeta+2|\zeta|^2)du'\right)
\\<& \exp(-\frac 12\int_0^{u_*}\trchb(u',\ub=0,\vartheta) du')\int_0^{\epsilon}|\chih|^2_\gamma(u=0,\ub,\vartheta)d\ub,
\end{split}
\end{equation*}
$\epsilon$ can be chosen sufficiently small so that 
$$\trch(u_*,\ub=\epsilon,\vartheta)<0\mbox{ for every $\vartheta$}.$$
\end{proof}

\bibliographystyle{hplain}
\bibliography{L2e14}

\end{document}

%% file: KhanPenrose.pdf_t
\begin{picture}(0,0)%
\includegraphics{KhanPenrose.pdf}%
\end{picture}%
\setlength{\unitlength}{2763sp}%
\begingroup\makeatletter\ifx\SetFigFont\undefined%
\gdef\SetFigFont#1#2#3#4#5{%
  \reset@font\fontsize{#1}{#2pt}%
  \fontfamily{#3}\fontseries{#4}\fontshape{#5}%
  \selectfont}%
\fi\endgroup%
\begin{picture}(4824,3021)(2389,-3973)
\put(5926,-3811){\makebox(0,0)[lb]{\smash{{\SetFigFont{8}{9.6}{\rmdefault}{\mddefault}{\updefault}{$\ub=0$}%
}}}}
\put(3226,-3811){\makebox(0,0)[lb]{\smash{{\SetFigFont{8}{9.6}{\rmdefault}{\mddefault}{\updefault}{$u=0$}%
}}}}
\end{picture}%

%% file: propagation.pdf_t
\begin{picture}(0,0)%
\includegraphics{propagation.pdf}%
\end{picture}%
\setlength{\unitlength}{3158sp}%
\begingroup\makeatletter\ifx\SetFigFont\undefined%
\gdef\SetFigFont#1#2#3#4#5{%
  \reset@font\fontsize{#1}{#2pt}%
  \fontfamily{#3}\fontseries{#4}\fontshape{#5}%
  \selectfont}%
\fi\endgroup%
\begin{picture}(4824,1806)(2389,-3355)
\put(6526,-2461){\makebox(0,0)[lb]{\smash{{\SetFigFont{10}{12.0}{\rmdefault}{\mddefault}{\updefault}{$H_0$}%
}}}}
\put(4576,-3286){\makebox(0,0)[lb]{\smash{{\SetFigFont{10}{12.0}{\rmdefault}{\mddefault}{\updefault}{$S_{0,0}$}%
}}}}
\put(3076,-1936){\makebox(0,0)[lb]{\smash{{\SetFigFont{10}{12.0}{\rmdefault}{\mddefault}{\updefault}{$\Hb_{\ub_s}$}%
}}}}
\put(5776,-2386){\makebox(0,0)[lb]{\smash{{\SetFigFont{10}{12.0}{\rmdefault}{\mddefault}{\updefault}{$\underline{H}_0$}%
}}}}
\put(4576,-1861){\makebox(0,0)[lb]{\smash{{\SetFigFont{10}{12.0}{\rmdefault}{\mddefault}{\updefault}{$S_{0,\ub_s}$}%
}}}}
\end{picture}%

%% file: Interaction.pdf_t
\begin{picture}(0,0)%
\includegraphics{Interaction.pdf}%
\end{picture}%
\setlength{\unitlength}{2763sp}%
\begingroup\makeatletter\ifx\SetFigFont\undefined%
\gdef\SetFigFont#1#2#3#4#5{%
  \reset@font\fontsize{#1}{#2pt}%
  \fontfamily{#3}\fontseries{#4}\fontshape{#5}%
  \selectfont}%
\fi\endgroup%
\begin{picture}(3624,1856)(1189,-2173)
\put(3676,-1786){\makebox(0,0)[lb]{\smash{{\SetFigFont{8}{9.6}{\rmdefault}{\mddefault}{\updefault}{$S_{0,\ub_S}$}%
}}}}
\put(1801,-1786){\makebox(0,0)[lb]{\smash{{\SetFigFont{8}{9.6}{\rmdefault}{\mddefault}{\updefault}{$S_{u_s,0}$}%
}}}}
\end{picture}%

%% file: region.pdf_t
\begin{picture}(0,0)%
\includegraphics{region.pdf}%
\end{picture}%
\setlength{\unitlength}{2763sp}%
\begingroup\makeatletter\ifx\SetFigFont\undefined%
\gdef\SetFigFont#1#2#3#4#5{%
  \reset@font\fontsize{#1}{#2pt}%
  \fontfamily{#3}\fontseries{#4}\fontshape{#5}%
  \selectfont}%
\fi\endgroup%
\begin{picture}(3924,1885)(2689,-2834)
\put(4351,-2236){\makebox(0,0)[lb]{\smash{{\SetFigFont{8}{9.6}{\rmdefault}{\mddefault}{\updefault}{$I_1$}%
}}}}
\put(2926,-2461){\makebox(0,0)[lb]{\smash{{\SetFigFont{8}{9.6}{\rmdefault}{\mddefault}{\updefault}{$I_2$}%
}}}}
\put(3751,-2761){\makebox(0,0)[lb]{\smash{{\SetFigFont{8}{9.6}{\rmdefault}{\mddefault}{\updefault}{$\epsilon$}%
}}}}
\put(3376,-2761){\makebox(0,0)[lb]{\smash{{\SetFigFont{8}{9.6}{\rmdefault}{\mddefault}{\updefault}{$\epsilon$}%
}}}}
\put(6001,-2161){\makebox(0,0)[lb]{\smash{{\SetFigFont{8}{9.6}{\rmdefault}{\mddefault}{\updefault}{$u$}%
}}}}
\put(6451,-2161){\makebox(0,0)[lb]{\smash{{\SetFigFont{8}{9.6}{\rmdefault}{\mddefault}{\updefault}{$\ub$}%
}}}}
\end{picture}%

%% file: TrappedSurfaceNew.pdf_t
\begin{picture}(0,0)%
\includegraphics{TrappedSurfaceNew.pdf}%
\end{picture}%
\setlength{\unitlength}{2763sp}%
\begingroup\makeatletter\ifx\SetFigFont\undefined%
\gdef\SetFigFont#1#2#3#4#5{%
  \reset@font\fontsize{#1}{#2pt}%
  \fontfamily{#3}\fontseries{#4}\fontshape{#5}%
  \selectfont}%
\fi\endgroup%
\begin{picture}(2799,2481)(1189,-2830)
\put(2176,-2011){\makebox(0,0)[lb]{\smash{{\SetFigFont{8}{9.6}{\rmdefault}{\mddefault}{\updefault}{$\Hb_0$}%
}}}}
\put(3826,-2761){\makebox(0,0)[lb]{\smash{{\SetFigFont{8}{9.6}{\rmdefault}{\mddefault}{\updefault}{$H_0$}%
}}}}
\put(2326,-661){\makebox(0,0)[lb]{\smash{{\SetFigFont{8}{9.6}{\rmdefault}{\mddefault}{\updefault}{Trapped surface=$\{\ub=\epsilon,\,u=u_*\}$}%
}}}}
\end{picture}%

%% file: frame.pdf_t
\begin{picture}(0,0)%
\includegraphics{frame.pdf}%
\end{picture}%
\setlength{\unitlength}{3158sp}%
\begingroup\makeatletter\ifx\SetFigFont\undefined%
\gdef\SetFigFont#1#2#3#4#5{%
  \reset@font\fontsize{#1}{#2pt}%
  \fontfamily{#3}\fontseries{#4}\fontshape{#5}%
  \selectfont}%
\fi\endgroup%
\begin{picture}(4824,2253)(2389,-3355)
\put(6526,-2461){\makebox(0,0)[lb]{\smash{{\SetFigFont{10}{12.0}{\rmdefault}{\mddefault}{\updefault}{$H_0$}%
}}}}
\put(5851,-2461){\makebox(0,0)[lb]{\smash{{\SetFigFont{10}{12.0}{\rmdefault}{\mddefault}{\updefault}{$\underline{H}_0$}%
}}}}
\put(4576,-3286){\makebox(0,0)[lb]{\smash{{\SetFigFont{10}{12.0}{\rmdefault}{\mddefault}{\updefault}{$S_{0,0}$}%
}}}}
\put(5701,-1261){\makebox(0,0)[lb]{\smash{{\SetFigFont{10}{12.0}{\rmdefault}{\mddefault}{\updefault}{$e_3$}%
}}}}
\put(6226,-1261){\makebox(0,0)[lb]{\smash{{\SetFigFont{10}{12.0}{\rmdefault}{\mddefault}{\updefault}{$e_4$}%
}}}}
\end{picture}%

%% file: L2e16.170418.nocolor.bbl
\begin{thebibliography}{10}

\bibitem{AKNS}
M.~J. Ablowitz, D.~J. Kaup, A.~C. Newell, and H.~Segur.
\newblock Method for solving the sine-gordon equation.
\newblock {\em Phy. Rev. Lett.}, 30:1262--1264, 1973.

\bibitem{Alinhac1}
S.~Alinhac.
\newblock Interaction d'ondes simples pour des \'equations compl\`etement
  non-lin\'eaires.
\newblock In {\em S\'eminaire sur les \'equations aux d\'eriv\'ees partielles,
  1985--1986}, pages Exp.\ No.\ VIII, 11. \'Ecole Polytech., Palaiseau, 1986.

\bibitem{BaHo}
C.~Barrabes and P.~A. Hogen.
\newblock {\em Singular null hypersurfaces in general relativity}.
\newblock World Scientific, London, 2003.

\bibitem{Beals}
M.~Beals.
\newblock {\em Propagation and interaction of singularities in nonlinear
  hyperbolic problems}.
\newblock Progress in Nonlinear Differential Equations and their Applications,
  3. Birkh\"auser Boston Inc., Boston, MA, 1989.

\bibitem{Bicak}
J.~Bicak.
\newblock Selected solutions of einstein's field equations: their role in
  general relativity and astrophysics.
\newblock {\em Lect. Notes Phys.}, 540:1--126, 2000.

\bibitem{Brinkmann}
M.~W. Brinkmann.
\newblock {On Riemann spaces conformal to Euclidean space}.
\newblock {\em Proc. Natl. Acad. Sci. U.S.A.}, 9:1--3, 1923.

\bibitem{Chr}
D.~Christodoulou.
\newblock {\em The formation of black holes in general relativity}.
\newblock EMS Monographs in Mathematics. European Mathematical Society (EMS),
  Z\"urich, 2009, arXiv:0805.3880.

\bibitem{CK}
D.~Christodoulou and S.~Klainerman.
\newblock {\em The global nonlinear stability of the {M}inkowski space},
  volume~41 of {\em Princeton Mathematical Series}.
\newblock Princeton University Press, Princeton, NJ, 1993.

\bibitem{EinsteinRosen}
A.~Einstein and N.~Rosen.
\newblock On gravitational waves.
\newblock {\em J. Franklin Inst.}, 223:43--54, 1937.

\bibitem{GGKM}
C.~S. Gardner, J.~M. Greene, M.~D. Kruskal, and R.~M. Miura.
\newblock Korteweg-de{V}ries equation and generalization. {VI}. {M}ethods for
  exact solution.
\newblock {\em Comm. Pure Appl. Math.}, 27:97--133, 1974.

\bibitem{Gr}
J.~B. Griffiths.
\newblock {\em Colliding plane waves in general relativity}.
\newblock Oxford UP, Oxford, 1991.

\bibitem{GrPo}
J.~B. Griffiths and J.~Podolsky.
\newblock {\em Exact space-times in {E}instein's general relativity}.
\newblock Cambridge UP, Cambridge, 2009.

\bibitem{HE1}
I.~Hauser and F.~J. Ernst.
\newblock {Initial value problem for colliding gravitational plane waves I}.
\newblock {\em J. Math. Phy.}, 30:872--887, 1989.

\bibitem{HE2}
I.~Hauser and F.~J. Ernst.
\newblock {Initial value problem for colliding gravitational plane waves II}.
\newblock {\em J. Math. Phy.}, 30:2322--2336, 1989.

\bibitem{HE3}
I.~Hauser and F.~J. Ernst.
\newblock {Initial value problem for colliding gravitational plane waves III}.
\newblock {\em J. Math. Phy.}, 31:871--881, 1990.

\bibitem{KhanPenrose}
K.~A. Khan and R.~Penrose.
\newblock {Scattering of two impulsive gravitational plane waves}.
\newblock {\em Nature}, 229:185--186, 1971.

\bibitem{KN}
S.~Klainerman and F.~Nicol{\`o}.
\newblock {\em The evolution problem in general relativity}, volume~25 of {\em
  Progress in Mathematical Physics}.
\newblock Birkh\"auser Boston Inc., Boston, MA, 2003.

\bibitem{KlRo1}
S.~Klainerman and I.~Rodnianski.
\newblock On emerging scarred surfaces for the {E}instein vacuum equations.
\newblock {\em Discrete Contin. Dyn. Syst.}, 28(3):1007--1031, 2010.

\bibitem{KlRo}
S.~Klainerman and I.~Rodnianski.
\newblock On the formation of trapped surfaces.
\newblock {\em Acta Math.}, 208(2):211--333, 2012, arXiv:0912.5097.

\bibitem{L21}
S.~Klainerman, I.~Rodnianski, and J.~Szeftel.
\newblock The bounded {$L^2$} curvature conjecture.
\newblock {\em Invent. Math.}, 202(1):91--216, 2015, arXiv:1204.1767.

\bibitem{LeSm}
P.~G. LeFloch and J.~Smulevici.
\newblock Global geometry of {$T^2$}-symmetric spacetimes with weak regularity.
\newblock {\em C. R. Math. Acad. Sci. Paris}, 348(21-22):1231--1233, 2010.

\bibitem{LeSte2}
P.~G. LeFloch and J.~M. Stewart.
\newblock {The characteristic initial value problem for plane symmetric
  spacetimes with weak regularity}.
\newblock {\em Class. Quant. Grav.}, 28:145019, 2011, arXiv:1004.2343.

\bibitem{L}
J.~Luk.
\newblock On the local existence for the characteristic initial value problem
  in general relativity.
\newblock {\em Int. Math. Res. Not. IMRN}, (20):4625--4678, 2012,
  arXiv:1107.0898.

\bibitem{LR}
J.~Luk and I.~Rodnianski.
\newblock Local propagation of impulsive gravitational waves.
\newblock {\em Comm. Pure Appl. Math.}, 68(4):511--624, 2015, arXiv:1209.1130.

\bibitem{MM2}
Y.~Martel and F.~Merle.
\newblock Description of two soliton collision for the quartic g{K}d{V}
  equation.
\newblock {\em Ann. of Math. (2)}, 174(2):757--857, 2011.

\bibitem{MM1}
Y.~Martel and F.~Merle.
\newblock Inelastic interaction of nearly equal solitons for the quartic
  g{K}d{V} equation.
\newblock {\em Invent. Math.}, 183(3):563--648, 2011.

\bibitem{Metivier}
G.~M{\'e}tivier.
\newblock The {C}auchy problem for semilinear hyperbolic systems with
  discontinuous data.
\newblock {\em Duke Math. J.}, 53(4):983--1011, 1986.

\bibitem{NutkuHalil}
Y.~Nutku and M.~Halil.
\newblock Colliding impulsive gravitational waves.
\newblock {\em Phys. Rev. Lett.}, 39:1379--1382, 1977.

\bibitem{Penrose65}
R.~Penrose.
\newblock A remarkable preperty of plane waves in general relativity.
\newblock {\em Rev. Mod. Phys.}, 37:215--220, 1965.

\bibitem{Penrose72}
R.~Penrose.
\newblock The geometry of impulsive gravitational waves.
\newblock In {\em General relativity (papers in honour of {J}. {L}. {S}ynge)},
  pages 101--115. Clarendon Press, Oxford, 1972.

\bibitem{RauchReed}
J.~Rauch and M.~C. Reed.
\newblock Propagation of singularities for semilinear hyperbolic equations in
  one space variable.
\newblock {\em Ann. of Math. (2)}, 111, 1980.

\bibitem{Stuart1}
D.~Stuart.
\newblock Dynamics of abelian {H}iggs vortices in the near {B}ogomolny regime.
\newblock {\em Comm. Math. Phys.}, 159(1):51--91, 1994.

\bibitem{Stuart2}
D.~Stuart.
\newblock The geodesic approximation for the {Y}ang-{M}ills-{H}iggs equations.
\newblock {\em Comm. Math. Phys.}, 166(1):149--190, 1994.

\bibitem{L22}
J.~Szeftel.
\newblock {Parametrix for wave equations on a rough background I: regularity of
  the phase at initial time}.
\newblock 2012, arXiv:1204.1768.

\bibitem{L23}
J.~Szeftel.
\newblock {Parametrix for wave equations on a rough background II: construction
  and control at initial time}.
\newblock 2012, arXiv:1204.1769.

\bibitem{L24}
J.~Szeftel.
\newblock {Parametrix for wave equations on a rough background III: space-time
  regularity of the phase}.
\newblock 2012, arXiv:1204.1770.

\bibitem{L25}
J.~Szeftel.
\newblock {Parametrix for wave equations on a rough background IV: control of
  the error term}.
\newblock 2012, arXiv:1204.1771.

\bibitem{Szekeres1}
P.~Szekeres.
\newblock {Colliding gravitational waves}.
\newblock {\em Nature}, 228:1183--1184, 1970.

\bibitem{Szekeres2}
P.~Szekeres.
\newblock {Colliding plane gravitational waves}.
\newblock {\em J. Math. Phy.}, 13:286--294, 1972.

\bibitem{Tipler}
F.~J. Tipler.
\newblock Singularities from colliding gravitational waves.
\newblock {\em Phys. Rev. D}, 22:2929--2932, 1980.

\bibitem{Yurtsever}
U.~Yurtsever.
\newblock Colliding almost-plane gravitational waves: Colliding exact plane
  waves and general properties of almost-plane-wave spacetimes.
\newblock {\em Phys. Rev. D}, 33:2803--2817, 1988.

\bibitem{Yurtsever88}
U.~Yurtsever.
\newblock {Structure of the singularities produced by colliding plane waves}.
\newblock {\em Phys. Rev. D}, 37:1706--1730, 1988.

\bibitem{ZS}
V.~E. Zakharov and A.~B. Shabat.
\newblock Exact theory of two-dimensional self-focusing and one-dimensional
  self-modulation of waves in nonlinear media.
\newblock {\em Soviet Physics - JETP}, 34:62--69, 1972.

\end{thebibliography}
